\documentclass{article}

\usepackage{arxiv}
\usepackage[utf8]{inputenc} 
\usepackage[T1]{fontenc}    
\usepackage[colorlinks=true, citecolor=blue]{hyperref}       
\usepackage{url}            
\usepackage{booktabs}       
\usepackage{amsfonts}       
\usepackage{nicefrac}       
\usepackage{microtype}      
\usepackage{graphicx}
\usepackage{natbib}
\usepackage{doi}
\usepackage{amsmath,amsthm,amssymb,bm}
\usepackage{bm}
\usepackage{xcolor}
\usepackage{algpseudocode}
\usepackage{algorithm}
\usepackage{caption}
\usepackage{subcaption}
\usepackage{dsfont}
\usepackage{makecell}
\newcommand{\R}{\mathbb{R}}

\newcommand{\N}{\mathbb{N}}

\newcommand{\E}{\mathbb{E}}
\newcommand{\NN}{\mathcal{N}}
\newcommand{\MN}{\mathcal{MN}}
\newcommand{\msp}{\mathcal{MSP}}
\let\sp\relax
\newcommand{\sp}{\mathcal{SP}}
\newcommand{\mgp}{\mathcal{MGP}}

\newcommand{\T}{\mathcal{T}}
\newcommand{\norm}[1]{\left\lVert#1\right\rVert}
\newcommand{\abs}[1]{\left|#1\right|}
\newcommand{\xhat}[1]{\hat{\bm{x}}^{#1}}
\newcommand{\xhatj}[2]{\hat{x}_{#2}^{#1}}
\newcommand{\Xhat}[1]{\hat{\bm{X}}^{#1}}
\newcommand{\Xhatj}[2]{\hat{X}_{#2}^{#1}}
\newcommand{\innerproduct}[2]{\langle #1, #2 \rangle}
\newcommand{\notprop}{\propto\kern-1\@ptsize pt \diagup}

\DeclareMathOperator\md{MD}
\DeclareMathOperator\mmd{MMD}

\DeclareMathOperator\diag{diag}
\DeclareMathOperator\tr{tr}
\DeclareMathOperator\vect{vec}
\DeclareMathOperator\cov{Cov}

\DeclareMathOperator\row{row}
\DeclareMathOperator\col{col}
\renewcommand{\ker}{\mathrm{ker}}

\DeclareMathOperator\pds{PDS}
\DeclareMathOperator\M{M}

\DeclareMathOperator*{\argmin}{arg\,min}
\DeclareMathOperator*{\mb}{m}

\DeclareMathOperator\fmmd{fMMD}
\DeclareMathOperator\fmd{fMD}
\DeclareMathOperator\K{\bm{\mathrm{K}}} 
\DeclareMathOperator\C{\bm{\mathcal{K}}} 
\DeclareMathOperator\ueig{\xi} 
\DeclareMathOperator\meig{\bm{\psi}} 
\DeclareMathOperator\mbasis{\bm{\phi}} 
\DeclareMathOperator\unibasis{{\phi}} 

\newtheorem{definition}{Definition}[subsection]

\newtheorem{remark}{Remark}[subsection]

\theoremstyle{plain}
\newtheorem{theorem}{Theorem}[subsection]
\newtheorem{proposition}{Proposition}[subsection]
\newtheorem{lemma}{Lemma}[subsection]
\newtheorem{corollary}{Corollary}[subsection]

\title{Explainable Outlier Detection for Multivariate Functional Data}

\author{
	\href{https://orcid.org/0000-0002-3430-8308}{\includegraphics[scale=0.06]{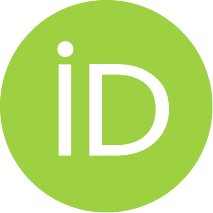}
		\hspace{1mm}Marcus Mayrhofer}\\
	TU Wien\\
	\And
	\href{https://orcid.org/0000-0003-0329-0595}{\includegraphics[scale=0.06]{orcid.pdf}
		\hspace{1mm}Una Radojičić} \\
	TU Wien\\
	\And
	\href{https://orcid.org/0000-0002-5752-429X}{\includegraphics[scale=0.06]{orcid.pdf}
		\hspace{1mm}Horst Lewitschnig} \\
	Infineon Technologies Austria AG\\
	\And
	\href{https://orcid.org/0000-0002-8014-4682}{\includegraphics[scale=0.06]{orcid.pdf}
		\hspace{1mm}Peter Filzmoser} \\
	TU Wien     
}

\hypersetup{
	pdftitle={Explainable Outlier Detection for Multivariate Functional Data},
	pdfsubject={stat.ME},
	pdfauthor={Marcus~Mayrhofer, Una~Radojičić, Horst~Lewitschnig, Peter~Filzmoser},
	pdfkeywords={Matrix-variate distributions, Covariance with Kronecker structure, Explainable artificial intelligence, Image data, Minimum covariance determinant, Shapley values},
}

\usepackage{mathtools}

\usepackage{setspace}
\usepackage{xurl}           
\def\spacingset#1{\renewcommand{\baselinestretch}{#1}\small\normalsize}

\begin{document}
	\maketitle
	
	\begin{abstract}
		This work addresses the challenges of robust covariance estimation and interpretable outlier detection for multivariate functional data with separable covariance structure. We develop a method that simultaneously improves robustness and interpretability in this context by establishing a connection between stochastic processes with separable covariance structures and the corresponding matrix-variate distribution of their basis representations. Leveraging this connection, we employ the recently developed matrix-variate counterpart of the Minimum Covariance Determinant estimator (MMCD) in conjunction with a truncated multivariate functional Mahalanobis semi-distance to robustly estimate mean and covariance for multivariate functional data. For interpretable outlier detection, we generalize multivariate outlier explanations based on Shapley values to decompose overall multivariate functional outlyingness into time-coordinate-specific contributions. Importantly, we reduce the otherwise exponential computational complexity (relative to the number of components) to linear complexity, while retaining the key properties of the Shapley value. This integrated framework combines robust Mahalanobis distances, MMCD estimators, and Shapley value-based outlyingness decomposition to provide a robust and interpretable approach for analyzing multivariate functional data with separable covariance structures. The effectiveness of this approach is demonstrated through both theoretical analysis and practical applications, including simulations and real-world examples.
	\end{abstract}

	\keywords{Robust statistics \and Shapley value \and Multivariate analysis \and Separable covariance}
	
	
\section{Introduction}\label{section:introduction
}

Functional Data Analysis (FDA) encompasses statistical models and methods to analyze data that are naturally represented as functions. With the advancement of modern data collection tools, multivariate functional observations are increasingly common since data are now often recorded repeatedly across multiple time points. These functional observations can be seen as finite-dimensional realizations of continuous stochastic processes, providing a framework for modeling and analyzing such data effectively \citep{CUEVAS2014, Wang2016}.

Unlike traditional approaches that treat data as vectors or matrices, models with a functional structure inherently account for key characteristics of the underlying random process generating the observations, such as smoothness \citep{ramsay_silverman_2005, ferraty2006nonparametric}. In this context, the estimation of mean and covariance plays a central role in understanding the underlying structure and variability. However, rather than direct covariance estimation, Functional Principal Component Analysis (FPCA) is often the dominant focus in the literature. \cite{shang2014survey} provides a comprehensive review of methods for univariate FPCA, \cite{chiou2014multivariate} provide extensions to multivariate functional data, and \citet{happgreven2018} establish and discuss the connection between univariate and multivariate FPCA.

Multivariate functional data routinely appear in climate monitoring, medical data analysis, signal processing, and financial modeling. In these settings, robust methods are very important; they identify the dominant patterns while flagging anomalies, which is what makes the resulting conclusions trustworthy. In particular, the presence of outliers can severely distort estimates of mean and covariance, as well as the performance of non-robust FPCA. Thus, robust approaches and tools for outlier detection are needed. Several methods for robust FPCA have been developed in the univariate functional setting; see, e.g., \cite{boente2015s,boente2021robust} for an overview. While there are methods for distance-based outlier detection in the univariate functional setting, see, e.g., \cite{galeano2015mahalanobis, ghiglietti2017statistical, berrendero2020mahalanobis, oguamalam2024minimum},  non-parametric, depth-based approaches are more commonly employed in the multivariate case. \cite{hubert2015multivariate} provide an overview of outlier detection in multivariate functional data. More recent developments in this area include methods such as the Functional Isolation Forest \citep{Staerman2019} and the model-based contaminated-mixture model approach~\citep{AmovinAssagba2022}.

Given that much of the focus in multivariate FDA is on FPCA and depth-based outlier detection, there is a clear need for robust mean and covariance estimation methods. Our contribution is a comprehensive framework for robust covariance estimation and explainable, distance-based outlier detection in smooth multivariate functional data, with a particular focus on processes exhibiting a separable covariance structure. Robustness is achieved through the use of Matrix Minimum Covariance Determinant (MMCD) estimators \citep{mayrhofer2024robust} for reliable covariance estimation, combined with a generalization of the trimmed functional Mahalanobis distance of \cite{galeano2015mahalanobis} to the multivariate functional setting for outlier identification. Beyond detection, a central contribution of our work lies in the explainability of outliers: by leveraging Shapley values \citep{shapley1953,mayrhofer2022}, we provide a principled decomposition of the trimmed multivariate functional Mahalanobis distance into additive, time- and component-specific contributions. This enables not only the identification of outlying curves but also a detailed understanding of which components and at what times drive their anomalous behavior.%

The structure of the paper is as follows: Section~\ref{section:Preliminaries} outlines the theoretical framework and the already existing notion of univariate functional Mahalanobis distance. In Section~\ref{section:fMMD}, we rigorously define the trimmed multivariate functional Mahalanobis distance, previously applied in \cite{martino2019k}, and derive several of its theoretical properties under separable covariance structures.
Section~\ref{section:parameter_estimation} focuses on processes expressed on a finite basis, and contains the first main theoretical contribution of this work, showing how the separability of the multivariate covariance operator of a random function translates onto the distribution of the random coefficient matrix of its smoothed counterpart. This leads to the connection between the trimmed multivariate functional Mahalanobis distance and the matrix-variate Mahalanobis distance. This connection is key to constructing an efficient algorithm, and a detailed outline of the method is provided. %
Section~\ref{section:shapley} contains the second major contribution: a Shapley value-based framework for explainable outlier detection in functional data. The derived framework formalizes how outliers are explained and enables the identification of the specific components and time points that contribute most to an observation’s anomalous behavior. Computational details and implementation steps are also provided.
Section~\ref{section:simulations} evaluates the performance of our method for outlier detection and interpretation via Shapley values, as well as covariance estimation through extensive simulations and comparisons with state-of-the-art approaches. The simulation study includes cases consistent with the model’s assumptions and settings where these are violated, achieved by introducing non-separable covariance structures. %
Section~\ref{section:examples} illustrates the practical use of the proposed robust framework and Shapley value–based outlier explanations through real-world applications, including separability testing and comparisons with methods that do not make this assumption.
Finally, Section~\ref{section:discussion} summarizes the main findings and concludes. %
Proofs of all the results, as well as additional simulation results, are given in the Appendix.

\section{Preliminaries}\label{section:Preliminaries}
\paragraph{Multivariate Stochastic Processes:}
Let $(\Omega, \mathcal{A}, P)$ be a probability space and $\T \subset \R$ a compact interval, commonly thought of as time, then 
$$\bm{X} = \{\bm{X}(t,\omega), t \in \T\} = \{(X_1(t,\omega),\dots,X_p(t,\omega))', t \in \T\} : \T \times \Omega \to \R^p$$
is a time-continuous vector-valued stochastic process. Hence, $\bm{X}$ is a collection of random variables defined on a common probability space index by the continuous set $\T$; for every fixed $t \in \T$ the process defines a random variable $\bm{X}(t)$, and for every fixed $\omega \in \Omega$ a sample path or trajectory $\bm{X}(.,\omega)$, i.e., a function of $t \in \T$. 

In functional data analysis (FDA) it is commonly assumed that those realizations are elements of a Hilbert space such as $\mathcal{H}:= L_p^2(\T) = L^2(\T) \times \cdots \times L^2(\T)$, the space of $p$-dimensional square-integrable functions \citep{jacques2014model, Wang2016}. The inner product and norm of %
$\bm{x} = (x_1,\dots,x_p)'$ and $\bm{y} = (y_1,\dots,y_p)'$ in $\mathcal{H}$ are given by
\begin{align} \label{eq:inner_product_L2p}
    \langle \bm{x}, \bm{y} \rangle_\mathcal{H} = \sum_{j = 1}^p \innerproduct{x_j}{y_j} = \sum_{j = 1}^p \int_{\T} x_j(t) y_j(t) \mathrm{d}t \quad \text{and} \quad \| \bm{x} \|_{\mathcal{H}} = \langle \bm{x}, \bm{x} \rangle_\mathcal{H}^{1/2},
\end{align} 
respectively. The Euclidean scalar product of $\bm{u}, \bm{v} \in \R^p$ is denoted as $\innerproduct{\bm{u}}{\bm{v}}_2$ and the Euclidean norm as $\norm{\bm{u}}_2$. Subscripts will be omitted when context permits.
A stochastic process is called an $L^2$ process if and only if it has finite second moments; $\E[\norm{\bm{X}(t)}^2] < \infty$ for all $t \in \T$. 
In the following, we consider $L^2$-continuous processes, i.e., $L^2$ processes for which  $\displaystyle\lim_{h \to 0} \E[\norm{\bm{X}(t+h) - \bm{X}(t)}^2]=0$
for every $t \in \T$, see, e.g., \cite{ash2014topics} for more details. 

In this setting, each component $X_j$ of $\bm{X}$ is an $L^2$-continuous stochastic process for all $j = 1,\dots,p$ with continuous mean and covariance function given by
\begin{align} \label{eq:mean_cov_function}
    \bm{\mu}(t) = \E[\bm{X}(t)] = 
        (\E[X_1(t)],\dots,\E[X_p(t)])'\in\mathbb{R}^p,\quad 
    \K(s,t) =  [\kappa_{ij}(s,t)]\in\mathbb{R}^{p\times p},
\end{align}
respectively. Here, $\kappa_{ij}$, $i,j=1\dots,p$, are (cross) covariance functions (kernels)
given by
\begin{align*} %
    \kappa_{ij}(s,t) = \cov(X_{i}(s),X_{j}(t)) = \E\left[(X_{i}(s) - \mu_{i}(s))(X_{j}(t) - \mu_{j}(t))\right].
\end{align*}
 The covariance operator  $\C: \mathcal{H} \rightarrow \mathcal{H}$ of $\bm{X}$ associated with kernel $\K(s,t)$ is defined as 
$$\C\bm{x}(s) = \int_{\T} \K(s,t)\bm{x}(t) \mathrm{d}t, \quad \bm{x} \in \mathcal{H}.$$
Since $\bm{X}$ is $L^2$-continuous, the covariance operator $\C$ is a Hilbert-Schmidt operator, and the multivariate Mercer's theorem \citep{withers1974mercer, daw2022overview} implies that there exist countable sequences of continuous orthonormal eigenfunctions $\{\meig_k\}_{k \geq 1}$ and non-negative decreasing eigenvalues $\{\pi_k\}_{k \geq 1}$ with $\sum_{k = 1}^\infty \pi_k < \infty$ such that 
\begin{align*}%
    \C\meig_k = \pi_k\meig_k \quad \text{and} \quad \K(s,t) &= \sum_{k = 1}^\infty \pi_k \meig_k(s) \meig_k'(t).
\end{align*}
The multivariate Karhunen-Loève representation theorem then implies that there exists a unique sequence of uncorrelated random variables $\{\beta_k\}_{k \geq 1}$ such that
\begin{align*}%
    \bm{X}(t) = \bm{\mu}(t) + \sum_{k = 1}^{\infty} \beta_k \meig_k \quad \text{with} \quad \beta_k =\langle \bm{X} - \bm{\mu},\meig_k \rangle = \int_{\T} \meig_k'(t)(\bm{X}(t) - \bm{\mu}(t)) \mathrm{d}t,
\end{align*}
where $\beta_k \sim \NN(0,\pi_k)$ if $\bm{X}$ is a multivariate Gaussian process \citep{daw2022overview}. 

In continuation, an $L^2$-continuous multivariate stochastic process $\bm{X}$ with mean function $\bm{\mu}$ and covariance function $\K$ is denoted as $\bm{X} \sim \msp(\bm{\mu}, \K)$.

\paragraph{Notion of Mahalanobis Distance:}

Before we discuss the functional setting, let us review the concept of the Mahalanobis distance for random vectors. For a $p$-variate random vector $\bm{x}$ from a population with mean $\bm{\mu} \in \R^p$ and covariance matrix $\bm{\Sigma} \in \pds(p)$, the squared Mahalanobis distance of a random vector $\bm{x}$ (from mean $\bm{\mu}$, with respect to covariance $\bm{\Sigma}$) is given by 
\begin{align}
    \label{eq:md}
    \md^2(\bm{x},\bm{\mu},\bm{\Sigma}) = \md^2(\bm{x}) = (\bm{x} - \bm{\mu})' \bm{\Sigma}^{-1} (\bm{x} - \bm{\mu}).
\end{align}
Here, $ \pds(p)$ denotes the set of all $(p\times p)$ positive definite symmetric matrices. Let $\bm{V}\bm{D}\bm{V}' = \bm{\Sigma}$ denote the spectral decomposition of $\bm{\Sigma}$, where $\bm{D} = \diag(\lambda_1,\dots,\lambda_p)$ is a diagonal matrix containing the ordered eigenvalues $\lambda_1 \geq \dots \geq \lambda_p$ of $\bm{\Sigma}$, and the matrix $\bm{V} \in \R^{p \times p}$ contains the corresponding eigenvectors. Based on the spectral decomposition, we can rewrite $\md^2(\bm{x})$ in terms of the principal components $\bm{z}=(z_1,\ldots ,z_p)' = \bm{V}'(\bm{x}-\bm{\mu})$ as {$ \md^2(\bm{x}) = (\bm{x} - \bm{\mu})' \bm{\Sigma}^{-1} (\bm{x} - \bm{\mu}) = (\bm{x} - \bm{\mu})' \bm{V}\bm{D}^{-1}\bm{V}' (\bm{x} - \bm{\mu}) = \sum_{j = 1}^p \lambda_j^{-1}z_j^2.$}

In the context of functional data, the covariance operator \(\C\) serves as the functional analog to the covariance matrix in multivariate statistics. Therefore, when defining the Mahalanobis distance for functional data in a manner analogous to the multivariate case, the inverse of the covariance operator plays a crucial role in the formulation. However, as a Hilbert-Schmidt operator, the covariance operator is, in general, not invertible, and a regularized covariance operator should be used instead. For univariate functional data, there have been several proposals on how to define a notion of Mahalanobis distance in infinite-dimensional $L^2$ space: \cite{galeano2015mahalanobis} introduced a method 
based on a spectral cutoff regularization, where the covariance operator is truncated to a finite number of components, making it invertible. Specifically, let $X \in L^2(\T)$ be a univariate stochastic process with mean $\mu$ and covariance $\kappa$, denoted as $X \sim \sp(\mu, \kappa)$, then its squared truncated functional Mahalanobis distance with truncation level $m \in \N$ is given by 
\begin{equation}\label{eq:fmd_uni}
    \fmd^2(X,\mu;\kappa,m) = \fmd^2(X;m) = \sum_{i = 1}^m\frac{1}{\lambda_i}\innerproduct{X-\mu}{\ueig_i}^2,
\end{equation}
where $(\lambda_i,\ueig_i)$, $i = 1,\dots,p$, $\lambda_1 \geq\dots\geq \lambda_m > 0$, denote the first $m$ eigenpairs of the covariance operator $\mathcal{K}$ with kernel $\kappa$.
This approach is computationally efficient and well-suited for smoothed functions represented by a finite basis. \cite{ghiglietti2017statistical} proposed an alternative method that introduces regularization through an additional parameter, offering greater flexibility and addressing convergence issues encountered in \cite{galeano2015mahalanobis}. \cite{berrendero2020mahalanobis} further extended the Mahalanobis distance definition by incorporating smoothing through reproducing kernel Hilbert spaces (RKHS), embedding regularization directly into the distance computation.

\section{Multivariate Functional Mahalanobis Distance}\label{section:fMMD}
In this section, we introduce the notion of multivariate truncated Mahalanobis semi-distance, which extends the univariate functional version defined by \cite{galeano2015mahalanobis}. This particular choice of univariate functional Mahalanobis distance was made for computational simplicity and for the favorable properties when applied to functions represented by a finite basis. While \cite{martino2019k} discussed the concept empirically, this paper offers, to the best of our knowledge, the first formal definition and theoretical study of the Mahalanobis distance in the multivariate functional framework.

\begin{definition}\label{def:fMMD(x,y)}
    Let $\bm{X} \sim \msp(\bm{\mu},\K)$ and $\bm{Y} \sim \msp(\bm{\mu},\K)$.  For $M\in\mathbb{N}$ such that $\pi_1 \geq\dots \geq \pi_M>0$, the squared truncated functional multivariate Mahalanobis semi-distance (fMMD) between
    $\bm{X}$ and $\bm{Y}$ (w.r.t. $\K$) is given by
\begin{equation*}%
    \fmmd^2(\bm{X},\bm{Y}; \K,M)=\sum_{k=1}^M\frac{1}{\pi_k} \innerproduct{\bm{X}-\bm{Y}}{\meig_k}^2,
\end{equation*}
where $(\pi_k,\meig_k)$ denotes the $k$th eigenpair of the covariance operator $\C$ with kernel $\K$, $k = 1,\dots,M$, and $M \in \N$ determines the spectral cutoff.
\end{definition}
Using Definition \ref{def:fMMD(x,y)}, we define the squared truncated functional multivariate Mahalanobis semi-distance of $\bm{X} \sim \msp(\bm{\mu},\K)$ (w.r.t. $\bm{\mu}$ and $\K$) as
\begin{equation}\label{eq:fMMD}
\fmmd^2(\bm{X};M):=\fmmd^2(\bm{X},\bm{\mu};\K,M)=\sum_{k=1}^M\frac{1}{\pi_k} \innerproduct{\bm{X}-\bm{\mu}}{\meig_k}^2.
\end{equation}

As shown in \cite{galeano2015mahalanobis} for univariate functions and discussed in \cite{martino2019k} for the multivariate setting, $\fmmd$ in Definition~\ref{def:fMMD(x,y)} is a semi-distance, since it lacks the identifiability condition due to truncation. I.e., if the projections of $\bm{X}$ and $\bm{Y}$ coincide on $\mathrm{span}(\meig_1,\dots,\meig_M)$, for a fixed $M \in \mathbb{N}$, then $\fmmd(\bm{X},\bm{Y})=0$, even if $\bm{X} \neq \bm{Y}$. For the sake of conciseness, we will simply write Mahalanobis distance.

\subsection{\texorpdfstring{$L^2$}{L2}-Multivariate Stochastic Processes}

The following lemma shows that $\fmmd$ given in Equation~\eqref{eq:fMMD} is affine invariant.%

\begin{lemma}\label{lemma:affine_invariance}
    Let $\bm{X} \sim \msp(\bm{\mu}_{\bm{X}},\K_{\bm{X}})$ be %
    such that $\fmmd^2(\bm{X},\bm{\mu}_{\bm{X}};\K_{\bm{X}},M)$ is well defined for $M\in\mathbb{N}$. Then, for any regular matrix $\bm{A} \in \R^{p\times p}$, fixed vector-variate function $\bm{\nu}:\T\to\mathbb{R}^p$ and $\bm{Y} = \bm{A}\bm{X} + \bm{\nu}$, the following holds,
    \begin{itemize}
        \item[i)] $\bm{Y} \sim \msp(\bm{\mu}_{\bm{Y}},\K_{\bm{Y}})$, with $\bm{\mu}_{\bm{Y}} = \bm{A}\bm{\mu}_{\bm{X}}+\bm{\nu}, \K_{\bm{Y}} = \bm{A}\K_{\bm{X}}\bm{A}'$,
        \item[ii)] $\fmmd^2(\bm{Y},\bm{\mu}_{\bm{Y}};\K_{\bm{Y}},M)=\fmmd^2(\bm{X},\bm{\mu}_{\bm{X}};\K_{\bm{X}},M).$
    \end{itemize}
\end{lemma}

Lemma \ref{lemma:independent processes} gives another desirable property of $\fmmd$ (Definition \ref{def:fMMD(x,y)}); in the case where the components of $\bm{X}$ are uncorrelated, its squared Mahalanobis distance reduces to the weighted sum of univariate functional Mahalanobis distances~\eqref{eq:fmd_uni},
where the amount of truncation depends on the magnitude of the eigenvalues of the individual processes. %

\begin{lemma}\label{lemma:independent processes}
   Let $\bm{X} \sim \msp(\bm{\mu},\K)$ be a multivariate process with uncorrelated components $X_j\in L^2(\T)$, $j=1,\dots,p$. Let further $\pi_1\geq\dots\geq\pi_M>\pi_{M+1}$ be the largest $M$ eigenvalues of the covariance operator $\C=\mathrm{diag}(\mathcal{K}_1,\dots,\mathcal{K}_p)$ associated with the covariance function $\K = \diag(\kappa_1,\dots,\kappa_p)$. 
   Then, for $M=m_1+\cdots +m_p$,
   $$
   \fmmd^2(\bm{X},\bm{\mu};\K,M)   = \sum_{j=1}^p\fmd^2(X_j,\mu_j;\kappa_j,m_j),
   $$
   with $(\lambda_i^{(j)},\ueig_i^{(j)}), i = 1,\dots,m_j$ denoting the $m_j$ largest eigenpairs of the covariance operator $\mathcal{K}_j$ with kernel $\kappa_j$, and $m_j=|\{\lambda_1^{(j)},\dots\lambda_M^{(j)}\}\cap \{\pi_1,\dots\pi_M\}|$ is the number of eigenvalues of the covariance operator $\mathcal{K}_j$ that belongs to  $\{\pi_1,\dots,\pi_M\}$, $i=1,\dots,p$, where we count also the multiplicities in $\{\lambda_1^{(j)},\dots\lambda_M^{(j)}\}$.
\end{lemma}
Lemma~\ref{lemma:independent processes} implicitly implies that the (uncorrelated) components in the multivariate process should be transformed to similar scales. We also note that the additive property of $\fmmd$ discussed in Lemma~\ref{lemma:independent processes}
 extends to processes with uncorrelated blocks of components. However, we refrain from formally presenting this more general result to keep the paper accessible and avoid unnecessary technical complexity.

A key challenge in computing \(\fmmd\) lies in accurately estimating the eigenfunctions and the covariance operator. The common approach is to vectorize the $p$-variate process \(\bm{X} = (X_1, \dots, X_p)': \T \to \mathbb{R}^p\) by concatenating the individual processes and then applying the univariate method \citep{ramsay_silverman_2005}. %
However, this approach disregards the structure in $\bm{X}$ and can substantially inflate the dimensionality, making computation increasingly demanding as $p$ grows. Interestingly, even when the separability assumption is violated, the variance–bias trade-off often results in the constrained estimator outperforming the unconstrained one for small to moderate sample sizes; a phenomenon we observed in the simulation study of Appendix~\ref{subsection:sample_vs_mmle}.
Hence, in the following, we focus on the family of multivariate processes with separable covariance structure; for an overview, see, e.g., \cite{chen2021space}.

\subsection{Separable Covariance Processes}

A multivariate stochastic process $\bm{X}\sim\mathcal{MSP}(\bm{\mu},\K)$ has separable covariance structure if for every $s,t\in\T$, the covariance $\K$ can be decomposed into
\begin{equation}\label{eq:sep_process}
    \K(s,t)=\bm{\Sigma}^{\row}\,\kappa(s,t),
\end{equation}
for $\bm{\Sigma}^{\row}\in\pds(p)$ representing the cross-covariance structure between the \(p\) components, and a positive definite kernel $\kappa$ capturing the common  temporal covariance structure. 
We write $\bm{X}\sim\mathcal{MSP}(\bm{\mu},\bm{\Sigma}^{\row},\kappa)$. The separability property significantly simplifies the covariance estimation since it allows for the within-function and the between-component second-order dependence to be studied (and interpreted) separately; see, e.g., \cite{chen2021space,chen2023multivariate,genton2007separable,cressie1999classes,rodriguez1974design}.

\begin{remark}
    It should be noted that $\bm{\Sigma}^{\row}$ and $\kappa(s,t)$ in decomposition~\eqref{eq:sep_process} %
    are only identifiable up to a multiplicative constant;  for any $c>0$, 
    $\K(s,t)=(c\bm{\Sigma}^{\row})\,(c^{-1}\kappa(s,t))$. This ambiguity has little practical relevance. However, for reproducibility, %
    we fix the scale of $\kappa$ using the strategy presented in \cite{mayrhofer2024robust}. %
\end{remark}

The equivalent of Lemma \ref{lemma:independent processes} for the processes with separable covariance structure is given in Corollary \ref{cor:independent processes}.

\begin{corollary}\label{cor:independent processes}
   Let $\bm{X} \sim \msp(\bm{\mu},\bm{\Sigma}^{\row},\kappa)$ be a multivariate process with separable covariance  and components $X_j\in L^2(\T)$, $j=1,\dots,p$. Then, for $M$ being a multiple of $p$, i.e., $M=mp>0$, with $m \in \mathbb{N}$, 
   the following holds:
   \begin{itemize}
       \item[(i)] $\displaystyle\fmmd^2(\bm{X},\bm{\mu};\bm{\Sigma}^{\row},\kappa,M)   = \sum_{j=1}^p\fmd^2(\bm{e}_j'(\bm{\Sigma}^{\row})^{-1/2} \bm{X},\bm{e}_j'(\bm{\Sigma}^{\row})^{-1/2}\bm{\mu};\kappa,m)$. 
       \item[(ii)] If $\bm{X}$ has uncorrelated components, i.e., $\bm{\Sigma}^{\row}  = \diag(\sigma_{1}^2,\dots,\sigma_{p}^2)$ with $\sigma_1,\dots,\sigma_p>0$, then $$
       \fmmd^2(\bm{X},\bm{\mu};\bm{\Sigma}^{\row},\kappa,M)   = 
       \sum_{j=1}^p \frac{1}{\sigma_j^2}\fmd^2(X_j,\mu_j;\kappa,m).
       $$
   \end{itemize}
   Here, $(\lambda_i,\ueig_i), i = 1,\dots,m$, correspond to the first $m$ eigenpairs of the covariance operator $\mathcal{K}$ with kernel $\kappa$, and $\bm{e}_j$ is the $j$th vector of the canonical basis of $\mathbb{R}^p$.
\end{corollary}

The separability of the covariance structure transfers further to a certain separability of the eigendecomposition of $\C$, thus simplifying the calculation of eigenpairs $(\pi_i,\meig_i)$, $i\geq 1$, to a separate univariate functional eigendecomposition of the covariance associated with $\kappa$, and a multivariate eigendecomposition of $\bm{\Sigma}^{\row}$. {As argued in \cite{aston2017tests}, even if separability is misspecified, separable eigenfunctions still provide a basis to represent the data. Although this may not guarantee optimal efficiency in such cases, they can still achieve near-optimal performance under suitable assumptions \citep{chen2017modelling}.} For more information, see Appendix~\ref{appendix:sub:FPCA}.

One of the most prominent members of the separable-covariance class of processes is a multivariate Gaussian process: $\bm{X}$ is a multivariate Gaussian process if every finite collection of realizations has a matrix-variate normal distribution. 
We provide a brief overview of matrix normal distribution in Appendix~\ref{appendix:sub:matrix_normal}; for more details on matrix-variate distributions, see, e.g., \cite{gupta1999, gupta2012elliptically}.
Multivariate Gaussian processes are fully characterized by their first and second moments, and in continuation, we write $\bm{X}\sim \mathcal{MGP}(\bm{\mu},\bm{\Sigma}^{\row},\kappa)$. For a formal definition, see Appendix~\ref{appendix:sub:multivariate_gaussian_process}, and for an overview of properties of multivariate Gaussian processes, see, e.g., \cite{chen2017modelling,chen2023multivariate}. 
Lemma \ref{lemma:fmmd_distribution} gives the distribution of $\fmmd$ under the assumption of Gaussianity. 

\begin{lemma}\label{lemma:fmmd_distribution}
    Let $\bm{X} \sim \mgp(\bm{\mu},\bm{\Sigma}^{\row},\kappa)$ be a multivariate Gaussian process. For $M>0$,
    $$ \fmmd^2(\bm{X},\bm{\mu};\bm{\Sigma}^{\row},\kappa,M)\sim \chi^2(M),
    $$
     where $\chi^2(M)$ is the chi-square distribution with $M$ degrees of freedom.
\end{lemma}

\section{Robust Parameter Estimation for Separable Processes}
\label{section:parameter_estimation}
In practice, we do not observe continuous functions, but rather a discrete set of functional values. As discussed in \cite{BASNA2022}, a fundamental step in FDA is often to transform these discretely recorded data into a functional form, allowing each observed function to be evaluated at any point within its continuous domain \(t \in \T\). Typically, the functional object is approximated by
linear combinations of a finite number of basis functions, where this representation is exact only for functions of finite rank.

\subsection{Finite Basis Representation}

Let $\bm{X}_1,\dots,\bm{X}_n$, where for $i = 1,\dots,n.$,  $\bm{X}_i = [X_{i,k}(t_l)]_{1\leq k\leq p,\,1\leq l\leq q}\in\mathbb{R}^{p\times q}$, 
be an i.i.d. sample of the multivariate random processes $\mathcal{MSP}(\bm{\mu},\bm{\Sigma}^{\row},\kappa)$ with separable covariance $\K=\bm{\Sigma}^{\row}\kappa$, observed at time points $t_1,\dots,t_q\in\T$, for $q\in\N$.

Given $m\in\mathbb{N}$, let $\mbasis{(t)}=(\unibasis_1(t),\dots \unibasis_m(t))'$, $t\in\T$ be a fixed basis that spans an $m$-dimensional subspace of $L^2(\T)$. We transform each discretely observed $\bm{X}_i$, $i=1,\dots,n$,
to a \textit{functional form} by representing its components as a linear combination of basis functions in $\mbasis$:
\begin{equation*}%
    X_{i,j}^{(m)}(t)=\sum_{k=1}^m{a}_{j,k}^{(i)}{}\unibasis_k(t)=\bm{a}_{j}^{(i)}{}'\mbasis(t),\quad j=1,\dots,p,\,i=1,\dots,n,
\end{equation*}
where superscript $(m)$ emphasizes that ${X}_{i,j}^{(m)}$, $j=1,\dots,p,\,i=1,\dots,n$, is at most a \textit{rank-$m$} process, meaning that its covariance operator has at most  $m$ non-zero eigenvalues. 

Collecting all coefficients 
$\bm{a}_j^{(i)}=(a_{j,1}^{(i)},\ldots ,a_{j,m}^{(i)})'$, $j=1,\dots,p$, corresponding to the $i$th observation, in a matrix $\bm{A}_i=(\bm{a}_1^{(i)},\hdots ,\bm{a}_p^{(i)})\in\mathbb{R}^{m\times p}$, we can represent each observation as
\begin{equation}\label{eq:Xi in basis}
    \bm{X}_{i}^{(m)}(t)=\bm{A}_i'\mbasis(t),\quad i=1,\dots,n.
\end{equation}
For simplicity of the notation, we drop the superscript $(m)$ in the following and write $\displaystyle  \bm{X}_{i}(t)=\bm{A}_i'\mbasis(t)$, $i=1,\dots,n.$

The coefficients $\bm{a}_j^{(i)}$, $j=1,\dots,p,\,i=1,\dots,n$, are usually determined through least squares estimation, where the goal is to minimize the difference between the observed data and their approximation while ensuring smoothness; see \cite{ramsay_silverman_2005} for more details. Common choices of basis functions include splines, wavelets, and Fourier bases, among others; see, e.g., \cite{ramsay_silverman_2005} for an overview. The choice of the basis $\mbasis$ is beyond the scope of this paper, and for simplicity we use B-splines; see, e.g.,~\cite{Eilers1996}. For more details on a connection between a finite-basis representation \eqref{eq:Xi in basis} and noise smoothing in additive noise models, see Appendix~\ref{appendix:sub:smoothing}. %
The following theorem shows how the separability of the covariance $\K$ translates onto the distribution of a random matrix $\bm{A}$, and gives a direct connection between $\fmmd^2$ of $\bm{X}$ and the squared matrix Mahalanobis distance of $\bm{A}$, defined as  
\begin{align}
    \label{eq:mmd}
    \mmd^2(\bm{A}, \bm{\M}, \bm{\Sigma}^{\col}, \bm{\Sigma}^{\row}) 
    &= \mmd^2(\bm{A})
    = \tr((\bm{\Sigma}^{\row})^{-1} (\bm{A} - \bm{\M})' (\bm{\Sigma}^{\col})^{-1} (\bm{A} - \bm{\M}))
\end{align}
for an $m \times p$ random matrix $\bm{A}$, with mean matrix $\bm{\M} \in \R^{m \times p}$, and covariances $\bm{\Sigma}^{\col} \in \pds(m)$ and $\bm{\Sigma}^{\row} \in \pds(p)$, respectively \citep{mayrhofer2024robust}.%

\begin{theorem}\label{theorem:fmmd_basis}
    Let $\boldsymbol{X}(t)=\bm{A}'\mbasis(t)$ be a rank $m$ separable covariance process with mean $\bm{\mu}$ and covariance $\bm{K}=\boldsymbol{\Sigma}^{\row}\kappa$, with a regular matrix
    $\bm{A}=(\bm{a}_1,\hdots ,\bm{a}_p)\in\mathbb{R}^{m\times p}$, and a vector of basis functions
    $\mbasis=(\unibasis_1,\dots,\unibasis_m)'$.
     Then the following holds:
    \begin{itemize}
        \item[(i)] $\bm{A}$ has a matrix-variate distribution with mean $\bm{M}_{\bm{A}}$  and covariance $\mathrm{Cov}(\vect(\bm{A}))=\boldsymbol{\Sigma}^{\row}\otimes\boldsymbol{\Sigma}^{\col}$, for $\boldsymbol{\Sigma}^{\col}\in\pds(m)$, such that for every $s,t\in \T$, 
        $$
        \bm{M}_{\bm{A}}'\mbasis(t)=\bm{\mu}(t)\quad \text{and}\quad\mbasis'(s)\bm{\Sigma}^{\col}\mbasis(t)=\kappa(s,t).
        $$
        \item[(ii)] $\fmmd^2{(\boldsymbol{X};mp)}=\mathrm{tr}\left((\boldsymbol{\Sigma}^{\row})^{-1} (\bm{A}-\bm{M}_{\bm{A}})' (\boldsymbol{\Sigma}^{\col})^{-1}(\bm{A}-\bm{M}_{\bm{A}})\right)=\mmd^2{(\bm{A})}$.
        \item[(iii)] If additionally  $\bm{X}\sim \mathcal{MGP}(\bm{\mu},\bm{\Sigma}^{\row},\kappa)$ is a multivariate Gaussian process, then $\bm{A}\sim\mathcal{MN}(\bm{M}_{\bm{A}},\allowbreak\bm{\Sigma}^{\col},\allowbreak\bm{\Sigma}^{\row})$ follows a matrix normal distribution with mean $\bm{M}_{\bm{A}}$ and positive definite covariances $\bm{\Sigma}^{\row}$ and $\bm{\Sigma}^{\col}$ as in (i). %
    \end{itemize} 
\end{theorem}
The equivalent of Theorem~\ref{theorem:fmmd_basis} for univariate processes of finite rank is given in Corollary~\ref{corollary:fmd} in Appendix~\ref{appendix:fmmd_proofs}. Theorem \ref{theorem:fmmd_basis}  thus establishes a formal equivalence between the functional and matrix-variate Mahalanobis distances under separability, thereby linking functional robustness to matrix-variate robustness theory.

Details on covariance estimation for the random matrices can be found in \cite{Dutilleul1999,soloveychik2016gaussian}.

As the primary objective is robust covariance estimation, we employ robust estimators of the moments of the random matrix \(\bm{A}\), specifically using the MMCD estimators for mean and covariance, as introduced in \cite{mayrhofer2024robust}. The MMCD estimators $(\hat{\bm{\M}}_{\bm{A}},\hat{\bm{\Sigma}}^{\row}, \hat{\bm{\Sigma}}^{\col})$ are highly robust, efficient and consistent estimators for matrix-variate elliptical distributions. We briefly review the MMCD method in Appendix~\ref{section:MMCD}.

We emphasize that there are computational arguments for using smoothed data: First, parameter estimation and distance computation are performed on the coefficient matrix, which has a much lower dimension than the raw data matrix, resulting in a substantial reduction in computation time, see Appendix~\ref{appendix:simulations_time} for a detailed comparison. Second, the required sample size and breakdown point of the MMCD estimators depend on the ratio of the number of rows to columns. As this ratio approaches one, fewer samples are needed, and the breakdown point increases. 

In practice, both the type and the dimension of the basis determine the degree of smoothing. In our applications, we choose the number of basis functions so that the main functional structure of the data is retained without excessive smoothing. Standard automatic selection criteria, such as cross-validation aimed at minimizing reconstruction error, are not directly suited to the present setting, since an accurate reconstruction of anomalous curves is not necessarily desirable for outlier detection. Developing a robust data-driven procedure for selecting the basis dimension would therefore constitute a separate methodological problem and is beyond the scope of this work.

A detailed pseudocode for robust parameter estimation of separable covariance processes using MMCD estimators is given in Algorithm~\ref{alg:alg1} in Appendix~\ref{appendix:algorithm}. 
Since covariance estimation and FPCA are closely related, we outline how to compute the robust functional principal components in the separable covariance setting in Algorithm~\ref{alg:alg2} in Appendix~\ref{appendix:algorithm}.

\section{Explainable Outlier Detection}~\label{section:shapley}
The combination of the MMCD estimators \citep{mayrhofer2024robust} with the truncated multivariate functional Mahalanobis distances provides a reliable framework for outlier detection. To understand why an observation is outlying, we propose a method for additively decomposing the truncated (multivariate) functional Mahalanobis distance into time-coordinate-specific outlyingness contributions. As in \cite{mayrhofer2022}, we use Shapley values \citep{shapley1953} to obtain those decompositions; they were originally introduced in cooperative game theory \citep{peters2008} and gained popularity in the field of Explainable AI~\citep{lundberg2017}. 

Intuitively, the Shapley value decomposition attributes an observation’s overall outlyingness to specific components and domain subintervals. In the functional context, this enables localized interpretability, showing when and which component contributes most to the outlyingness.

As a starting point, we briefly review Shapley values in the multivariate setting. For ease of reading, a summary of the notation used in this section is provided in Appendix~\ref{appendix:shapley_proofs}. For a $p$-variate observation $\bm{x} = (x_1,\ldots ,x_p)'$ from a population with mean $\bm{\mu} = (\mu_1,\ldots ,\mu_p)'$, covariance matrix $\bm{\Sigma} \in \pds(p)$,  and $P=\{1,\ldots ,p\}$ the index set of the variables, the outlyingness contributions $\bm{\theta}(\bm{x},\bm{\mu},\bm{\Sigma}) = \bm{\theta}(\bm{x}) = (\theta_1(\bm{x}),\ldots, \theta_p(\bm{x}))'$
assign each variable its average marginal contribution to the squared Mahalanobis distance~\eqref{eq:md}, i.e.,
\begin{align}
    \theta_k(\bm{x},\bm{\mu},\bm{\Sigma}) = \sum_{S \subseteq P\setminus\{k\}} \frac{\abs{S}!(p-\abs{S}-1)!}{p!} \Delta_k \md^2(\xhat{S})
    = (x_k-\mu_k) \sum_{j =1}^p (x_j-\mu_j) \omega_{jk},
    \label{eq:shapley_md}
\end{align}
simply written as $\theta_k(\bm{x}) = \theta_k(\bm{x},\bm{\mu},\bm{\Sigma})$, with marginal outlyingness contributions
\begin{equation}
    \Delta_k \md^2(\xhat{S}) := \md^2(\xhat{S\cup\{k\}})-\md^2(\xhat{S}) \quad \text{and} \quad
    \xhatj{S}{j}:= \begin{cases}
        x_j & \text{if } j \in S\\
        \mu_j & \text{if } j \notin S
    \end{cases}
    \label{eq:define_xj}
\end{equation}
as the components of $\xhat{S}$. Here, $\omega_{jk}$ is the element $(j,k)$ of $\bm{\Omega} = \bm{\Sigma}^{-1}$. The efficiency property of the Shapley value implies that $\sum_{k = 1}^p \theta_k(\bm{x}) = \md^2{(\bm{x})}$. 
For more details on outlier explanation based on Shapley values and its properties for multivariate and matrix-variate data, see Appendix~\ref{appendix:shapley_proofs}.

\subsection{Univariate Functional Data}%
Consider two univariate stochastic processes $X, Y \in L^2(\T)$ with $\T = \T_1 \cup \cdots \cup \T_d$, $\T_a \cap \T_b = \emptyset$ for all $a \neq b$ $\in \{1, \ldots ,d\}$, then the inner product $\innerproduct{X}{Y} = \sum_{a = 1}^d \innerproduct{X}{Y}_{\T_a}$ with $\innerproduct{X}{Y}_{\T_a} = \int_{\T_a} X(t)Y(t) \mathrm{d}t$.
To generalize Equation~\eqref{eq:define_xj} to the functional setting, consider $X \sim \sp(\mu, \kappa)$, with marginal outlyingness contribution on the subinterval $\T_a$ to $\fmd^2(X,\mu;\kappa,m)$ defined as
\begin{equation}
    \label{eq:delta_univ_time}
    \Delta_{\T_a} \fmd^2(\hat{X}^{R},\mu;\kappa,m) := \fmd^2(\hat{X}^{R\cup\{a\}},\mu;\kappa,m)-\fmd^2(\hat{X}^{R},\mu;\kappa,m) 
\end{equation} 
with $R \subseteq D \setminus \{a\}, D = \{1,\dots,d\},$ and 
\begin{equation*}
    \hat{X}^{R}(t):= 
    \begin{cases}
        X(t) & \text{if } t \in \bigcup_{b \in R} \T_b\\
        \mu(t) & \text{if } t \notin \bigcup_{b \in R} \T_b 
    \end{cases}.
\end{equation*}
\begin{proposition}\label{proposition:shapley_fmd}
    For $X \sim \sp(\mu, \kappa)$ and $\Delta_{\T_a} \fmd^2(\hat{X}^{R},\mu;\kappa,m)$ as in Equation~\eqref{eq:delta_univ_time},
    the time-specific outlyingness contribution within the subinterval $\T_a$ based on the Shapley value is given by
    \begin{align}
        \theta_{\T_a}(X,\mu;\kappa,m) :=& \sum_{R \subseteq D\setminus\{a\}} \frac{\abs{R}!(d-\abs{R}-1)!}{(d)!} \Delta_{\T_a} \fmd^2(\hat{X}^{R},\mu;\kappa,m) \nonumber\\
        &= \sum_{i = 1}^m \frac{1}{\lambda_i}\innerproduct{X - \mu}{\ueig_i}_{\T_a} \innerproduct{X - \mu}{\ueig_i}, \label{eq:shapley_time}
    \end{align}
    where $(\lambda_i,\ueig_i), i = 1,\dots,m,$ denote the eigenpairs of the covariance %
    $\mathcal{K}$ with kernel $\kappa$.
\end{proposition}
The following lemma outlines how to efficiently compute ~\eqref{eq:shapley_time} for smooth functions represented in a basis. Let $\mbasis(t) = (\unibasis_1(t),\dots,\unibasis_m(t))', t \in \T$ with $m \in \N$ be a family of basis functions in $L^2(\T)$. The rank $m \times m$ matrix of inner products of $\mbasis$ is denoted as $\bm{W}=\int_{\T}\mbasis(t)\mbasis'(t)\mathrm{d}t$ and $\bm{W}_{\T_a} = \int_{\T_a}\mbasis(t)\mbasis'(t)\mathrm{d}t$ represents the matrix of inner products restricted to $\T_a \subseteq \T$. 
\begin{lemma}\label{lemma:shapley_using_coefficents_univariate}
    Let $X \sim \sp(\mu, \kappa)$ be a rank $m \in \N$ stochastic process as in Corollary~\ref{corollary:fmd}, with $X(t) = \bm{a}'\mbasis(t)$, $\mu(t) = \bm{m}_{\bm{a}}'\mbasis(t)$, and $\kappa(s,t) = \mbasis'(s)\bm{\Sigma}\mbasis(t)$, for $s,t \in \T$, then
    \begin{align*}
        \theta_{\T_a}(X,\mu;\kappa,m) 
        &= (\bm{a}-\bm{m}_{\bm{a}})'\bm{W}_{\T_a} \bm{W}^{-1}\bm{\Sigma}^{-1} (\bm{a}-\bm{m}_{\bm{a}}).
    \end{align*} 
\end{lemma}

\subsection{Multivariate Functional Data}%
Let us consider the $p$-variate stochastic process $\bm{X} \sim \msp(\bm{\mu},\K)$, with $\bm{X} = (X_1,\ldots ,X_p)'$, $P = \{1,\ldots ,p\}$ the index set of variables, and $(\pi_k,\meig_k), k = 1,\dots,M$, the eigenpairs of the covariance operator $\C$ with kernel $\K$.
To obtain the marginal outlyingness contribution of the $k$th coordinate function in the time interval $\T_a$ to $\fmmd^2$ based on the Shapley value, we modify Equation~\eqref{eq:define_xj} and define them as
\begin{equation}
    \Delta_{k,\T_c} \fmmd^2(\Xhat{S,R},\bm{\mu};\K,M) := \fmmd^2(\Xhat{S\cup\{k\}, R \cup \{c\}},\bm{\mu};\K,M)-\fmmd^2(\Xhat{S, R},\bm{\mu};\K,M) \label{eq:delta_time_multivariate}
\end{equation} 
with $\Xhat{S,R}=(\Xhatj{S,R}{1},\ldots ,\Xhatj{S,R}{p})'$,
\begin{equation}
\label{eq:define_fXj_a}
    \Xhatj{S,R}{j}(t):= \begin{cases}
        X_j(t) & \text{if } j \in S \land t \in \bigcup_{a \in R} \T_a\\
        \mu_j(t) & \text{if } j \notin S \lor t \notin \bigcup_{a \in R} \T_a 
    \end{cases}.
\end{equation}

\begin{proposition}~\label{proposition:shapley_cell}
    For $\bm{X} \sim \msp(\bm{\mu},\K)$ and $\Delta_{k,\T_c} \fmmd^2(\Xhat{S,R},\bm{\mu};\K,M)$ as in Equation~\eqref{eq:define_fXj_a},
    the outlyingness contributions of the $k$th coordinate in the time-interval $\T_a$ to $\fmmd^2(\bm{X},\bm{\mu};\K,M)$ based on the Shapley value are given by
    \begin{align}
        \Theta_{k,\T_{a}}(\bm{X},\bm{\mu};\K,M)
        = \sum_{i = 1}^M \frac{1}{\pi_i}\innerproduct{X_k - \mu_k}{\psi_{i,k}}_{\T_a} \innerproduct{\bm{X} - \bm{\mu}}{\meig_i}, \label{eq:shapley_time_and_coordinate_multivariate}
    \end{align}
    with $\psi_{i,k} = \meig_i'\bm{e}_k$ denoting the $k$th component of the $i$th eigenfunction $\meig_i$ of covariance operator $\C$ with kernel $\K$.
\end{proposition}
The proof is given in Appendix~\ref{appendix:shapley_proofs}; it relies on concatenating the coordinate functions and Proposition~\ref{proposition:shapley_fmd}. We can modify Equations~\eqref{eq:delta_time_multivariate} and \eqref{eq:define_fXj_a} to obtain time-specific outlyingness contributions by replacing all coordinate functions $X_j(t)$ by their mean $\mu_j(t)$ for a given interval $t \in \T_a$ for all $j = 1,\dots,p$ coordinates instead of only one. This yields the time-specific outlyingness contributions
\begin{align*}
    \theta_{\T_{a}}(\bm{X},\bm{\mu};\K,M) &= \sum_{i = 1}^M \frac{1}{\pi_i}\innerproduct{\bm{X} - \bm{\mu}}{\meig_i}_{\T_a} \innerproduct{\bm{X} - \bm{\mu}}{\meig_i}\\
    &= \sum_{k = 1}^p \Theta_{k,\T_{a}}(\bm{X},\bm{\mu};\K,M). %
\end{align*}
Similarly we can obtain coordinate-specific outlyingness contributions by replacing $X_k(t)$ by its mean $\mu_k(t)$ for the entire interval $\T$ in Equation~\eqref{eq:delta_time_multivariate} to obtain
\begin{align*}
    \theta_k(\bm{X},\bm{\mu};\K,M) &= \sum_{i=1}^M\frac{1}{\pi_i} \innerproduct{X_{k} - \mu_k}{\psi_{i,k}}  \innerproduct{\bm{X} - \bm{\mu}}{\meig_{i}} = \sum_{a = 1}^d \Theta_{k,\T_{a}}(\bm{X},\bm{\mu};\K,M), 
\end{align*}
see Appendix~\ref{appendix:shapley_proofs} for more details. The following corollary outlines how to compute the coordinate-specific outlyingness contributions for a multivariate stochastic process with a separable covariance structure.

\begin{corollary}~\label{corollary:shapley_separable_cell}
    For a multivariate stochastic process $\bm{X}\sim\msp(\bm{0},\bm{\Sigma}^{\row},\kappa)$ with separable covariance operator $\C = \bm{\Sigma}^{\row}\mathcal{K}$ and covariance kernel $\K(s,t) = \bm{\Sigma}^{\row} \kappa(s,t)$, Equation~\eqref{eq:shapley_time_and_coordinate_multivariate} becomes
    \begin{align*}
        \Theta_{k,\T_{a}}(\bm{X},\bm{\mu};\bm{\Sigma}^{\row}\kappa,M) 
        = \sum_{i = 1}^m\sum_{j = 1}^p \frac{1}{\lambda_i^{\ker}\lambda_j^{\row}}\left(\innerproduct{X_k - \mu_k}{\ueig_i}_{\T_a}{v}^{\row}_{j,k} \sum_{l = 1}^p\innerproduct{X_l - \mu_l}{\ueig_i}_{\T}{v}^{\row}_{j,l}\right).
    \end{align*}
    Here $(\lambda^{\ker}_i,\ueig_i),\, i = 1,\dots,m$, denote the $m$ largest eigenpairs of $\mathcal{K}$, $(\lambda^{\row}_j,\bm{v}^{\row}_j),\, j = 1,\dots,p$, the eigenpairs of $\bm{\Sigma}^{\row}$, and ${v}^{\row}_{j,k} = \bm{e}_k'\bm{v}^{\row}_j$.
\end{corollary}

We can efficiently compute the outlyingness contributions in the time interval $\T_{a}$ for all $p$ variables simultaneously using matrix operations. Let $\tilde{\bm{A}}^{\T} \in \R^{p \times m}$ with entries $(\alpha_{ji}^{\T}) = \innerproduct{X_j - \mu_j}{\ueig_i}_{\T}$ and $\tilde{\bm{A}}^{\T_a} \in \R^{p \times m}$ with entries $(\alpha_{ji}^{\T_a}) = \innerproduct{X_j - \mu_j}{\ueig_i}_{\T_a}$ denote the matrices of inner products of the coordinate functions $X_j, j = 1,\dots,p,$ with the functional principal components $\ueig_i, i = 1,\dots,m,$ and $\bm{D}^{\ker} = \diag(\lambda^{\ker}_1,\dots,\lambda^{\ker}_m)$ the diagonal matrix of the corresponding ordered eigenvalues $\lambda^{\ker}_1 \geq \cdots \geq \lambda^{\ker}_m$ of the kernel function $\kappa$. Then the vector $\bm{\Theta}_{\T_a}(\bm{X},\bm{\mu};\K,M)$ with entries $\Theta_{k,\T_{a}}(\bm{X},\bm{\mu};\bm{\Sigma}^{\row}\kappa,M)$, for $k = 1,\dots,p$, can be computed as
\begin{align} \label{eq:shapley_cell_computation}
    \bm{\Theta}_{\T_a}(\bm{X},\bm{\mu};\K,M) = \diag((\bm{\Sigma}^{\row})^{-1}\tilde{\bm{A}}^{\T_a} (\bm{D}^{\ker})^{-1} (\tilde{\bm{A}}^{\T})'),
\end{align}
for each interval subinterval $\T_a, a = 1,\dots,d$, of $\T$. For smooth multivariate functional data represented in a finite basis $\mbasis = (\unibasis_1,\dots,\unibasis_m)'$, the outlyingness scores can be computed using the coefficients. 
\begin{lemma}\label{lemma:shapley_coef_fmmd_cell}
     Let $\bm{X} \sim \msp(\bm{\mu}, \bm{\Sigma}^{\row}, \kappa)$ be a rank $M \in \N$ multivariate stochastic process as in Theorem~\ref{theorem:fmmd_basis}, with $\bm{X}(t)=\bm{A}'\mbasis(t)$, $\bm{\mu}(t) = \bm{M}_{\bm{A}}'\mbasis(t)$, and $\kappa(s,t) = \mbasis'(s)\bm{\Sigma}^{\col}\mbasis(t)$, for $s,t \in \T$. Then the following holds,
    \begin{align}
    \label{eq:singlecomp}
        \Theta_{\T_a,k}(\bm{X},\bm{\mu};\K,M) 
        &=  \sum_{j = 1}^p \frac{1}{\lambda_j^{\row}}{v}^{\row}_{j,k}(\bm{a}_k-\bm{m}_{\bm{A},k})'\bm{W}_{\T_a}\bm{W}^{-1}(\bm{\Sigma}^{\col})^{-1}(\bm{A}-\bm{M}_{\bm{A}})'\bm{v}^{\row}_{j},
    \end{align}
    with $(\lambda^{\row}_j,\bm{v}^{\row}_j),\, j = 1,\dots,p$, the eigenpairs of $\bm{\Sigma}^{\row}$, and ${v}^{\row}_{j,k} = \bm{e}_k'\bm{v}^{\row}_j$.
\end{lemma}
Using matrix operations, the variable contributions in (\ref{eq:singlecomp}) can be obtained as the elements of the vector
\begin{align*}
    \bm{\Theta}_{\T_a}(\bm{X},\bm{\mu};\K,M) = \diag((\bm{\Sigma}^{\row})^{-1}(\bm{A}-\bm{M}_{\bm{A}}) \bm{W}_{\T_a}\bm{W}^{-1}(\bm{\Sigma}^{\col})^{-1}(\bm{A}-\bm{M}_{\bm{A}})').
\end{align*}
Here, $\diag(.)$ of a square matrix denotes the vector of diagonal entries.

\section{Simulations}\label{section:simulations}
Outlier detection in multivariate functional data is challenging due to the presence of diverse anomaly types, such as shifts, shape deviations, spikes, and changes in dependence structure. A simulation study was conducted to compare the proposed distance-based approach with existing distance- and depth-based methods under various contamination settings. The analysis was performed in \texttt{R}~\citep{R_language}.%
The \emph{distance-based} approach uses Mahalanobis distance, which is either based on the trimmed functional Mahalanobis distance~\eqref{eq:fMMD} applied to B-spline coefficient matrices, or the matrix Mahalanobis distance~\eqref{eq:mmd} applied to the raw data. The parameters for the Mahalanobis distance are estimated using classical matrix maximum likelihood estimation (MMLE, \citep{Dutilleul1999}) or the robust MMCD %
approach \citep{mayrhofer2024robust}. Both approaches are implemented in the \texttt{R} package \texttt{robustmatrix} \citep{robustmatrix}. %
    
On the other hand, \emph{depth-based} methods determine outliers by measuring the centrality of a function within a data cloud and define outliers as observations with low depth values, which indicates that they lie far from the central bulk of the data, see, e.g., \cite{zuo2000general} for more details. 
There are various depth-based outlyingness measures for multivariate functional data, such as Stahel-Donoho outlyingness/projection depth (fSDO; \cite{stahel1981breakdown, donoho1982breakdown, zuo2003projection}), skewness-adjusted outlyingness/skewness-adjusted projection depth (fAO; \cite{hubert2008outlier, hubert2015multivariate}), or directional outlyingness/directional projection depth (fDO; \cite{rousseeuw2018measure}), all implemented in the \texttt{R} package \texttt{mrfDepth}~\citep{mrfDepth}. 
Additionally, we consider the Magnitude-Shape Plot (MS; \cite{dai2018multivariate}) from the \texttt{R} package \texttt{fdaoutlier}~\citep{fdaoutlier}, which is based on directional outlyingness as defined by \cite{dai2019directional}. These methods are applied to the raw and smoothed data as well as to the coefficient matrices. 
Finally, we consider the multivariate functional isolation forest (MFIF) method by \cite{Staerman2019}, which extends isolation forest to functional data by projecting functions onto a finite-dimensional dictionary and applying isolation-based splitting in the resulting coefficient space. The scalar product, required to define the random splits, is set to equal weights between the $L_p^2$ inner product \eqref{eq:inner_product_L2p} and the $L_p^2$ inner product of the derivatives, balancing detection of shape and location outliers. We also have to choose from the predefined dictionaries and use Brownian motion (MFIF brown), Gaussian wavelets (MFIF gauss), and Self (MFIF self). The method is available at \url{https://github.com/GuillaumeStaermanML/FIF} and was imported into \texttt{R} via \texttt{reticulate} \citep{reticulate}.

\subsection{Setup}
The random functions are drawn from a multivariate Gaussian process with a separable covariance structure by generating finite-dimensional realizations at $q = 100$ time points. 

The data are smoothed using cubic B-spline basis functions without a roughness penalty. To assess the sensitivity of the proposed procedure to the amount of smoothing, we consider $m\in\{10,20,30\}$ basis functions, in addition to the analysis based on the raw data. These values cover representations ranging from relatively strong smoothing to considerably more flexible approximations of the observed curves. The results obtained with 20 and 30 basis functions are generally very similar to each other and to those based on the raw data, whereas in some settings 10 basis functions oversmooth the curves and fail to capture relevant functional features, resulting in poorer performance. Thus, the simulations indicate that the proposed procedure is relatively insensitive to the precise number of basis functions, provided that the data are not excessively smoothed. Since the aim of the simulation study is to compare the relative performance of the methods rather than to optimize the smoothing step, we keep the same basis dimensions across the respective simulation settings.
The rank truncation is set equal to the number of basis functions, which is sufficient given the moderate basis dimensions considered; for substantially larger bases, additional truncation would effectively correspond to a PCA-type post-smoothing step and is therefore left for future work. We consider sample sizes of $n \in \{300,1000\}$ observations with  $p \in \{3,50\}$ coordinate functions. 
The two dimensions are chosen to cover distinct regimes: $p = 3$ represents the low-dimensional case in which the vectorized covariance can still be estimated reliably if the sample size is large enough, whereas $p = 50$ is the regime the proposed structured estimator targets, where vectorization inflates the dimension to $pq = 5000$ and becomes the binding constraint. We therefore report $p = 50$ in the main text and the corresponding $p = 3$ results in Appendix~\ref{appendix:simulations_gaussian}.
For the covariance structure between the coordinate functions, we adopt the matrices proposed by \citet{agostinelli2015robust}, denoted by $\bm{\Sigma}^{\row}$, which have random entries and typically yield low to moderate correlations.
For the covariance function $\kappa$ we consider both Ornstein-Uhlenbeck $\kappa_{\text{OU}}$ and Matérn-type $\kappa_{\text{Matérn}}$ covariance structures, defined as 
\begin{align*}
    \kappa_{\text{OU}}(s,t) = \sigma_1^2 \exp\left(\frac{-\abs{s-t}}{\sigma_2}\right)
    \quad \text{and} \quad
    \kappa_{\text{Matérn}}(s,t) = \sigma^2 \frac{2^{1-\nu}}{\Gamma(\nu)}(\tau\abs{s-t})^{\nu}K_{\nu}(\tau\abs{s-t}),
\end{align*}
respectively. For the clean data, we use parameters $\sigma_1 = \sqrt{0.3}$ and $\sigma_2 = 0.3$ for the Ornstein-Uhlenbeck covariance function, and $\sigma = 1$, $\tau = 5$, and $\nu = 0.5$ for the Matérn-type covariance function. Except for the isolated outliers, the mean function of every coordinate is given by $\mu_j(t) = 30t(1-t)^{1.5}$, for $j = 1,\dots,p$. Similar choices for the mean and covariance function were considered by \cite{arribas2014shape}, \cite{dai2018multivariate}, and \cite{oguamalam2024minimum} for outlier detection in univariate functional data. Three contamination scenarios are examined, with $\varepsilon_{\mathrm{cord}} \in \{0.1, 0.5, 1\}$ denoting the proportion of contaminated coordinate functions, to evaluate the impact of contamination level on robustness and detection accuracy. Outliers are added to the datasets by randomly replacing a fraction $\varepsilon \in \{0.05,0.1,0.2,0.3\}$ of the clean observations for $p=3$, and, due to the higher computational cost, $\varepsilon \in \{0.1,0.3\}$ for $p = 50$. We consider four types of outliers: shift, shape, covariance-induced, and isolated, which are defined as follows: %

\emph{Shift outliers} are created in $\lfloor{\varepsilon_{\mathrm{cord}} \cdot p\rfloor}$ randomly chosen coordinates by introducing perturbations along the first eigenfunction $\ueig_1$, capturing the largest mode of variation, while \emph{shape outliers} are created by perturbing along the tenth eigenfunction $\ueig_{10}$, affecting local features without a global shift of the functions. The coordinate functions are given by 
$X_j^{\text{shift}}(t) = X_j + \lambda_{\mathrm{shift}} \ueig_1$ 
and  
$X_j^{\text{shape}}(t) = X_j + \lambda_{\mathrm{shape}} \ueig_{10}$, 
for $j = 1, \dots, p$. For $p = 3$, the perturbation magnitudes are set to $\lambda_{\mathrm{shift}} \in \{15, 22.5, 30\}$ and $\lambda_{\mathrm{shape}} \in \{5, 8, 11\}$, whereas for $p = 50$, they are $\lambda_{\mathrm{shift}} \in \{8, 10, 12\}$ and $\lambda_{\mathrm{shape}} \in \{1.5, 2, 2.5\}$. These magnitudes are scaled according to the dimensionality $p$ so that smaller values yield subtle, harder-to-detect outliers, while larger values introduce more pronounced deviations. As the outlier magnitude increases, the resulting contamination exerts a stronger influence on the sample mean and covariance structure.

To introduce \emph{covariance-induced outliers}, we modify the covariance structure of $\lfloor{\varepsilon_{\mathrm{cord}} \cdot p\rfloor}$ randomly chosen coordinates using the Matérn-type covariance function. By altering the smoothness parameter $\nu \in \{0.1,0.2,0.5\}$ and range parameter $\tau \in \{7,10,15\}$, we generate functions with unusually high variability or erratic behavior compared to the regular observations.

The final setting considers \emph{isolated outliers}, which deviate from the regular observations in $\lfloor{\varepsilon_{\mathrm{cord}} \cdot p\rfloor}$ randomly chosen coordinates at randomly selected time points per coordinate, while the remaining function remains unchanged. The mean of the clean data is $\mu_j = 4t$ while the outliers have random mean functions $\mu_j = 4t + \lambda(-1)^{u}\left(1.8 - \frac{1}{\sqrt{0.02\pi}}\exp\left(\frac{-(t-\alpha)^2}{0.02}\right)\right)$, $\lambda \in \{0.2,0.5,0.8\}$, for $j = 1,\dots,p$. Here $u \sim Bernoulli(0.5)$ is a Bernoulli random variable, and $\alpha \sim U(0.25,0.75)$ follows a uniform distribution.
In contrast to the first three settings, the mean function is random, and hence, the outliers neither follow the same distribution nor form a cluster. 
Figure~\ref{fig:outlier_examples} visualizes the first coordinate function for each of the four considered outlier types.
\begin{figure}[!h]
    \centering
    \includegraphics[width=1\linewidth]{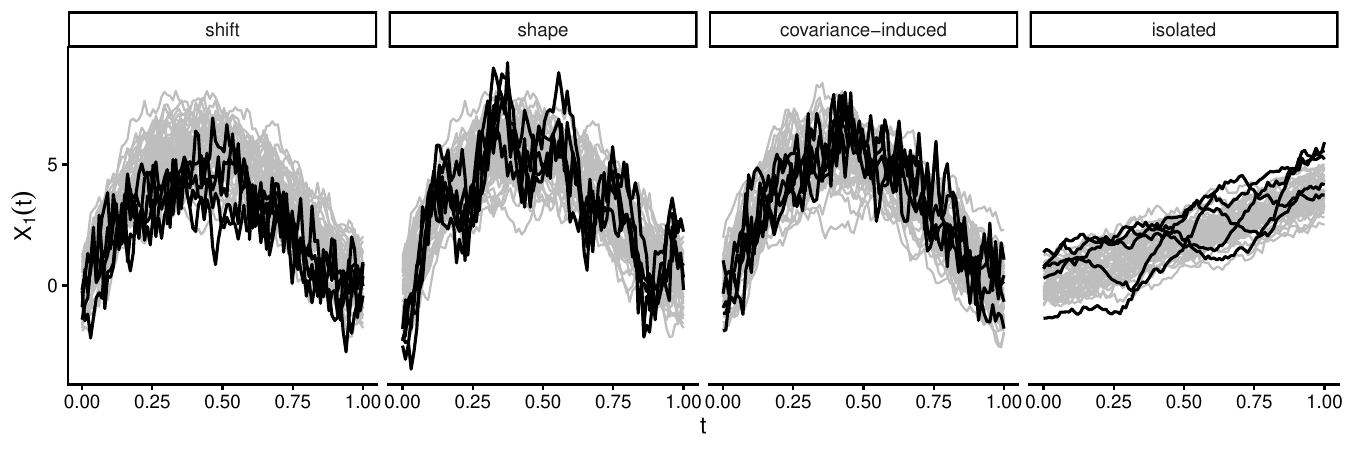}
    \caption{The first component %
    of the multivariate process $\bm{X}(t) = (X_1(t),X_2(t),X_3(t))$, for the 4 considered outlier types. %
    Each plot shows 45 clean sample curves (gray) and 5 outliers (black). For shift outliers $\lambda_{\mathrm{shift}} = 15$ and shape outliers $\lambda_{\mathrm{shape}} = 5$, and $\kappa = \kappa_{\text{OU}}$. For the covariance-induced outliers, $\nu = 0.2$ and $\tau = 10$, and for the isolated outliers, $\lambda = 0.5$.}
    \label{fig:outlier_examples}
\end{figure}

To provide a more comprehensive evaluation, we extend the simulation design beyond the separable and Gaussian setting. The non-separable processes are obtained by adding $q_{ns} \in \{1,\dots,4\}$ independent separable processes with Matérn covariance function. Thus, the resulting non-separable covariance is a sum of $q_{ns}$ separable components. Due to the already large number of simulation settings, here we focus only on shift outliers and generate outliers in each of $p=10$ coordinates by shifting the $q_{ns}$th component in the direction of its $q_{ns}$th eigenfunction. To include a non-Gaussian setting, we generated data from t-distributions with varying degrees of freedom. 

Each simulation setting is replicated $100$ times. Table \ref{tab1} summarizes all the parameter combinations considered in the study.

\begingroup
\setlength{\abovecaptionskip}{0pt}
\setlength{\belowcaptionskip}{0pt}
\spacingset{1}
\begin{table}[!h]
\caption{Simulation setup. All settings are replicated 100 times, except MFIF
for $p=50$, which uses 10 replications.}
\centering
\small
\begin{tabular}[t]{llll}
\toprule
 & Separable (Gaussian) & Non-separable & Heavy-tailed ($t$) \\
\midrule
$n$ & 300, 1000 & 1500, 5000 & 1000 \\
$p$ & 3, 50 & 10 & 10 \\
$\varepsilon$ & 0.05, 0.1, 0.2, 0.3 ($p=3$) & 0.05, 0.1, 0.2, 0.3 & 0.1, 0.3 \\
              & 0.1, 0.3 ($p=50$) & & \\
$\varepsilon_{\mathrm{coord}}$ & 0.1, 0.5, 1 & 1 & 0.1, 0.5, 1 \\
$m$ & $q$ (raw), 10, 20, 30 & $q$ (raw), 10, 20, 30 & $q$ (raw), 10, 30 \\
Covariance function & Matérn, Ornstein--Uhlenbeck & Sum of $q_{ns}$ Matérn & Matérn \\
$q_{ns}$ & -- & 1, 2, 3, 4 & -- \\
Degrees of freedom & -- & -- & 3, 5, 8, 12, 15 \\
Outlier magnitude & small, medium, large & small, large & small, large \\
Outlier type & shift, shape, isolated, & shift & shift \\
             & covariance-induced & & \\
\bottomrule
\end{tabular}\label{tab1}
\end{table}
\endgroup

\subsection{Results}
The outlier detection performance of the methods is compared based on precision, recall, and their harmonic mean, i.e., the F-score \citep{TahaHanbury2015}. To provide a threshold-independent performance measure, we include the AUC (Area Under the ROC Curve) score, which measures how well a model separates inliers from outliers across all possible thresholds, ranging from 0.5 (random) to 1 (perfect). For all methods compared in the simulation study, we apply the default thresholding strategies used in their respective \texttt{R} implementations to classify observations as outliers. Specifically, for the depth-based methods, we compute the local outlyingness and classify a curve as outlying if it is flagged as such at any observed time point. For the MS-plot approach, we use the default threshold based on the $F$-distribution, provided by the \texttt{msplot} function from the \texttt{fdaoutlier} package. The MFIF method does not yield an outlyingness indicator, but only an anomaly score, and is only evaluated based on AUC. It is also the most computationally demanding method across all settings and is therefore evaluated on 10 rather than 100 replications for $p = 50$, which excludes it from the rank-based comparison below. A detailed analysis of computation times is given in Table~\ref{tab:time} and
Figure~\ref{fig:time} in Appendix~\ref{appendix:simulations}. 
For both the MMLE and the proposed MMCD-based Mahalanobis distance, we employ a $\chi_{0.99}^2$-cutoff on the Mahalanobis distance values, following standard robust multivariate practice. 

To evaluate mean and covariance estimation, we consider
\begin{align*}
    \frac{1}{p\abs{\T}}\int \norm{\bm{\mu}(t) - \hat{\bm{\mu}}(t))}_2^2 dt 
    \quad\text{and}\quad
    \frac{1}{(p\abs{\T})^2}\iint \norm{\bm{\Sigma}^{\row} k(s,t) - \hat{\bm{\Sigma}}^{\row} \hat{k}(s,t))}_F^2 ds dt,
\end{align*}
respectively. Here $\abs{\T} = \max(\T) - \min(\T)$ denotes the length of the interval $\T$. For the distance-based methods, these errors can be computed based on the estimated parameters. Since the depth-based approaches are non-parametric, we first apply the outlier detection methods to identify and remove the outliers before robustly computing the maximum likelihood estimates on the cleaned subset. Additionally, relative errors are computed by dividing the estimation errors of each method by the benchmark estimation error attained by the ML estimates computed on the clean data.

Among the depth-based methods, the average results across most settings are only slightly influenced by smoothing. Therefore, only the results based on the raw data are reported. For a comparison of the depth-based methods, see Appendix~\ref{appendix:simulations_depth}.
For the distance-based methods, smoothing has a more apparent influence, and results for both raw and smoothed data are reported. 

To get an overview of the overall performance for the separable simulation settings from Table \ref{tab1}, we compare the methods by ranking them according to the given performance measures.  %
All ranks are ordered such that the method with the lowest rank is best. The dots show the average ranks, and the intervals are based on non-parametric multiple comparisons using the Friedman and the post-hoc Nemenyi tests; see \cite{hollander2013nonparametric} for details. %
The results for $p=50$ are given in Figure~\ref{fig:test_p50} and a similar plot for $p = 3$ is given in Figure \ref{fig:test_p3} of Appendix \ref{appendix:simulations_gaussian}.
For shift outliers, the robust distances computed on the smoothed data perform best across all scores. For the shape outliers, the robust distances computed on the smoothed data work best for outlier detection, but the raw robust covariance estimates are superior to the smooth ones. The classical raw covariance estimator performs well because outliers affect the 10th eigenfunction. Since the MS plot hardly detects outliers here, its covariance performance is similar to the classical estimator. 
For isolated and covariance-induced outliers, distance-based methods work significantly better than depth-based approaches. Outlier detection based on non-robust distances is very precise but lacks in recall, which is also reflected in higher covariance estimation errors compared to their robust counterparts. For isolated outliers, the number of basis functions has a strong influence on outlier detection, with better performance for a higher number of basis functions.

\begin{figure}[!h]
    \centering
    \includegraphics[width=1\linewidth]{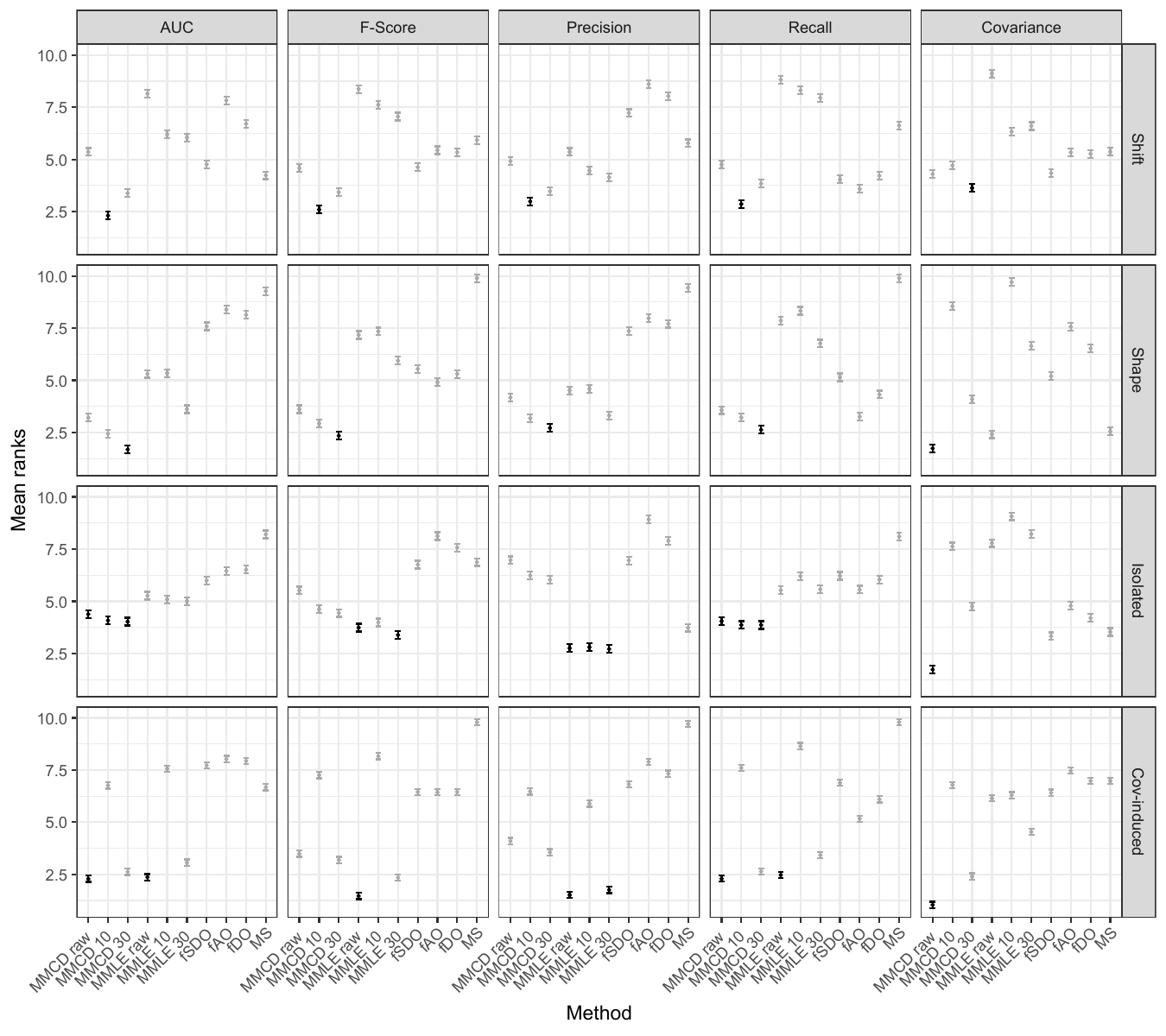}
    \caption{Rank-based comparison of methods across all separable simulation settings in Table \ref{tab1} for $p=50$. Intervals are based on the Friedman and Nemenyi tests; methods not significantly different (99\% level) from the best are shown in black, others in gray. Horizontal facets correspond to performance metrics, vertical facets to outlier types. MFIF is omitted, because the rank test requires the same number of replications for every method.}
    \label{fig:test_p50}
\end{figure}

\begin{figure}[p]
    \centering
    \includegraphics[width=1\linewidth]{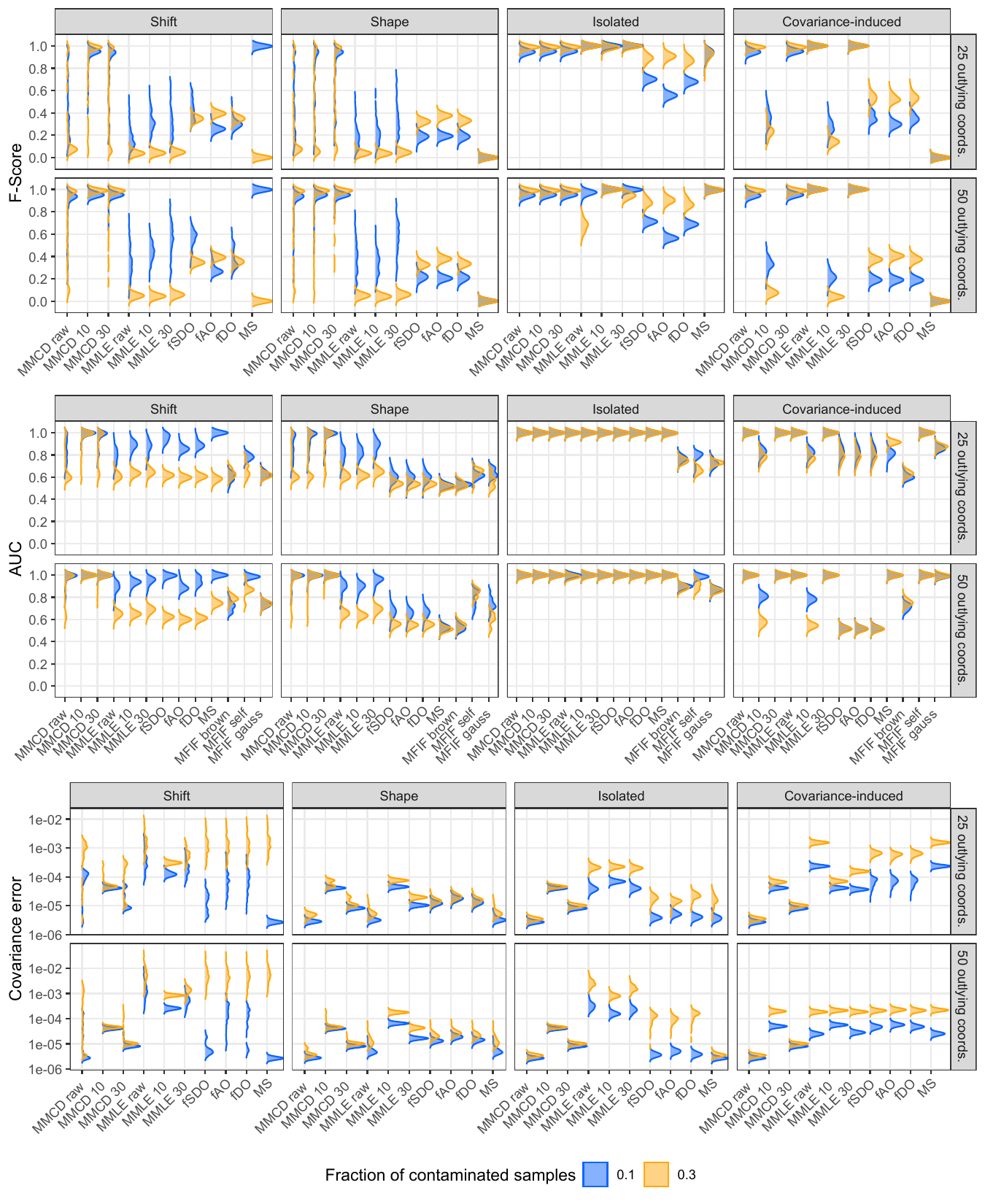}
    \caption{Density plots of F-score, AUC, and log covariance estimation error for a representative setting with $n = 1000$ and $p = 50$ under the Matérn covariance and medium outlier magnitude, shown across outlier types (shift, shape, isolated, covariance-induced), contamination levels (25 and 50 coordinates), and outlier proportions ($\varepsilon = 0.1, 0.3$). MFIF appears only in the AUC panel, as it provides no outlyingness cutoff, and its densities are based on 10 instead of 100 replications.%
    }
    \label{fig:gaussian_main}
\end{figure}

The rank-based comparison provides an overview across all settings, but does not show details for specific settings.
Figure~\ref{fig:gaussian_main} shows density plots of the F-score, AUC, and covariance estimation error for one representative setting, the parameter choices are described in the figure caption. Similar plots across the remaining settings are given in Appendix \ref{appendix:simulations_gaussian}. While this plot confirms the findings of the rank-based comparison, it also highlights settings where performance varies across parameter settings and simulation runs. For example, the MS plot performs well for shift outlier detection in settings with low contamination, but it breaks down for higher contamination. Moreover, in most settings, we see a clear connection between F-score and AUC across all methods. However, the depth-based methods for isolated outliers are an exception to that. While the depths yield a correct ordering of the data, as is evident by the high AUC, the default cutoff of the methods fails to flag those samples as outliers. This issue is even more pronounced for the MS plot for covariance-induced outliers. 
The MFIF, shown for the AUC only, depends strongly on the choice of dictionary and matches the robust distances for covariance-induced outliers with the Self dictionary, but falls clearly below the competing methods for isolated outliers. 
Although the robust distances generally perform better than the competing methods for detecting shift and shape outliers, we observe several simulation runs with relatively low F-scores. This tends to occur in scenarios where only a few coordinates are contaminated. Outliers can only be identified when the contaminated samples have a sufficiently strong impact on the mean and covariance to make them distinguishable through the Mahalanobis distance. Because the same shift is applied to all coordinates and the covariance function is computed as an average over them, smaller contaminations are often masked. At the same time, these are precisely the cases in which outliers have little effect on covariance estimation, raising the question of whether such modified observations should truly be considered outliers. This pattern is further reflected in the covariance estimation errors: contaminated settings in which the robust methods show larger errors tend to correspond to smaller errors for the classical estimators.

Figure~\ref{fig:AUC_non_separable_p10} presents density plots of AUC values for outlier detection in a representative non-separable setting ($p = 10$, $n = 1500$). The results are displayed across different numbers of non-separable components and varying magnitudes of shift contamination, and are further grouped by the fraction of contaminated observations.\\ %
Interestingly, the outlier detection performance of the proposed approach appears largely unaffected by the introduction of non-separability, particularly when the data are smoothed. 
This suggests that the proposed distance-based procedure remains effective even when the separability assumption is violated. Additional results supporting this conclusion are provided in Appendix~\ref{subsection:sample_vs_mmle}.
Moreover, results in Figures \ref{fig:gaussian_main}-\ref{fig:AUC_non_separable_p10}, as well as complementary results in Appendix~\ref{appendix:simulations}, indicate that the proposed approach is not particularly affected by the number of B-splines used when smoothing the data, as long as the data are not oversmoothed.

\begin{figure}[!h]
    \centering
    \includegraphics[width=1\linewidth]{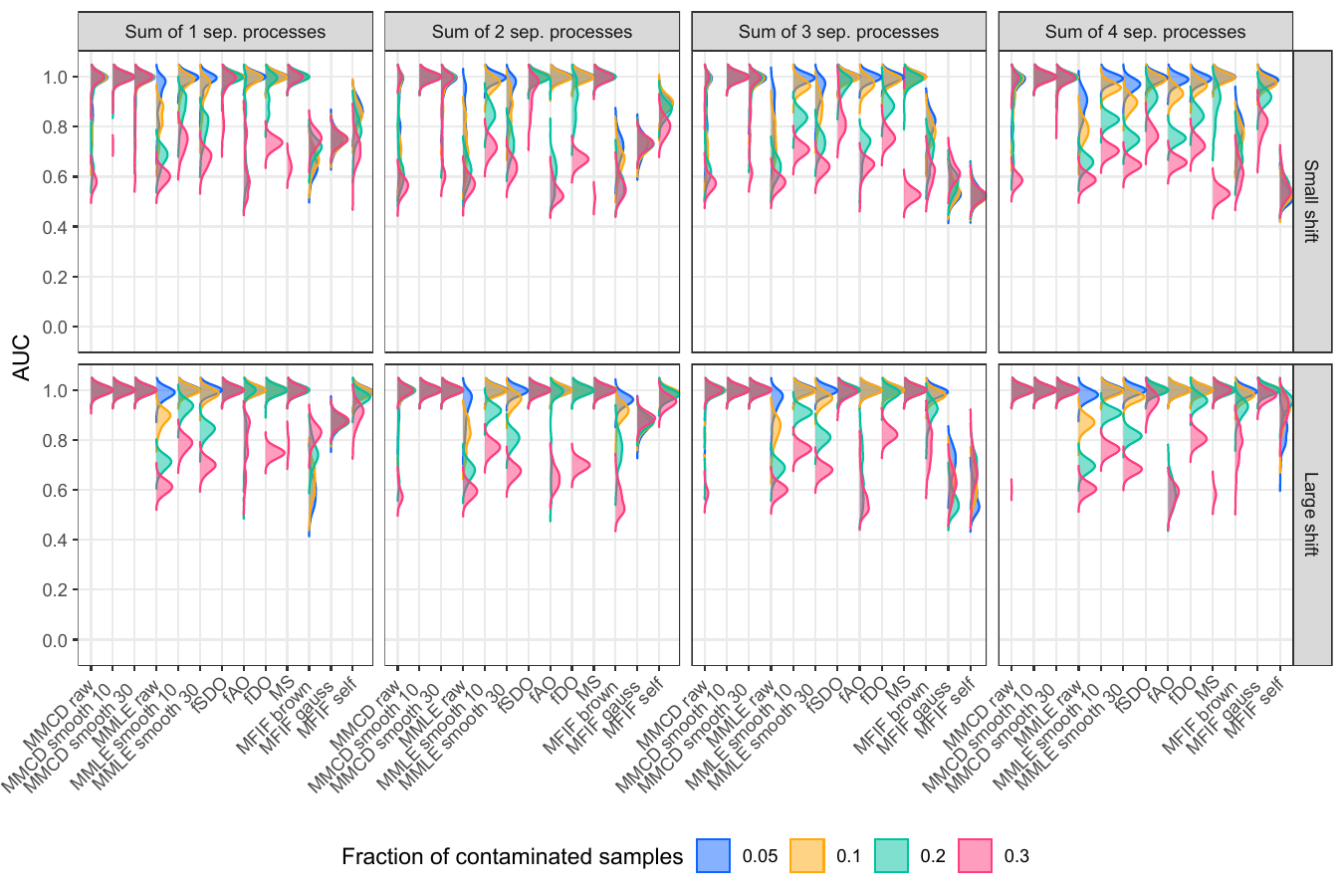}
    \caption{AUC values for detecting small (upper row) and large (lower row) shift outliers in non-separable processes obtained as sums of $q_{\text{ns}} = 1, 2, 3, 4$ (columns) independent separable components with Matérn covariance functions and $p=10$ components.}
    \label{fig:AUC_non_separable_p10}
\end{figure}

The results in Appendix~\ref{supplement:non_gaussian} show that under heavy-tailed distributions, all methods exhibit reduced precision, as observations from the distribution tails are more difficult to distinguish from true outliers. While robust Mahalanobis-based methods continue to achieve high AUC and recall, they tend to overflag tail observations and are, in that manner, somewhat less precise than depth-based alternatives. Nevertheless, the overall qualitative conclusions are similar to the Gaussian setting, while the performance improves when contamination is confined to a subset of coordinates.

\subsection{Evaluating Shapley values} \label{subsection:shapley_sim}

To assess the ability of the Shapley values to localize outlyingness, we evaluate their performance separately with respect to the affected coordinates and the affected regions of the functional domain. For coordinate-specific localization, we compute the AUC of the component-specific Shapley values and compare it with the AUC obtained from the pointwise absolute deviation from the uncontaminated population mean curve and from the local outlyingness scores of the MS-plot. For time-specific localization, the functional domain is partitioned into varying numbers of intervals, and AUC is used to assess how accurately the methods identify the intervals containing the contamination. We compare the time-specific Shapley values with component-wise absolute deviations and local outlyingness measures of the depth-based methods, which are provided at the observed time points. As a benchmark for the Shapley values obtained using the robust MMCD estimator with (d=30) basis functions, we also report results based on the true population mean and covariance structure.\\
Given the focus on local behavior, we consider isolated outliers generated by adding sine-wave perturbations of varying length and frequency to a random subset of $(p_0 \in {1,\dots,9})$ coordinate functions. Results for a representative setting with $n=1000$ observations, $p=10$ coordinates, and medium outlier magnitude are presented in Figure~\ref{fig:shapley_main}. The left column of the figure shows coordinate-specific localization, whereas the right column shows time-specific localization for partitions of the functional domain into $5$, $10$, and $100$ intervals. The rows correspond to the two covariance structures considered. The method labeled \textit{Shapley} refers to the Shapley values computed using the true population mean and covariance structure, whereas \textit{Shapley MMCD 30} refers to the Shapley values based on the robust MMCD estimates obtained using 30 basis functions.

\begin{figure}[!h]
    \centering
    \includegraphics[width=1\linewidth]{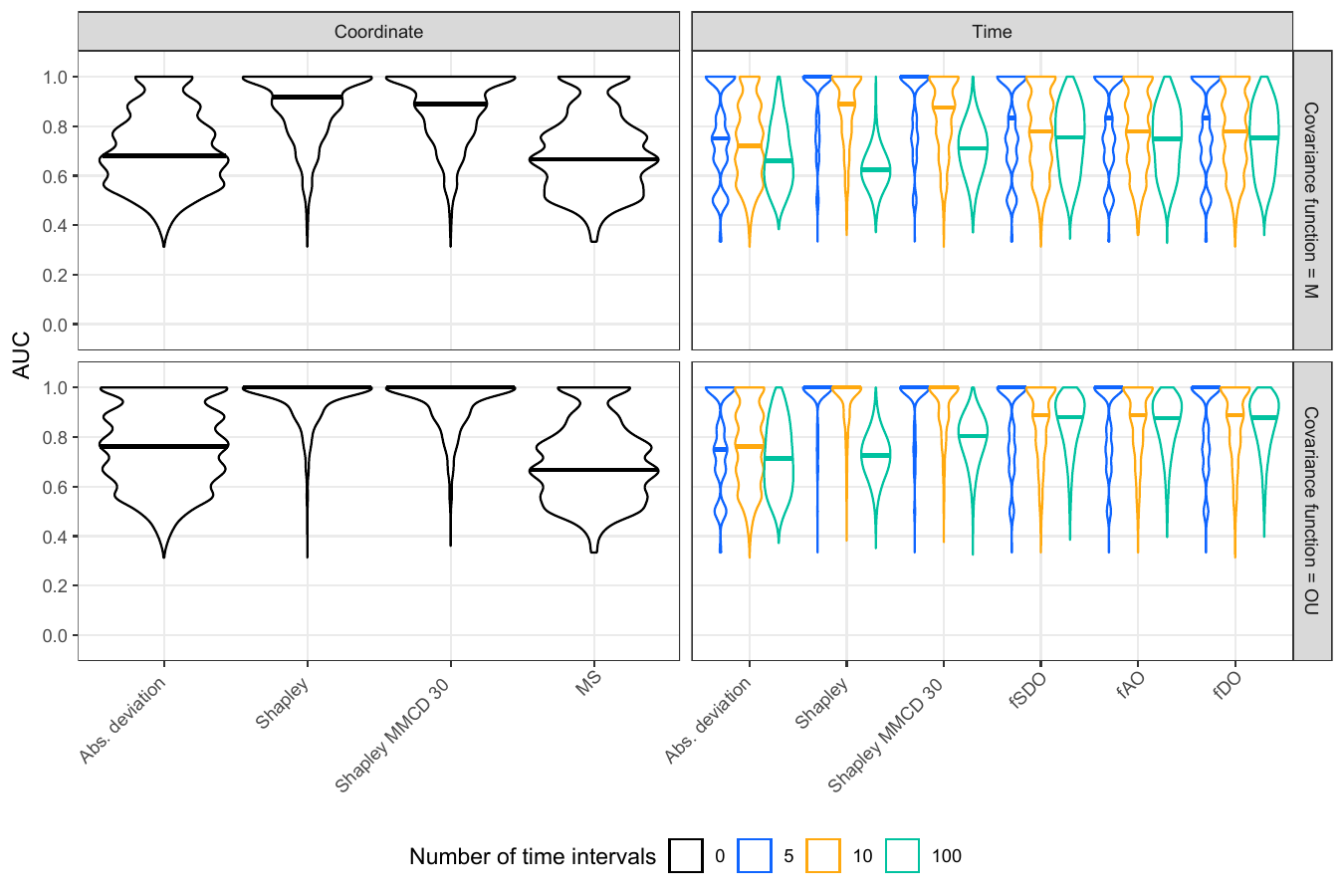}
    \caption{
    Mirrored density plots of AUC values for localization of isolated outliers. The left column evaluates coordinate-specific localization, while the right column evaluates time-specific localization using partitions of the functional domain into $5$, $10$, or $100$ intervals. Rows correspond to the Matérn (M) and Ornstein–Uhlenbeck (OU) covariance structures. The value $0$ in the legend denotes coordinate-specific evaluation without partitioning of the time domain. Higher AUC indicates better localization. Results are shown for $p=10$ coordinates and medium outlier magnitude. 
}
    \label{fig:shapley_main}
\end{figure}

The left column of Figure~\ref{fig:shapley_main} shows that the Shapley value-based approach provides the most accurate coordinate-specific localization of outlyingness, outperforming both the MS-plot and the naive coordinate-wise absolute deviation. Importantly, the Shapley values based on the robust MMCD covariance estimate perform comparably to those computed using the true population covariance, indicating that estimation of the covariance structure introduces only a small loss in localization accuracy.\\
The right column of Figure~\ref{fig:shapley_main} illustrates the effect of the temporal resolution used for localizing outlyingness. At the finest resolution, corresponding to 100 intervals, none of the methods achieves particularly high AUC values, and the Shapley-based approach performs similarly to, or slightly worse than, the depth-based methods, depending on the covariance structure. As the number of intervals decreases, however, the Shapley values identify the contaminated regions more accurately and generally outperform the competing methods.\\
This behavior is consistent with the temporal dependence inherent in functional data. When the domain is partitioned into many narrow intervals, a local perturbation is not necessarily confined to a single interval, since neighboring regions can exhibit correlated deviations due to the smoothness and covariance structure of the underlying process. Because the Shapley decomposition explicitly accounts for this dependence, outlyingness may be distributed across the contaminated interval and its neighbors, reducing the AUC when exact interval-wise localization is used as the target. With coarser partitions, these correlated deviations are more likely to be contained within the same interval, allowing the Shapley values to localize the affected region more reliably. Thus, Figure~\ref{fig:shapley_main} highlights a trade-off between temporal resolution and localization accuracy, while also demonstrating the advantage of the proposed approach at practically relevant levels of temporal aggregation.

\section{Real data applications}\label{section:examples}
We would like to emphasize that, although separability may not hold exactly in the real data examples, we use the separable estimator because the sample sizes ($N=75$ and $N=115$) are relatively small to moderate. In this setting, the reduction in variability can outweigh the bias caused by possible model misspecification. This is also supported by our simulation results, where the separable estimator often performs better than the unconstrained estimator in the non-separable model, for similar sample sizes.
\paragraph{El~Niño–Southern Oscillation:}
The \textit{El~Niño–Southern Oscillation} (ENSO) is a recurring climate phenomenon characterized by periodic changes in sea surface temperature (SST) and atmospheric pressure in the equatorial Pacific Ocean \citep{trenberth1997definition}. ENSO strongly influences global temperature and precipitation patterns and is typically classified into three phases: \textit{El~Niño} (warm), \textit{La~Niña} (cold), and the intermediate \textit{Neutral} state.\\
We analyze monthly SST data from the U.S.~Climate Prediction Center (CPC), based on the \textit{Extended Reconstructed Sea Surface Temperature, Version~5} (ERSSTv5) dataset \citep{huang2017extended}. SST is measured in four standard Niño regions (Niño~1+2, Niño~3, Niño~3.4, and Niño~4), shown in Figure~\ref{fig:map} in Appendix~\ref{appendix:el_nino}. The CPC defines the onset of an El~Niño (La~Niña) episode when the three-month running mean SST anomaly in the Niño~3.4 region exceeds $\pm 0.5^{\circ}$C relative to a 30-year baseline, known as the \textit{Oceanic Niño Index};  see \citep{trenberth1997definition} for more details. Following common practice, months are grouped into annual periods spanning June to May.\\
The dataset covers 75 periods from 1950–1951 to 2024–2025, each consisting of 12 monthly SST measurements per region. Overall, the dataset includes 585 neutral months, 223 La~Niña months, and 92 El~Niño months, corresponding to 19 neutral, 36 La~Niña, and 20 El~Niño periods.

As noted by \citet{hanley2003quantitative}, the various Niño indices capture different aspects of ENSO variability. The Niño~4 index tends to show relatively weak warming during canonical El~Niño events but pronounced and persistent cooling during La~Niña phases. In contrast, the Niño~1+2 region exhibits the earliest and strongest warming during El~Niño and sharp, short-lived cooling during La~Niña. Niño~3 and Niño~3.4 are most representative of canonical El~Niño episodes. These regional differences indicate that reliance on a single index may obscure important spatial or phase-dependent patterns of ENSO variability. Therefore, we compare a univariate analysis based on the Niño~3.4 region with a multivariate approach that jointly incorporates sea surface temperature anomalies from all four Niño regions.

As a first step, the raw SST data are smoothed using six B-spline basis functions without a roughness penalty. Since each functional observation represents one annual cycle from June to May, this choice is made  to capture the four seasonal periods, along with two additional basis functions that accommodate transitions across annual boundaries. An overview of the resulting smooth curves is provided in Figure~\ref{fig:plt_overview_smooth2} in Appendix~\ref{appendix:el_nino}.

To assess the model assumptions, the separability of the covariance structure is tested using the procedure of \citet{aston2017tests} and is rejected at the 5\% significance level (p-value$~<0.001$). Nevertheless, results from the non-separable simulation study in Section \ref{section:simulations} indicate that, even when separability is violated, the robust Mahalanobis distance computed under a separable covariance approximation remains an informative measure of outlyingness. Because the Gaussianity assumption is used solely to define the theoretical cutoff for outliers, we further examine whether the empirical Mahalanobis distances follow a $\chi^2$ distribution. This hypothesis is rejected by the Kolmogorov-Smirnov test, likely due to the presence of outliers and the test’s sensitivity to departures from normality. However, the Q--Q plot comparing the empirical quantiles of the estimated robust squared Mahalanobis distances with the theoretical $\chi^2$ quantiles shows that the potential inliers align closely with the reference distribution, while deviations occur primarily in the upper tail; see Figure~\ref{fig:plt_qq_all}. This pattern supports the use of the $\chi^2$ reference as a practical diagnostic baseline for identifying outlying observations.

We then compare the robust functional Mahalanobis distance ($\fmd$) based on the SST data from the \textit{Niño~3.4} region and its multivariate extension ($\fmmd$) using SST data from all four regions.

\begin{figure}[!h]
    \begin{minipage}[t]{0.33\textwidth}
        \vspace*{0mm}
        \centering
        \includegraphics[width=1\linewidth]{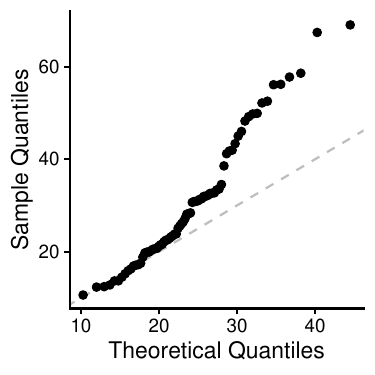}
        \caption{Q--Q plot of robust squared Mahalanobis distances against $\chi^2$ quantiles for all $75$ periods.}
        \label{fig:plt_qq_all}
    \end{minipage}
    \hspace{0.02\linewidth}
    \begin{minipage}[t]{0.64\textwidth}
        \vspace*{0mm}
        \includegraphics[width=1\linewidth]{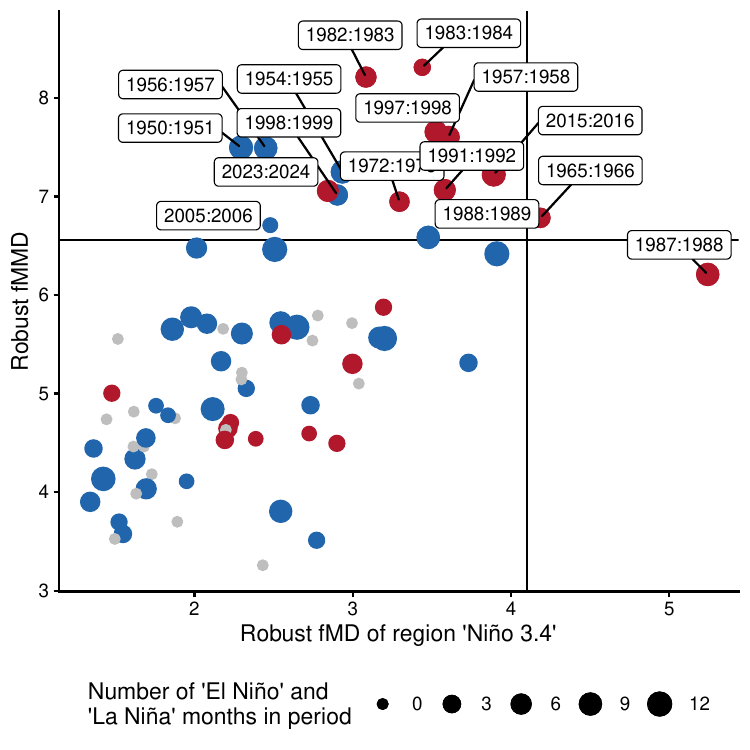}
        \caption{Distance-distance comparison of the robust $\fmd$ based on the \textit{Niño~3.4} region and the $\fmmd$ based on all four Niño regions. Vertical and horizontal lines indicate the $0.99$ quantile cutoffs of the corresponding $\sqrt{\chi^2}$ distributions with 6 and 24 degrees of freedom, respectively. Points represent 75 annual periods from 1950:1951 to 2024:2025, sized by the number of El~Niño or La~Niña months and colored by phase (red: El~Niño; blue: La~Niña; gray: neutral).}
        \label{fig:plt_uni_vs_multi}
    \end{minipage}
\end{figure}

Figure~\ref{fig:plt_uni_vs_multi} shows that the multivariate treatment of the SST data identifies more unusual observations, many of which correspond to strong El~Niño periods that are not flagged as outlying in the univariate analysis. This suggests that incorporating SST information from all four Niño regions provides a more sensitive framework for detecting extreme ENSO episodes.
To compare the outlier detection performance with competing methods in a way that is not influenced by the choice of cutoff, we compute the Spearman rank correlation between the estimated outlyingness measures and a numerical proxy for ENSO activity, defined as the highest temperature anomaly per period based on the Niño~3.4 index. Table~\ref{tab:correlations_enso} summarizes the results and shows that the proposed robust $\fmmd$ shows the strongest association ($\rho = 0.65$), exceeding that of all competing methods.

\begin{table}[!ht]
\centering
\caption{Spearman rank correlations between estimated outlyingness measures and the ENSO activity proxy based on the Niño~3.4 index.}
\label{tab:correlations_enso}
    \begin{tabular}{lccccccccc}
    \toprule
     & \thead{Robust\\$\fmmd$} & \thead{Robust\\$\fmd$} & \thead{fSDO} & \thead{fAO} & \thead{fDO} & \thead{MFIF\\self} & \thead{MFIF\\brown} & \thead{MS} & \thead{MFIF\\gauss} \\
    \midrule
    $\rho$ & \textbf{0.65} & 0.56 & 0.54 & 0.55  & 0.55  & 0.49  & 0.34  & 0.07  & 0.04 \\
    \bottomrule
    \end{tabular}
\end{table}

We extend the analysis beyond identifying outlying periods by using the Shapley decomposition to highlight the temporal and regional contributions responsible for deviations from the global SST trend. The corresponding Shapley values for El~Niño and La~Niña years are displayed in Figure~\ref{fig:shapley_lines_outliers}, where the color intensity indicates the magnitude of outlyingness, and the color hue reflects whether the deviation is above (red) or below (blue) the global mean. Figure~\ref{fig:sup_nino} in the Appendix provides Shapley values for representative El~Niño (2015:2016) and La~Niña (1956:1957) years. Figure~\ref{fig:shapley_lines_outliers} reveals that, for El~Niño episodes, outlyingness is primarily driven by elevated temperatures in the Niño~3.4 and Niño~3 regions, whereas for La~Niña episodes, it is mainly associated with cold anomalies in the Niño~4 and Niño~3 regions.
 For example, for 2015:2016, the largest contributions arise from the Niño~3.4 region in autumn (Aug–Dec) and winter (Jan–Apr), with additional contributions from Niño~3 during late summer and spring, consistent with the characteristic warming pattern of strong El~Niño events.

\begin{figure}[!h]
    \centering
    \includegraphics[width = 1\linewidth]{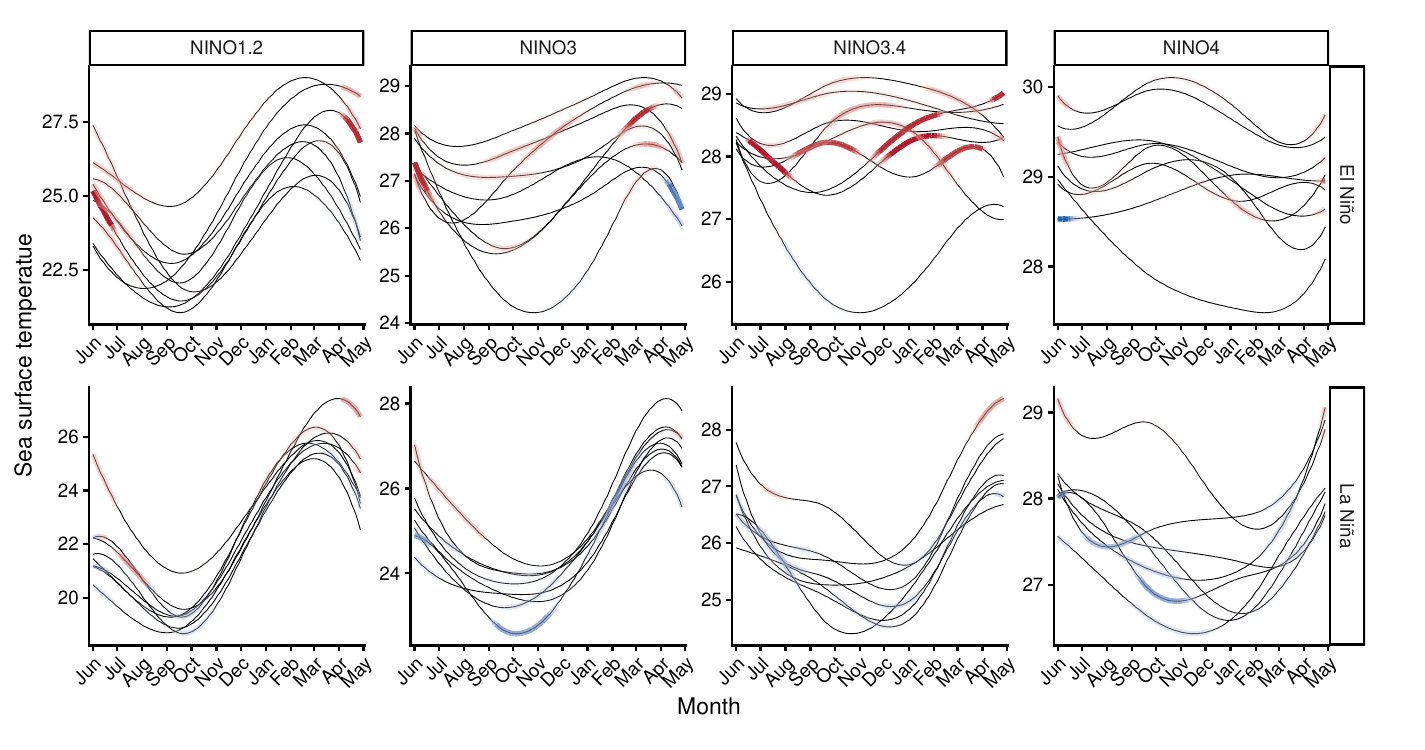}
    \caption{Smoothed SST curves across the four Niño regions (columns) during El~Niño and La~Niña years (rows). The color hue indicates whether the deviation is above (red) or below (blue) the global mean, and the color intensity reflects the magnitude of outlyingness.
    }
    \label{fig:shapley_lines_outliers}
\end{figure}

\paragraph{Resistance spot welding:}
To further demonstrate the effectiveness of the proposed outlier detection framework, we apply it to resistance spot welding (RSW) data from the automotive manufacturing industry. This dataset comprises $n = 115$ dynamic resistance curves (DRCs), which track electrical resistance during welding. Abnormal conditions, such as excessive current or low electrode pressure, can cause expulsion, which occurred in $70$ cases, corresponding to more than half of the samples. However, expulsion is an isolated effect in this dataset that does not affect all coordinates; we want to test whether our approach gives meaningful results in a cellwise high-dimensional setting. Since manual labeling is costly, we consider the functional outlier detection procedures from the simulation study to automatically detect expulsions. For more details on the dataset and the welding process, we refer to \cite{centofanti2025cellwise} and the references therein. 

Rather than treating the data in this original format as a sequence of 5 welding points at 750 time points, as in \cite{centofanti2025cellwise}, we instead regard it as multivariate functions with five coordinate functions, each representing a welding point recorded at 150 time points; see Figure~\ref{fig:welding_curves_and_shapley}. Applying the same testing procedures as in the ENSO example, we reject both the separability and Gaussianity assumptions. Further details are provided in Appendix~\ref{appendix:welding}. Nonetheless, we proceed with the analysis described below.

In Figure~\ref{fig:welding_roc}, we compare the methods based on the ROC curves and their AUC values. The robust distances achieve the highest AUC (0.95), matching the cellwise robust high-dimensional approach proposed by \cite{centofanti2025cellwise}, while being considerably faster to compute. The multivariate functional isolation forest \citep{Staerman2019} achieves a comparable AUC of 0.94 with a Gaussian dictionary (MFIF gauss), but drops to 0.61 and 0.50 with the Brownian (MFIF brown) and Self (MFIF self) dictionaries. Because the dictionary must be chosen in advance without label information, this sensitivity poses a practical difficulty in an unsupervised setting. MMLE, the non-robust counterpart of MMCD, achieves an AUC of 0.89, indicating that robust estimation contributes noticeably to detection performance. The remaining methods reach AUC values between 0.50 and 0.70, suggesting that the contamination pattern in this dataset is not well captured by their underlying outlier notions.
For the distance-based methods, we smooth the data using $m=30$ B-spline basis functions. Visual inspection showed that the resulting curves were reasonably smooth, with no apparent oversmoothing. Varying $m$ between $10$ and $50$ confirms that the specific choice of $m$ is not critical once the smoothing error stabilizes at $m \geq 20$, with an AUC value between $0.94$ and $0.95$ for $m = 20, 25, \dots, 50$. 
This choice is also consistent with the simulation results, which show stable performance for such basis representations.

\begin{figure}[!h]
    \centering
    \includegraphics[width = 0.7\linewidth]{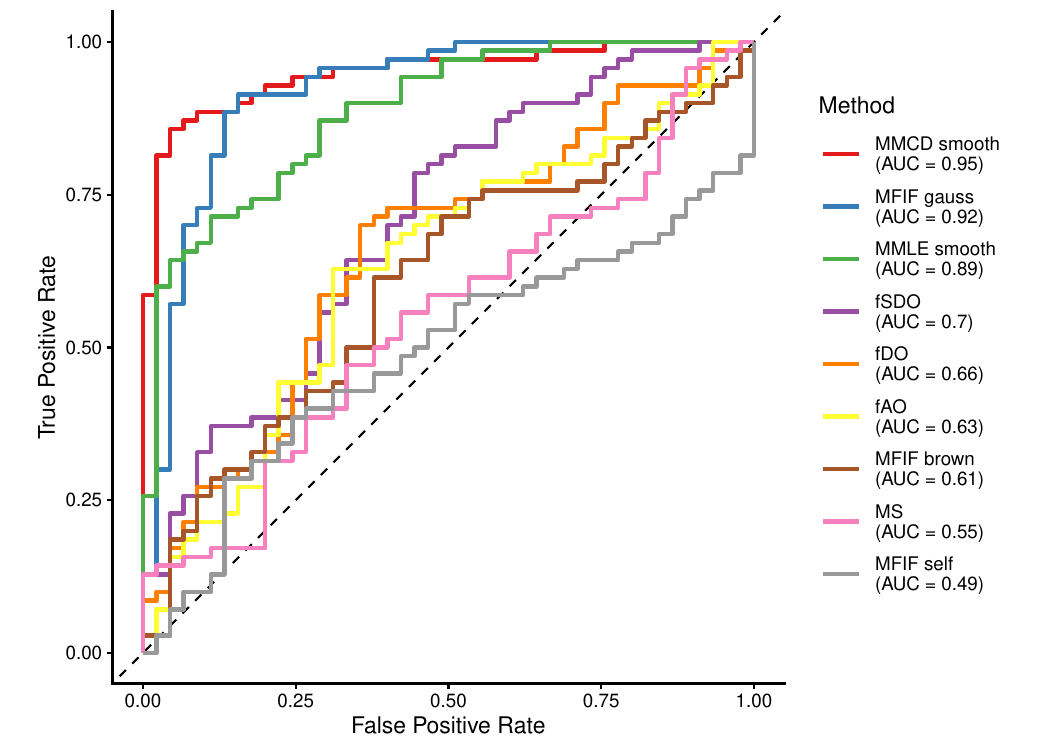}
    \caption{ROC curves and their AUC values to evaluate outlier detection performance in the resistance spot welding data.}
    \label{fig:welding_roc}
\end{figure}

As demonstrated in the ENSO example, we can apply Shapley values to identify which coordinate functions are outlying in which time intervals. Experts can use this information to increase efficiency in manually confirming the labels of new data, as well as to explain the outlyingness of existing samples. This is illustrated in Figure~\ref{fig:welding_curves_and_shapley}, where we computed time-coordinate-specific Shapley values for 10 time intervals for two samples with large robust Mahalanobis distances.

\begin{figure}[!h]
    \centering
    \includegraphics[width = 1\linewidth]{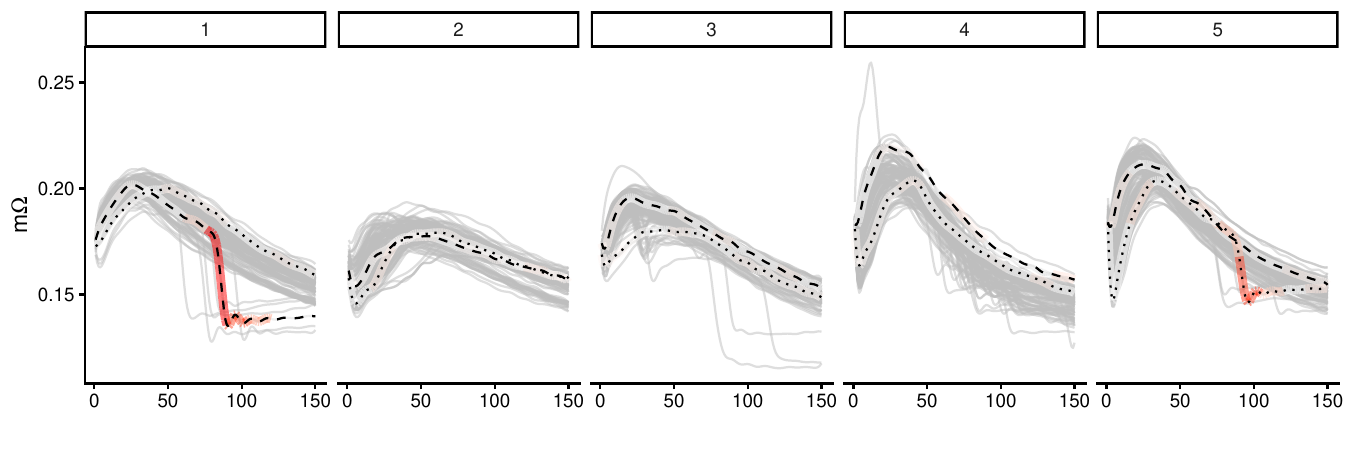}
    \caption{Smoothed welding DRCs with two of the outliers highlighted by a dashed/dotted line. The color intensity of the band around the two curves is based on time-coordinate-specific Shapley values computed at 10 equally sized time intervals.}
    \label{fig:welding_curves_and_shapley}
\end{figure}

\section{Discussion and Conclusions}\label{section:discussion}
We introduced a multivariate functional outlier detection approach based on the Mahalanobis distance, a method widely used for traditional multivariate data. While most existing approaches rely on depth-based measures \citep{hubert2015multivariate}, our framework provides a distance-based alternative grounded in classical multivariate analysis. Defining a robust Mahalanobis distance for multivariate functional data is nontrivial, particularly due to the need for robust estimation of the associated covariance structure.
We proposed a functional Mahalanobis (semi-)distance and established some of its key properties. It is affine invariant, and in the case of uncorrelated components, the proposed distance reduces to the sum of the univariate functional Mahalanobis distances of \cite{galeano2015mahalanobis}. Under the additional assumption of separability of the covariance operator, we showed how a multivariate random function relates in the distribution to the coefficient matrix obtained from a basis representation of the smoothed data. This connection enables a direct link between the trimmed multivariate functional Mahalanobis distance and the matrix-variate Mahalanobis distance, for which robust estimators have been developed in \cite{mayrhofer2024robust}.

The proposed robust distances can be computed either from raw data or from a smoothed version of the data. The simulations based on multivariate separable Gaussian processes demonstrate that both versions outperform established benchmark methods across various functional outlier types in the considered settings and validate the stability of the smoothed distances with respect to the number of basis functions. 

The simulation study under a misspecified model indicates additional regimes in which the proposed approach is very useful. Under non-separable covariance structures, the separable estimator could be interpreted as a structured approximation rather than as a correctly specified model. In the settings considered, this approximation often yields more accurate covariance estimates and more stable Mahalanobis diagnostics than the unregularized sample covariance on vectorized data, especially when the sample size is moderate relative to the dimension. This behavior reflects a bias--variance trade-off: the separability constraint introduces bias but can substantially reduce estimation variance. More flexible or regularized non-separable estimators may behave differently. Heavy-tailed simulations show reduced precision across all methods, as tail observations become harder to distinguish from outliers. Robust Mahalanobis-based methods maintain strong AUC and recall, while depth-based approaches are more precise. Overall, the conclusions remain similar to the Gaussian setting, though performance depends on both contamination structure and tail behavior.

A further contribution is outlier explainability using Shapley values, which allow for an additive decomposition of the univariate and multivariate squared functional Mahalanobis distance. The outlyingness contributions can be evaluated for the individual components of the multivariate functions, for non-overlapping time domains, or simultaneously as time-coordinate-specific contributions. 
The Shapley simulation supports the use of the proposed decomposition as an interpretability tool. In settings with localized contamination, the Shapley values reliably identify the affected coordinates and, when the time domain is partitioned into moderately sized intervals, also localize the relevant time regions more accurately than competing local outlyingness measures. For very fine partitions, performance decreases because temporally correlated deviations spread the outlyingness contribution to neighboring intervals, which is consistent with the dependence structure of functional data rather than a failure of the decomposition.
The real-data analyses further illustrate the interpretability of the Shapley decomposition and indicate that the decomposition often aligns with the expert knowledge, providing a clear explanation of multivariate functional outlyingness in practice.

While the separability assumption simplifies the modeling and estimation of covariance structures for multivariate functional data, it may not fully capture the complexity of real-world dependencies. Further future research could focus on developing methods that relax this assumption, allowing for non-separable covariance structures. %
Recent studies, such as \cite{aston2017tests,masak2023random} have highlighted the limitations of separable models and proposed tests for separability, but there remains a need for practical estimation methods under non-separability. New approaches, such as the expansion of the covariance into a series of separable terms \citep{masak2023separable}, have recently been proposed to deal with this problem. In particular, future research in this direction could involve sequential robust estimation of the separable components in the covariance expansion, thus providing a robust and computationally efficient framework when working with more general multivariate functional data. 

\subsection*{Acknowledgements}
We thank the anonymous reviewers for their careful reading and constructive comments, which have considerably improved this paper.
\subsection*{Funding}
This work was supported by the AI4CSM project, which has received funding from the ECSEL Joint Undertaking (JU) under grant agreement No 101007326; by the Austrian IKT der Zukunft programme via the Austrian Research Promotion Agency (FFG) and the Austrian Federal Ministry for Climate Action, Environment, Energy, Mobility, Innovation and Technology (BMK) under project No 884070; and by the Austrian Science Fund (FWF), project number I 5799-N.
\subsection*{Data and code availability}
The R code that reproduces all figures in this paper, along with the simulation results, is available at \url{https://github.com/MarcusMayrhofer/explainable-functional-outlier-detection}.
The El Niño data (U.S. Climate Prediction Center) are included. The welding data (\url{https://wis.kuleuven.be/statdatascience/robust/software}) and the fertility rates (\url{https://www.humanfertility.org}) are not redistributed; the repository documents how to obtain each.
\subsection*{Declaration of AI use}
AI coding assistants (Anthropic, Claude Opus 5) were used for proofreading, language editing, and coding assistance. The manuscript was drafted by the authors, and all methodological development, theoretical results, and interpretation of the findings are the work of the authors.
\subsection*{Disclosure statement}
The authors report there are no competing interests to declare.

	\bibliographystyle{apalike}
	\bibliography{references}

\clearpage
\appendix

\section{Further Preliminaries} 
\label{appendix:further_preliminaries}

\subsection{Matrix Normal Distribution} \label{appendix:sub:matrix_normal}
A random matrix $\bm{X} \in \R^{p \times q}$ follows a matrix normal distribution, denoted as $\bm{X} \sim \MN(\bm{\M}, \bm{\Sigma}^{\row}, \bm{\Sigma}^{\col})$, with mean $\bm{\M} \in \R^{p \times q}$, row covariance $\bm{\Sigma}^{\row} \in \pds(p)$, and column covariance $\bm{\Sigma}^{\col} \in \pds(q)$, 
if and only if its vectorized form $\vect(\bm{X}) \in \R^{pq}$ has a multivariate normal distribution $\NN(\vect(\bm{\M}),\bm{\Sigma}^{\col} \otimes \bm{\Sigma}^{\row})$ \citep{gupta1999}. Here, the class of all positive definite symmetric $a \times a$ matrices is denoted by $\pds(a)$. The vectorization operator $\vect(\cdot)$ stacks the columns of a matrix on top of each other, and $\otimes$ represents the Kronecker product. 
The probability density function (pdf) of a matrix normal random variable $\bm{X}$ is given by 
\begin{align}
    f(\bm{X}|\bm{\M}, \bm{\Sigma}^{\row}, \bm{\Sigma}^{\col}) = \frac{\exp(-\frac{1}{2}\tr((\bm{\Sigma}^{\col})^{-1} (\bm{X}-\bm{\M})' (\bm{\Sigma}^{\row})^{-1}(\bm{X}-\bm{\M})))}{(2\pi)^{\nicefrac{pq}{2}} \det(\bm{\Sigma}^{\col})^{\nicefrac{p}{2}} \det(\bm{\Sigma}^{\row})^{\nicefrac{q}{2}}}. \label{eq:matrix_normal_density}
\end{align}
Importantly, $\bm{\Sigma}^{\row}$ and $\bm{\Sigma}^{\col}$ are only identified up to a multiplicative constant $\kappa \neq 0$. Specifically, replacing $\bm{\Sigma}^{\row}$ by $\kappa\bm{\Sigma}^{\row}$ and $\bm{\Sigma}^{\col}$ by $\nicefrac{1}{\kappa}\bm{\Sigma}^{\col}$ leaves the pdf~\eqref{eq:matrix_normal_density} unchanged. To resolve the non-identifiability, one can fix a diagonal entry, the determinant, or norm of either matrix~\citep{ros2016existence, soloveychik2016gaussian}.

\subsection{Matrix Minimum Covariance Determinant Estimator (MMCD)}\label{section:MMCD}

For a random sample $\bm{X}_1,\dots,\bm{X}_n\in\mathbb{R}^{m\times p}$ from a matrix-elliptical semi-parametric distribution~\citep{gupta2012elliptically}, with the parametric part 
parameterized by the mean $\bm{\M}$ and covariance matrices $\bm{\Sigma}^{\row}$ and  $\bm{\Sigma}^{\col}$, \cite{mayrhofer2024robust} proposed robust mean and covariance estimators. The robust MMCD estimators 
$(\hat{\bm{\M}},\hat{\bm{\Sigma}}^{\row}, \hat{\bm{\Sigma}}^{\col})$ solve
\begin{align}
\label{eq:MMCD}        
\argmin_{\substack{\hat{\bm{\M}}_{H},\hat{\bm{\Sigma}}^{\row}_{H}, \hat{\bm{\Sigma}}^{\col}_{H}\\ H \subset \{1,\dots,n\}, \abs{H} = h}} p\ln(\det(\hat{\bm{\Sigma}}^{\col}_{H})) + q\ln(\det(\hat{\bm{\Sigma}}^{\row}_{H})),
\end{align}
where 
\begin{align}
    \hat{\bm{\M}}_{H} &= \frac{1}{h}\sum_{i \in H} \bm{X}_i,\nonumber\\
    \hat{\bm{\Sigma}}^{\row}_{H} &= \frac{1}{qh} \sum_{i \in H} (\bm{X}_i - \hat{\bm{\M}}_{H})(\hat{\bm{\Sigma}}^{\col}_{H})^{-1} (\bm{X}_i - \hat{\bm{\M}}_{H})',\label{eq:MMCD_row} \\  \hat{\bm{\Sigma}}^{\col}_{H} &= \frac{1}{ph} \sum_{i \in H} (\bm{X}_i - \hat{\bm{\M}}_{H})'(\hat{\bm{\Sigma}}^{\row}_{H})^{-1} (\bm{X}_i - \hat{\bm{\M}}_{H})\label{eq:MMCD_col}   .
\end{align}
Here,  $h = \alpha n, \alpha \in [0.5,1]$ represents the size of the \textit{clean} subset of the sample used for moment estimation. For $h=0.5n$, the estimators achieve a maximal breakdown point of $\nicefrac{n}{2} - \lfloor \nicefrac{p}{m} + \nicefrac{m}{p} \rfloor-1$. With the proper scaling, the method yields consistent estimators in this context, while the finite-sample efficiency can be further improved with an additional reweighting step. 

As there are no closed-form solutions for the robust MMCD-estimators, \cite{mayrhofer2024robust} proposed a nested iterative estimation procedure based on a \textit{concentration step} algorithm \citep{Rousseeuw1999} for solving~\eqref{eq:MMCD}, 
and an iterative \textit{flip-flop} algorithm \citep{Dutilleul1999} for computing the maximum likelihood estimates~\eqref{eq:MMCD_row}-\eqref{eq:MMCD_col}. Starting from any positive definite initialization, the proposed procedure is shown to converge almost surely to the positive definite covariance estimates, provided $h \geq \lfloor \nicefrac{p}{m} + \nicefrac{m}{p} \rfloor + 2$. The convergence also holds if the ellipticity assumption is violated. For technical and implementation details, see \cite{mayrhofer2024robust}. 

\subsection{Multivariate Gaussian Process}\label{appendix:sub:multivariate_gaussian_process}
A stochastic process $\bm{X}$ is a multivariate Gaussian process if every finite collection of realizations has a joint normal distribution. 
A finite collection of realizations from a stochastic process $\bm{X}$ at time points $\bm{t} = (t_1,\dots,t_q), t_1 < t_2 < \dots < t_q, t_k \in \T, k = 1,\dots,q,$ is denoted by 
$\bm{X}_{\bm{t}} = (\bm{X}(t_1),\dots,\bm{X}(t_q))' \in \R^{p \times q},$
yielding a matrix-variate sample. 
For the mean function, we have 
$\bm{\M}_{\bm{t}} = (\bm{\mu}(t_1),\dots,\bm{\mu}(t_q))' \in \R^{p \times q}$,
and the covariance function yields a block-partitioned matrix $\K_{\bm{t}} \in \R^{pq \times pq}$ as in Equation~\eqref{eq:mean_cov_function}, with %
entries $\kappa_{ij}(t_k,t_l)$, for $i,j = 1,\dots,p$ and $k,l = 1,\dots,q$.

The joint normality of \(\bm{X}_{\bm{t}}\) is described using a matrix-variate approach, which directly models the matrix-valued realizations as in \cite{chen2017modelling, chen2023multivariate}. In this case, \(\bm{X}_{\bm{t}} \sim \MN(\bm{M}_{\bm{t}}, \bm{\Sigma}^{\row}, \bm{\Sigma}_{\bm{t}}^{\col})\), where \(\bm{\Sigma}^{\row}\) represents the row-wise covariance matrix, capturing the dependencies between individual coordinate functions, and \(\bm{\Sigma}_{\bm{t}}^{\col} = (\kappa(t_k, t_l))_{k, l = 1}^q \in \mathbb{R}^{q \times q}\) is the column-wise covariance matrix, accounting for the temporal correlations between different time points. This formulation leverages a single kernel function \(\kappa(s, t)\) to model time dependencies, allowing the covariance structure to be factorized as \(\K(s, t) = \bm{\Sigma}^{\row} \kappa(s, t)\), significantly reducing complexity compared to using \(p(p+1)\) separate kernels, as mentioned in Equation~\eqref{eq:mean_cov_function}.

\begin{remark}
    Alternatively, the joint normality can also be expressed by vectorizing the matrix-valued realizations, as in \cite{alvarez2012kernels}. In this case, the columns of \(\bm{X}_{\bm{t}}\) are stacked into a vector \(\bm{x}_{\bm{t}} = \vect(\bm{X}_{\bm{t}})\), and the process is modeled as \(\bm{x}_{\bm{t}} \sim \NN(\bm{m}_{\bm{t}}, \K_{\bm{t}})\), where \(\bm{m}_{\bm{t}} = \vect(\bm{M}_{\bm{t}})\). 
\end{remark}

\subsection{Finite-basis Representation and Additive Noise Model}\label{appendix:sub:smoothing}

One of the key principles of FDA is to work with \textit{smooth} functions. This means that adjacent values are linked together to some degree and are unlikely to be too different from each other. If the functions were not smooth, there would be no significant advantage to treating them as functional data instead of just multivariate \citep{ramsay_silverman_2005}.

In practice, functional data are observed at discrete time points, and the raw observed data may contain noise or fluctuations that obscure the underlying patterns or trends present in the true functional form. We can formalize this using an additive noise model, see, e.g., \cite{ramsay_silverman_2005} and \cite{zhu2016bayesian},
\begin{align*}
    \bm{Y}(t) = \begin{pmatrix}
        Y_1(t)\\
        \vdots\\
        Y_p(t)
    \end{pmatrix} = \begin{pmatrix}
        X_1(t)\\
        \vdots\\
        X_p(t)
    \end{pmatrix} + \begin{pmatrix}
        \varepsilon_1(t)\\
        \vdots\\
        \varepsilon_p(t)
    \end{pmatrix} = 
    \bm{X}(t) + \bm{\varepsilon}(t),
\end{align*}
where $\bm{Y}$ is the observed process, and $\bm{X}$ is an underlying \textit{signal} process we are interested in, and $\bm{\varepsilon}$ are additive errors that are independent of $\bm{X}$. 

We employ smoothing techniques to reduce the noise and reveal the underlying structure or behavior of the process. The most common approach is to represent the observed process by basis functions. 
Let $\{\unibasis_k\}_{k \geq 1}$ denote a family of orthonormal basis functions of $L^2(\T)$, then each component $X_j$ of $\bm{X} \in \mathcal{H}$ can be expressed in terms of this basis as 
\begin{align*}
    X_j = \sum_{k = 1}^\infty {a}_{jk} \unibasis_k,\,j=1,\dots,p.
\end{align*}
Because the basis is fixed, the randomness of the stochastic process $\bm{X}$ is captured by the coefficients $a_{jk}$, $j = 1,\dots,p$, $k\geq 1$. 
Using only a sufficiently large number $m$ of basis functions, we can approximate $X_j$ arbitrarily well and rewrite the coordinates of the observed process as
\begin{align*}
     Y_j(t) = \sum_{k = 1}^m a_{jk} \unibasis_k(t) + \tilde{\varepsilon}_j(t),
\end{align*}
where the error term now consists of the approximation error and the measurement errors, i.e., 
$$\displaystyle\tilde{\varepsilon}_j(t) = \varepsilon_j(t) +  \sum_{k = m + 1}^\infty a_{jk} \unibasis_k(t).$$
Let $\mbasis=(\unibasis_1,\dots,\unibasis_m)'$ denote the vector consisting of the first $m$ basis functions of $\{\unibasis_k\}_{k \geq 1}$ and $\bm{a}_j = (a_{j1},\dots,a_{jm})' \in \R^{m}$ the vector of coefficients, we can write 
\begin{align*}
    Y_j(t) &= \bm{a}_j '\mbasis(t) + \tilde{\varepsilon}_j(t), \quad j=1\dots,p.
\end{align*}
By collecting the coefficients in a matrix $\bm{A} = (\bm{a}_1, \dots,\bm{a}_p)' \in \R^{p \times m}$ we can write the multivariate process as
\begin{align*}
    \bm{Y}(t) &= \bm{A}\mbasis(t) + \tilde{\bm{\varepsilon}}(t).
\end{align*} 
The coefficients $a_{jk}$, $j = 1,\dots,p$, $k = 1,\dots,m,$ are usually determined based on a least squares approach, and often a roughness penalty is involved; see \cite{ramsay_silverman_2005} for more details. 

\subsection{FPCA for Multivariate Stochastic Process with Separable Covariance} \label{appendix:sub:FPCA}

Let $\mathcal{K}$ denote the covariance operator corresponding to the covariance kernel $\kappa(s,t)$ of the multivariate stochastic process $\bm{X}\sim\msp(\bm{\mu},\bm{\Sigma}^{\row},\kappa)$. Then we have that 
\begin{equation*}
    \mathcal{K} \ueig_i(s) = \int_{\T}\kappa(s,t)\ueig_i(t)\mathrm{d}t=\lambda^{\ker}_i\ueig_i(s),\quad i=1,\dots,m,
\end{equation*}
and 
\begin{equation*}
    \bm{\Sigma}^{\row} \bm{v}^{\row}_j = \lambda^{\row}_j \bm{v}^{\row}_j ,\quad j=1,\dots,p.
\end{equation*}
For the multivariate covariance operator $\C$ with corresponding kernel $\K(s,t)$ we have
\begin{equation} \label{eq:eigen_eq}
    \C \meig_k(s) = \int_{\T} \K(s,t) \meig_k(t) \mathrm{d}t = \pi_k\meig_k(s) ,\quad k=1,\dots,M,
\end{equation}
where $M = pm$. In the separable setting, $\C = \bm{\Sigma}^{\row}\mathcal{K}$ with kernel $\K(s,t) = \bm{\Sigma}^{\row} \kappa(s,t)$. 
Consider the eigendecomposition $\bm{\Sigma}^{\row} = \bm{V}^{\row}\bm{D}^{\row}(\bm{V}^{\row})'$, where $\bm{V}^{\row} = ((\bm{v}^{\row}_1)',\dots,(\bm{v}^{\row}_p)')$ is the matrix of eigenvectors and $\bm{D}^{\row} = \diag(\lambda^{\row}_1,\dots,\lambda^{\row}_p)$ the diagonal matrix of ordered eigenvalues $\lambda^{\row}_1 \geq \cdots \geq \lambda^{\row}_p$, where we assume the uniqueness of the eigenvalues for simplicity; a proper generalization of the results holds also in the case of non-simple eigenvalues. Using the indexation $k=k(i,j)=1,\dots,M=pm$,  
\begin{align}\label{eq:MFPCA_decomposition}
\pi_k=\lambda_i^{\ker}\lambda_j^{\row},\quad \meig_k(t)=\ueig_i(t) \bm{v}^{\row}_j,\quad t\in \T,\quad i=1,\dots,m,\quad j=1\dots,p.
\end{align}
To see that the relations in \eqref{eq:MFPCA_decomposition} indeed hold, observe first that orthogonality of $\bm{V}^{\row}$ and orthonormality of $\ueig_i$, $i=1,\dots,m$  give the orthonormality of the corresponding products. Furthermore, 
\begin{align} \label{eq:eigen_eq_mpg}
\begin{split}
    \lambda^{\ker}_i \lambda^{\row}_j (\ueig_i(s) \bm{v}^{\row}_j)&=\lambda_j^{\row}\int_{\T}  \kappa(s,t) {\ueig_i(t) \bm{v}^{\row}_j}\,\mathrm{d}t \\
   &= \int_{\T}  \boldsymbol{\Sigma}^{\row}\kappa(s,t) {\ueig_i(t) \bm{v}^{\row}_j}\,\mathrm{d}t = \int_{\T}  \K(s,t) ({\ueig_i(t) \bm{v}^{\row}_j})\,\mathrm{d}t. 
\end{split}   
\end{align}
The uniqueness of the eigendecomposition \eqref{eq:eigen_eq} yields the desired claim. 
For a more general connection between multivariate FPCA of $\bm{X}\sim\mathcal{MSP}(\bm{\mu},\K)$ and univariate FPCA of its components $X_1,\dots,X_p$, see \cite{happgreven2018}.

\section{Proofs and Remarks on Multivariate Functional Mahalanobis Distance}~\label{appendix:fmmd_proofs}

\begin{proof}[Proof of Lemma \ref{lemma:affine_invariance}]
The affine equivariance of the mean and the covariance follows directly from the affine equivariance of their multivariate counterparts: \begin{align*}
        \bm{\mu}_{\bm{Y}} = \E[\bm{Y}(t)] = \E[\bm{A}\bm{X}(t)+\bm{\nu}] = \bm{A}\E[\bm{X}(t)]+\bm{\nu} = \bm{A}\bm{\mu}_{\bm{X}}+\bm{\nu},
    \end{align*}  
    \begin{align*}
        \K_{\bm{Y}}(s,t) = \cov(\bm{Y}(s),\bm{Y}(t)) 
        &= \E[(\bm{Y}(s)-\E[\bm{Y}(s)])(\bm{Y}(s)-\E[\bm{Y}(t)])']\\
        &= \E[(\bm{A}\bm{X}(s)-\E[\bm{A}\bm{X}(s)])(\bm{A}\bm{X}(s)-\E[\bm{A}\bm{X}(t)])']\\
        &= \E[(\bm{A}\bm{X}(s)-\bm{A}\E[\bm{X}(s)])(\bm{A}\bm{X}(s)-\bm{A}\E[\bm{X}(t)])']\\
        &= \bm{A}\E[(\bm{X}(s)-\E[\bm{X}(s)])(\bm{X}(s)-\E[\bm{X}(t)])']\bm{A}'\\
        &= \bm{A} \K_{\bm{X}}(s,t)\bm{A}'.
    \end{align*}

For simplicity of the notation assume $\bm{\mu}_{\bm{X}},\bm{\nu}=\bm{0}$, and denote $\C_{\bm{X}}$ and $\C_{\bm{Y}}$ to be the (matrices of) covariance operators associated with $\K_{\bm{X}}$ and $\K_{\bm{Y}}$, respectively. Then, for any $\bm{f}\in \mathcal{H}$, $u\in \T$,
\begin{align*}
    \C_{\bm{Y}}\bm{f}(u) 
    &= \int_{\T} \K_{\bm{Y}}(u,v)\bm{f}(v)\mathrm{d}v 
    = \bm{A}\int_{\T} \K_{\bm{X}}(u,v)(\bm{A}'\bm{f}(v))\mathrm{d}v = \bm{A}\C_{\bm{X}}(\bm{A}'\bm{f})(u). 
\end{align*}

Denoting $\C_{\bm{X}}^{(M)}$ and $\C_{\bm{Y}}^{(M)}$ to be the truncation of $\C_{\bm{X}}$ and $\C_{\bm{Y}}$ onto the corresponding first $M$ components, i.e., for $\bm{f}\in\mathcal{H}$, 
$$
\C_{\bm{X}}^{(M)}\bm{f}(u)=\sum_{i=1}^M\pi_{\bm{X},i}\langle\bm{f},\meig_{\bm{X},i}\rangle\meig_{\bm{X},i}(u),\quad 
\C_{\bm{Y}}^{(M)}\bm{f}(u)=\sum_{i=1}^M\pi_{\bm{Y},i}\langle\bm{f},\meig_{\bm{Y},i}\rangle\meig_{\bm{Y},i}(u),
$$
it is straightforward to verify that 
$$
(\C_{\bm{Y}}^{(M)})^{-1}\bm{f}(u)=(\bm{A}')^{-1}(\C_{\bm{X}}^{(M)})^{-1}(\bm{A}^{-1}\bm{f})(u),
$$
provided $M$ is such that $\C_{\bm{X}}^{(M)}$ is invertible. Additionally, we can write 
$$
\fmmd^2(\bm{X};k)=\langle (\C_{\bm{X}}^{(M)})^{-1}\bm{X},\bm{X}\rangle,\quad \fmmd^2(\bm{Y};k)=\langle (\C_{\bm{Y}}^{(M)})^{-1}\bm{Y},\bm{Y}\rangle.
$$
Finally, 
\begin{align*}
    \fmmd^2(\bm{Y},\bm{\mu}_{\bm{Y}};\K_{\bm{Y}},M)
    &=\langle (\C_{\bm{Y}}^{(M)})^{-1}\bm{Y},\bm{Y}\rangle\\
    &=\langle (\bm{A}')^{-1}(\C_{\bm{X}}^{(M)})^{-1}(\bm{A}^{-1}\bm{Y}),\bm{Y}\rangle
    =\langle (\bm{A}')^{-1}(\C_{\bm{X}}^{(M)})^{-1}\bm{X},\bm{A}\bm{X}\rangle\\
    &=\int_{\T}((\bm{A}')^{-1}(\C_{\bm{X}}^{(M)})^{-1}\bm{X}(u))'\bm{A}\bm{X}(u)\mathrm{d}u\\
    &=\int_{\T}((\C_{\bm{X}}^{(M)})^{-1}\bm{X})'\bm{X}(u)\mathrm{d}u\\
    &=\langle (\C_{\bm{X}}^{(M)})^{-1}\bm{X},\bm{X}\rangle
    =\fmmd^2(\bm{X},\bm{\mu}_{\bm{X}};\K_{\bm{X}},M).
\end{align*}

\end{proof}

\begin{proof}[Proof of Lemma \ref{lemma:independent processes}]
    Let $(\lambda_i^{(j)},\ueig_i^{j})$, be the $i$th eigenpair of the covariance $\mathcal{K}_j$ of the $j$-component of $\bm{X}$, $i\geq 1$, $j=1,\dots,p$. Construct now the following set of multivariate functions: $\ueig_i^{(j)}\bm{e}_j$, $i\geq 1$, $j=1,\dots,p$, where $\bm{e}_j$ is the $j$th vector of the canonical basis of $\mathbb{R}^p$. It is then straightforward to verify that this set is indeed orthonormal. Additionally, as components in $\bm{X}$ are uncorrelated, $\C=\mathrm{diag}(\mathcal{K}_1,\dots,\mathcal{K}_p)$. Simple algebra gives further $\C\ueig_i^{(j)}\bm{e}_j=\lambda_i^{(j)}\ueig_i^{(j)}\bm{e}_j$, for any $i\geq 1$ and $j=1\dots,p$. Thus, functions in $\{\ueig_i^{(j)}\bm{e}_j: i\geq 1,\,j=1,\dots,p\}$, are the eigenfunctions of $\C$, while $\lambda_i^{(j)},\,i\geq 1,\,j=1,\dots,p$ are the corresponding eigenvalues. In other words, the spectrum of $\C$ corresponds to the union of the spectra of individual covariance operators $\mathcal{K}_j$, $j=1,\dots,p$. Then, for $m_1,\dots,m_p$ as described in the statement of the result, the $M$ largest eigenpairs of $\C$ are $(\lambda_i^{(j)},\ueig_i^{(j)}\bm{e}_j)$, $j=1,\dots,p$, $i=1,\dots,m_j$. Observe that since $\pi_{M+1}<\pi_{M}$, these eigenvalues are chosen from the spectrum of $\C$ in a unique way. The following now holds:
    \begin{align*}
        \fmmd^2(\bm{X},\bm{\mu};\K,M) 
        &=  \sum_{j=1}^p\sum_{j=1}^{m_j} \frac{1}{\lambda_i^{(j)}}  \innerproduct{\bm{X} - \bm{\mu}}{\ueig_i^{(j)}\bm{e}_j}^2=\sum_{j=1}^p\left(\sum_{j=1}^{m_j} \frac{1}{\lambda_i^{(j)}}  \innerproduct{{X}_j - {\mu}_j}{\ueig_i^{(j)}}^2\right)\\
        &= \sum_{j=1}^p\fmd^2(X_j,\mu_j;\kappa_j,m_j).
    \end{align*}        
\end{proof}

\begin{proof}[Proof of Corollary \ref{cor:independent processes}]
    The proof of the statement (i) follows from the following claims: Lemma \ref{lemma:independent processes}, the fact that the components of the separable covariance processes, which have uncorrelated components, share, up to scale, common covariance kernel, and Lemma \ref{lemma:affine_invariance} by taking $\bm{A}=\bm{\Sigma}^{-1/2}$ and observing that the process transformed that way has uncorrelated components. Claim (ii) follows directly from (i).
\end{proof}

\begin{remark}[Choice of $M$ in Corollary~\ref{cor:independent processes}]\label{rem:M=mp}
    It is important to note that requiring \( M \) to be a multiple of \( p \) is not arbitrary, but rather a natural choice given the structure of the problem. To illustrate, if $\bm{\Sigma}^{\row}=\bm{I}_p$, the eigenvalues of the covariance operator \( \C \) appear with multiplicity \( p \), reflecting the inherent symmetries in the data. When projecting onto an \( M \)-dimensional space, \( M \) is typically chosen to capture a desired amount of explained variance, or based on a threshold related to the significance of the eigenfunctions. Given that each eigenvalue corresponds to \( p \) linearly independent components, it is reasonable to select all \( p \) components associated with any given eigenvalue when deciding on the projection dimension. This ensures that the projection retains the intrinsic structure of the data, avoiding arbitrary truncation of the eigenspaces and preserving the full contribution of the variance associated with each eigenvalue.  
\end{remark}

\begin{proof}[Proof of Lemma \ref{lemma:fmmd_distribution}]
Let $\bm{X}$ be the multivariate Gaussian process with covariance $\K=\kappa\bm{\Sigma}^{\row}$ and mean $\bm{\mu}$, where for the simplicity of the notation we assume that $\bm{\mu}=\bm{0}$. Let further $\pi_1\geq\pi_2\geq\cdots\geq \pi_M>0$, and $\meig_1,\dots,\meig_M$ the leading $M$ eigenvalues and eigenfunctions of covariance operator $\C$ associated with covariance function $\K$, respectively, i.e., $\C\meig_i=\pi_i\meig_i$, for $i=1,\dots,M$. Then $\beta_i=\langle\bm{X},\meig_i\rangle\sim\mathcal{N}(0,\pi_i)$, $i=1,\dots,M$ are uncorrelated, i.e., independent random variables; see, e.g., \cite{wang2008karhunen} for more details. Denoting  $\eta_i=\beta_i/\sqrt{\pi_i}\sim\mathcal{N}(0,1)$ to be i.i.d. random variables from standard normal distribution, we can write 
$$
\fmmd^2(\bm{X},\bm{\mu};\bm{\Sigma}^{\row},\kappa,M)=\sum_{i=1}^M\eta_i^2\sim\chi^2(M), 
$$
as a sum of $M$ {squared} independent standard normal random variables.
\end{proof}

\begin{proof}[Proof of Theorem \ref{theorem:fmmd_basis}]
\textit{(i)} %
    For $t\in \T$ let $\bm{X}(t)=\bm{A}'\mbasis(t)$, $\mbasis=(\unibasis_1,\dots,\unibasis_m)'$, for $\bm{A}=(\bm{a}_1\hdots\bm{a}_p)$. First, observe that for every $t\in\T$,  $\bm{\mu}(t)=\mathbb{E}(\bm{X}(t))=\mathbb{E}(\bm{A}'\mbasis(t))=\mathbb{E}(\bm{A}')\mbasis(t)=\mathbb{E}(\bm{A})'\mbasis(t)=\bm{M}_{\bm{A}}'\mbasis(t)$. For the simplicity of the notation, since all the quantities involved are centered, we further take $\bm{M}_{\bm{A}}=\bm{0}$. Proceed now first by assuming $\bm{\Sigma}^{\row}=\bm{I}_p$. Then, for $t\in \T$ 
   \begin{align}\label{eq:a_is}
       \kappa(t,t)\delta_{i,j}&=\mathbb{E}(X_i(t)X_j(t))=\bm{e}_i'\mathbb{E}(\bm{X}(t)\bm{X}'(t))\bm{e}_j\\%
       &=\bm{e}_i'\mathbb{E}\left(\bm{A}'\mbasis(t)(\bm{A}'\mbasis(t))'\right )\bm{e}_j=\mbasis'(t)\mathbb{E}\left(\bm{a}_i\bm{a}_j'\right )\mbasis(t),\quad i,j=1,\dots,p,
   \end{align}
where Kronecker delta $\delta_{i,j}=1$ if $i=j$ and $0$ otherwise. Since \eqref{eq:a_is} holds for every $t\in \T$, 
$$
\mathbb{E}\left(\bm{a}_i\bm{a}_j'\right )=0 \text{ for  }i\neq j,\quad \text{and}\quad \mathbb{E}\left(\bm{a}_i\bm{a}_i'\right )=\mathbb{E}\left(\bm{a}_j\bm{a}_j'\right ),  \text{ for  }i=1,\dots,p.
$$
Denoting $\bm{\Sigma}^{\col}:=\mathrm{Cov}(\bm{a}_1)$
$$
\mathrm{Cov}(\vect(\bm{A}))=\mathrm{Cov}\begin{pmatrix}
    \bm{a}_1\\
    \vdots\\
    \bm{a}_p
\end{pmatrix}=\begin{pmatrix}
    \bm{\Sigma}^{\col}&\bm{0}&\cdots&\bm{0}\\
    \bm{0}&\bm{\Sigma}^{\col}&\cdots&\bm{0}\\
    \vdots&&\ddots&\vdots\\
    \bm{0}&\cdots&&\bm{\Sigma}^{\col}
\end{pmatrix}=\bm{\Sigma}^{\col}\otimes\bm{I}_p,
$$
where for every $s,t\in \T$ the matrix $\bm{\Sigma}^{\col}$ satisfies 
\begin{equation}\label{eq:Sigma_col1}
\kappa(s,t)=\mbasis'(s)\bm{\Sigma}^{\col}\mbasis(t).    
\end{equation}

To see that $\bm{\Sigma}^{\col}$ is indeed positive, observe the following: As $\displaystyle\bm{W}:=\int_{\T}\mbasis(t)\mbasis'(t)\mathrm{d}t$ is a positive definite matrix, then there exists $m_0\in \mathbb{N}$ and $t_1,\dots,t_{m_0}$, such that the Riemann sum $\nicefrac{1}{m_0}\sum_{i=1}^{m_0} \mbasis(t_i)\mbasis'(t_i)
$ approximating the integral in $\bm{W}$ is a positive definite matrix. Therefore, there exist $m$ linearly independent vectors in $\{\mbasis(t_i): i=1,\dots,m_0\}$. Without loss of generality, assume that these are $\mbasis(t_1),\dots,\mbasis(t_m)$. Take now $\bm{y}\in\mathbb{R}^m$, $\bm{y}\neq 0$. Due to the independence of $\mbasis(t_1),\dots,\mbasis(t_m)$,  $\bm{y}$ can represented as $\bm{y}=\sum_{i = 1}^m c_i \mbasis(t_i)$, for some $(c_1,\dots,c_m)\neq \bm{0}$. Then $$\bm{y}' \bm{\Sigma}^{\col} \bm{y} = \sum_{i,j = 1}^m c_i c_j \mbasis'(t_i) \bm{\Sigma}^{\col} \mbasis(t_j) = \sum_{i,j = 1}^m c_i c_j \kappa(t_i,t_j) > 0,$$ 
where the last inequality holds due to the positive definiteness of $\kappa$.

Let now $\bm{X}=\bm{A}'\mbasis\sim\mathcal{MSP}(\bm{0},\bm{\Sigma}^{\row},\kappa)$. Lemma \ref{lemma:affine_invariance} then implies that   $\bm{Y}:=\bm{\Sigma}^{\row}{}^{-1/2}\bm{X}\sim \mathcal{MSP}(\bm{0},\bm{I}_p,\kappa)$. Additionally, for $\bm{A}_{\bm{Y}}=\bm{A}\bm{\Sigma}^{\row}{}^{-1/2}$ is $\bm{Y}=\bm{A}_{\bm{Y}}'\mbasis$. The first part of the proof now shows that 
$$
\mathrm{Cov}(\vect(\bm{A}_{\bm{Y}}))=\bm{I}_p\otimes \bm{\Sigma}^{\col},
$$
where $\bm{\Sigma}^{\col}$ depends only on the kernel $\kappa$ as given in \eqref{eq:Sigma_col1}. Note further that 
$$
\vect(\bm{A})=\vect(\bm{I}_m\bm{A}_{\bm{Y}}(\bm{\Sigma}^{\row})^{1/2})=\left((\bm{\Sigma}^{\row})^{1/2}\otimes\bm{I}_m\right)\vect(\bm{A}_{\bm{Y}}),
$$
finally giving
$$
\mathrm{Cov}(\vect(\bm{A}))=\left((\bm{\Sigma}^{\row})^{1/2}\otimes\bm{I}_m\right)\left(\bm{I}_p\otimes \bm{\Sigma}^{\col}\right)\left((\bm{\Sigma}^{\row})^{1/2}\otimes\bm{I}_m\right)=\bm{\Sigma}^{\row}\otimes \bm{\Sigma}^{\col}.
$$

\medskip\noindent \textit{(ii)} To prove the statement (ii), we begin by assuming $\bm{\Sigma}^{\row}=\bm{I}_p$. Additionally, without loss of generality and for the simplicity of the notation, take $\bm{\mu}=\bm{0}$. Let further $(\lambda_1,\ueig_1),\dots,(\lambda_m,\ueig_m)$, $\lambda_1\geq\cdots\lambda_m>0$  be the eigenpairs of $\kappa$, i.e. for every $s\in\T$
\begin{equation}\label{eq:Lemma ii_1}
\int_{\T}\kappa(s,t)\ueig_i(t)\mathrm{d}t=\lambda_i\ueig_i(s).
\end{equation}
Expressing the eigenfunctions $\ueig_i$, $i=1,\dots,m$ in $\mbasis$ basis gives  
\begin{equation}\label{eq:Lemma ii_2}
\ueig_i=\bm{b}_i'\mbasis,\quad i=1,\dots,m.
\end{equation}
For any $k\in\{1,\dots,p\}$, $i=1,\dots,m$, and $s\in\T$  \eqref{eq:Lemma ii_1} and \eqref{eq:Lemma ii_2} imply the following relations are equivalent:
\begin{align}\label{eq:Lemma ii_3}
    \int_{\T} \kappa(s,t)\ueig_i(t)\mathrm{d}t=\lambda_i\ueig_i(s)
    \iff &\int_{\T} \mathbb{E}(X_k(s)X_k(t))\ueig_i(t)\mathrm{d}t=\lambda_i\ueig_i(s),\nonumber\\
    \iff &\mbasis'(s)\mathbb{E}(\bm{a}_k\bm{a}_k')\left(\int_{\T} \mbasis(t)\mbasis'(t)\mathrm{d}t\right)\bm{b}_i=\lambda_i\mbasis'(s)\bm{b}_i,\nonumber \\
    \iff &\mbasis'(s)\mathbb{E}(\bm{a}_k\bm{a}_k')\bm{W}\bm{b}_i=\lambda_i\mbasis'(s)\bm{b}_i,
\end{align}
where $\displaystyle\bm{W}:=\int_{\T}\mbasis(t)\mbasis'(t)\mathrm{d}t$ is a positive definite matrix; see proof of part (i) for more details. Since \eqref{eq:Lemma ii_3} holds for every $s\in\T$, and since $\mathbb{E}(\textbf{a}_k\textbf{a}_k')=\boldsymbol{\Sigma}^{\col}$; see proof of (i) for details, we obtain
\begin{equation}\label{eq:Lemma ii_4}
\bm{\Sigma}^{\col}\bm{W}\bm{b}_i=\lambda_i\bm{b}_i \iff  (\bm{\Sigma}^{\col})^{-1}
        = \bm{W}^{1/2}\bm{\Sigma}^{\col}\bm{W}^{1/2}\bm{u}_i
        =  \lambda_i\bm{u}_i, 
\end{equation}
where $\bm{u}_i=\bm{W}^{1/2}\bm{b}_i$, $i=1,\dots,m$. Equation~\eqref{eq:Lemma ii_4} also implies that $(\lambda_i,\bm{u}_i)$, $i=1,\dots,m$ are eigenpairs of symmetric matrix $\bm{W}^{1/2}\bm{\Sigma}^{\col}\bm{W}^{1/2}$. Finally, \eqref{eq:Lemma ii_4} further implies
\begin{align*}
\bm{W}^{1/2}\bm{\Sigma}^{\col}\bm{W}^{1/2} = \sum_{i = 1}^m  \lambda_i \bm{u}_i \bm{u}_i' =  \bm{W}^{1/2} \left(\sum_{i = 1}^m  \lambda_i\bm{b}_i \bm{b}_i'\right) \bm{W}^{1/2}
    \iff  \bm{\Sigma}^{\col} = \sum_{i = 1}^m  \lambda_i\bm{b}_i \bm{b}_i'.
\end{align*}
Moreover,
\begin{align}
\bm{W}^{-1/2}(\bm{\Sigma}^{\col})^{-1}\bm{W}^{-1/2}&=(\bm{W}^{1/2}\bm{\Sigma}^{\col}\bm{W}^{1/2})^{-1} 
        = \sum_{i = 1}^m  \frac{1}{\lambda_i} \bm{u}_i \bm{u}_i'
        =  \bm{W}^{1/2} \left(\sum_{i = 1}^m  \frac{1}{\lambda_i}\bm{b}_i \bm{b}_i'\right) \bm{W}^{1/2}\nonumber\\
           & \iff  (\bm{\Sigma}^{\col})^{-1}
        =  \bm{W} \left(\sum_{i = 1}^m  \frac{1}{\lambda_i} \bm{b}_i \bm{b}_i'\right) \bm{W}.\label{eq:inv_eigen}
\end{align}
Corollary \ref{cor:independent processes}, together with relation  $\langle X_j,\ueig_i\rangle=\langle\bm{a}_j'\mbasis,\bm{b}_i'\mbasis\rangle=\bm{a}_j'\bm{W}\bm{b}_i$  now implies 
\begin{align*}
\fmmd^2(\bm{X};mp)&=\sum_{i=1}^m\sum_{j=1}^p\lambda_i^{-1}\langle X_j,\ueig_i\rangle^2
=\sum_{i=1}^m\sum_{j=1}^p\lambda_i^{-1}(\bm{a}_j'\bm{W}\bm{b}_i)^2\\
&=\sum_{j=1}^p\bm{a}_j'\bm{W}\sum_{i=1}^m\left(\lambda_i^{-1}\bm{b}_i\bm{b}_i'\right)\bm{W}\bm{a}_j=\sum_{j=1}^p\bm{a}_j'(\bm{\Sigma}^{\col})^{-1}\bm{a}_j\\
&=\mathrm{tr}(\bm{A}'(\bm{\Sigma}^{\col})^{-1}\bm{A})=\mmd^2(\bm{A}).
\end{align*}
Affine invariance of $\fmmd$; Lemma \ref{lemma:affine_invariance} and $\mmd$ \cite[Lemma 3.0.1]{mayrhofer2024robust} complete the proof of (ii).

\medskip\noindent \textit{(iii)} %
As in (i), for $t\in\T$,  $\bm{\mu}(t)=\mathbb{E}(\bm{X}(t))=\mathbb{E}(\bm{A}'\mbasis(t))=\mathbb{E}(\bm{A}')\mbasis(t)=\mathbb{E}(\bm{A})'\mbasis(t)=\bm{M}_{\bm{A}}'\mbasis(t)$. For the simplicity of the notation, we again take $\bm{M}_{\bm{A}}=\bm{0}$.  

Proceed first by assuming that $\bm{\Sigma}^{\row}=\bm{I}_p$. Take further $t_1,\dots,t_m$, such that $\mbasis_{t_1,\dots,t_m}:=(\mbasis(t_1)\dots\mbasis(t_m))$ is a full column rank matrix; for the proof of the existence see proof of (i). Gaussianity of $\bm{X}$ then implies that
$$
(\bm{X}(t_1),\dots,\bm{X}(t_{\mb}))=\bm{A}'(\mbasis(t_1),\dots,\mbasis(t_{\mb}))\sim \mathcal{MN}(\bm{0},\bm{I}_p,\bm{\Sigma}_{t_1,\dots,t_{\mb}}^{\col}),
$$
where $\bm{\Sigma}_{t_1,\dots,t_{\mb}}^{\col}=[\kappa(t_i,t_j)]_{i,j}$ is a positive definite matrix, due to the positive definiteness of kernel $\kappa$. Therefore, 
$$
\vect(\bm{I}_p\bm{A}'\mbasis_{t_1,\dots,t_m})=\left(\mbasis_{t_1,\dots,t_m}\otimes \bm{I}_p\right)\vect(\bm{A}')\sim \mathcal{N}_{mp}(\bm{0},\bm{\Sigma}_{t_1,\dots,t_{\mb}}^{\col}\otimes\bm{I}_p),
$$
i.e.,
$$
\vect(\mbasis_{t_1,\dots,t_m}\bm{A}\bm{I}_p)=\left( \bm{I}_p\otimes\mbasis_{t_1,\dots,t_m}'\right)\vect(\bm{A})\sim \mathcal{N}_{mp}(\bm{0},\bm{I}_p\otimes\bm{\Sigma}_{t_1,\dots,t_{\mb}}^{\col}).
$$
Regularity of $\mbasis_{t_1,\dots,t_m}$ and the fact that inversion of Kronecker product of two matrices, as well as the product of two Kronecker products, retains the Kronecker structure completes the first part of the proof, where the particular form of the covariance $\bm{\Sigma}^{\col}$ is given by (i). A short note to the reader: Given the general result in (i), it was enough to show that $\bm{A}$ has normally distributed entries. 

Finally, for $\bm{X}\sim \msp(\bm{0},\bm{\Sigma}^{\row},\kappa)$, let $\bm{Y}:=(\bm{\Sigma}^{\row})^{-1/2}\bm{X}=\left(\bm{A}(\bm{\Sigma}^{\row})^{-1/2}\right)'\mbasis
\sim\msp(\bm{0},\bm{I}_p,\kappa)$. The first part of the proof shows that 
$$
\bm{A}(\bm{\Sigma}^{\row})^{-1/2}\sim\MN(\bm{0},\bm{\Sigma}^{\col},\bm{I}_p).
$$
Matrix affine equivariance of matrix normal distribution~ \citep{gupta1999}  finally gives that 
$$
\bm{A}\sim\MN(\bm{0},\bm{\Sigma}^{\col},\bm{\Sigma}^{\row}),
$$
thus completing the proof.
\end{proof}

\begin{corollary}\label{corollary:fmd}
    Let $X(t) = \bm{a}'\mbasis(t)$ be a rank $m \in \N$ stochastic process with mean $\mu$ and covariance $\kappa$, with coefficients $\bm{a} \in \R^m$ and basis $\mbasis=(\unibasis_1,\dots,\unibasis_m)'$. Then the following holds:
    \begin{itemize}
        \item[(i)] $\bm{a}$ has a multivariate distribution with mean $\bm{m}_{\bm{a}}$ and covariance $\cov(\bm{a}) = \bm{\Sigma} \in \pds(m)$ such that
        $\bm{m}_{\bm{a}}'\mbasis(t) = \bm{\mu}(t) \quad \text{and} \quad \mbasis'(s)\bm{\Sigma}\mbasis(t) = \kappa(s,t)$ for all $s,t  \in \T$.
        \item[(ii)] $\fmd^2(X;m) = (\bm{a} - \bm{m}_{\bm{a}})'\bm{\Sigma}^{-1}(\bm{a} - \bm{m}_{\bm{a}}) = \md^2(\bm{a})$.
        \item[(iii)] If $X$ is a Gaussian process, then $\bm{a} \sim \NN(\bm{m}_{\bm{a}}, \bm{\Sigma})$ has a multivariate normal distribution with mean $\bm{m}_{\bm{a}}$ and covariance matrix $\bm{\Sigma}$.
    \end{itemize}
\end{corollary}
The corollary follows directly from Theorem~\ref{theorem:fmmd_basis}.

\section{Algorithms}\label{appendix:algorithm}
\begin{algorithm}[!h]
\onehalfspacing
    \caption{Robust estimation of mean and covariance function}
    \label{alg:alg1}
    \begin{algorithmic}[1] 
        \Require{$\bm{\mathfrak{X}} = (\bm{X}_1,\dots,\bm{X}_n),\, \bm{X}_i \in \R^{p \times q},\, i = 1,\dots,n,\, \mbasis = (\unibasis_1,\dots,\unibasis_m)',\, \T$}
            \State \textbf{Create functional data object} 
            \Statex \quad {Estimate coefficient matrices $\bm{\mathfrak{A}} = (\bm{A}_1,\dots,\bm{A}_n)$ by smoothing $\bm{\mathfrak{X}}$;
            \Statex \quad Obtain finite basis representation $\bm{X}_i(t) = \bm{A}_i'\mbasis(t),\, i = 1,\dots,n,\, t \in \T$};
            \State \textbf{MMCD} \citep[Algorithm 2]{mayrhofer2024robust}
            \Statex \quad Run MMCD procedure on $\bm{\mathfrak{A}}$ and get $(\hat{\bm{\M}}_{\bm{A},H^\ast},\hat{\bm{\Sigma}}^{\row}_{H^\ast}, \hat{\bm{\Sigma}}^{\col}_{H^\ast},\mmd(\bm{\mathfrak{A}}))$;
            \State \textbf{Obtain functional data objects for mean and covariance}
            \Statex \quad $\hat{\bm{\mu}}(t) = \hat{\bm{\M}}_{\bm{A},H^\ast}\mbasis(t)$;
            \Statex \quad $\hat{\bm{\Sigma}}^{\row} = \hat{\bm{\Sigma}}^{\col}_{H^\ast}$;
            \Statex \quad $\hat{\kappa}(s,t) = \mbasis'(s)\hat{\bm{\Sigma}}^{\row}_{H^\ast}\mbasis(t)$;
            \Statex \quad $\fmmd(\bm{X}_i) = \fmmd(\bm{X}_i,\hat{\bm{\mu}};\hat{\bm{\Sigma}}^{\row},\hat{\kappa},mp) = \mmd(\bm{A}_i,\hat{\bm{\M}}_{\bm{A},H^\ast};\hat{\bm{\Sigma}}^{\row}_{H^\ast}, \hat{\bm{\Sigma}}^{\col}_{H^\ast})$
            \Ensure{$\hat{\bm{\mu}}, \hat{\bm{\Sigma}}^{\row}, \hat{\kappa}, \left(\fmmd(\bm{X}_1),\dots,\fmmd(\bm{X}_n)\right)$}
    \end{algorithmic}
\end{algorithm}

Algorithm~\ref{alg:alg1} yields robust estimators. Similarly, non-robust counterparts can be obtained by replacing MMCD by the iterative matrix maximum likelihood estimation (MMLE) procedure of \cite{Dutilleul1999} in \texttt{step 2}.

Algorithm~\ref{alg:alg1} can be easily adapted for the analysis of raw data: \texttt{step 1} in the algorithm is omitted, and $p \times q$ matrices of raw data observations are supplied to the MMCD in \texttt{step 2}. The pointwise estimates of mean and covariance evaluated at observed time points $t_1,\dots,t_q$ are the output of MMCD \texttt{step 2}. Usually, post-smoothing is applied to those estimates to extend them to the functional setting, see \cite{ramsay_silverman_2005}. However, this is beyond the scope of this paper. The rationale for using the algorithm on raw data lies in the fact that, for certain classes of separable covariance processes (e.g., Gaussian and Student’s-t processes, see \cite{chen2021space, chen2023multivariate}), all finite-dimensional projections belong to the same family of matrix-variate distributions with separable covariance structure. Moreover, in those cases, parameters estimated on finite-dimensional projections correspond to pointwise evaluations of the process parameters, as described above.

\begin{algorithm}[!h]
    \caption{Robust FPCA for separable processes}
    \label{alg:alg2}
    \begin{algorithmic}[1] 
    \onehalfspacing
        \Require{$\bm{\mathfrak{X}} = (\bm{X}_1,\dots,\bm{X}_n),\, \bm{X}_i \in \R^{p \times q},\, i = 1,\dots,n,\, \mbasis = (\unibasis_1,\dots,\unibasis_m)',\, \T$}
            \State \textbf{Create functional data object} 
            \Statex \quad {Estimate coefficient matrices $\bm{\mathfrak{A}} = (\bm{A}_1,\dots,\bm{A}_n)$ by smoothing $\bm{\mathfrak{X}}$;
            \Statex \quad Obtain finite basis representation $\bm{X}_i(t) = \bm{A}_i'\mbasis(t),\, i = 1,\dots,n,\, t \in \T$};
            \State \textbf{MMCD} (Algorithm 2 in \cite{mayrhofer2024robust})
            \Statex \quad Run MMCD procedure on $\bm{\mathfrak{A}}$ and get $(\hat{\bm{\M}}_{\bm{A},H^\ast},\hat{\bm{\Sigma}}^{\row}_{H^\ast}, \hat{\bm{\Sigma}}^{\col}_{H^\ast},\mmd(\bm{\mathfrak{A}}))$;
            \State \textbf{Compute coefficient representations of FPCs} 
            \Statex \quad Compute matrix of inner products of basis functions $\bm{W}=\int_{\T}\mbasis(t)\mbasis'(t)\mathrm{d}t$;
            \Statex \quad Eigendecomposition of $\bm{W}^{1/2}\hat{\bm{\Sigma}}^{\col}_{H^\ast}\bm{W}^{1/2} = \bm{U}\bm{\Lambda}\bm{U}'$;
            \Statex \quad \quad Matrix of eigenvalues $\bm{\Lambda} = \diag(\lambda_1^{\ker},\dots,\lambda_m^{\ker})$;
            \Statex \quad \quad Matrix of eigenvectors $\bm{U} = (\bm{u}_1,\dots,\bm{u}_m)$;
            \Statex \quad Compute coefficients $\bm{b}_j = \bm{W}^{-1/2}\bm{u}_j,\,j = 1,\dots,m,$ of FPCs;
            \State \textbf{Obtain univariate FPCs}
            \Statex \quad Univariate eigenpairs $(\ueig_j,\lambda_j^{\ker})$ with $\ueig_j(t) = \bm{b}_j'\mbasis(t),\,t \in \T,\,j = 1,\dots,m$;
            \State \textbf{Obtain multivariate FPCs}
            \Statex \quad Eigendecomposition of $\hat{\bm{\Sigma}}^{\row} = \hat{\bm{\Sigma}}^{\col}_{H^\ast} = \bm{V}\bm{\Gamma}\bm{V}'$;
            \Statex \quad \quad Matrix of eigenvalues $\bm{\Gamma} = \diag(\lambda_1^{\row},\dots,\lambda_p^{\row})$;
            \Statex \quad \quad Matrix of eigenvectors $\bm{V} = (\bm{v}_1^{\row},\dots,\bm{v}_1^{\row})$;
            \Statex \quad Define indexation $k=k(j,l)=1,\dots,M=pm$ for $j = 1,\dots,m,\,l = 1,\dots,p$;
            \Statex \quad Multivariate eigenpairs $(\pi_k,\meig_k) = (\lambda_j^{\ker}\lambda_k^{\row}, \bm{v}_l^{\row}\ueig_j)$;
        \Ensure{$(\pi_1,\dots,\pi_M),(\meig_1,\dots,\meig_M)$}
    \end{algorithmic}
\end{algorithm}

\section{Outlier Explanations} \label{appendix:shapley_proofs}

Before we provide the proofs for the results presented in Section~\ref{section:shapley}, we summarize the notation used in this section in Table~\ref{tab:notation_shapley}, and provide further details for outlier explanations for multivariate and matrix-variate data. 

\begin{table}[p]
\centering
\caption{Summary of notation used in Section~\ref{section:shapley}.}
\label{tab:notation_shapley}

\small\begin{subtable}[H]{1\linewidth}
\centering
\caption{Univariate functional data.}
\begin{tabular}{lp{10cm}}
\toprule
\textbf{Symbol} & \textbf{Description} \\
\midrule
$X(t)\in L^2(\mathcal{T})$ & Univariate stochastic process \\
$\mu(t)$ & Mean function of $X$ \\
$\kappa(s,t)$ & Covariance kernel of $X$ \\
$\mathcal{K}$ & Covariance operator of $X$ \\
$\psi_i$, $\pi_i$ & Eigenfunctions and eigenvalues of $\mathcal{K}$ \\
$\mathcal{T} = \bigcup_{a=1}^d \mathcal{T}_a$ & Functional domain partitioned into $d$ disjoint subintervals \\
$\langle X, Y\rangle_{\mathcal{T}_a} = \int_{\mathcal{T}_a} X(t)Y(t)\,dt$ & Inner product on subinterval $\mathcal{T}_a$ \\
$D=\{1,\dots,d\}$ & Index set of $d$ domain subintervals\\
$R\subseteq D$ & Index set of chosen subintervals.\\
$\hat{X}^{R}(t)$ & Process equal to $X(t)$ on $\bigcup_{b \in R} \mathcal{T}_b$ and to the mean $\mu(t)$ elsewhere.\\
$\Delta_{\T_a} \fmd^2(\hat{X}^{R},\mu;\kappa,m)$ & Marginal outlyingness contribution in the time interval $\T_a$ to $\fmd^2$ for time intervals indexed $R$\\
$\theta_{\T_a}(X,\mu;\kappa,m)$ & Time-specific Shapley outlyingness contribution on $\T_a$ \\
\bottomrule
\end{tabular}
\end{subtable}
\vspace{0.5em}

\small\begin{subtable}[t]{\linewidth}
\centering
\caption{Multivariate functional data.}
\begin{tabular}{lp{10cm}}
\toprule
\textbf{Symbol} & \textbf{Description} \\
\midrule
$\bm X(t) = (X_1(t), \ldots, X_p(t))'\in L_p^2(\T)$ & Multivariate separable stochastic process with $p$ components \\
$\bm\mu(t) = (\mu_1(t), \ldots, \mu_p(t))'$ & Mean function of $\bm X$ \\
$\bm \Sigma_{\text{row}} \in \mathrm{PDS}(p)$ & Cross-covariance matrix between $p$ components \\
$\kappa(s,t)$ & Temporal covariance kernel \\
$\bm K(s,t) = \Sigma_{\text{row}}\,\kappa(s,t)$ & Separable covariance kernel of $\bm X$ \\
 $P = \{1, \ldots, p\}$ & Index set of the variables\\
$S \subseteq P$ & Chosen subset of variables\\
 $\mathcal{T} = \bigcup_{a=1}^d \mathcal{T}_a$ & Functional domain partitioned into $d$ disjoint subintervals \\
$\langle X, Y\rangle_{\mathcal{T}_a} = \int_{\mathcal{T}_a} X(t)Y(t)\,dt$ & Inner product on subinterval $\mathcal{T}_a$ \\
$D=\{1,\dots,d\}$ & Index set of $d$ domain subintervals\\
$R\subseteq D$ & Index set of chosen subintervals.\\
$ \Delta_{k,\T_c} \fmmd^2(\Xhat{S,R},\bm{\mu};\K,M)$ & Marginal outlyingness contribution of the $k$th coordinate function in the time interval $\T_a$ to $\fmmd^2$ for time intervals in $R$ and components in $S$\\
$\Xhat{S,R}=(\Xhatj{S,R}{1},\ldots ,\Xhatj{S,R}{p})'$ & $\Xhatj{S,R}{k}(t)$ is equal to $X_k(t)$ on $\bigcup_{b \in R} \mathcal{T}_b$ if $k\in S$ and to $\mu_k(t)$ otherwise\\
$\Theta_{k,\T_{a}}(\bm{X},\bm{\mu};\K,M)$ & Outlyingness contribution of the $k$th coordinate in the interval $\T_a$\\
$\theta_k(\bm X,\bm \mu;\bm K,M)$ & Coordinate-specific Shapley outlyingness contribution for component $k$ \\
$\theta_{T_a}(\bm X,\bm \mu;\bm K,M)$ & Time-specific Shapley outlyingness contribution in interval $T_a$ \\
$\bm\phi(t) = (\phi_1(t), \ldots, \phi_M(t))'$ & Vector of basis functions used for smoothing \\
$\bm W = \int_T \phi(t)\phi'(t)\,dt$ & Matrix of inner products of basis functions over $\T$ \\
$\bm W_{\T_a} = \int_{T_a} \phi(t)\phi'(t)\,dt$ & Inner-product matrix restricted to interval $\T_a$ \\
$\bm v^{\text{row}}_j$, $\lambda^{\text{row}}_j$ & Eigenvectors and eigenvalues of $\bm \Sigma_{\text{row}}$ \\
$\xi_i$, $\lambda^{\text{row}}_i$ & Eigenfunctions and eigenvalues of kernel $\kappa(s,t)$ \\
\bottomrule
\end{tabular}
\end{subtable}
\end{table}

The outlyingness decomposition $\bm{\theta}(\bm{x})$ \eqref{eq:shapley_md} is based on the Shapley value, it is the only decomposition of the squared Mahalanobis distance based on Equation~\eqref{eq:define_xj} that fulfills the following properties:
\begin{itemize}
\item{\emph{Efficiency:}}
The contributions $\theta_j(\bm{x})$, for $j = 1,\ldots ,p$, sum up to the squared Mahalanobis distance of $\bm{x}$, hence $\sum_{j = 1}^p \theta_j(\bm{x}) = \md^2(\bm{x})$.
\item{\emph{Symmetry:}}
If $\md^2(\xhat{S \cup \{j\}}) = \md^2(\xhat{S \cup \{k\}})$ holds for all subsets $S \subseteq P\setminus \{j,k\}$
for two coordinates $j$ and $k$, then $\theta_j(\bm{x}) = \theta_k(\bm{x})$.
\item{\emph{Monotonicity:}}
Let $\bm{\mu}, \tilde{\bm{\mu}} \in \R^p$ be two vectors and $\bm{\Sigma}, \tilde{\bm{\Sigma}} \in \pds(p)$ be two matrices. If 
\begin{align*}
        \md_{\bm{\mu},\bm{\Sigma}}^2(\xhat{S \cup \{j\}}) - \md_{\bm{\mu},\bm{\Sigma}}^2(\xhat{S}) \geq \md_{\tilde{\bm{\mu}},\tilde{\bm{\Sigma}}}^2(\xhat{S \cup \{j\}}) - \md_{\tilde{\bm{\mu}},\tilde{\bm{\Sigma}}}^2(\xhat{S})
\end{align*}
holds for all subsets $S \subseteq P$, then $\theta_j(\bm{x},\bm{\mu},\bm{\Sigma}) \geq \theta_j(\bm{x},\tilde{\bm{\mu}},\tilde{\bm{\Sigma}})$.
\end{itemize}
This means that the $k$th coordinate of the Shapley value $\theta_k(\bm{x})$ represents the average marginal contribution of the $k$th coordinate to the squared Mahalanobis distance. This is calculated by averaging over all marginal outlyingness contributions $\Delta_k \md^2(\xhat{S})$ across all possible subsets $S \subseteq P \setminus \{k\}$. The second equality in Equation~\eqref{eq:shapley_md} demonstrates that the exponential computational complexity of the Shapley value can be reduced to a linear complexity in this setting. For a proof, we refer to \cite{mayrhofer2022}. 
Equation~\eqref{eq:shapley_md} provides further insight into the Shapley value by comparing it to the squared Mahalanobis distance $\md^2(\bm{x}) = \sum_{j,k = 1}^p (x_j-\mu_j)(x_k-\mu_k)\omega_{jk}$. While $\md^2(\bm{x})$ yields an outlyingness measure that aggregates the contributions $(x_j-\mu_j)(x_k-\mu_k)\omega_{jk}$ of all $p$ variables, the outlyingness scores $\theta_k(\bm{x})$ only consider the contributions associated with the $k$th coordinate. 
In \cite{mayrhofer2024robust}, the concept was extended to the 
matrix-variate setting. For a random matrix $\bm{X} \in \R^{p \times q}$ with mean $\bm{\M} \in \R^{p \times q}$, row covariance $\bm{\Sigma}^{\row} \in \pds(p)$ and column covariance $\bm{\Sigma}^{\col} \in \pds(q)$, cellwise, rowwise, and columnwise outlyingness contributions to the squared matrix Mahalanobis distance~\eqref{eq:mmd} are given by
\begin{align}
    \bm{\Theta}(\bm{X}) &= (\bm{X}-\bm{\M}) \circ (\bm{\Sigma}^{\row})^{-1} (\bm{X}-\bm{\M}) (\bm{\Sigma}^{\col})^{-1} \in \R^{p \times q}, \label{eq:shapley_cell}\\
    \bm{\theta}_{\row}(\bm{X}) &=  \diag((\bm{\Sigma}^{\row})^{-1} (\bm{X}-\bm{\M}) (\bm{\Sigma}^{\col})^{-1} (\bm{X}-\bm{\M})') \in \R^{p} \text{, }\\
    \bm{\theta}_{\col}(\bm{X}) &= \diag((\bm{X}-\bm{\M})'(\bm{\Sigma}^{\row})^{-1} (\bm{X}-\bm{\M}) (\bm{\Sigma}^{\col})^{-1}) \in \R^{q} ,
\end{align}
respectively, where $\circ$ is an element-wise product.
The cellwise Shapley values~\eqref{eq:shapley_cell} are based on the multivariate Shapley values~\eqref{eq:shapley_md} of vectorized observations. 
The row- and columnwise Shapley values can be obtained by adding the cellwise Shapley values for the respective row or column, or by adjusting the individual contributions for row-wise replacements.

\subsection{Univariate Functional Data}

\begin{proof}[Proof of Proposition~\ref{proposition:shapley_fmd}]
    The outlyingness contribution $\theta_{\T_a}(X;m) = \theta_{\T_a}(X,\mu;\kappa,m)$ of the $k$th coordinate to $\fmd^2(X,m)  = \fmd^2(X,\mu;\kappa,m)$ is given as the weighted average of the marginal outlyingness contributions $\Delta_{\T_a} \fmd^2(\hat{X}^{R};m) = \Delta_{\T_a} \fmd^2(\hat{X}^{R},\mu;\kappa,m)$. The marginal outlyingness contributions $\Delta_{\T_a} \fmd^2(\hat{X}^{R};m)$ can be simplified as follows:
    \begin{align*}
        \Delta_{\T_a} \fmd^2(\hat{X}^{R};m) &= \fmd^2(\hat{X}^{R\cup\{a\}};m)-\fmd^2(\hat{X}^{R};m)\\
        &= \sum_{i = 1}^m \frac{1}{\lambda_i} \left( \left(\sum_{b = 1}^d \innerproduct{\hat{X}^{R \cup \{a\}}}{\ueig_i}_{\T_b}\right)^2 - \left(\sum_{b = 1}^d \innerproduct{\hat{X}^{R}}{\ueig_i}_{\T_b}\right)^2 \right)\\
        &= \sum_{i = 1}^m \frac{1}{\lambda_i} \left( \sum_{b \in R \cup \{a\}} \sum_{c \in R \cup \{a\}} \innerproduct{X}{\ueig_i}_{\T_b} \innerproduct{X}{\ueig_i}_{\T_c}- \sum_{b \in R}  \sum_{c \in R} \innerproduct{X}{\ueig_i}_{\T_b} \innerproduct{X}{\ueig_i}_{\T_c}\right) \\
        &= \sum_{i = 1}^m \frac{1}{\lambda_i} \left( \sum_{b \in R} \sum_{c \in R} \innerproduct{X}{\ueig_i}_{\T_b} \innerproduct{X}{\ueig_i}_{\T_c} + 2\innerproduct{X}{\ueig_i}_{\T_a}\sum_{b \in R} \innerproduct{X}{\ueig_i}_{\T_b} \right. \\
        &\left.\qquad\qquad\quad+  \innerproduct{X}{\ueig_i}_{\T_a}^2 - \sum_{b \in R}  \sum_{c \in R} \innerproduct{X}{\ueig_i}_{\T_b} \innerproduct{X}{\ueig_i}_{\T_c}\right) \\
        &= \sum_{i = 1}^m \frac{1}{\lambda_i} \left(2\innerproduct{X}{\ueig_i}_{\T_a}\sum_{b \in R} \innerproduct{X}{\ueig_i}_{\T_b} +  \innerproduct{X}{\ueig_i}_{\T_a}^2\right).
    \end{align*} 

    Further, we write
    \begin{align*}
        w(\abs{R}) := \frac{r!(d-r-1)!}{d!},
    \end{align*}
    with $r = \abs{R}$ for which $\sum_{R \subseteq D\setminus\{a\}} w(\abs{R}) = 1$ holds. With this, the time-specific outlyingness contribution within the subinterval $\T_a$ based on the Shapley value simplifies as follows

    {\allowdisplaybreaks
    \begin{align*}
        \theta_{\T_a}(X,\mu;\kappa,m)
         &= \sum_{R \subseteq D\setminus\{a\}} \frac{\abs{R}!(d-\abs{R}-1)!}{(d)!} \Delta_{\T_a} \fmd^2(\hat{X}^{R};m)\\
         &= \sum_{R \subseteq D\setminus\{a\}} w(\abs{R}) \left(\sum_{i = 1}^m \frac{1}{\lambda_i} \left(2\innerproduct{X}{\ueig_i}_{\T_a}\sum_{b \in R} \innerproduct{X}{\ueig_i}_{\T_b} +  \innerproduct{X}{\ueig_i}_{\T_a}^2\right)\right)\\
         &= \sum_{i = 1}^m \frac{1}{\lambda_i} \innerproduct{X}{\ueig_i}_{\T_a}^2 + 2\sum_{i = 1}^m \frac{1}{\lambda_i}\innerproduct{X}{\ueig_i}_{\T_a}\sum_{R \subseteq D\setminus\{a\}} w(\abs{R})  \sum_{b \in R} \innerproduct{X}{\ueig_i}_{\T_b} \\
         &= \sum_{i = 1}^m \frac{1}{\lambda_i} \innerproduct{X}{\ueig_i}_{\T_a}^2 + 2\sum_{i = 1}^m \frac{1}{\lambda_i}\innerproduct{X}{\ueig_i}_{\T_a} \sum_{r = 1}^{d-1} w(r) \sum_{\substack{R \subseteq D\setminus\{a\}\\ \abs{R} = r}} \sum_{b \in R} \innerproduct{X}{\ueig_i}_{\T_b} \\
         &= \sum_{i = 1}^m \frac{1}{\lambda_i} \innerproduct{X}{\ueig_i}_{\T_a}^2 + 2\sum_{i = 1}^m \frac{1}{\lambda_i}\innerproduct{X}{\ueig_i}_{\T_a} \sum_{r = 1}^{d-1} w(r) \binom{d-2}{r-1} \sum_{b \in D\setminus \{a\}} \innerproduct{X}{\ueig_i}_{\T_b} \\
         &= \sum_{i = 1}^m \frac{1}{\lambda_i} \innerproduct{X}{\ueig_i}_{\T_a}^2 + 2\sum_{i = 1}^m \frac{1}{\lambda_i}\innerproduct{X}{\ueig_i}_{\T_a} \sum_{r = 1}^{d-1} \frac{r}{d(d-1)} \sum_{b \in D\setminus \{a\}} \innerproduct{X}{\ueig_i}_{\T_b} \\
         &= \sum_{i = 1}^m \frac{1}{\lambda_i} \innerproduct{X}{\ueig_i}_{\T_a}^2 + \sum_{i = 1}^m \frac{1}{\lambda_i}\innerproduct{X}{\ueig_i}_{\T_a} \sum_{b \in D\setminus \{a\}} \innerproduct{X}{\ueig_i}_{\T_b} \\
         &= \sum_{i = 1}^m \frac{1}{\lambda_i}\innerproduct{X}{\ueig_i}_{\T_a} \sum_{b \in D} \innerproduct{X}{\ueig_i}_{\T_b} \\
         &= \sum_{i = 1}^m \frac{1}{\lambda_i}\innerproduct{X}{\ueig_i}_{\T_a} \innerproduct{X}{\ueig_i}.
    \end{align*}}
\end{proof}

\begin{proof}[Proof of Lemma~\ref{lemma:shapley_using_coefficents_univariate}]
    We have that $X(t) = \bm{a}'\mbasis(t)$, $\mu(t) = \bm{m}_{\bm{a}}'\mbasis(t)$, and $\kappa(s,t) = \mbasis'(s)\bm{\Sigma}\mbasis(t)$, for $s,t \in \T$. Based on Proposition~\ref{proposition:shapley_fmd} we obtain
    \begin{align*}
        \phi_{\T_a}(X,\mu,\kappa;m) 
        &= \sum_{i = 1}^m \frac{1}{\lambda_i}\innerproduct{X - \mu}{\ueig_i}_{\T_a} \innerproduct{X - \mu}{\ueig_i}\\
        &= \sum_{i = 1}^m \frac{1}{\lambda_i}\innerproduct{(\bm{a} - \bm{m}_{\bm{a}})'\bm{\phi}}{\bm{\phi}'\bm{b}_i}_{\T_a} \innerproduct{(\bm{a} - \bm{m}_{\bm{a}})'\bm{\phi}}{\bm{\phi}'\bm{b}_i}\\
        &= \sum_{i = 1}^m \frac{1}{\lambda_i}(\bm{a}-\bm{m}_{\bm{a}})'\innerproduct{\mbasis}{\mbasis'}_{\T_a}\bm{b}_i (\bm{a}-\bm{m}_{\bm{a}})'\innerproduct{\mbasis}{\mbasis'}\bm{b}_i\\
        &= (\bm{a}-\bm{m}_{\bm{a}})'\bm{W}_{\T_a} \underbrace{\left(\sum_{i = 1}^m \frac{1}{\lambda_i} \bm{b}_i \bm{b}_i'\right)\bm{W}}_{= \bm{W}^{-1}(\bm{\Sigma}^{\col})^{-1} \text{, see Eq.~\eqref{eq:inv_eigen}}} (\bm{a}-\bm{m}_{\bm{a}}) \\
        &= (\bm{a}-\bm{m}_{\bm{a}})'\bm{W}_{\T_a} \bm{W}^{-1}\bm{\Sigma}^{-1} (\bm{a}-\bm{m}_{\bm{a}}),
    \end{align*}
    where $\bm{W}=\innerproduct{\mbasis}{\mbasis'}_{\T}:=\int_{\T}\mbasis(t)\mbasis'(t)\mathrm{d}t$ and $\bm{W}_{\T_a} = \innerproduct{\mbasis}{\mbasis'}_{\T_a} := \int_{\T_a}\mbasis(t)\mbasis'(t)\mathrm{d}t$.
\end{proof}

\subsection{Multivariate Functional Data}%
\begin{proof}[Proof of Proposition~\ref{proposition:shapley_cell}]
    Let $Y$ denote the univariate process concatenating each coordinate function $X_j$ for $j = 1,\dots,p$ of $\bm{X}$ which is then defined on the interval $\Tilde{\T} = \underbrace{\T \cup \cdots \cup \T}_{p\text{ times}}$. 
    With this concatenating approach, the mean function for $\Tilde{t} \in [\T \cdot (j-1),\T \cdot j] \subset \Tilde{\T}, j = 1,\dots,p,$ is given as
    \begin{align*}
        \mu(\Tilde{t}) = \mu_j(t),
    \end{align*}
    with $t \in \T$,
    and the covariance function for $\Tilde{s} \in [\T \cdot (j-1),\T \cdot j] \subset \Tilde{\T}, j = 1,\dots,p,$ and $\Tilde{t} \in [\T \cdot (k-1),\T \cdot k] \subset \Tilde{\T}, k = 1,\dots,p,$ is given as 
    \begin{align*}
        \kappa(\Tilde{s},\Tilde{t}) = \kappa_{jk}(s,t),
    \end{align*}
    with $s,t \in \T$. The eigendecomposition of $Y$ is then given as 
    \begin{align*}
        \mathcal{K}\ueig_i= \pi_i \ueig_i  \quad \text{and} \quad \kappa(\Tilde{s},\Tilde{t}) &= \sum_{i = 1}^\infty \pi_i \ueig_i(\Tilde{s}) \ueig_i'(\Tilde{t}),
    \end{align*}
    with $\ueig_i(\Tilde{t}) = \meig_i(t)\bm{e}_j = \psi_{i,j}(t)$ for $\Tilde{t} \in [\T \cdot (j-1),\T \cdot j] \subset \Tilde{\T}, j = 1,\dots,p$. 
    Let $\T_{a_{k}} \subseteq [\T \cdot (k-1),\T \cdot k] \subset \Tilde{\T}, a \in \{1,\dots,d\}, k \in \{1,\dots,p\}$, $\tilde{D} = \{1,\dots,pd\}$, and $\Tilde{\T} = \bigcup_{j = 1}^p \bigcup_{b = 1}^d \T_{b_j}$, where $\T_{b_j} = \T_{b}$ and $\T = \bigcup_{b = 1}^d \T_{b}$.
    Then we can define time-coordinate-specific outlyingness contributions based on Shapley values using Proposition~\ref{proposition:shapley_fmd}, which yields
    \begin{align*}
        \Theta_{k,\T_{a}}(\bm{X},\bm{\mu};\K,M) &=
        \theta_{\T_{a_{k}}}(Y,\mu;\kappa,M) \\
        &= \sum_{R \subseteq \Tilde{D}\setminus\{a_{k}\}} \frac{\abs{R}!(d-\abs{R}-1)!}{(d)!} \Delta_{\T_{a_{k}}} \fmd^2(\hat{Y}^{R},\mu;\kappa,M) \\
        &= \sum_{i = 1}^M \frac{1}{\pi_i}\innerproduct{Y - \mu}{\ueig_i}_{\T_{a_{k}}} \innerproduct{Y - \mu}{\ueig_i}_{\Tilde{\T}}
        = \sum_{i = 1}^M \frac{1}{\pi_i}\innerproduct{X_k - \mu_k}{\psi_{i,k}}_{\T_a} \innerproduct{\bm{X} - \bm{\mu}}{\meig_i}_{\T}
    \end{align*}
\end{proof}

\begin{proof}[Proof of Corollary~\ref{corollary:shapley_separable_cell}]
    From Equation~\eqref{eq:eigen_eq_mpg} we get that $\psi_{kl}(t) = \meig_k'(t)\bm{e}_l = \ueig_i(t) \bm{v}_j^{\row} \bm{e}_l = \ueig_i(t) {v}^{\row}_{j,l}$, for some $i,j,k$, hence
    \begin{align*}
        \Theta_{k,\T_{a}}(\bm{X},\bm{\mu};\bm{\Sigma}^{\row}\kappa,M) 
        &= \sum_{i = 1}^M \frac{1}{\pi_i}\innerproduct{X_k}{\psi_{i,k}}_{\T_a} \innerproduct{\bm{X}}{\meig_i}_{\T} \\
        &= \sum_{i = 1}^m\sum_{j = 1}^p \frac{1}{\lambda_i^{\ker}\lambda_j^{\row}}\left(\innerproduct{X_k}{\ueig_i {v}^{\row}_{j,k}}_{\T_a} \sum_{l = 1}^p\innerproduct{X_l}{\ueig_i {v}^{\row}_{j,l}}_{\T}\right).\\
        &= \sum_{i = 1}^m\sum_{j = 1}^p \frac{1}{\lambda_i^{\ker}\lambda_j^{\row}}\left(\innerproduct{X_k}{\ueig_i}_{\T_a}{v}^{\row}_{j,k} \sum_{l = 1}^p\innerproduct{X_l}{\ueig_i}_{\T}{v}^{\row}_{j,l}\right).\\
        &= \sum_{i = 1}^m\sum_{j = 1}^p \frac{1}{\lambda_i^{\ker}\lambda_j^{\row}}\left(\alpha_{ki}^{\T_a}{v}^{\row}_{j,k} \sum_{l = 1}^p\alpha_{li}^{\T}{v}^{\row}_{j,l}\right).
    \end{align*}
\end{proof}

\begin{proof}[Proof of Lemma~\ref{lemma:shapley_coef_fmmd_cell}]
    We have that $\bm{X}(t)=\bm{A}'\mbasis(t)$, $\bm{\mu}(t) = \bm{M}_{\bm{A}}'\mbasis(t)$, and $\kappa(s,t) = \mbasis'(s)\bm{\Sigma}^{\col}\mbasis(t)$, for $s,t \in \T$, see Theorem~\ref{theorem:fmmd_basis}. In combination with Corollary~\ref{corollary:shapley_separable_cell} we obtain 
    \begin{align*}
        \Theta_{\T_a,k}(\bm{X},\bm{\mu};\K,M) 
        &=  \sum_{i = 1}^m\sum_{j = 1}^p \frac{1}{\lambda_i^{\ker}\lambda_j^{\row}}\left(\innerproduct{X_k - \mu_k}{\ueig_i}_{\T_a}{v}^{\row}_{j,k} \sum_{l = 1}^p\innerproduct{X_l}{\ueig_i}_{\T}{v}^{\row}_{j,l}\right)\\
        &=  \sum_{i = 1}^m\sum_{j = 1}^p \frac{1}{\lambda_i^{\ker}\lambda_j^{\row}}\left(\innerproduct{(\bm{a}_k-\bm{m}_{\bm{A},k})'\mbasis}{\mbasis'\bm{b}_i}_{\T_a}{v}^{\row}_{j,k} \sum_{l = 1}^p\innerproduct{(\bm{a}_l-\bm{m}_{\bm{A},l})'\mbasis}{\mbasis'\bm{b}_i}_{\T}{v}^{\row}_{j,l}\right)\\
        &=  \sum_{i = 1}^m\sum_{j = 1}^p \frac{1}{\lambda_i^{\ker}\lambda_j^{\row}}\left((\bm{a}_k-\bm{m}_{\bm{A},k})'\innerproduct{\mbasis}{\mbasis'}_{\T_a}\bm{b}_i{v}^{\row}_{j,k} \sum_{l = 1}^p\bm{b}_i'\innerproduct{\mbasis}{\mbasis'}_{\T}(\bm{a}_l-\bm{m}_{\bm{A},l}){v}^{\row}_{j,l}\right)\\
        &=  \sum_{i = 1}^m\sum_{j = 1}^p \frac{1}{\lambda_i^{\ker}\lambda_j^{\row}}\left((\bm{a}_k-\bm{m}_{\bm{A},k})'\bm{W}_{\T_a}\bm{b}_i{v}^{\row}_{j,k} \sum_{l = 1}^p\bm{b}_i'\bm{W}(\bm{a}_l-\bm{m}_{\bm{A},l}){v}^{\row}_{j,l}\right)\\
        &=  \sum_{j = 1}^p \frac{1}{\lambda_j^{\row}}{v}^{\row}_{j,k}(\bm{a}_k-\bm{m}_{\bm{A},k})'\bm{W}_{\T_a}\underbrace{\left(\sum_{i = 1}^m\frac{1}{\lambda_i^{\ker}}\bm{b}_i \bm{b}_i'\right)\bm{W}}_{= \bm{W}^{-1}(\bm{\Sigma}^{\col})^{-1} \text{, see Eq.~\eqref{eq:inv_eigen}}}\sum_{l = 1}^p(\bm{a}_l-\bm{m}_{\bm{A},l}){v}^{\row}_{j,l}\\
        &=  \sum_{j = 1}^p \frac{1}{\lambda_j^{\row}}{v}^{\row}_{j,k}(\bm{a}_k-\bm{m}_{\bm{A},k})'\bm{W}_{\T_a}\bm{W}^{-1}(\bm{\Sigma}^{\col})^{-1}\sum_{l = 1}^p(\bm{a}_l-\bm{m}_{\bm{A},l}){v}^{\row}_{j,l}\\
        &=  \sum_{j = 1}^p \frac{1}{\lambda_j^{\row}}{v}^{\row}_{j,k}(\bm{a}_k-\bm{m}_{\bm{A},k})'\bm{W}_{\T_a}\bm{W}^{-1}(\bm{\Sigma}^{\col})^{-1}(\bm{A}-\bm{M}_{\bm{A}})'\bm{v}^{\row}_{j},
    \end{align*}
    where $\bm{W}=\innerproduct{\mbasis}{\mbasis'}_{\T}:=\int_{\T}\mbasis(t)\mbasis'(t)\mathrm{d}t$ and $\bm{W}_{\T_a} = \innerproduct{\mbasis}{\mbasis'}_{\T_a} := \int_{\T_a}\mbasis(t)\mbasis'(t)\mathrm{d}t$.
\end{proof}

\subsubsection{Time- and Coordinate-specific Outlyingness Contributions}
For the time-specific or coordinate-specific outlyingness contributions, we can either rely on the concatenation approach used to prove Proposition~\ref{proposition:shapley_cell}, or we can directly work with the multivariate process, which is illustrated in the following for the coordinate-specific contributions. Note that the same approach is also possible with time-coordinate-specific contributions, but it would introduce laborious notation, and hence, concatenation was used instead.

Let us consider the $p$-variate stochastic process $\bm{X} \sim \msp(\bm{\mu},\K)$, with $\bm{X} = (X_1,\ldots ,X_p)'$, $P = \{1,\ldots ,p\}$ the index set of variables, and $(\pi_k,\meig_k), k = 1,\dots,M$ the eigenpairs of the covariance operator $\C$ with kernel $\K$. We determine the contribution of the $k$th coordinate function $X_k$ to $\fmmd(\bm{X},\bm{\mu};\K,M)$ using Shapley values. Similar as in Equation~\eqref{eq:define_xj}, the marginal outlyingness contributions to $\fmmd^2$ are given by
\begin{equation*}
    \Delta_k \fmmd^2(\Xhat{S},\bm{\mu};\K,M) := \fmmd^2(\Xhat{S\cup\{k\}},\bm{\mu};\K,M)-\fmmd^2(\Xhat{S},\bm{\mu};\K,M) %
\end{equation*} 
with $\Xhat{S}=(\Xhatj{S}{1},\ldots ,\Xhatj{S}{p})'$, $S \subseteq P$, and
\begin{equation}
\label{eq:define_fXj}
\Xhatj{S}{j}:= \begin{cases}
        X_j & \text{if } j \in S\\
        \mu_j & \text{if } j \notin S
    \end{cases}.
\end{equation}

\begin{proposition} \label{proposition:shapley_fmmd}
For $\bm{X} \sim \msp(\bm{\mu},\K)$ and $\Delta_k \fmmd^2(\Xhat{S},\bm{\mu};\K,M)$ as in Equation~\eqref{eq:delta_univ_time},
the coordinate-specific outlyingness contribution to $\fmmd^2(\bm{X},\bm{\mu};\K,M)$ of the $k$th coordinate function based on the Shapley value is given by
    \begin{align}
        \theta_k(\bm{X},\bm{\mu};\K,M) 
         :=& \sum_{S \subseteq P\setminus\{k\}} \frac{\abs{S}!(p-\abs{S}-1)!}{p!} \Delta_k \fmmd^2(\Xhat{S},\bm{\mu};\K,M) \nonumber\\
         =& \sum_{i=1}^M\frac{1}{\pi_i} \innerproduct{X_{k} - \mu_k}{\psi_{i,k}}  \innerproduct{\bm{X} - \bm{\mu}}{\meig_{i}}, \label{eq:shapley_coordinate_multivariate}
    \end{align}
    with $\psi_{i,k} = \meig_i'\bm{e}_k$ denoting the $k$th component of the $i$th eigenfunction $\meig_i$ of covariance operator $\C$ with kernel $\K$. 
\end{proposition}
\begin{proof}
    The outlyingness contribution $\theta_k(\bm{X};M) = \theta_k(\bm{X},\bm{\mu};\K,M)$ of the $k$th coordinate to $\fmmd^2(\bm{X};M) = \fmmd^2(\bm{X},\bm{\mu};\K,M)$ is given as the weighted average of the marginal outlyingness contributions $\Delta_k \fmmd^2(\Xhat{S};M) = \Delta_k \fmmd^2(\Xhat{S},\bm{\mu};\K,M)$. Without loss of generality, we assume that the data are centered to simplify the notation, i.e., $\bm{\mu}(t) = \bm{0} = (0,\dots,0)'\in \R^p, t \in \T$, simply denoted as $\bm{\mu} = \bm{0}$. Then the marginal outlyingness contributions can be written as
    {\allowdisplaybreaks\begin{align*}
        \Delta_k \fmmd^2(\Xhat{S};M)
        &= \fmmd^2(\Xhat{S\cup\{k\}};M)-\fmmd^2(\Xhat{S};M)
        = \sum_{i=1}^M\frac{1}{\pi_i} \innerproduct{\Xhat{S\cup\{k\}}}{\meig_i}^2 -  \sum_{i=i}^M\frac{1}{\pi_i} \innerproduct{\Xhat{S}}{\meig_i}^2 \\
        &= \sum_{i=1}^M\frac{1}{\pi_i} \left(\left(\sum_{j = 1}^p \innerproduct{\Xhatj{S\cup\{k\}}{j}}{\psi_{i,j}}\right)^2 - \left(\sum_{j = 1}^p \innerproduct{\Xhatj{S}{j}}{\psi_{i,j}}\right)^2\right) \\
        &= \sum_{i=1}^M\frac{1}{\pi_i} \left(\sum_{j \in S \cup k} \sum_{l \in S \cup k} \innerproduct{X_{l}}{\psi_{i,l}} \innerproduct{X_j}{\psi_{i,j}} - \sum_{j \in S} \sum_{l \in S} \innerproduct{X_{l}}{\psi_{i,l}} \innerproduct{X_j}{\psi_{i,j}}\right) \\
        &= \sum_{i=1}^M\frac{1}{\pi_i} \left(\sum_{j \in S \cup k} \sum_{l \in S} \innerproduct{X_{l}}{\psi_{i,l}} \innerproduct{X_j}{\psi_{i,j}} + \innerproduct{X_{k}}{\psi_{i,k}}\sum_{j \in S \cup k}  \innerproduct{X_j}{\psi_{i,j}}  \right.\\ 
        &\left.\qquad\qquad\quad- \sum_{j \in S} \sum_{l \in S} \innerproduct{X_{l}}{\psi_{i,l}} \innerproduct{X_j}{\psi_{i,j}}\right)\\
        &= \sum_{i=1}^M\frac{1}{\pi_i} \left(\sum_{j \in S} \sum_{l \in S} \innerproduct{X_{l}}{\psi_{i,l}} \innerproduct{X_j}{\psi_{i,j}} + \innerproduct{X_{k}}{\psi_{i,k}} \sum_{j \in S \cup k}  \innerproduct{X_j}{\psi_{i,j}} \right. \\
        &\left.\qquad\qquad\quad+ \innerproduct{X_k}{\psi_{i,k}} \sum_{l \in S} \innerproduct{X_{l}}{\psi_{i,l}}  - \sum_{j \in S} \sum_{l \in S} \innerproduct{X_{l}}{\psi_{i,l}} \innerproduct{X_j}{\psi_{i,j}} \right)\\
        &= \sum_{i=1}^M\frac{1}{\pi_i} \left(\innerproduct{X_{k}}{\psi_{i,k}}^2 + 2\innerproduct{X_{k}}{\psi_{i,k}} \sum_{j \in S}  \innerproduct{X_j}{\psi_{i,j}}\right) \\
    \end{align*} }
    
    Further, we write
    \begin{align*}
        w(\abs{S}) := \frac{\abs{S}!(p-\abs{S}-1)!}{p!},
    \end{align*}
    for which $\sum_{S \subseteq P\setminus\{k\}} w(\abs{S}) = 1$ holds.
    Then the contribution of the $k$th coordinate to the squared truncated functional Mahalanobis distance $\fmmd^2(\bm{X};M)$ based on the Shapley value is given by 
    \begin{align*}
             \theta_k(\bm{X};M) 
             &= \sum_{S \subseteq P\setminus\{k\}} w(\abs{S}) \Delta_k \fmmd^2(\Xhat{S};M)  \\
             &= \sum_{S \subseteq P\setminus\{k\}} w(\abs{S}) \left(\sum_{i=1}^M\frac{1}{\pi_i} \left(\innerproduct{X_{k}}{\psi_{i,k}}^2 + 2\innerproduct{X_{k}}{\psi_{i,k}} \sum_{j \in S}  \innerproduct{X_j}{\psi_{i,j}}\right)\right)\\
             &= \sum_{i=1}^M\frac{1}{\pi_i} \innerproduct{X_{k}}{\psi_{i,k}}^2  + 2\sum_{S \subseteq P\setminus\{k\}} w(\abs{S}) \left(\sum_{i=1}^M\frac{1}{\pi_i} \innerproduct{X_{k}}{\psi_{i,k}} \sum_{j \in S}  \innerproduct{X_j}{\psi_{i,j}}\right)\\
             &= \sum_{i=1}^M\frac{1}{\pi_i} \innerproduct{X_{k}}{\psi_{i,k}}^2  + 2 \sum_{i=1}^M\frac{1}{\pi_i} \innerproduct{X_{k}}{\psi_{i,k}} \sum_{s = 1}^{p-1}  w(s) \sum_{\substack{S \subseteq P\setminus\{k\}\\ \abs{S} = s}} \sum_{j \in S}  \innerproduct{X_j}{\psi_{i,j}}\\
             &= \sum_{i=1}^M\frac{1}{\pi_i} \innerproduct{X_{k}}{\psi_{i,k}}^2  + 2 \sum_{i=1}^M\frac{1}{\pi_i} \innerproduct{X_{k}}{\psi_{i,k}} \sum_{s = 1}^{p-1}  w(s) \binom{p-2}{s-1} \sum_{j \in P\setminus\{k\}}  \innerproduct{X_j}{\psi_{i,j}}\\
             &= \sum_{i=1}^M\frac{1}{\pi_i} \innerproduct{X_{k}}{\psi_{i,k}}^2  + 2 \sum_{i=1}^M\frac{1}{\pi_i} \innerproduct{X_{k}}{\psi_{i,k}} \sum_{s = 1}^{p-1}  \frac{s }{p(p-1)} \sum_{j \in P\setminus\{k\}}  \innerproduct{X_j}{\psi_{i,j}}\\
             &= \sum_{i=1}^M\frac{1}{\pi_i} \innerproduct{X_{k}}{\psi_{i,k}}^2  +  \sum_{i=1}^M\frac{1}{\pi_i} \innerproduct{X_{k}}{\psi_{i,k}} \sum_{j \in P\setminus\{k\}}  \innerproduct{X_j}{\psi_{i,j}}\\
             &= \sum_{i=1}^M\frac{1}{\pi_i} \left(\innerproduct{X_{k}}{\psi_{i,k}}^2  +  \innerproduct{X_{k}}{\psi_{i,k}} \sum_{j \in P\setminus\{k\}}  \innerproduct{X_j}{\psi_{i,j}}\right)\\
             &= \sum_{i=1}^M\frac{1}{\pi_i} \left(\innerproduct{X_{k}}{\psi_{i,k}} \sum_{j = 1}^p  \innerproduct{X_j}{\psi_{i,j}}\right)
             = \sum_{i=1}^M\frac{1}{\pi_i} \left(\innerproduct{X_{k}}{\psi_{i,k}}  \innerproduct{\bm{X}}{\meig_{i}}\right)
    \end{align*}
\end{proof} 

\begin{corollary} \label{corollary:shapley_separable}
    Let $\bm{X}\sim\msp(\bm{0},\bm{\Sigma}^{\row},\kappa)$ with covariance operator $\C = \bm{\Sigma}^{\row}\mathcal{K}$ and covariance kernel $\K(s,t) = \bm{\Sigma}^{\row} \kappa(s,t)$, then 
    \begin{align*}
        \theta_k(\bm{X},\bm{\mu};\bm{\Sigma}^{\row}\kappa,M) 
        = \sum_{i=1}^m \sum_{j=1}^p \frac{1}{\lambda^{\ker}_i \lambda^{\row}_j} \left(\innerproduct{X_k - \mu_k}{\ueig_i} {v}^{\row}_{j,k} \sum_{l = 1}^p  \innerproduct{X_l - \mu_l}{\ueig_i} {v}^{\row}_{j,l}\right).
    \end{align*}
    Here, $(\lambda^{\ker}_i,\ueig_i),\, i = 1,\dots,m$, denote the $m$ largest eigenpairs of $\mathcal{K}$, $(\lambda^{\row}_j,\bm{v}^{\row}_j),\, j = 1,\dots,p$, the eigenpairs of $\bm{\Sigma}^{\row}$, and ${v}^{\row}_{j,k} = \bm{e}_k'\bm{v}^{\row}_j$.
\end{corollary}
\begin{proof}
    From Equation~\eqref{eq:eigen_eq_mpg} we get that $\psi_{kl}(t) = \meig_k'(t)\bm{e}_l = \ueig_i(t) \bm{v}_j^{\row} \bm{e}_l = \ueig_i(t) {v}^{\row}_{j,l}$, for some $i,j,k$, hence
    \begin{align*}
        \theta_k(\bm{X},\bm{\mu};\K,M) 
        &= \sum_{i=1}^M\frac{1}{\pi_i} \left(\innerproduct{X_{k}}{\psi_{i,k}} \sum_{l = 1}^p  \innerproduct{X_l}{\psi_{i,l}}\right) \\
        &= \sum_{i=1}^m \sum_{j=1}^p \frac{1}{\lambda^{\ker}_i \lambda^{\row}_j}\left(\innerproduct{X_{k}}{\ueig_i {v}^{\row}_{j,k}} \sum_{l = 1}^p  \innerproduct{X_l}{\ueig_i {v}^{\row}_{j,l}}\right) \\
        &= \sum_{i=1}^m \sum_{j=1}^p \frac{1}{\lambda^{\ker}_i \lambda^{\row}_j} \left(\innerproduct{X_{k}}{\ueig_i } {v}^{\row}_{j,k} \sum_{l = 1}^p  \innerproduct{X_l}{\ueig_i} {v}^{\row}_{j,l}\right) \\
        &= \sum_{i=1}^m \sum_{j=1}^p \frac{1}{\lambda^{\ker}_i \lambda^{\row}_j} \left(\alpha_{ki} {v}^{\row}_{j,k} \sum_{l = 1}^p  \alpha_{li} {v}^{\row}_{j,l}\right) \\
    \end{align*}
\end{proof}

We can efficiently compute the outlyingness contributions for all $p$ variables using matrix operations using the same notation as in \eqref{eq:shapley_cell_computation}, i.e., the vector $\bm{\theta}(\bm{X},\bm{\mu};\K,M)$ with entries $\theta_k(\bm{X},\bm{\mu};\bm{\Sigma}^{\row}\kappa,M)$, for $k = 1,\dots,p$, can be computed as 
\begin{align*}
    \bm{\theta}(\bm{X},\bm{\mu};\K,M) = \diag(\tilde{\bm{A}}^{\T} (\bm{D}^{\ker})^{-1} (\tilde{\bm{A}}^{\T})' (\bm{\Sigma}^{\row})^{-1}).
\end{align*}

\begin{lemma}\label{lemma:shapley_using_coefficents_multivariate}
    Let $\bm{X} \sim \msp(\bm{\mu}, \bm{\Sigma}^{\row}, \kappa)$ be a rank $M \in \N$ multivariate stochastic process as in Theorem~\ref{theorem:fmmd_basis}, with $\bm{X}(t)=\bm{A}'\mbasis(t)$, $\bm{\mu}(t) = \bm{M}_{\bm{A}}'\mbasis(t)$, and $\kappa(s,t) = \mbasis'(s)\bm{\Sigma}^{\col}\mbasis(t)$, for $s,t \in \T$. Then it holds that
    \begin{align*}
        \theta_k(\bm{X},\bm{\mu};\K,M) 
        &=  \sum_{j=1}^p \frac{1}{\lambda^{\row}_j} {v}^{\row}_{j,k}  (\bm{a}_k-\bm{m}_{\bm{A},k})' (\bm{\Sigma}^{\col})^{-1} (\bm{A}-\bm{M}_{\bm{A}})'\bm{v}^{\row}_j,
    \end{align*}
    with $(\lambda^{\row}_j,\bm{v}^{\row}_j),\, j = 1,\dots,p$, the eigenpairs of $\bm{\Sigma}^{\row}$, and ${v}^{\row}_{j,k} = \bm{e}_k'\bm{v}^{\row}_j$.
\end{lemma}
\begin{proof}[Proof of Lemma~\ref{lemma:shapley_using_coefficents_multivariate}]
    We have that $\bm{X}(t)=\bm{A}'\mbasis(t)$, $\bm{\mu}(t) = \bm{M}_{\bm{A}}'\mbasis(t)$, and $\kappa(s,t) = \mbasis'(s)\bm{\Sigma}^{\col}\mbasis(t)$, for $s,t \in \T$, see Theorem~\ref{theorem:fmmd_basis}. Combined with Corollary~\ref{corollary:shapley_separable} we obtain

    {\allowdisplaybreaks\begin{align*}
        \theta_k(\bm{X},\bm{\mu};\K,M) 
        &= \sum_{i=1}^m \sum_{j=1}^p \frac{1}{\lambda^{\ker}_i \lambda^{\row}_j} \left(\innerproduct{X_{k}}{\ueig_i } {v}^{\row}_{j,k} \sum_{l = 1}^p  \innerproduct{X_l}{\ueig_i} {v}^{\row}_{j,l}\right) \\
        &= \sum_{i=1}^m \sum_{j=1}^p \frac{1}{\lambda^{\ker}_i \lambda^{\row}_j} \left(\bm{a}_k'\innerproduct{\mbasis}{\mbasis'}\bm{b}_i {v}^{\row}_{j,k} \sum_{l = 1}^p  \bm{a}_l'\innerproduct{\mbasis}{\mbasis'}\bm{b}_i {v}^{\row}_{j,l}\right) \\
        &= \sum_{i=1}^m \sum_{j=1}^p \frac{1}{\lambda^{\ker}_i \lambda^{\row}_j} \left(\bm{a}_k'\bm{W}\bm{b}_i {v}^{\row}_{j,k} \sum_{l = 1}^p  \bm{a}_l'\bm{W}\bm{b}_i {v}^{\row}_{j,l}\right) \\
        &=  \sum_{j=1}^p \frac{1}{\lambda^{\row}_j} {v}^{\row}_{j,k} \sum_{l = 1}^p  \bm{a}_k' \underbrace{\bm{W}\left(\sum_{i=1}^m \frac{1}{\lambda^{\ker}_i} \bm{b}_i \bm{b}_i' \right)\bm{W}}_{= (\bm{\Sigma}^{\col})^{-1} \text{, see Eq.~\eqref{eq:inv_eigen}}} \bm{a}_l {v}^{\row}_{j,l} \\
        &=  \sum_{j=1}^p \frac{1}{\lambda^{\row}_j} {v}^{\row}_{j,k} \sum_{l = 1}^p  \bm{a}_k' (\bm{\Sigma}^{\col})^{-1} \bm{a}_l {v}^{\row}_{j,l} \\
        &=  \sum_{j=1}^p \frac{1}{\lambda^{\row}_j} {v}^{\row}_{j,k}  \bm{a}_k'(\bm{\Sigma}^{\col})^{-1}  \bm{A}'\bm{v}^{\row}_j,
    \end{align*}}
    where $\bm{W}=\innerproduct{\mbasis}{\mbasis'}:=\int_{\T}\mbasis(t)\mbasis'(t)\mathrm{d}t$. 
\end{proof}

Using matrix operations we obtain the vector $\bm{\theta}(\bm{X},\bm{\mu};\K,M)$ with entries $\theta_k(\bm{X},\bm{\mu};\K,M)$, for $k = 1,\dots,p$, by
\begin{align*}
    \bm{\theta}(\bm{X},\bm{\mu};\K,M) = \diag((\bm{\Sigma}^{\row})^{-1} (\bm{A} - \bm{M}_{\bm{A}}) (\bm{\Sigma}^{\col})^{-1}  (\bm{A} - \bm{M}_{\bm{A}})').
\end{align*}

\newpage

\section{Further Examples}\label{appendix:examples}

This section includes further details regarding the \textit{El Niño} example we presented in Section~\ref{section:examples}, and a second detailed real data application based on age-specific female fertility data from the \cite{HFD2024}.

\subsection{El Niño}\label{appendix:el_nino}

Figure~\ref{fig:map} shows the four Niño regions in the Pacific Ocean. 

\begin{figure}[!h]
    \centering
    \includegraphics[width = 1\linewidth]{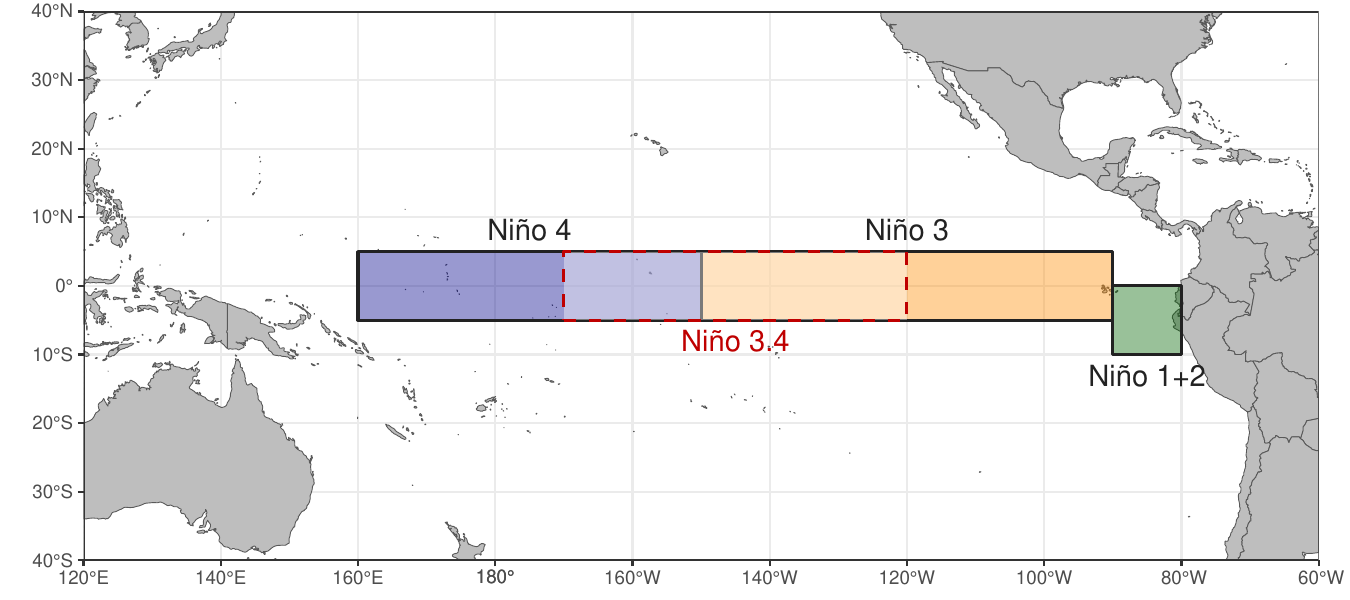}
    \caption{Map of the 4 regions in the equatorial Pacific Ocean where SST related to ENSO are measured: Niño 1+2 (0° to 10°S, 90°W to 80°W), Niño 3 (5°N to 5°S, 150°W to 90°W), Niño 4 (5°N to 5°S, 160°E to 150°W), Niño 3.4 (5°N to 5°S, 170°W to 120°W).}
    \label{fig:map}
\end{figure}

Figure~\ref{fig:plt_overview_smooth2} displays the smoothed SST measurements of all $n = 75$ periods. Regular observations are depicted in gray, while outliers are colored in either red or blue, indicating whether they correspond to El Niño or La Niña periods, respectively. The thicker black line is the smoothed robust mean function. To prevent label overlap, each period is labeled with the last two digits of its starting year (e.g., label 50 represents the period 1950:1951). Most outlying La Niña periods show the same trend as the regular observations but with a downward shift, while the detected El Niño periods often show both a different trend and an upward shift. Figure~\ref{fig:sup_nino} provides the corresponding Shapley decompositions for representative El~Niño (2015:2016) and La~Niña (1956:1957) periods. For the El~Niño event, the largest contributions arise from the Niño~3.4 region during autumn (Aug-Dec) and late winter/early spring (Jan-Apr), with additional contributions from Niño~3 in late spring (Apr-May), reflecting the characteristic warming in the central and eastern Pacific. In contrast, the La~Niña event is primarily driven by negative contributions in the Niño~3 region in late summer (Jul-Dec) and early spring (Feb-Apr), as well as in the Niño~4 region in summer and autumn (Jul-Dec), consistent with the corresponding cooling pattern.

\begin{figure}[!h]
    \centering
    \includegraphics[width = 1\linewidth]{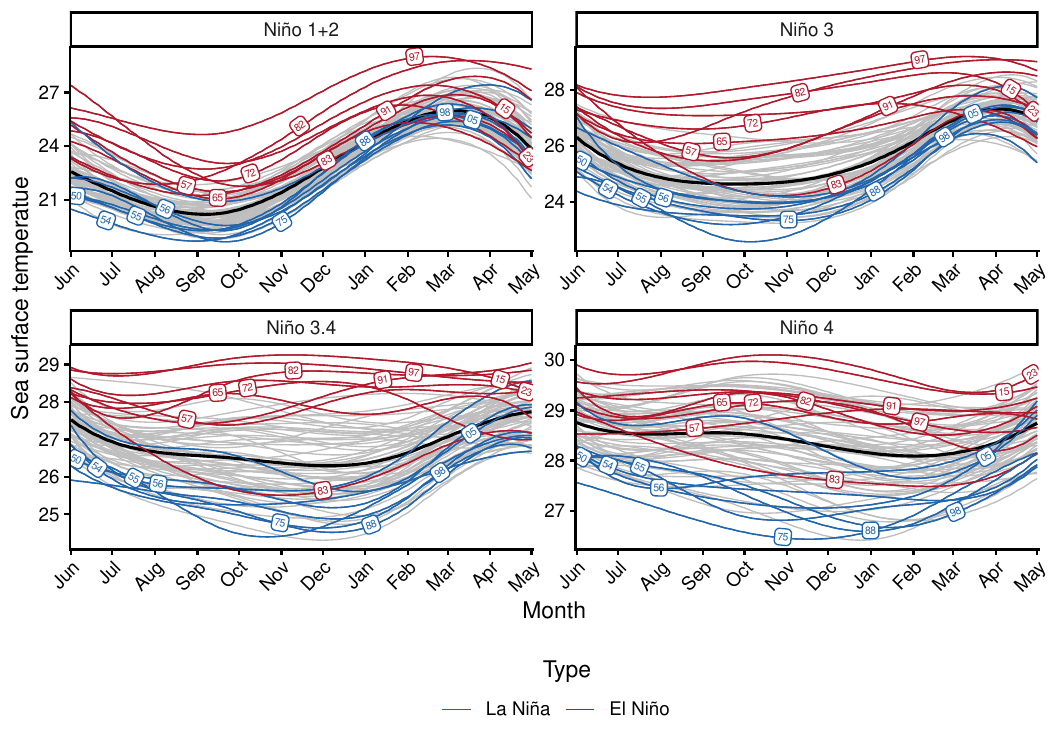}
    \caption{Smoothed SST measurements of all four Niño regions with outlying observations colored either red or blue, depending on whether the outliers are El Niño or La Niña periods.}
    \label{fig:plt_overview_smooth2}
\end{figure}

\begin{figure}[!h]
    \centering
    \includegraphics[width = 1\linewidth]{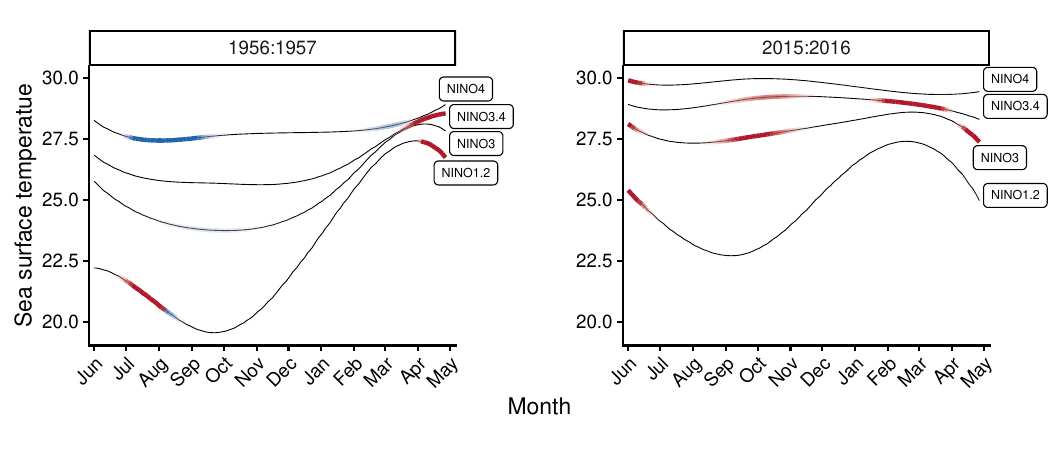}
    \caption{Smoothed SST curves for the years 1956:1957 (left) and 2015:2016 (right across the four Niño regions (columns) during El~Niño and La~Niña years (rows). The color hue indicates whether the deviation is above (red) or below (blue) the global mean, and the color intensity reflects the magnitude of outlyingness.}
    \label{fig:sup_nino}
\end{figure}

\newpage
\subsection{Welding example - additional results}\label{appendix:welding}
As indicated in Section~\ref{section:examples}, we tested both the separability and normality assumptions. Separability was assessed following the procedure proposed by \cite{aston2017tests}, whereas normality was examined by verifying whether the obtained Mahalanobis distances follow a chi-squared distribution—this being the only point in the analysis where the Gaussianity assumption is explicitly required. Both hypotheses were rejected at the 0.05 significance level.
For completeness, Figure \ref{fig:plt_qq_all_welding} presents the Q--Q plot of the robust multivariate functional Mahalanobis distances. Interestingly, the figure reveals a natural grouping of the data into three clusters. The cluster with the smallest Mahalanobis distances corresponds to the samples in the $h$-subset of the MMCD algorithm, noting that these also contain a fraction of outliers, as irregular curves constitute more than 50\% of the total observations. Towards the upper end of the first cluster, we also observe a steeper increase in Mahalanobis distances, as some outliers are forced into the $h$-subset. The second cluster represents the contaminated curves that we detect using the cutoff, which deviate slightly compared to the third cluster, characterized by the largest Mahalanobis distances, which are more evident by visual inspection.
Despite the violation of model assumptions, the ordering of the Mahalanobis distances remains meaningful and informative, as further evidenced by the high AUC achieved by the proposed method.
\begin{figure}[!h]
        \centering
        \includegraphics[width=0.7\linewidth]{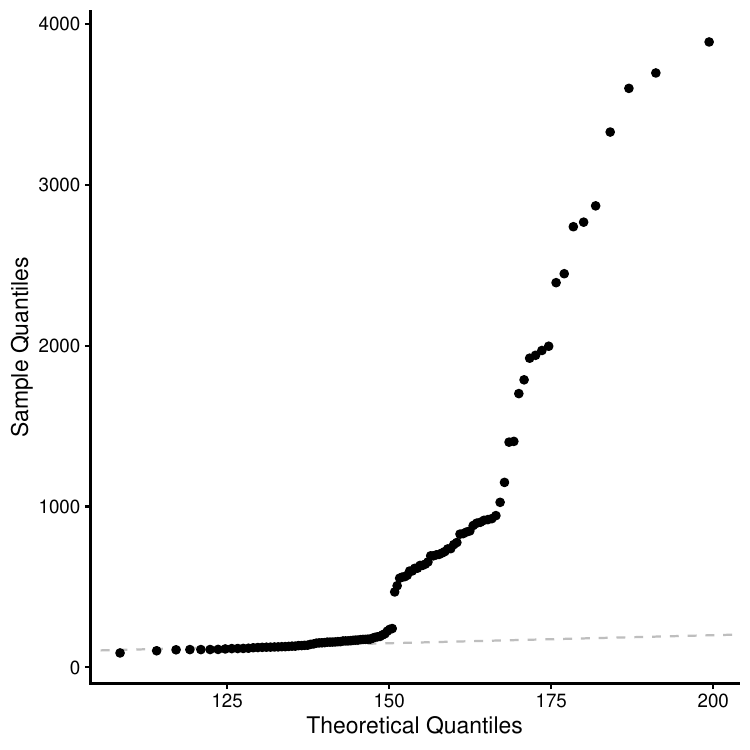}
        \caption{Q--Q plot of robust squared Mahalanobis distances against $\chi^2$ quantiles.}
        \label{fig:plt_qq_all_welding}
\end{figure}

\subsection{Fertility Rates}\label{appendix:fertility}
We analyze the annual age-specific female fertility curves for several countries/regions from the \cite{HFD2024}. Specifically, we consider the age-specific fertility rate (ASFR), defined as 
$$
\mathrm{ASFR}(s,t) = \frac{\text{number of live births to women aged $s$ in year $t$ }}{\text{population of women aged $s$ in year $t$}},
$$
for women aged $15$ to $45$. In the context of ASFRs, we refer to women towards the lower end of the spectrum as younger women, women aged around 30 years old as middle-aged women, and women at the higher end of the spectrum as older women. We selected the subset of $n = 22$ countries/regions with no missing values for the years between $1960$ and $2019$. %
To facilitate the interpretability of the results, we aggregate the annual ASFRs into five-year intervals (1960:1964, 1965:1969, \dots, 2015:2019), which results in observations that are naturally arranged in $12 \times 31$ matrices for each country. %
This matrix structure reflects the average ASFRs for each of the 12 five-year periods across the ages from 15 to 45.

Our goal is to treat these ASFRs as functional data with continuous age $s$ for each of the $12$ five-year intervals. The matrix-variate data are transformed into multivariate functional data through the following process: First, a log transformation is applied to the ASFR data. Next, the log-transformed data are smoothed into multivariate functional data using a cubic B-spline basis consisting of 10 basis functions. Finally, the smoothed coefficient matrices are exponentiated to return the data to their original scale.
This method ensures the positivity of the smoothed ASFRs, thereby maintaining the inherent characteristics of the fertility rates. 

The smoothed curves are shown in Figure~\ref{fig:p_smooth}. Here, every plot shows the fertility curves for one of the five-year intervals. Overall, fertility is clearly declining over the years, and women give birth at older ages as time progresses. Moreover, the curves are similar in the last period while there is more difference between the countries in the earlier years. We see that some countries/regions form a cluster with left-skewed fertility, i.e., women give birth at younger ages. On the other hand, Norway (NOR) stands out because of very high fertility in the first years and right-skewed fertility in the later years. 
\begin{figure}[!h]
    \centering
    \includegraphics[width=1\linewidth]{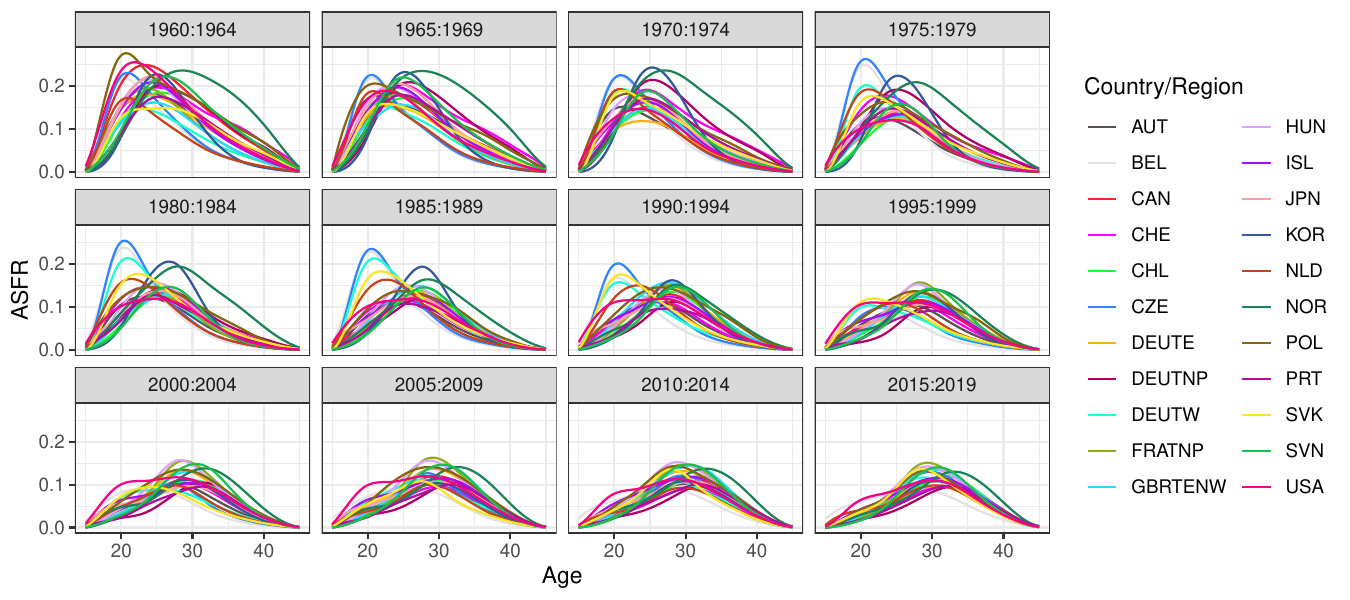}
    \caption{Smoothed age-specific fertility curves for all 22 countries/regions. %
    }
    \label{fig:p_smooth}
\end{figure}

As discussed in Section \ref{section:examples}, we examined both the separability and normality assumptions. The separability assumption was evaluated using the method proposed by \cite{aston2017tests}, while normality was assessed by testing whether the Mahalanobis distances follow a chi-squared distribution—the only step in our framework where the Gaussianity assumption is explicitly invoked. Both assumptions were rejected at the 0.05 significance level.\\
For completeness, Figure \ref{fig:plt_qq_all_feritlity} presents the Q–Q plot of the robust multivariate functional Mahalanobis distances. Similar to the observations from the Welding example (Figure~\ref{fig:plt_qq_all_welding}), the Q–Q plot suggests a heterogeneous structure in the data, indicating the presence of multiple underlying subgroups. In particular, Figure~\ref{fig:p_smooth} shows that in the early 1960s most countries exhibited peak fertility at younger ages, whereas in more recent decades there has been a general shift toward later-age peak fertility. The timing of this transition, however, differed substantially across countries, leading to a heterogeneous—and somewhat clustered—structure in the data based on this demographic pattern.
\begin{figure}[!h]
        \centering
        \includegraphics[width=0.7\linewidth]{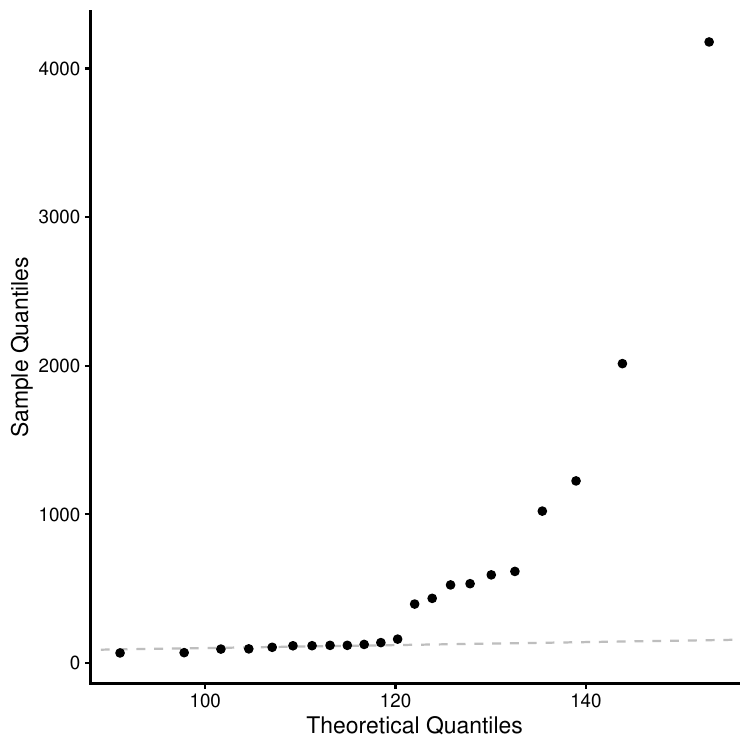}
        \caption{Q--Q plot of robust squared Mahalanobis distances against $\chi^2$ quantiles.}
        \label{fig:plt_qq_all_feritlity}
\end{figure}

We compare the classic and robust $\fmmd$ of the smoothed samples in the distance-distance plot shown in Figure~\ref{fig:p_dd_label}. The distances are based on parameter estimates from the MMLE and MMCD procedures, respectively, applied to the coefficient matrices. Outliers are detected using $\sqrt{\chi^2_{0.99}(10 \cdot 12)}$ as the cutoff, where $\chi^2_{0.99}(d)$ denotes the 0.99 quantile of a $\chi^2$-distribution with $d$ degrees of freedom. 
When comparing the classic and robust $\fmmd$, we observe that several countries, such as the United States (USA) and Germany (DEUTNP), are masked when using the classic distance.

The influence of outliers on the principal component functions is illustrated in Figure~\ref{fig:p_pc_MMLE_vs_MMCD}. The classic and robust mean functions are similar, but their eigenfunctions are quite different. The first classic eigenfunction captures increased fertility, especially for younger women, and explains 49\% of the variance, while the first robust eigenfunction, which mirrors the form of the robust mean, indicates overall high or low fertility levels and explains around 62\% of the variance. This difference arises because the robust eigenfunctions are estimated from a subset of the data that excludes countries with exceptionally high fertility rates among younger or older women. 
The second classic and robust eigenfunctions illustrate shifts in fertility towards younger women, moving away from middle-aged women. While the classic eigenfunction explains 38\% percent of the variance and also describes changes in ASFRs for older women aged 35 and above, the robust eigenfunction shows more isolated effects for women under 35 and explains 21\% percent of the variance. 
The first two robust eigenfunctions provide a clear and coherent interpretation, i.e., overall high or low fertility from the first eigenfunction and higher and lower fertility for younger or middle-aged women based on the second eigenfunction. In contrast, the classic eigenfunctions exhibit contrasting effects on the fertility of young women, complicating their interpretation. The third robust and classical eigenfunctions explain about 8\% and 12\% of the variance, respectively, and are quite similar, showing a concentrated increase in fertility around age 25, with lower fertility rates for both younger and older women.

\begin{figure}
     \centering
     \begin{subfigure}[t]{0.49\textwidth}
         \centering
         \includegraphics[width=\textwidth]{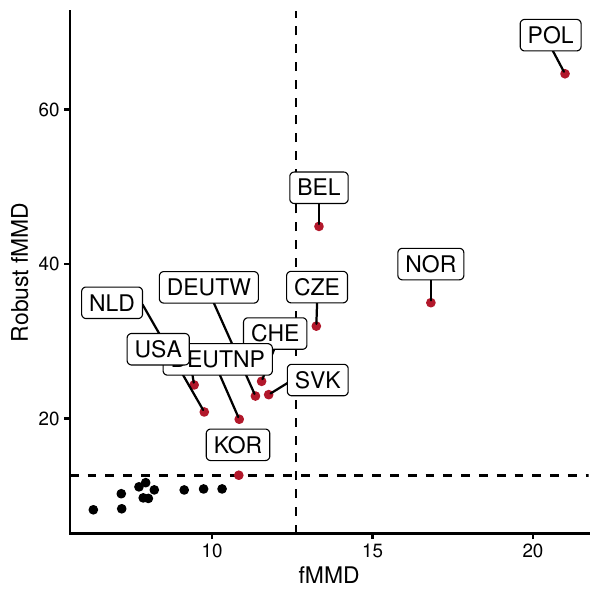}
         \caption{Comparison of classic and robust Mahalanobis distances of the ASFRs.}
         \label{fig:p_dd_label}
     \end{subfigure}
     \hfill
     \begin{subfigure}[t]{0.49\textwidth}
         \centering
         \includegraphics[width=\textwidth]{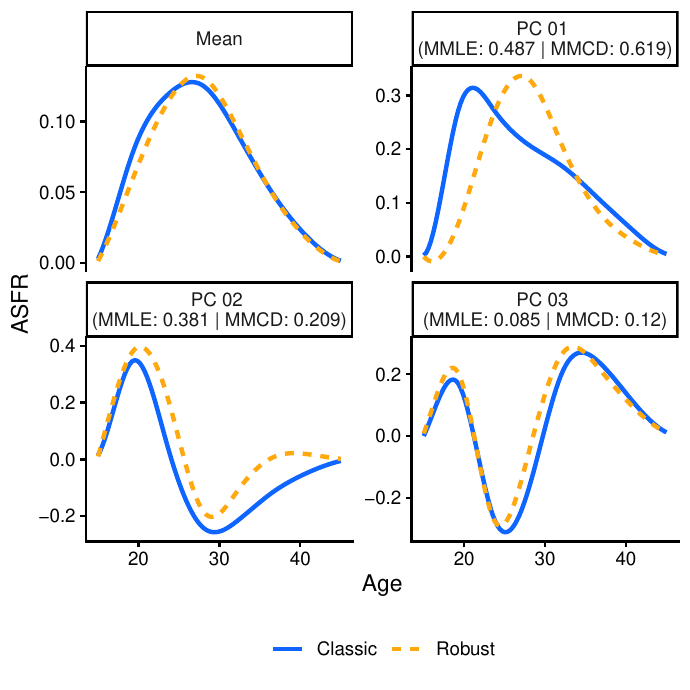}
         \caption{Comparison of classic (MMLE) and robust (MMCD) mean as well as the first three eigenfunctions. The explained variance of the eigenfunctions is given in parentheses.}
         \label{fig:p_pc_MMLE_vs_MMCD}
     \end{subfigure}
     \hfill
    \caption{Robust analysis of the smoothed ASFRs.}
    \label{fig:three graphs}
\end{figure}

To get a better understanding of the reasons for the outlyingness, we analyze the year-specific and/or age-specific outlyingness contributions of the outlying countries, visualized in Figures~\ref{fig:p_shapley_row_col} and \ref{fig:p_shapley_cell}. The outlyingness contributions are computed based on the robustly estimated mean and covariance function. Hence, below or above average refers to the robustly estimated location based on the MMCD estimators. 
Using the same color scheme as in Section~\ref{section:examples}, years and/or ages with high outlyingness contributions are colored in red if they are positive and blue if they are negative. To ensure comparability between observations with varying levels of outlyingness, color intensity is based on relative Shapley values, i.e., Shapley values divided by $\fmmd^2$. Moreover, age-specific outlyingness contributions in Figure~\ref{fig:p_shapley_cell} are aggregated for the same three-year periods as in Figure~\ref{fig:p_shapley_row_col}. In addition to the outlyingness contributions, the ASFRs are plotted in the background of Figure~\ref{fig:p_shapley_cell}.

Taking Belgium (BEL) as an example, the outlyingness scores reveal the following: The Age-specific outlying contributions highlight higher than expected fertility for women aged 15 to 24 years and lower than average fertility for women aged 33 to 36 years. The year-specific outlyingness shows the highest outlyingness contributions between 1990 and 2019, with the exception of 1995:1999. When we look at Belgium in more detail in Figure~\ref{fig:p_shapley_cell}, the period between 1970 and 1974 provides an interesting case even though the outlyingness contribution is not that high: The overall fertility was below average, and thus the corresponding entry in Figure~\ref{fig:p_shapley_row_col} is colored blue. However, the age group that contributed most to the outlyingness is women aged 15 to 18 years with above-average fertility. Overall, Belgium shows above-average fertility for young women and below-average fertility for older women. The outlyingness contributions in the earlier years are mostly driven by high fertility for young women, while in more recent years, low fertility for women aged 25 to 35 contributes substantially.

\begin{figure}
    \centering
    \includegraphics[width=1\linewidth]{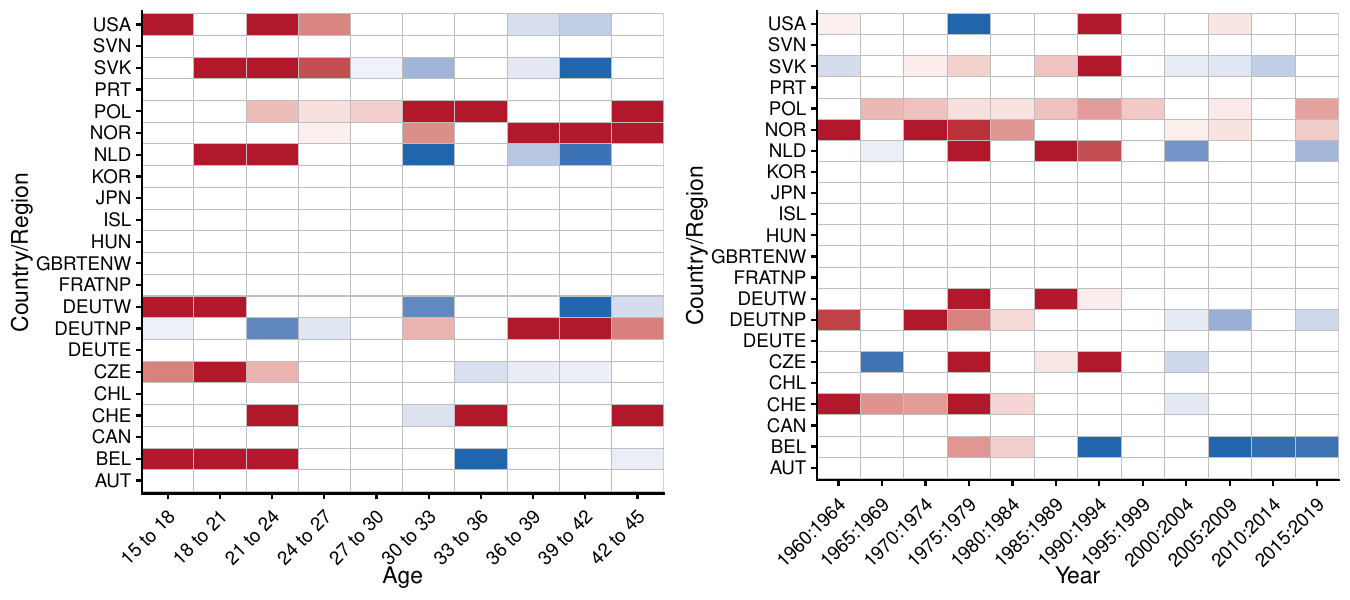}
    \caption{Age-specific (left) or year-specific (right) outlyingness contributions based on Shapley values for the smoothed ASFRs.}
    \label{fig:p_shapley_row_col}
\end{figure}

\begin{figure}
    \centering
    \includegraphics[width=1\linewidth]{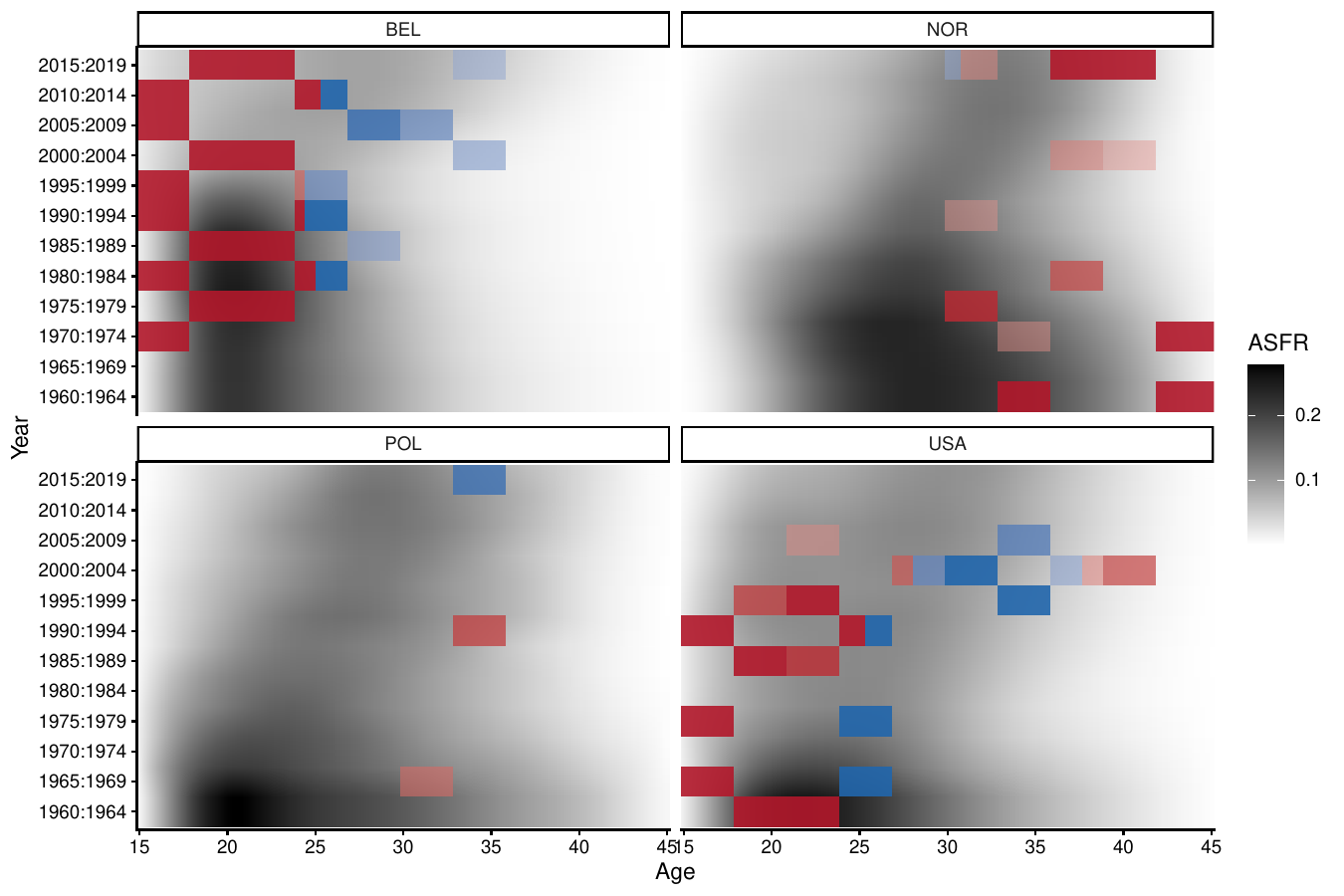}
    \caption{Age-specific and year-specific outlyingness contributions based on Shapley values for the smoothed ASFRs of Belgium (upper left), Norway (upper right), Poland (lower left), and the US (lower right).}
    \label{fig:p_shapley_cell}
\end{figure}

\newpage
\section{Further Simulation Results}\label{appendix:simulations}

\subsection{Non-Gaussian Setting}\label{supplement:non_gaussian}

For the non-Gaussian setting, the random functions are drawn from a multivariate t-distribution with $\nu = 3, 5, 8, 12$, and $15$ degrees of freedom, and a separable covariance structure is generated by drawing finite-dimensional realizations at $q = 100$ time points. We consider a sample size of $n = 1000$ observations with $p = 10$ coordinate functions. For the covariance structure between the coordinate functions, we adopt the matrices proposed by \citet{agostinelli2015robust}, denoted by $\bm{\Sigma}^{\row}$, which have random entries and typically yield low to moderate correlations. For the covariance function $\kappa$, we consider a Matérn-type $\kappa_{\text{Matérn}}$ covariance with the same specification as in the Gaussian setting. 

Outliers are added to the datasets by randomly replacing a fraction $\varepsilon \in \{0.1,0.3\}$ of the clean observations. We only consider shift outliers for the non-Gaussian setting, which are created in $\lfloor{\varepsilon_{\mathrm{cord}} \cdot p\rfloor}$ randomly chosen coordinates by introducing perturbations along the first eigenfunction $\ueig_1$, capturing the largest mode of variation. The coordinate functions are given by 
$X_j^{\text{shift}}(t) = X_j + \lambda_{\mathrm{shift}} \ueig_1$
for $j = 1, \dots, p$. The perturbation magnitudes are set to $\lambda_{\mathrm{shift}} \in \{1, 10\}$. 

The results are shown in Figures~\ref{fig:t_test}--\ref{fig:t_covariance}.
Overall, the robust distances computed from either the raw or smoothed data demonstrate strong performance in terms of AUC, recall, and covariance error. However, they tend to flag a substantial number of regular observations as outliers, since they cannot distinguish between data points originating from the distribution tails and genuine outliers. As the degrees of freedom $\nu$ increase, the precision of all methods increases, but the robust distances still fall short of their depth-based alternatives. In the MMCD procedure \citep{mayrhofer2024robust}, a consistency factor for the t-distribution could be derived instead of the consistency factor for Gaussian data, to better accommodate the heavy tails and improve precision. The MFIF, evaluated with all three dictionaries, stays close to an AUC of 0.5 throughout. A possible reason is that the contamination is a shift along an eigenfunction of the underlying process that is not necessarily well represented by the fixed dictionaries it projects onto.

Unlike in the Gaussian setting, the performance of all methods, except for MFIF, improves when contamination affects only a subset of the coordinates. When all coordinates are contaminated, the resulting shift outliers closely resemble observations from the tails of the distribution. In contrast, when only a subset of coordinates is affected, the underlying covariance structure becomes distorted, allowing the methods to more effectively detect the contamination.

\begin{figure}[!h]
    \centering
    \includegraphics[width=1\linewidth]{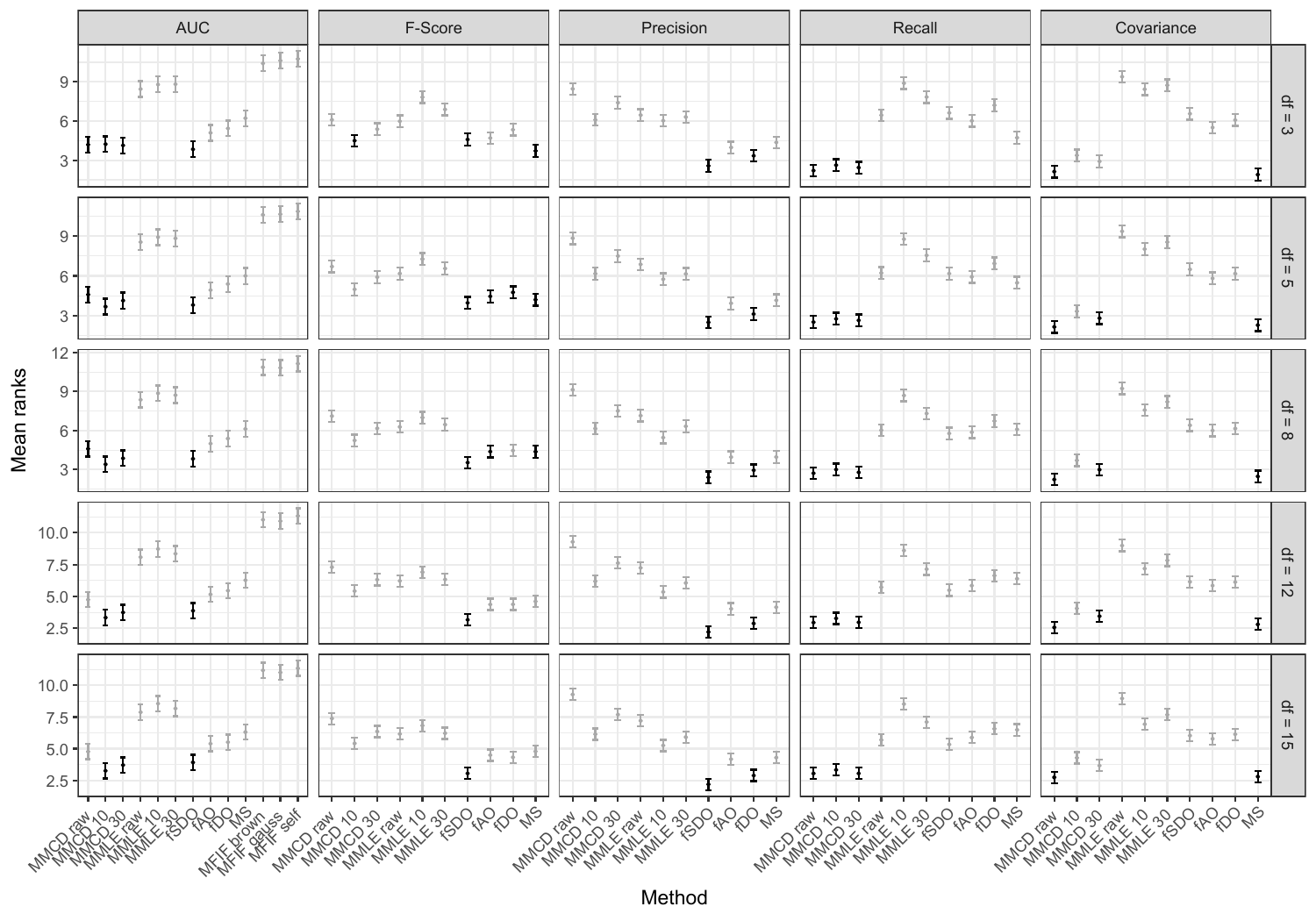}
    \caption{Rank-based comparison of methods across all simulation settings in the non-Gaussian setting for $p=10$. Intervals are based on the Friedman and Nemenyi tests; methods not significantly different (99\% level) from the best are shown in black, others in gray. The horizontal facets correspond to performance metrics, and the vertical facets to the degrees of freedom of the t-distribution.}
    \label{fig:t_test}
\end{figure}

\begin{figure}[p]
    \centering
    \includegraphics[width=1\linewidth]{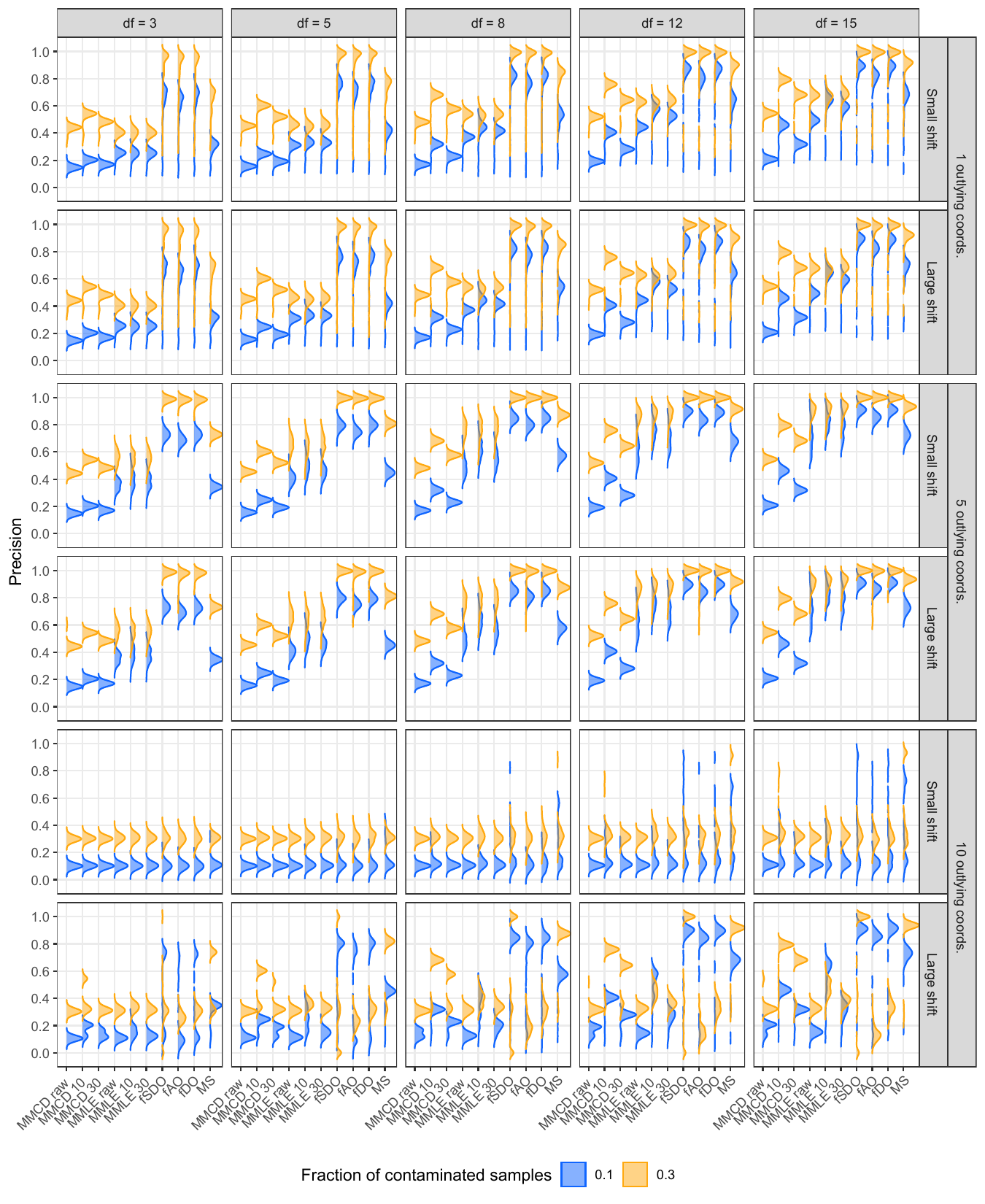}
    \caption{{Precision of shift outlier detection for a stochastic process with t-distributed innovations ($\nu = 3,5,8,12,$ and $15$ degrees of freedom) in a setting with $n = 1000$, $p = 10$, under the Matérn covariance, small and large outlier magnitudes, various coordinate contamination levels (1, 5, and 10 coordinates), and outlier proportions ($\varepsilon = 0.1, 0.3$).}}
    \label{fig:t_precision}
\end{figure}

\begin{figure}[p]
    \centering
    \includegraphics[width=1\linewidth]{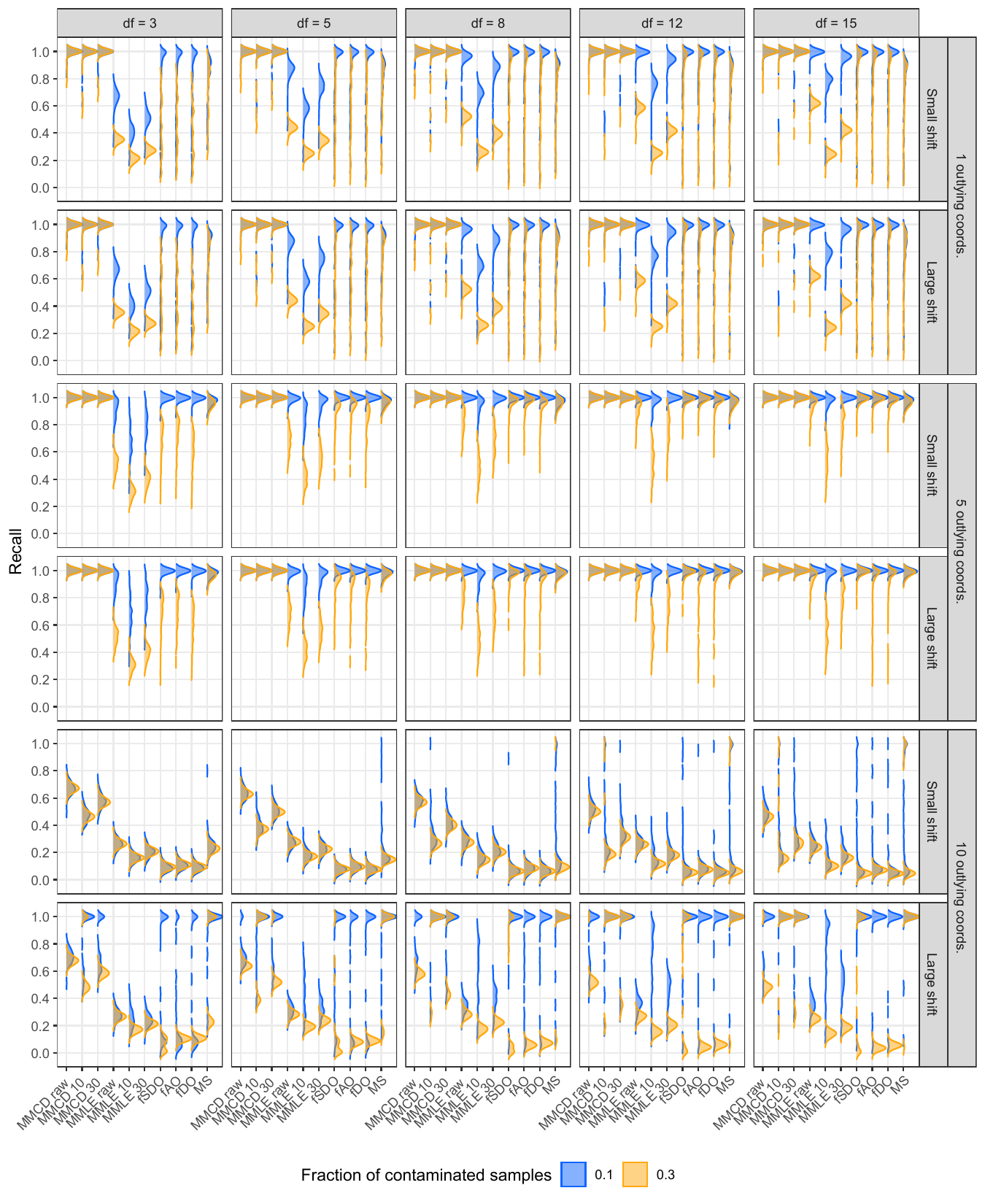}
    \caption{{Recall of shift outlier detection for a stochastic process with t-distributed innovations ($\nu = 3,5,8,12,$ and $15$ degrees of freedom) in a setting with $n = 1000$, $p = 10$, under the Matérn covariance, small and large outlier magnitudes, various coordinate contamination levels (1, 5, and 10 coordinates), and outlier proportions ($\varepsilon = 0.1, 0.3$).}}
    \label{fig:t_recall}
\end{figure}

\begin{figure}[p]
    \centering
    \includegraphics[width=1\linewidth]{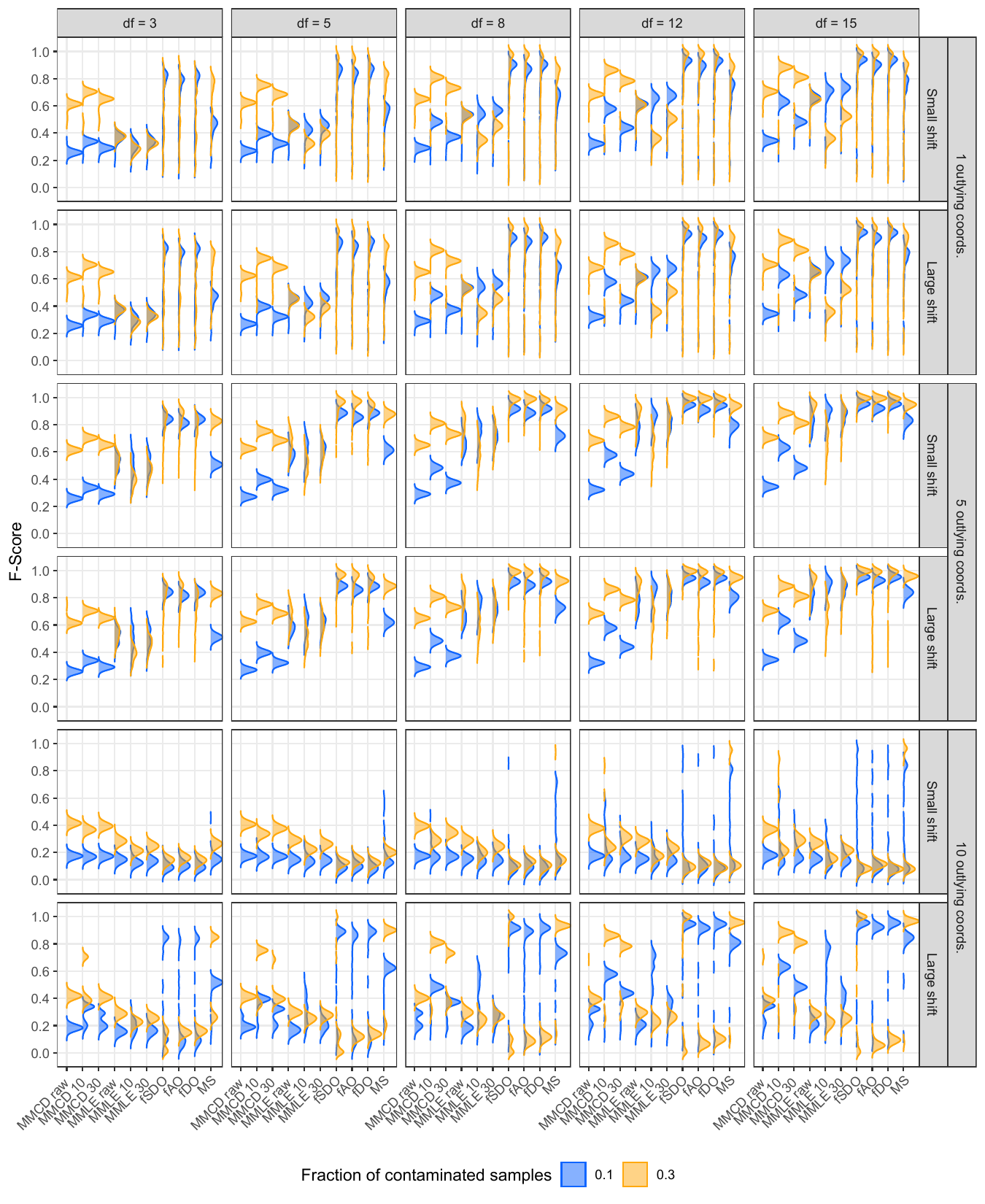}
    \caption{{F-score of shift outlier detection for a stochastic process with t-distributed innovations ($\nu = 3,5,8,12,$ and $15$ degrees of freedom) in a setting with $n = 1000$, $p = 10$, under the Matérn covariance, small and large outlier magnitudes, various coordinate contamination levels (1, 5, and 10 coordinates), and outlier proportions ($\varepsilon = 0.1, 0.3$).}}
    \label{fig:t_fscroe}
\end{figure}

\begin{figure}[p]
    \centering
    \includegraphics[width=1\linewidth]{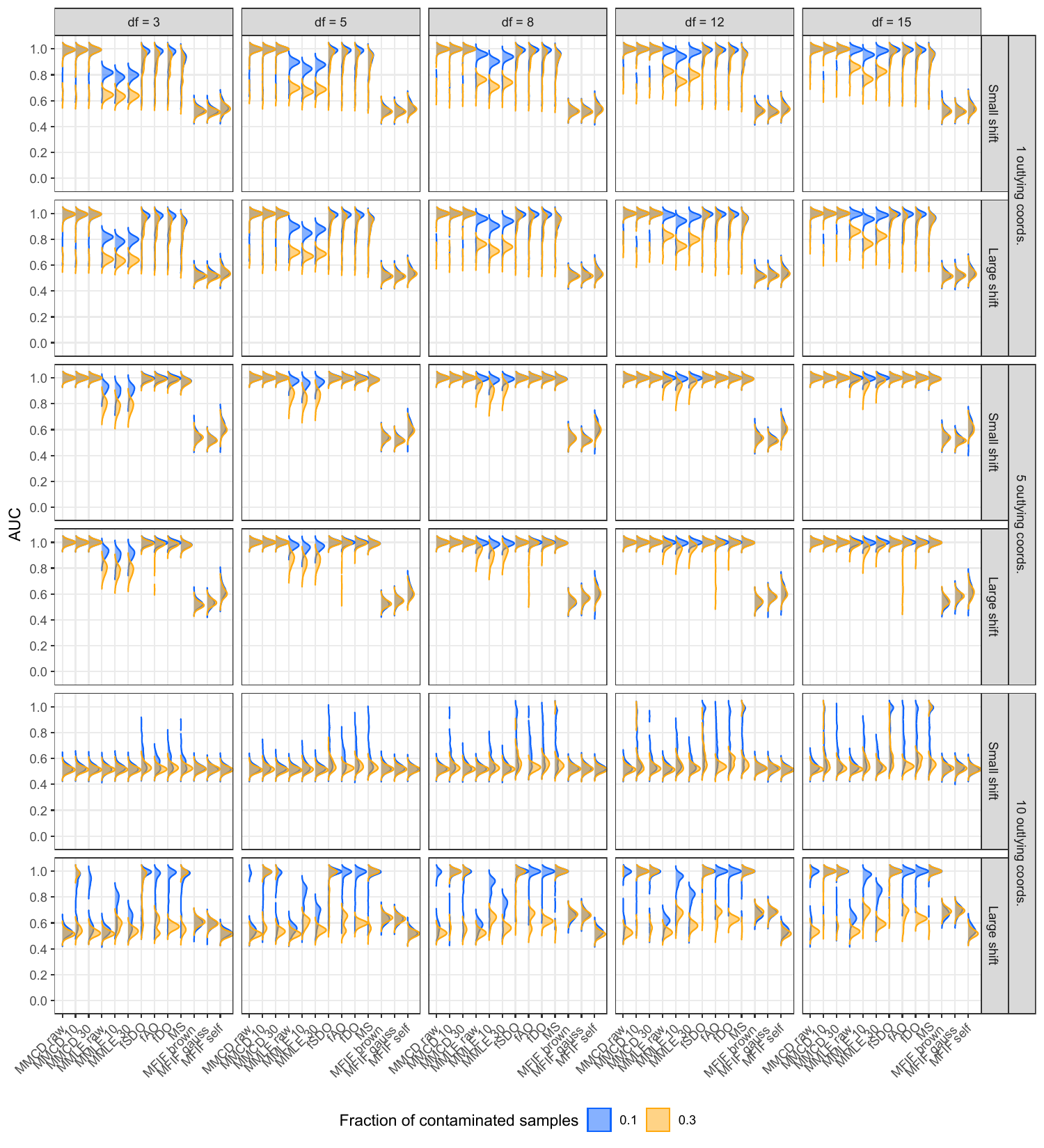}
    \caption{{AUC values of shift outlier detection for a stochastic process with t-distributed innovations ($\nu = 3,5,8,12,$ and $15$ degrees of freedom) in a setting with $n = 1000$, $p = 10$, under the Matérn covariance, small and large outlier magnitudes, various coordinate contamination levels (1, 5, and 10 coordinates), and outlier proportions ($\varepsilon = 0.1, 0.3$).}}
    \label{fig:t_AUC}
\end{figure}

\begin{figure}[p]
    \centering
    \includegraphics[width=1\linewidth]{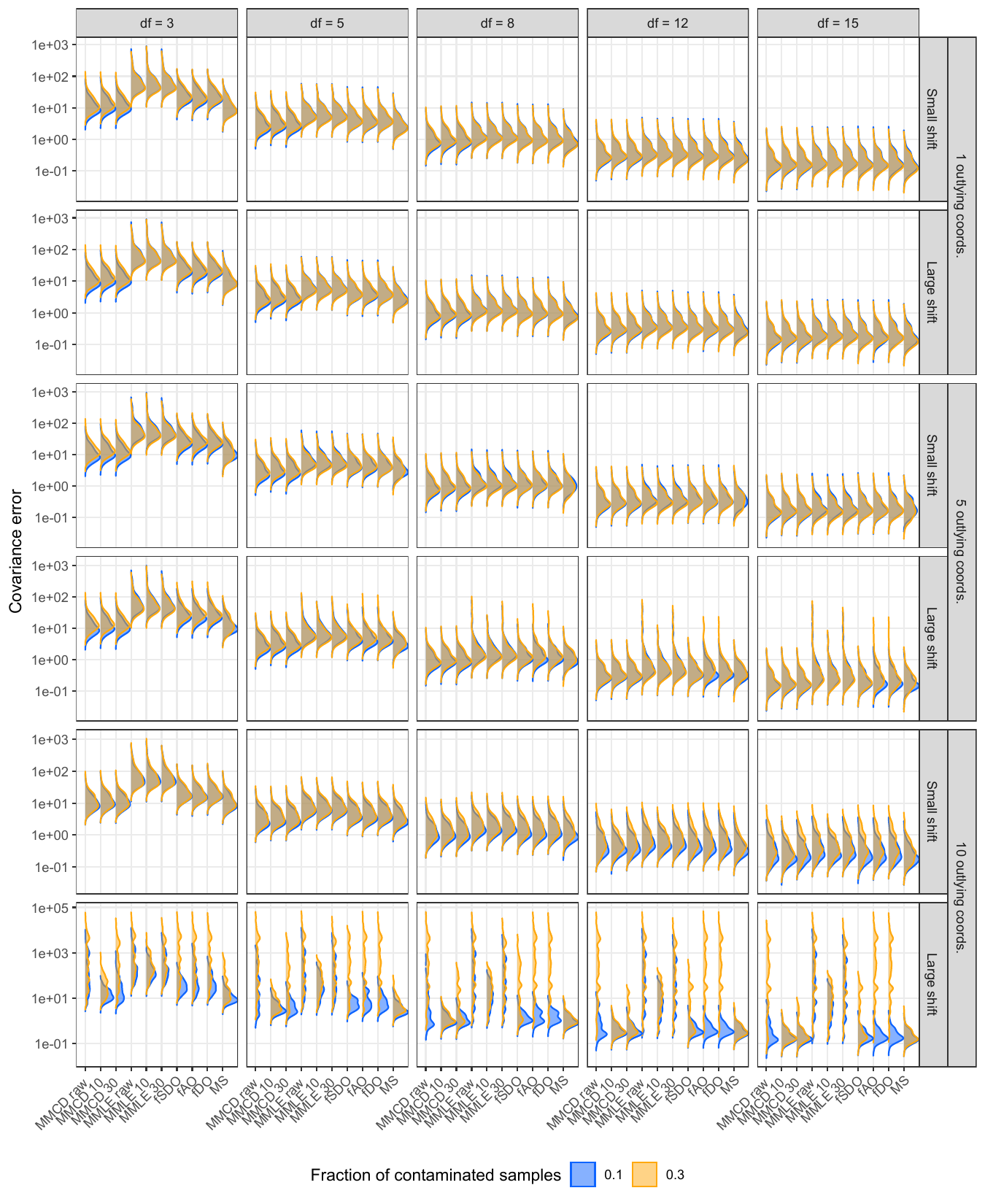}
    \caption{{Log covariance estimation error for a stochastic process with t-distributed innovations ($\nu = 3,5,8,12,$ and $15$ degrees of freedom) and outlyingness introduced by shift outliers in a setting with $n = 1000$, $p = 10$, under the Matérn covariance, small and large outlier magnitudes, various coordinate contamination levels (1, 5, and 10 coordinates), and outlier proportions ($\varepsilon = 0.1, 0.3$).}}
    \label{fig:t_covariance}
\end{figure}

\cleardoublepage
\subsection{Computation Time}\label{appendix:simulations_time}

Figure~\ref{fig:time} reports boxplots of computation times of single runs for each method, and Table~\ref{tab:time} lists the corresponding medians for all four simulation settings. The classical MMLE is the fastest method, while the robust MMCD estimates carry a moderate cost. The dominant factor is whether a method is applied to the raw data or to the basis coefficients: for $p = 50$ and $n = 1000$, the median MMCD time drops from about $12$ minutes on the raw data to about half a minute and two minutes with $10$ and $30$ basis functions, since the coefficient matrices are of much lower dimension. MFIF is the most expensive method across all settings, ranging from about 1 minute for $p = 3$, $n = 300$ to roughly 80 minutes for $p = 50$, $n = 1000$. Its cost is driven by both the dimension and the sample size: at $n = 1000$ it takes about $16$ minutes for $p = 10$ but $80$ minutes for $p = 50$, and increasing $n$ from $300$ to $1000$ at $p = 50$ multiplies its runtime by about $4.4$, against $1.3$ for the robust distances on the raw data. This is why it is evaluated on $10$ instead of $100$ replications for $p = 50$.

\begin{figure}[!h]
    \centering
    \includegraphics[width=1\linewidth]{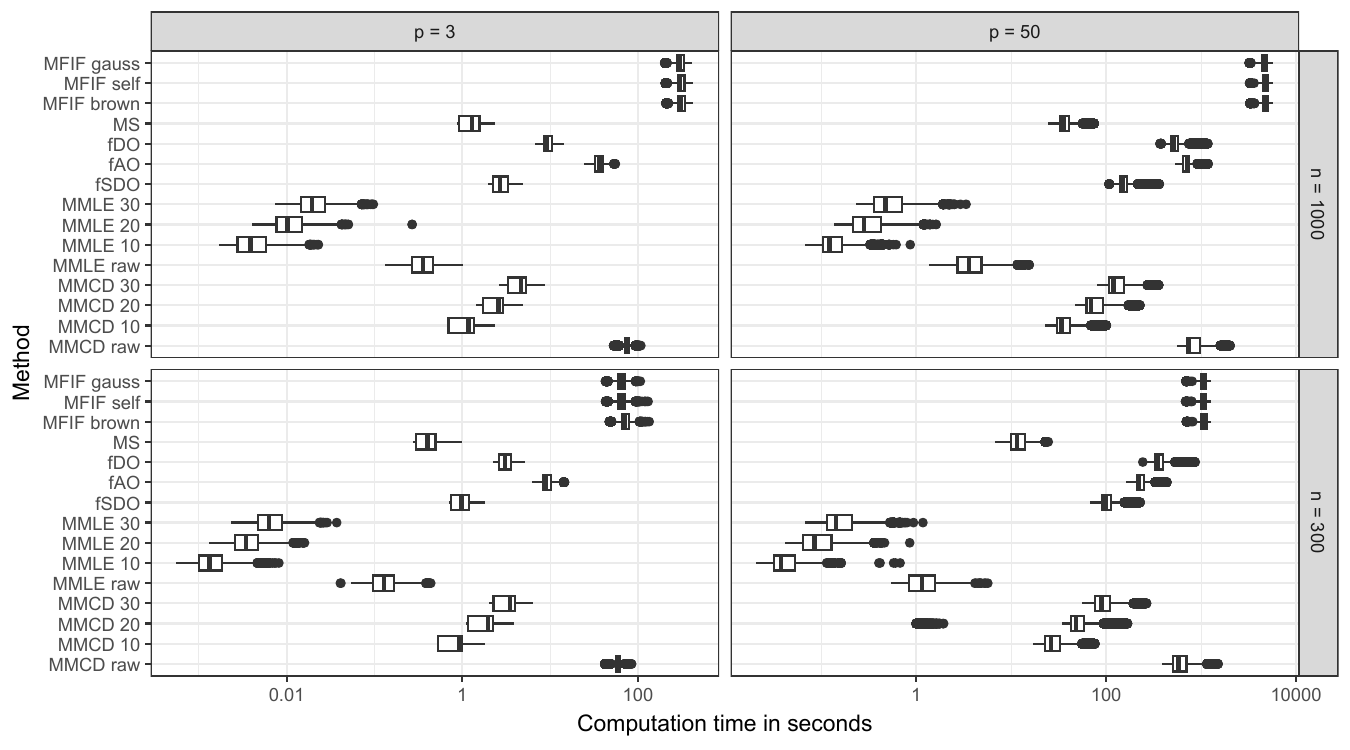}
    \caption{Comparison of log computation time in seconds for the Gaussian setting outlined in Section~\ref{section:simulations}, divided into facets according to dimensionality $p  \in \{3,50\}$ and number of observations $n \in \{300, 1000\}$. The MFIF timings rest on 10 replications per configuration for $p = 50$ and 100 in all other settings.}
    \label{fig:time}
\end{figure}

\begin{table}[!h]
\caption{Median computation time of a single run, in seconds. MMCD 10/20/30 denote the robust distances computed from $m$ B-spline coefficients, MMCD raw the same distances on the unsmoothed data. Note that the sample sizes differ between settings. The heavy-tailed setting uses $m\in\{10,30\}$ only. The MFIF timings rest on 10 replications per configuration for $p = 50$ and 100 in all other settings.}
\centering
\small
\setlength{\tabcolsep}{4pt}
\begin{tabular}{lrrrrrrr}
\toprule
 & \multicolumn{2}{c}{Gaussian $p=3$} & \multicolumn{2}{c}{Gaussian $p=50$} & \multicolumn{2}{c}{Non-separable $p=10$} & Heavy-tailed $p=10$ \\
\cmidrule(lr){2-3}\cmidrule(lr){4-5}\cmidrule(lr){6-7}\cmidrule(lr){8-8}
Method & $n=300$ & $n=1000$ & $n=300$ & $n=1000$ & $n=1500$ & $n=5000$ & $n=1000$ \\
\midrule
MFIF brown & 71.2 & 307.5 & 1070.6 & 4754.1 & 1292.9 & 3887.4 & 968.2 \\
MFIF gauss & 64.8 & 297.9 & 1056.2 & 4707.4 & 1326.9 & 3957.4 & 955.0 \\
MFIF self  & 64.8 & 302.0 & 1056.4 & 4772.2 & 1341.0 & 4004.9 & 957.0 \\
\midrule
MMCD raw & 58.7 & 75.0 & 579.3 & 747.3 & 285.8 & 470.2 & 123.6 \\
MMCD 10  &  0.9 &  1.2 &  26.3 &  33.9 &   5.7 &   7.5 &   3.2 \\
MMCD 20  &  2.0 &  2.6 &  47.9 &  69.3 &  15.2 &  21.7 &  -- \\
MMCD 30  &  3.5 &  4.6 &  89.3 & 119.5 &  28.6 &  40.6 &  12.4 \\
\midrule
fDO  & 3.1 &  9.2 & 347.8 & 513.6 &  79.2 & 298.6 &  38.9 \\
fSDO & 1.0 &  2.7 &  98.4 & 149.7 &  18.6 &  65.8 &   9.7 \\
fAO  & 8.9 & 35.7 & 222.7 & 706.5 & 236.8 & 920.7 & 116.3 \\
MS   & 0.4 &  1.3 &  11.6 &  35.8 &   7.0 &  23.9 &   3.4 \\
\bottomrule
\end{tabular}\label{tab:time}
\end{table}

\cleardoublepage
\subsection{Comparison of the Depth-based Approaches}\label{appendix:simulations_depth}

Figures~\ref{fig:gaussian_depth_based_fscore} - \ref{fig:gaussian_depth_based_cov} compare the performance of the depth-based procedures for the same setting we considered in Section~\ref{section:simulations}. The methods were either applied to the raw data, smoothed data (evaluated on the raw data's time grid), or the coefficient matrices (using 10 or 30 basis functions) of the smoothed data. Overall, the results indicate that the methods perform better on either the raw or smoothed data, rather than the coefficient matrices. An exception is the setting where we consider covariance-induced outliers, in which case the AUC reveals that the smoothed data or the coefficient matrices with 30 basis functions yield better results than the raw data.

\begin{figure}[p]
    \centering
    \includegraphics[width=1\linewidth]{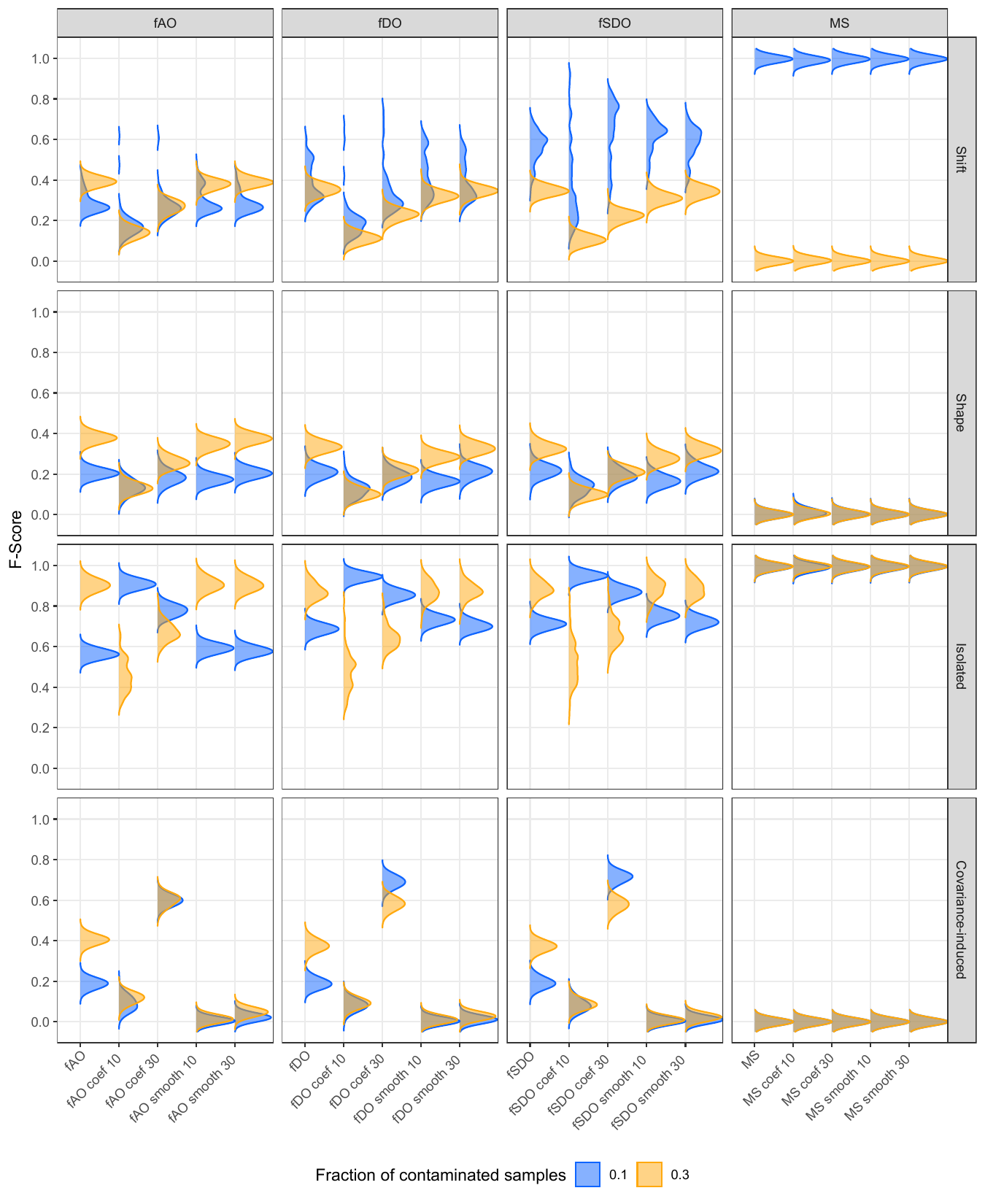}
    \caption{{Density plots of F-score for a representative setting with $n = 1000$ and $p = 50$ with all coordinates contaminated under the Matérn covariance and medium outlier magnitude, shown across outlier types (shift, shape, isolated, covariance-induced), and outlier proportions ($\varepsilon = 0.1, 0.3$).}}
    \label{fig:gaussian_depth_based_fscore}
\end{figure}

\begin{figure}[p]
    \centering
    \includegraphics[width=1\linewidth]{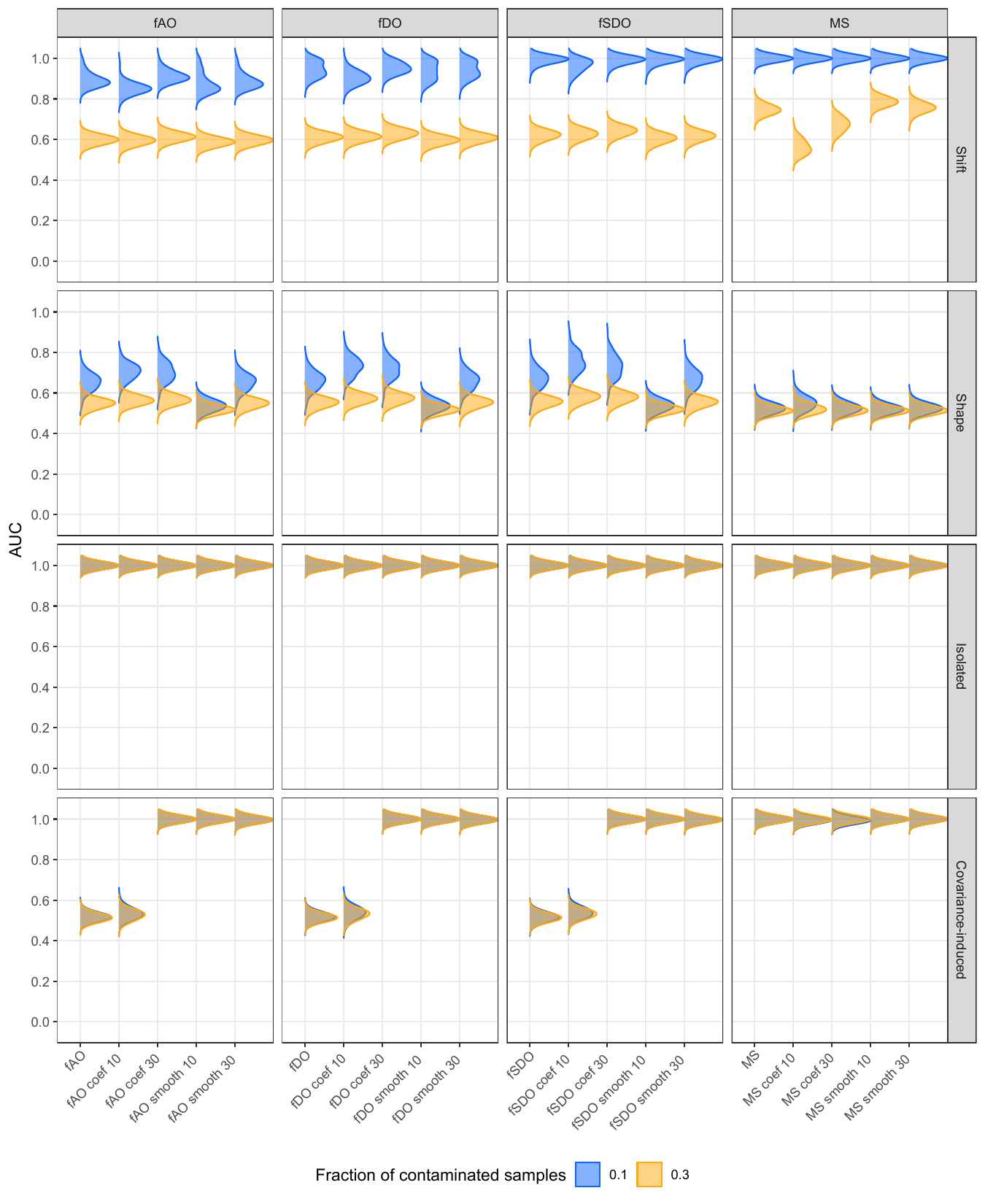}
    \caption{{Density plots of AUC for a representative setting with $n = 1000$ and $p = 50$ with all coordinates contaminated under the Matérn covariance and medium outlier magnitude, shown across outlier types (shift, shape, isolated, covariance-induced), and outlier proportions ($\varepsilon = 0.1, 0.3$).}}
    \label{fig:gaussian_depth_based_AUC}
\end{figure}

\begin{figure}[p]
    \centering
    \includegraphics[width=1\linewidth]{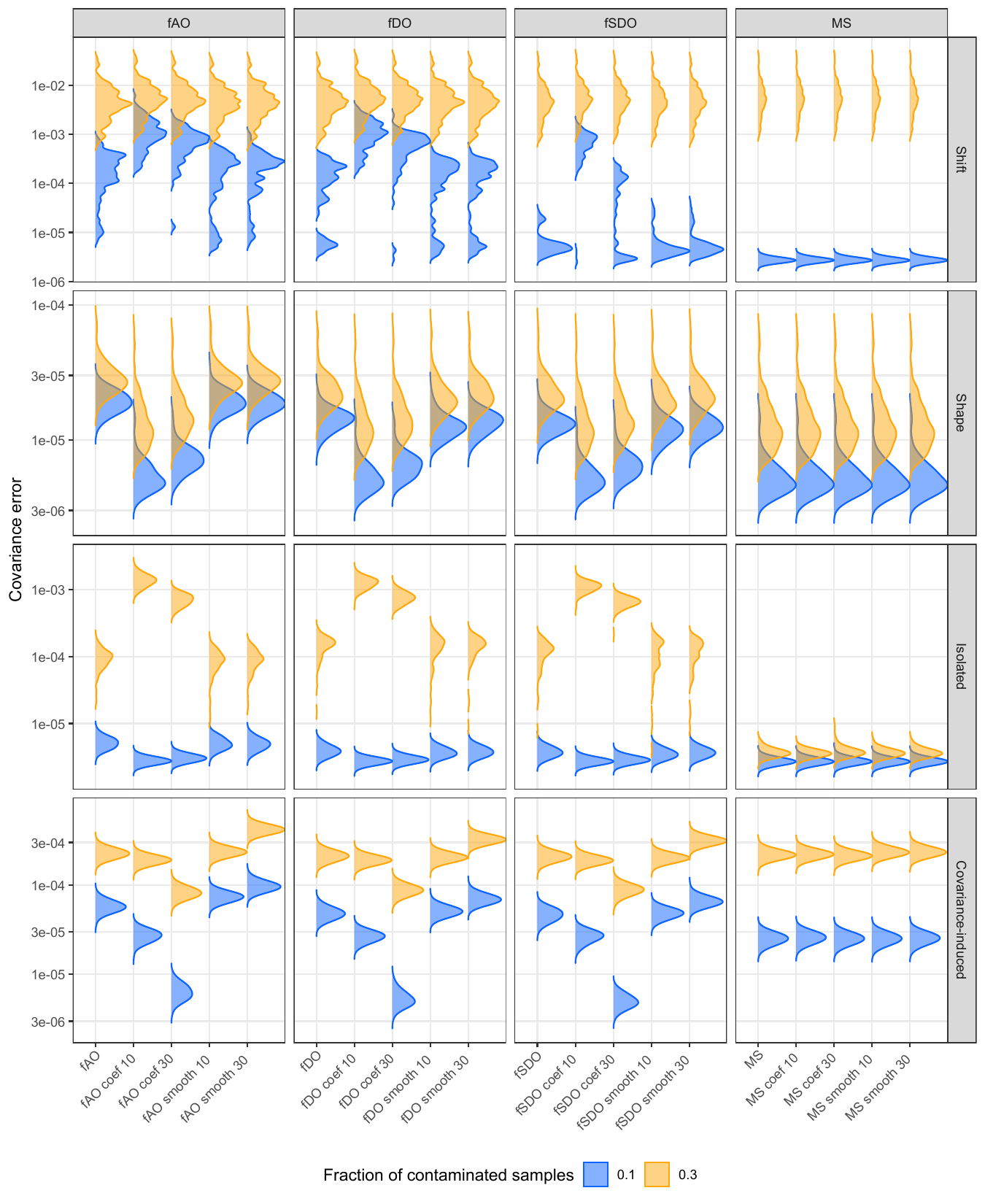}
    \caption{{Density plots of log covariance estimation error for a representative setting with $n = 1000$ and $p = 50$ with all coordinates contaminated under the Matérn covariance and medium outlier magnitude, shown across outlier types (shift, shape, isolated, covariance-induced), and outlier proportions ($\varepsilon = 0.1, 0.3$).}}
    \label{fig:gaussian_depth_based_cov}
\end{figure}

\newpage
\subsection{Non-Separable Setting} 
\label{subsection:sample_vs_mmle}
For the non-separable simulation setting presented in Section~\ref{section:simulations} (Table~\ref{tab1}), we also included the sample covariance and the matrix MLE covariance estimators based on the raw data. In Figure~\ref{fig:sample_vs_mmle}, we compare the Spearman correlation as well as the error in squared Mahalanobis distance with the true covariance used to generate the data as a baseline. 
In the non-separable setting considered here, the separable covariance estimator (MMLE) seems to yield more stable Mahalanobis diagnostics than the unregularized sample covariance on vectorized data, particularly in high-dimensional, limited-sample regimes. We also observe that the resulting structured estimators can provide more accurate covariance estimates than the unconstrained sample covariance, despite the violation of separability, particularly for the smaller sample sizes.\\
This reflects a bias–variance trade-off: although the separability assumption is misspecified, the structured estimator benefits from reduced variance. We emphasize that this comparison is against an unregularized non-separable estimator; more flexible or regularized non-separable approaches may behave differently.

Figures~\ref{fig:precision_non_separable_p10}-\ref{fig:AUC_non_separable_p10_all} present density plots of precision, recall, F-score, and AUC, respectively, for shift outlier detection in the non-separable setting ($p = 10$, $n = 1500, 5000$) described in Section~\ref{section:simulations}. Here, the AUC of the MFIF varies considerably with the dictionary and with the number of separable components, whereas the robust distances remain close to one throughout. Since the dictionary has to be fixed in advance and without label information, this sensitivity is a practical limitation, as also seen for the welding data in Section~\ref{section:examples}.

\begin{figure}[!h]
    \centering
    \includegraphics[width=1\linewidth]{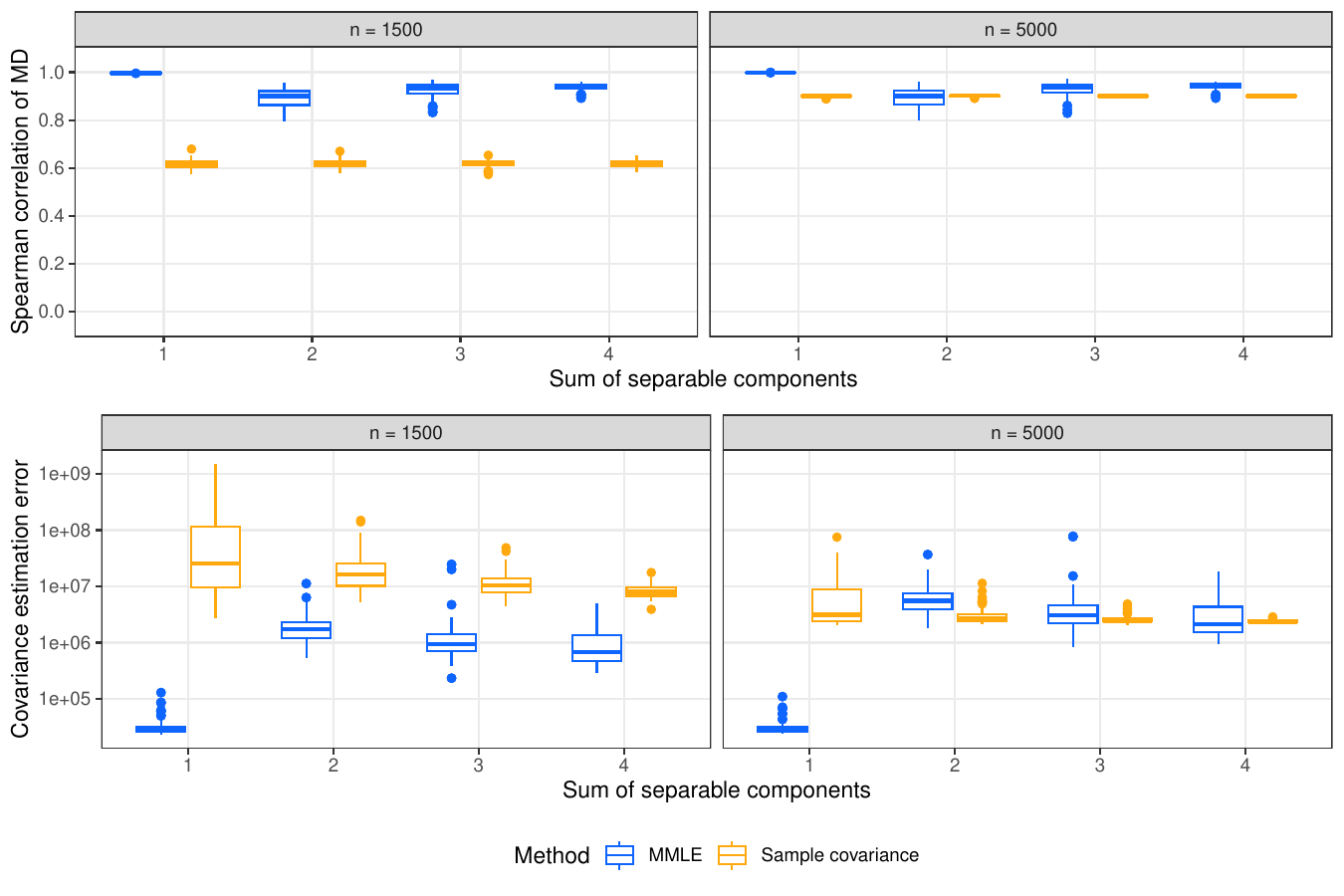}
    \caption{Comparison of Spearman correlation of squared Mahalanobis distances based on the sample covariance of the vectorized observations and the matrix MLE estimates in non-contaminated data for $p = 10$ as well as $n = 1500$ and $n = 5000$ in the upper panel. The lower panel shows the Mahalanobis distance estimation errors for the two estimators in the same setting.}
    \label{fig:sample_vs_mmle}
\end{figure}

\begin{figure}[p]
    \centering
    \includegraphics[width=1\linewidth]{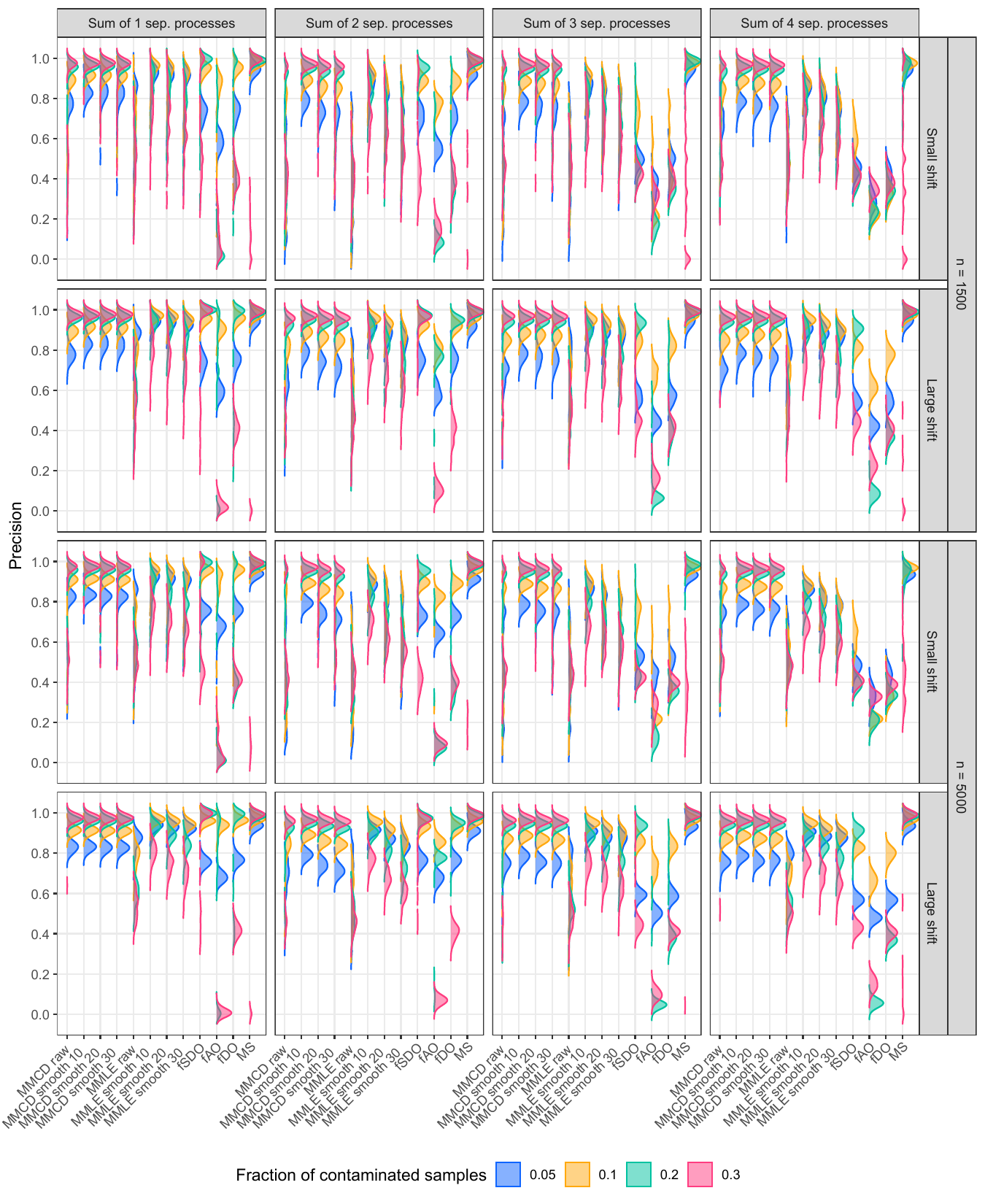}
    \caption{Precision for detecting small and large shift outliers in non-separable processes obtained as sums of $q_{\text{ns}} = 1, 2, 3, 4$ (columns) independent separable components with Matérn covariance functions and $p=10$ components for $n = 1500, 5000$ samples.}
    \label{fig:precision_non_separable_p10}
\end{figure}

\begin{figure}[p]
    \centering
    \includegraphics[width=1\linewidth]{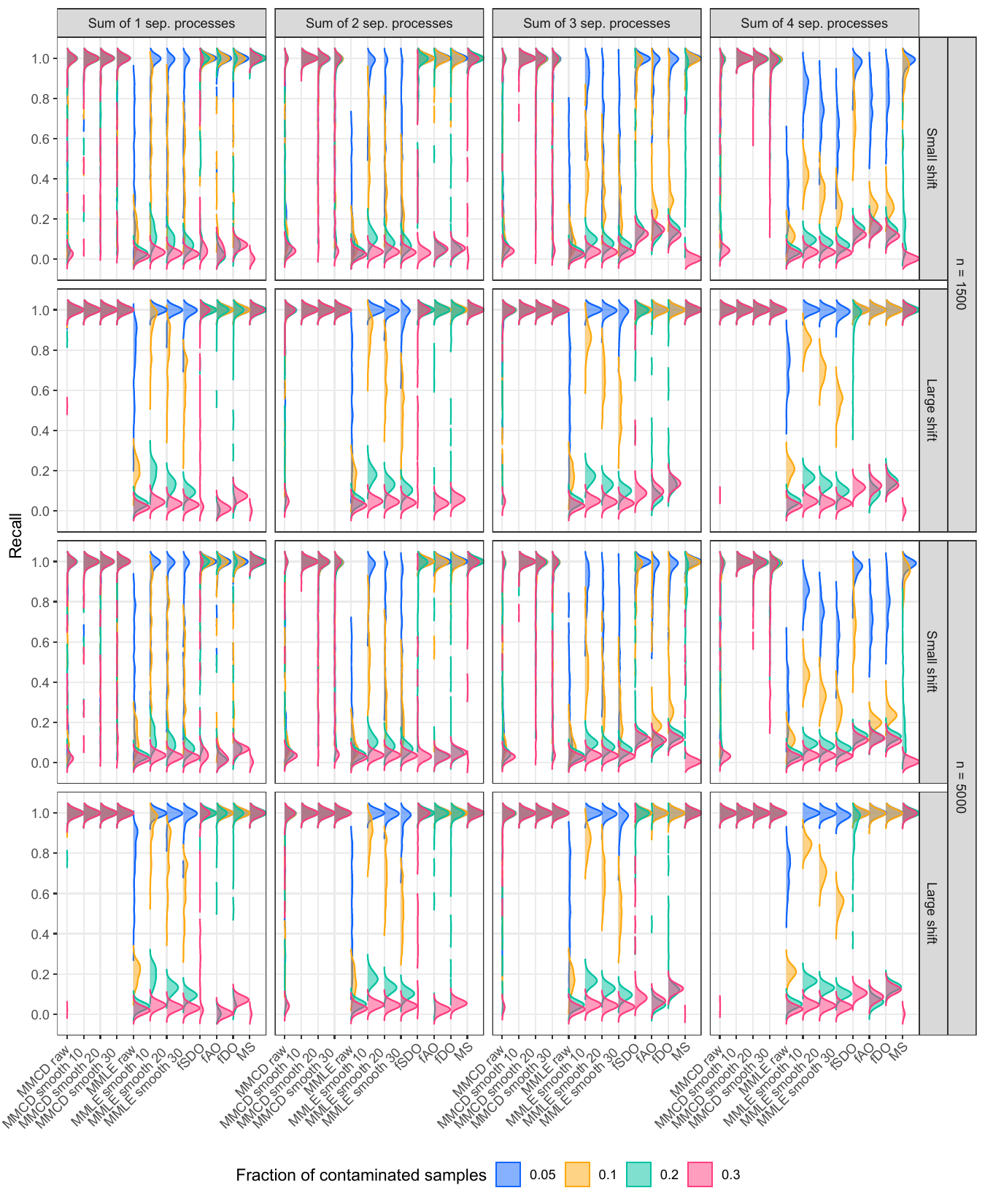}
    \caption{Recall for detecting small and large shift outliers in non-separable processes obtained as sums of $q_{\text{ns}} = 1, 2, 3, 4$ (columns) independent separable components with Matérn covariance functions and $p=10$ components for $n = 1500, 5000$ samples.}
    \label{fig:recall_non_separable_p10}
\end{figure}

\begin{figure}[p]
    \centering
    \includegraphics[width=1\linewidth]{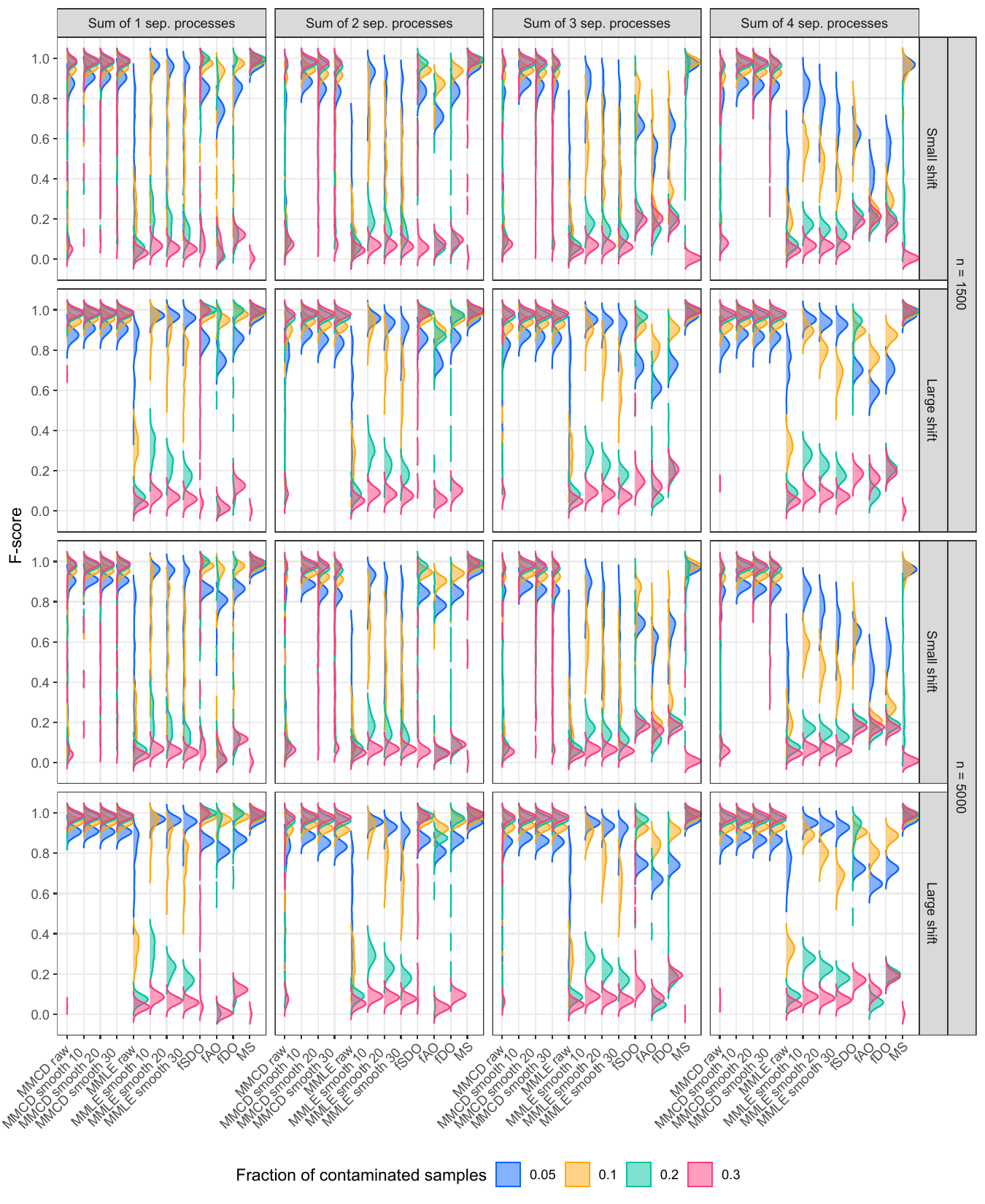}
    \caption{F-score for detecting small and large shift outliers in non-separable processes obtained as sums of $q_{\text{ns}} = 1, 2, 3, 4$ (columns) independent separable components with Matérn covariance functions and $p=10$ components for $n = 1500, 5000$ samples.}
    \label{fig:fscore_non_separable_p10}
\end{figure}

\begin{figure}[p]
    \centering
    \includegraphics[width=1\linewidth]{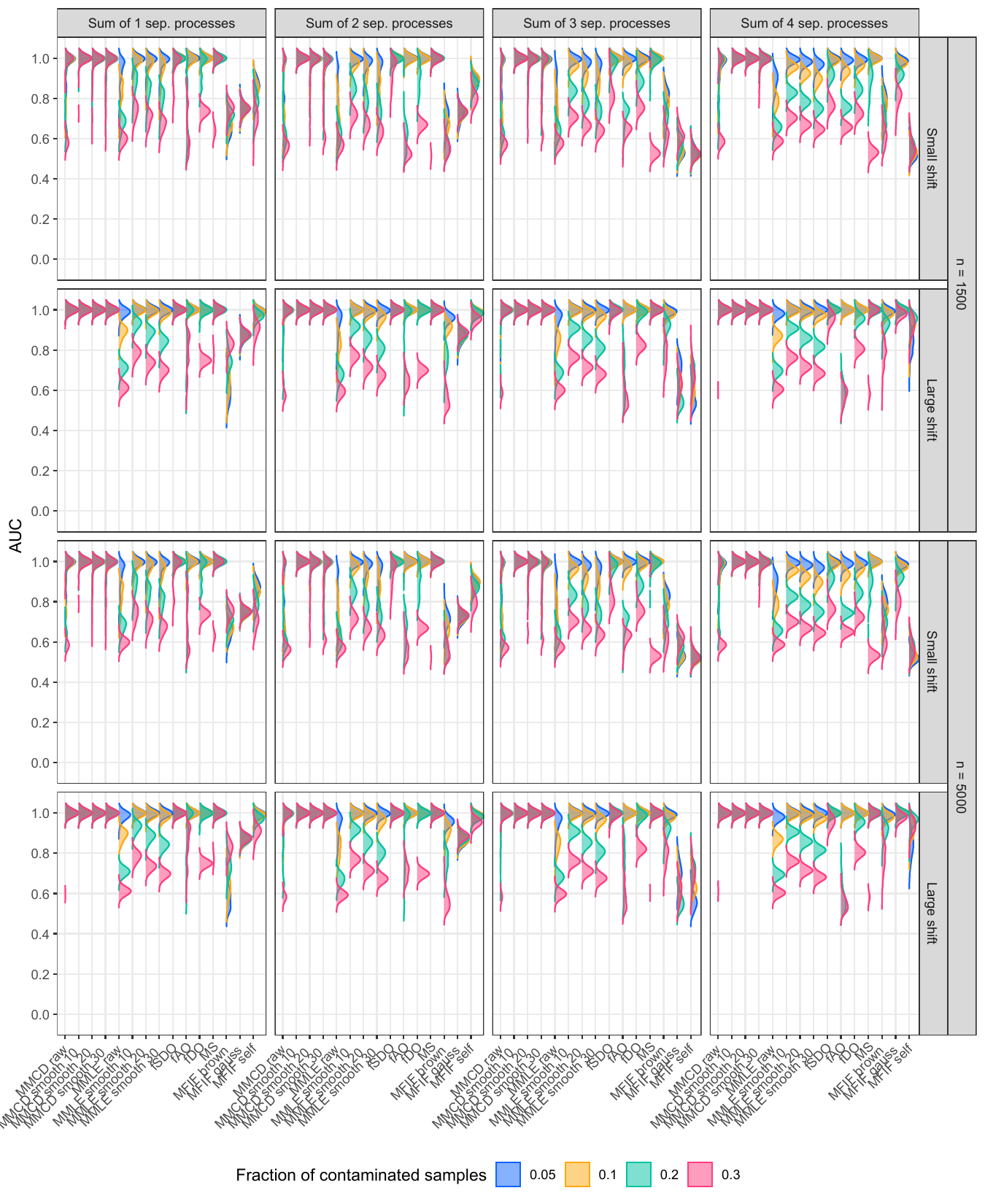}
    \caption{AUC for detecting small and large shift outliers in non-separable processes obtained as sums of $q_{\text{ns}} = 1, 2, 3, 4$ (columns) independent separable components with Matérn covariance functions and $p=10$ components for $n = 1500, 5000$ samples.}
    \label{fig:AUC_non_separable_p10_all}
\end{figure}

\cleardoublepage
\subsection{Additional Simulation Results in the Gaussian Setting} \label{appendix:simulations_gaussian}

This section reports the individual performance measures for the Gaussian settings of Table~\ref{tab1}, complementing the comparison of Section~\ref{section:simulations}. Figure~\ref{fig:test_p3} shows the rank-based comparison for $p = 3$ while Figures~\ref{fig:precision_shift}--\ref{fig:cov_cov} show precision, recall, F-score, AUC, and covariance estimation error for shift, shape, isolated, and covariance-induced outliers, each across $p = 3$ and $p = 50$, both covariance structures, and the small and large outlier magnitudes. Since MFIF returns an anomaly score without an outlyingness cutoff, it is included for AUC only. For $p = 50$ it is based on $10$ instead of $100$ replications, so its densities are less stable than those of the other methods.

\begin{figure}[!h]
    \centering
    \includegraphics[width=1\linewidth]{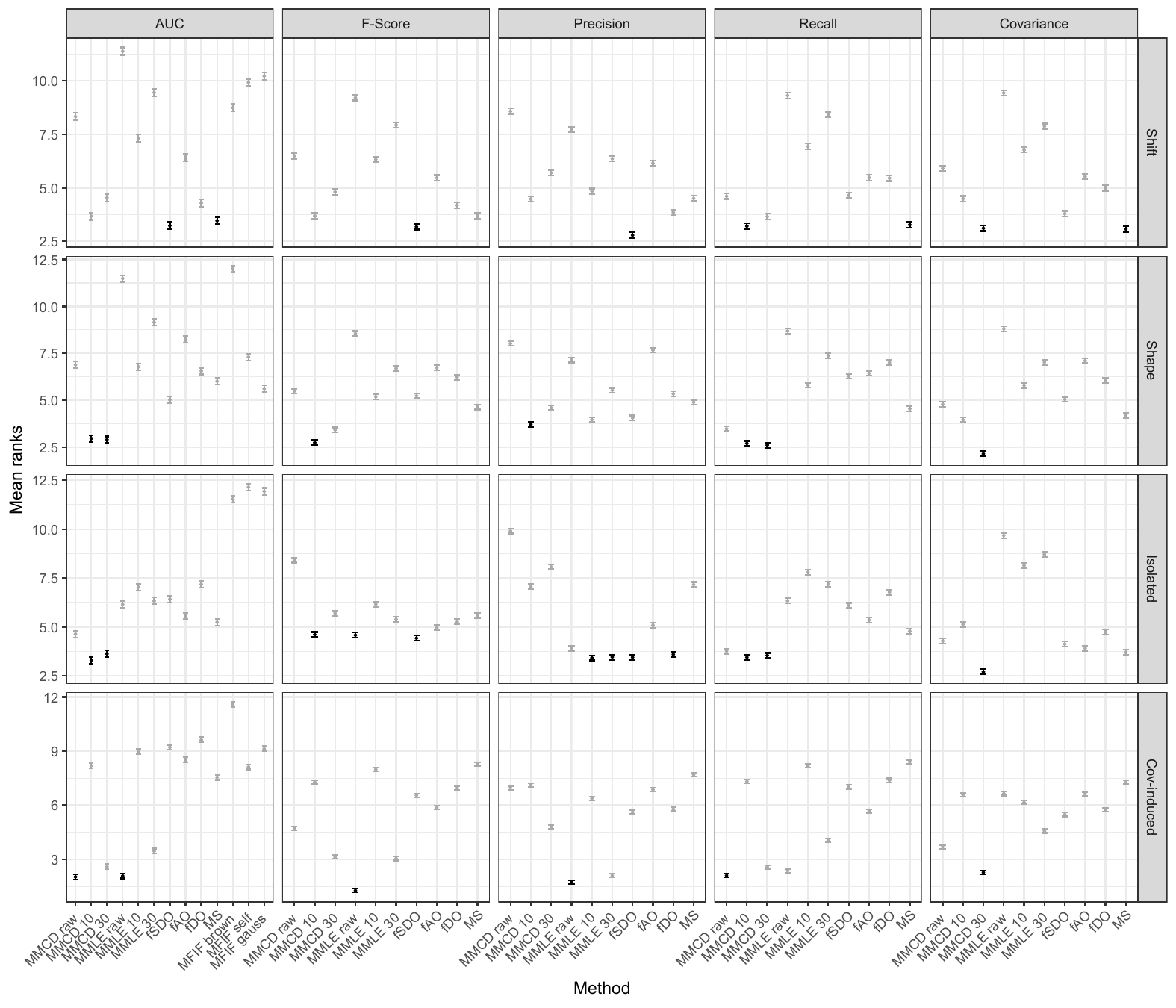}
    \caption{Rank-based comparison of methods across all simulation settings in Table \ref{tab1} for $p=3$. Intervals are based on the Friedman and Nemenyi tests; methods not significantly different (99\% level) from the best are shown in black, others in gray. Horizontal facets correspond to performance metrics, vertical facets to outlier types.}
    \label{fig:test_p3}
\end{figure}

\begin{figure}[p]
    \centering
    \includegraphics[width=1\linewidth]{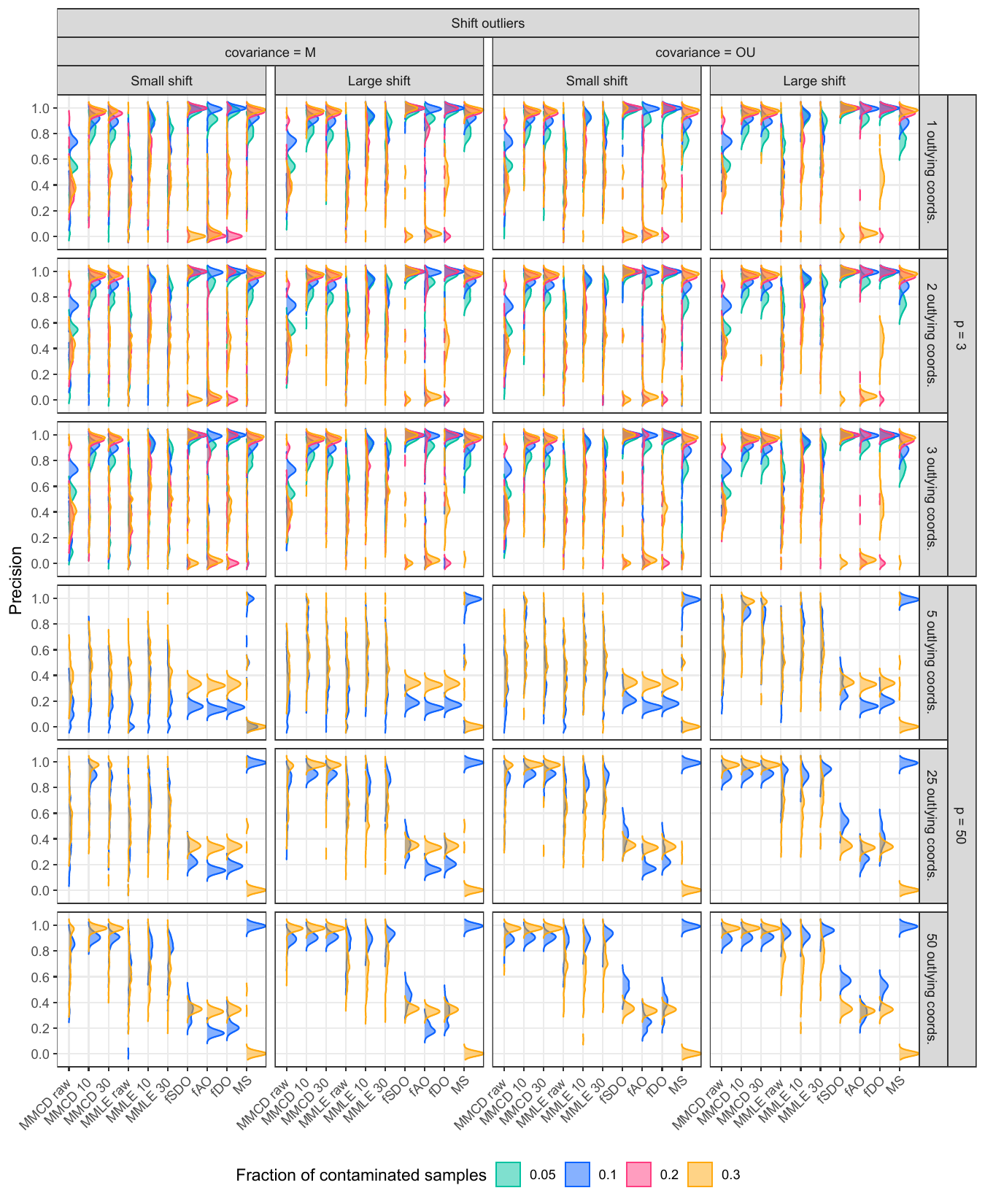}
    \caption{{Density plots of precision for shift outliers in the Gaussian setting with $n = 1000$ across $p = 3,50$, Ornstein-Uhlenbeck $\kappa_{\text{OU}}$ and Matérn-type $\kappa_{\text{Matérn}}$ covariance structures, small and large outlier magnitudes, coordinate contamination levels $\lfloor{\varepsilon_{\mathrm{cord}} \cdot p\rfloor}$ ($\varepsilon_{\mathrm{cord}} = 0.1, 0.5, 1$), and outlier proportions ($\varepsilon = 0.05, 0.1, 0.2, 0.3$).}}
    \label{fig:precision_shift}
\end{figure}

\begin{figure}[p]
    \centering
    \includegraphics[width=1\linewidth]{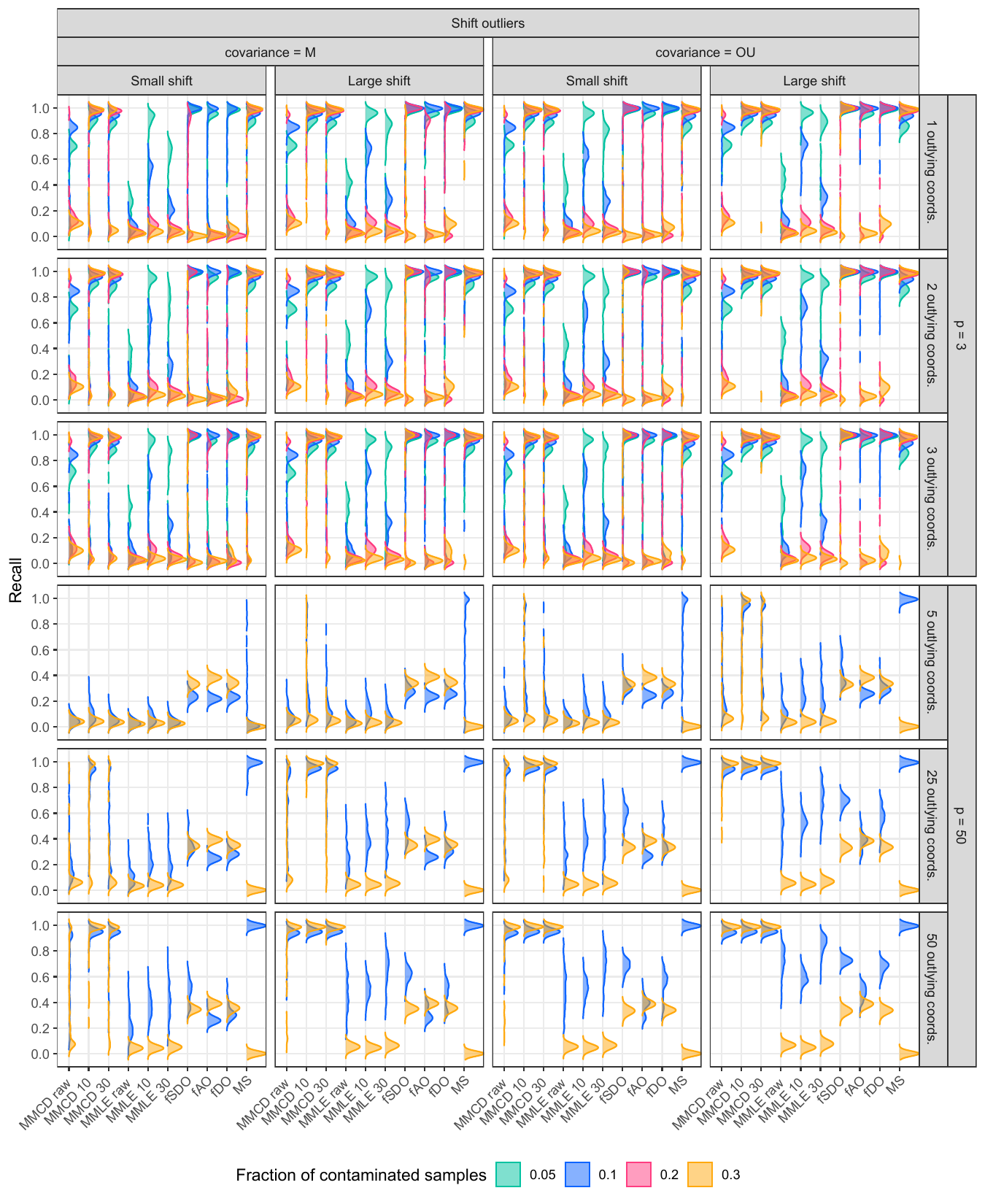}
    \caption{{Density plots of recall for shift outliers in the Gaussian setting with $n = 1000$ across $p = 3,50$, Ornstein-Uhlenbeck $\kappa_{\text{OU}}$ and Matérn-type $\kappa_{\text{Matérn}}$ covariance structures, small and large outlier magnitudes, coordinate contamination levels $\lfloor{\varepsilon_{\mathrm{cord}} \cdot p\rfloor}$ ($\varepsilon_{\mathrm{cord}} = 0.1, 0.5, 1$), and outlier proportions ($\varepsilon = 0.05, 0.1, 0.2, 0.3$).}}
    \label{fig:recall_shift}
\end{figure}

\begin{figure}[p]
    \centering
    \includegraphics[width=1\linewidth]{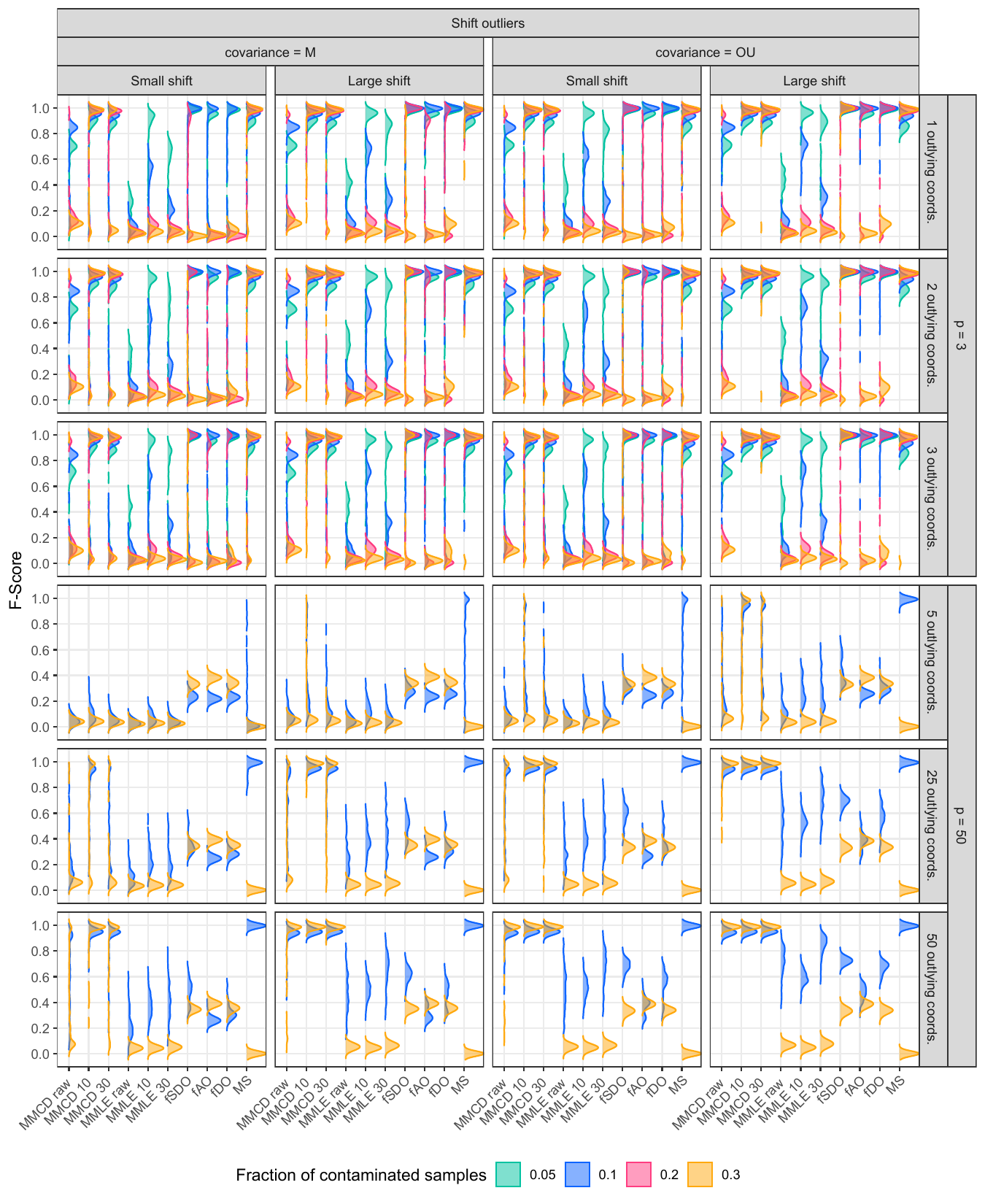}
    \caption{{Density plots of F-score for shift outliers in the Gaussian setting with $n = 1000$ across $p = 3,50$, Ornstein-Uhlenbeck $\kappa_{\text{OU}}$ and Matérn-type $\kappa_{\text{Matérn}}$ covariance structures, small and large outlier magnitudes, coordinate contamination levels $\lfloor{\varepsilon_{\mathrm{cord}} \cdot p\rfloor}$ ($\varepsilon_{\mathrm{cord}} = 0.1, 0.5, 1$), and outlier proportions ($\varepsilon = 0.05, 0.1, 0.2, 0.3$).}}
    \label{fig:fscore_shift}
\end{figure}

\begin{figure}[p]
    \centering
    \includegraphics[width=1\linewidth]{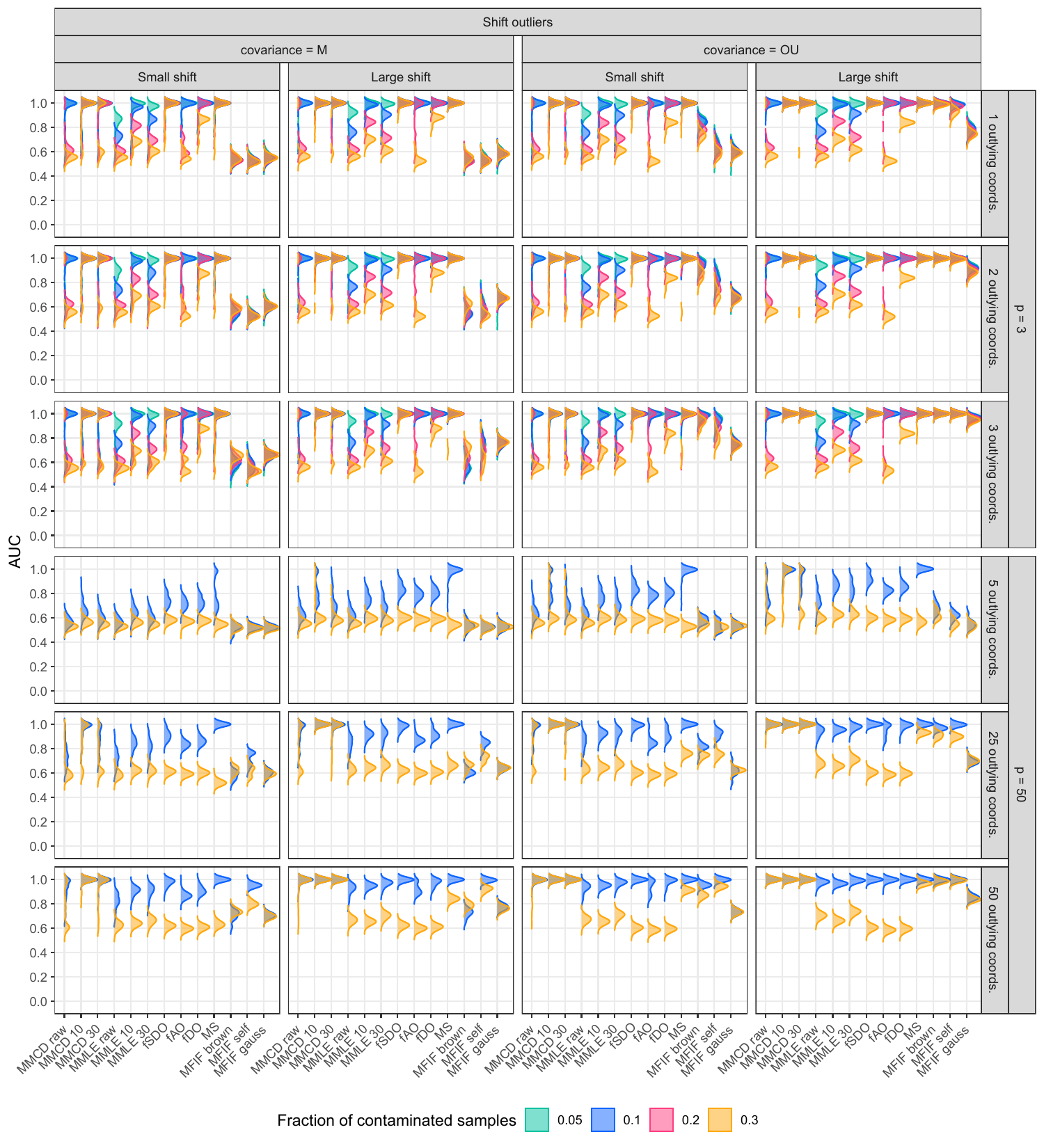}
    \caption{{Density plots of AUC values for shift outliers in the Gaussian setting with $n = 1000$ across $p = 3,50$, Ornstein-Uhlenbeck $\kappa_{\text{OU}}$ and Matérn-type $\kappa_{\text{Matérn}}$ covariance structures, small and large outlier magnitudes, coordinate contamination levels $\lfloor{\varepsilon_{\mathrm{cord}} \cdot p\rfloor}$ ($\varepsilon_{\mathrm{cord}} = 0.1, 0.5, 1$), and outlier proportions ($\varepsilon = 0.05, 0.1, 0.2, 0.3$).} For $p = 50$, MFIF densities are based on 10 replications instead of 100.}
    \label{fig:AUC_shift}
\end{figure}

\begin{figure}[p]
    \centering
    \includegraphics[width=1\linewidth]{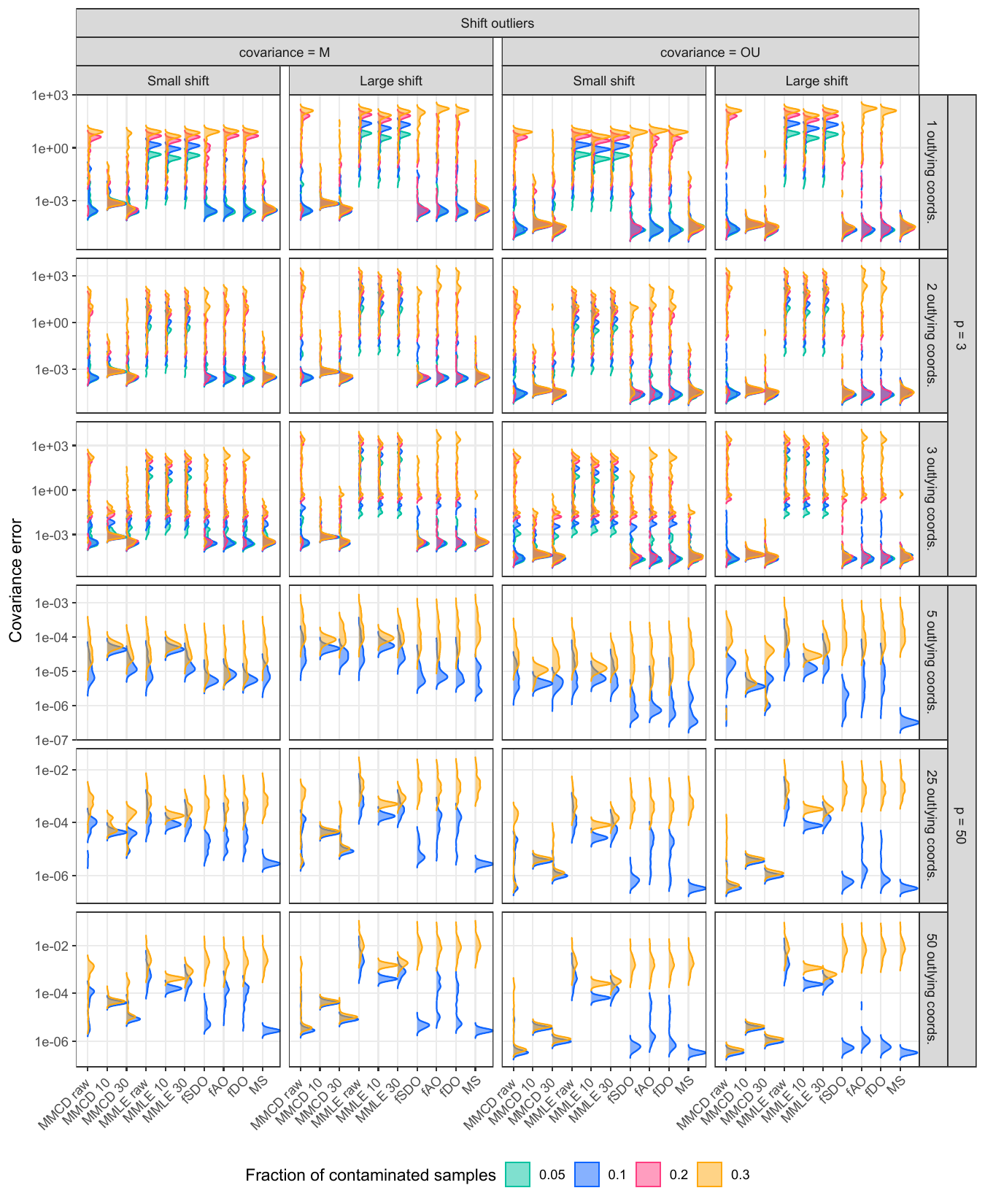}
    \caption{{Density plots of log covariance estimation error for shift outliers in the Gaussian setting with $n = 1000$ across $p = 3,50$, Ornstein-Uhlenbeck $\kappa_{\text{OU}}$ and Matérn-type $\kappa_{\text{Matérn}}$ covariance structures, small and large outlier magnitudes, coordinate contamination levels $\lfloor{\varepsilon_{\mathrm{cord}} \cdot p\rfloor}$ ($\varepsilon_{\mathrm{cord}} = 0.1, 0.5, 1$), and outlier proportions ($\varepsilon = 0.05, 0.1, 0.2, 0.3$).}}
    \label{fig:cov_shift}
\end{figure}

\begin{figure}[p]
    \centering
    \includegraphics[width=1\linewidth]{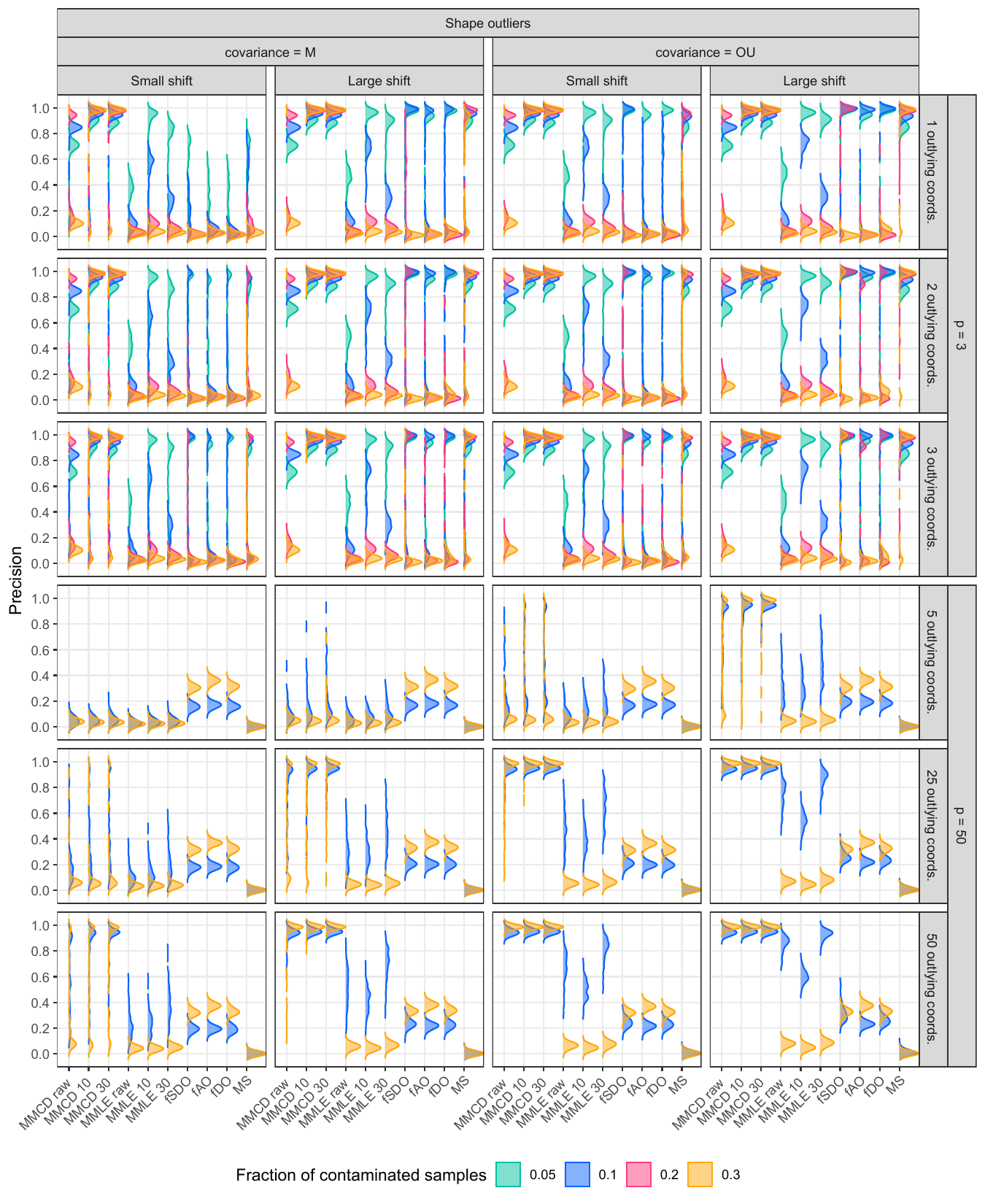}
    \caption{{Density plots of precision for shape outliers in the Gaussian setting with $n = 1000$ across $p = 3,50$, Ornstein-Uhlenbeck $\kappa_{\text{OU}}$ and Matérn-type $\kappa_{\text{Matérn}}$ covariance structures, small and large outlier magnitudes, coordinate contamination levels $\lfloor{\varepsilon_{\mathrm{cord}} \cdot p\rfloor}$ ($\varepsilon_{\mathrm{cord}} = 0.1, 0.5, 1$), and outlier proportions ($\varepsilon = 0.05, 0.1, 0.2, 0.3$).}}
    \label{fig:precision_shape}
\end{figure}

\begin{figure}[p]
    \centering
    \includegraphics[width=1\linewidth]{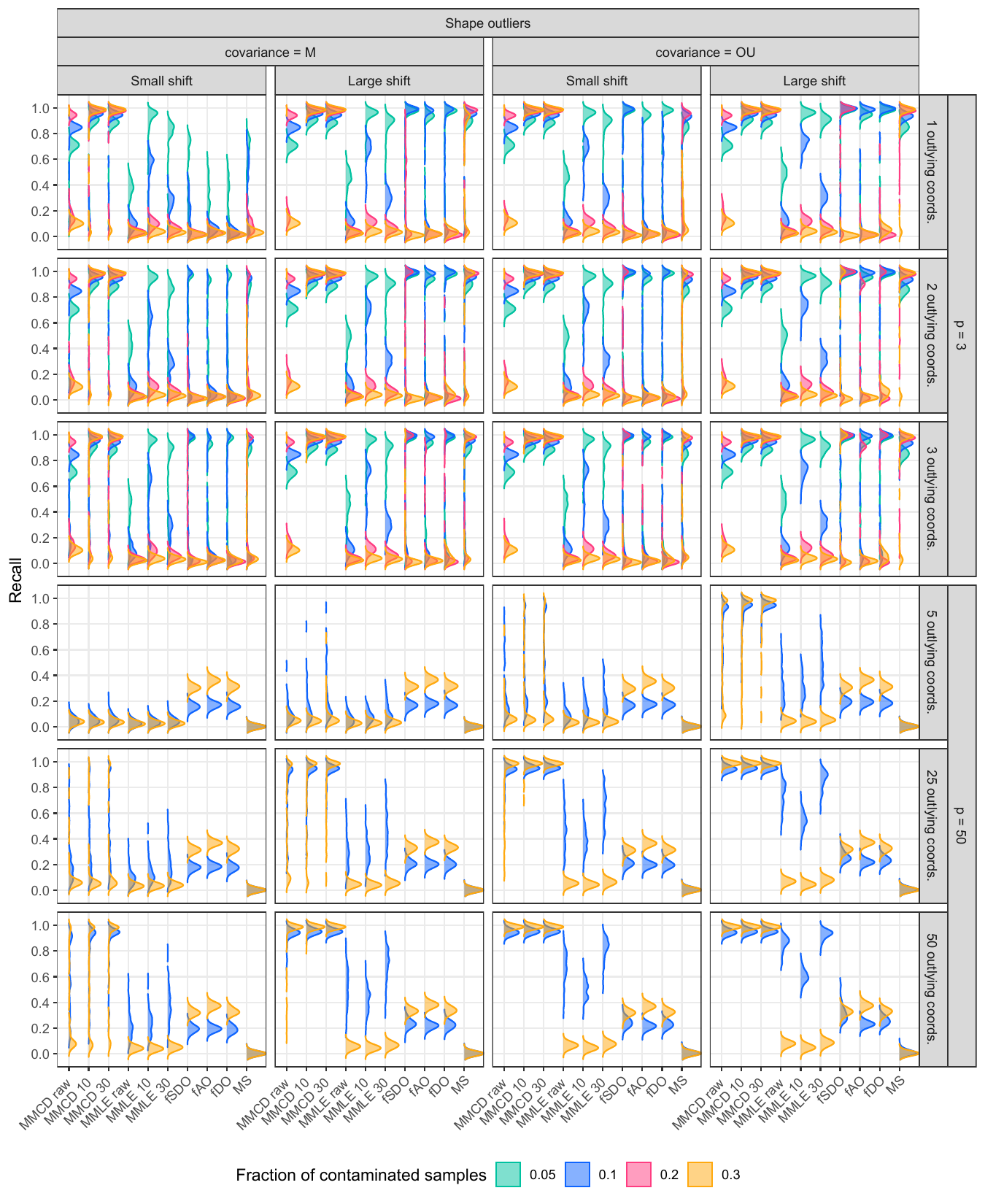}
    \caption{{Density plots of recall for shape outliers in the Gaussian setting with $n = 1000$ across $p = 3,50$, Ornstein-Uhlenbeck $\kappa_{\text{OU}}$ and Matérn-type $\kappa_{\text{Matérn}}$ covariance structures, small and large outlier magnitudes, coordinate contamination levels $\lfloor{\varepsilon_{\mathrm{cord}} \cdot p\rfloor}$ ($\varepsilon_{\mathrm{cord}} = 0.1, 0.5, 1$), and outlier proportions ($\varepsilon = 0.05, 0.1, 0.2, 0.3$).}}
    \label{fig:recall_shape}
\end{figure}

\begin{figure}[p]
    \centering
    \includegraphics[width=1\linewidth]{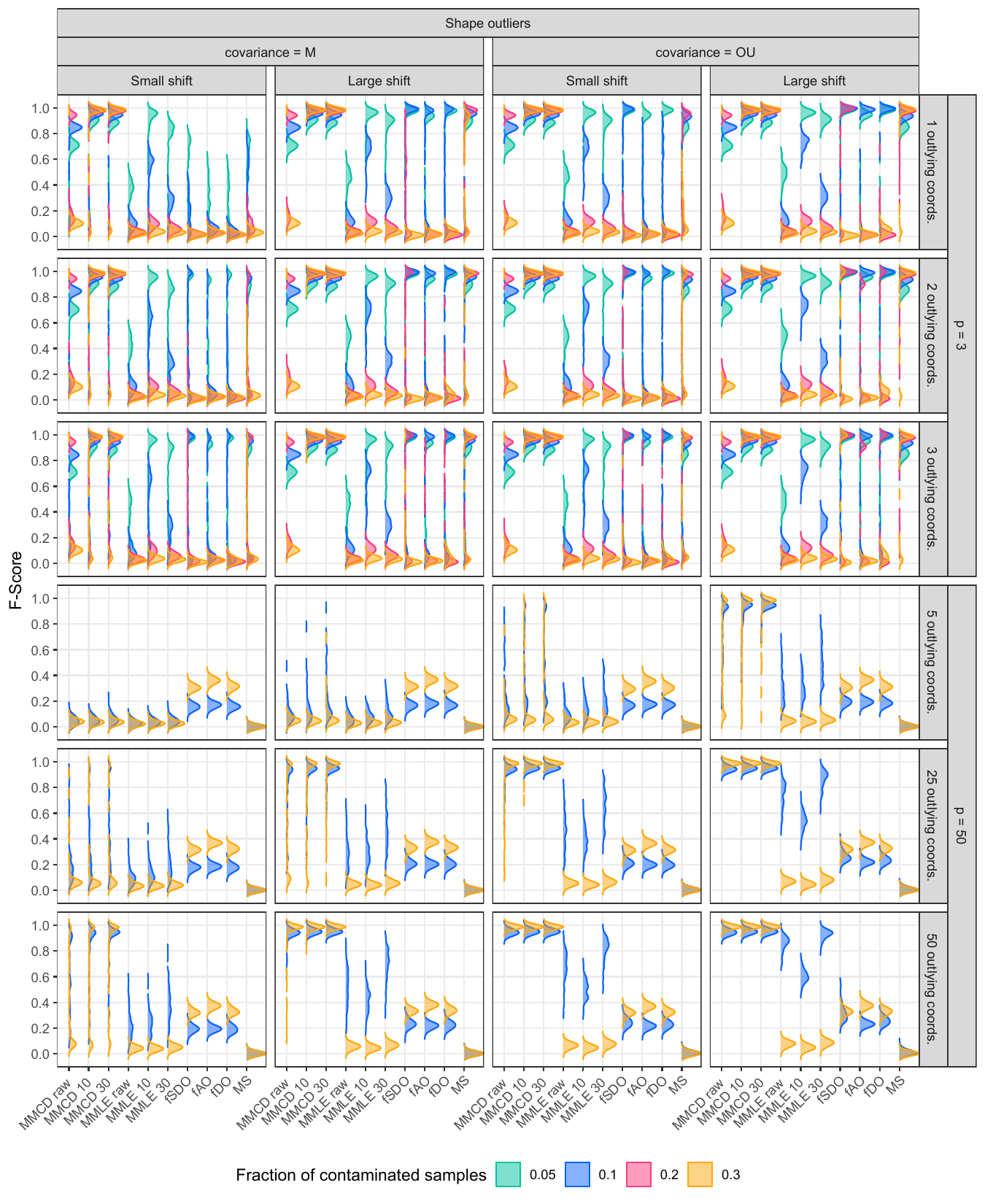}
    \caption{{Density plots of F-score for shape outliers in the Gaussian setting with $n = 1000$ across $p = 3,50$, Ornstein-Uhlenbeck $\kappa_{\text{OU}}$ and Matérn-type $\kappa_{\text{Matérn}}$ covariance structures, small and large outlier magnitudes, coordinate contamination levels $\lfloor{\varepsilon_{\mathrm{cord}} \cdot p\rfloor}$ ($\varepsilon_{\mathrm{cord}} = 0.1, 0.5, 1$), and outlier proportions ($\varepsilon = 0.05, 0.1, 0.2, 0.3$).}}
    \label{fig:fscore_shape}
\end{figure}

\begin{figure}[p]
    \centering
    \includegraphics[width=1\linewidth]{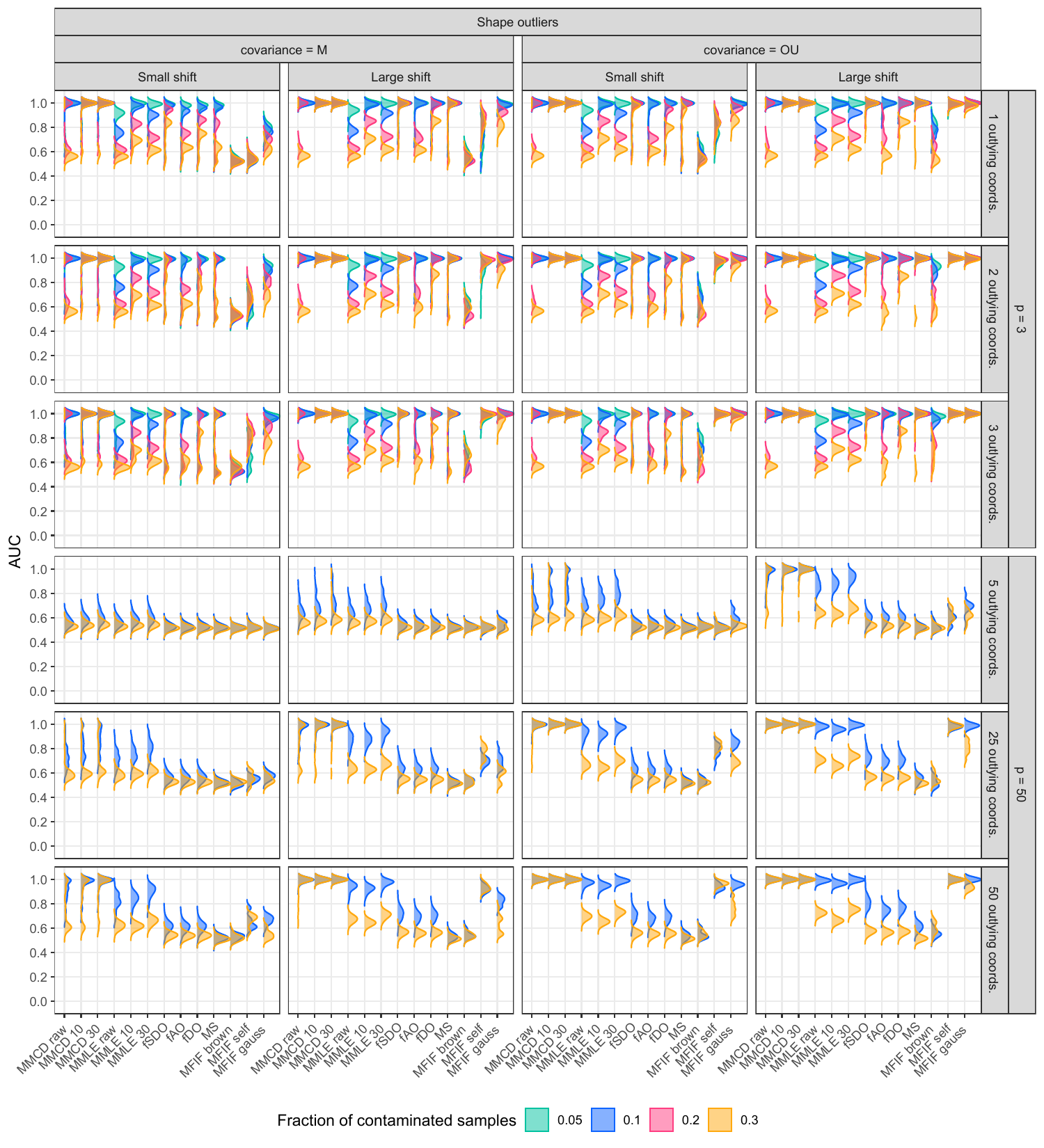}
    \caption{{Density plots of AUC values for shape outliers in the Gaussian setting with $n = 1000$ across $p = 3,50$, Ornstein-Uhlenbeck $\kappa_{\text{OU}}$ and Matérn-type $\kappa_{\text{Matérn}}$ covariance structures, small and large outlier magnitudes, coordinate contamination levels $\lfloor{\varepsilon_{\mathrm{cord}} \cdot p\rfloor}$ ($\varepsilon_{\mathrm{cord}} = 0.1, 0.5, 1$), and outlier proportions ($\varepsilon = 0.05, 0.1, 0.2, 0.3$).} For $p = 50$, MFIF densities are based on 10 replications instead of 100.}
    \label{fig:AUC_shape}
\end{figure}

\begin{figure}[p]
    \centering
    \includegraphics[width=1\linewidth]{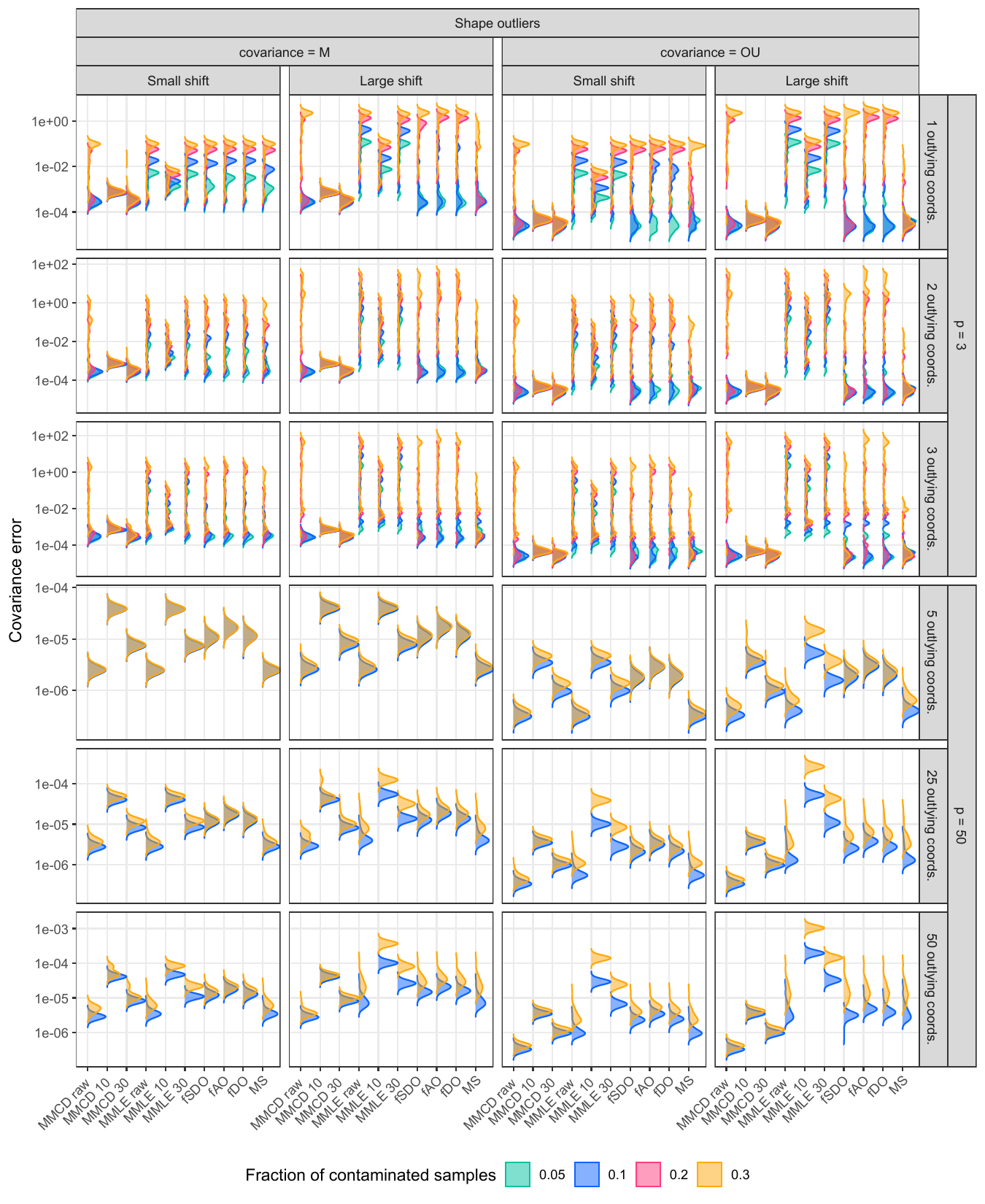}
    \caption{{Density plots of log covariance estimation error for shape outliers in the Gaussian setting with $n = 1000$ across $p = 3,50$, Ornstein-Uhlenbeck $\kappa_{\text{OU}}$ and Matérn-type $\kappa_{\text{Matérn}}$ covariance structures, small and large outlier magnitudes, coordinate contamination levels $\lfloor{\varepsilon_{\mathrm{cord}} \cdot p\rfloor}$ ($\varepsilon_{\mathrm{cord}} = 0.1, 0.5, 1$), and outlier proportions ($\varepsilon = 0.05, 0.1, 0.2, 0.3$).}}
    \label{fig:cov_shape}
\end{figure}

\begin{figure}[p]
    \centering
    \includegraphics[width=1\linewidth]{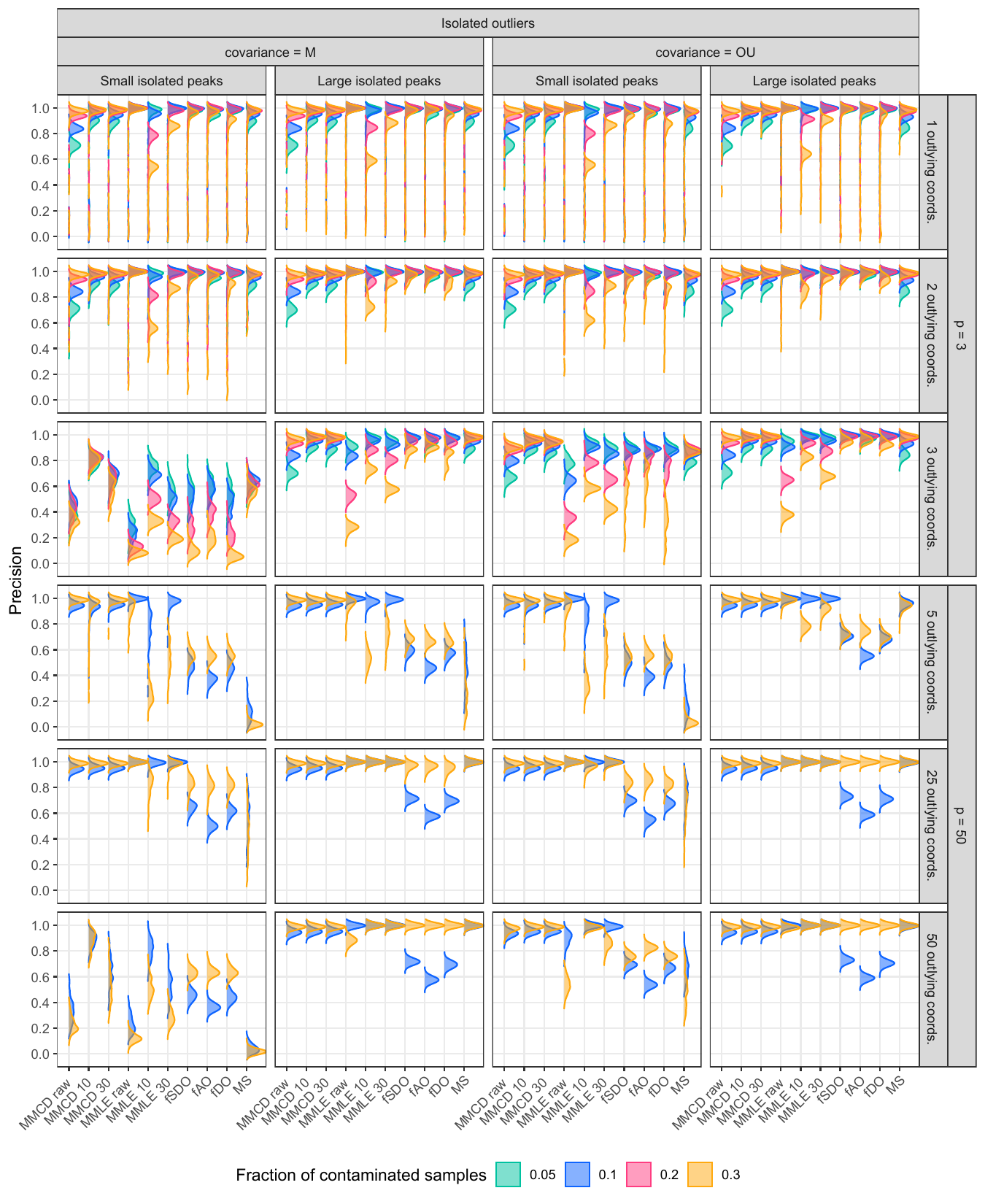}
    \caption{{Density plots of precision for isolated outliers in the Gaussian setting with $n = 1000$ across $p = 3,50$, Ornstein-Uhlenbeck $\kappa_{\text{OU}}$ and Matérn-type $\kappa_{\text{Matérn}}$ covariance structures, small and large outlier magnitudes, coordinate contamination levels $\lfloor{\varepsilon_{\mathrm{cord}} \cdot p\rfloor}$ ($\varepsilon_{\mathrm{cord}} = 0.1, 0.5, 1$), and outlier proportions ($\varepsilon = 0.05, 0.1, 0.2, 0.3$).}}
    \label{fig:precision_iso}
\end{figure}

\begin{figure}[p]
    \centering
    \includegraphics[width=1\linewidth]{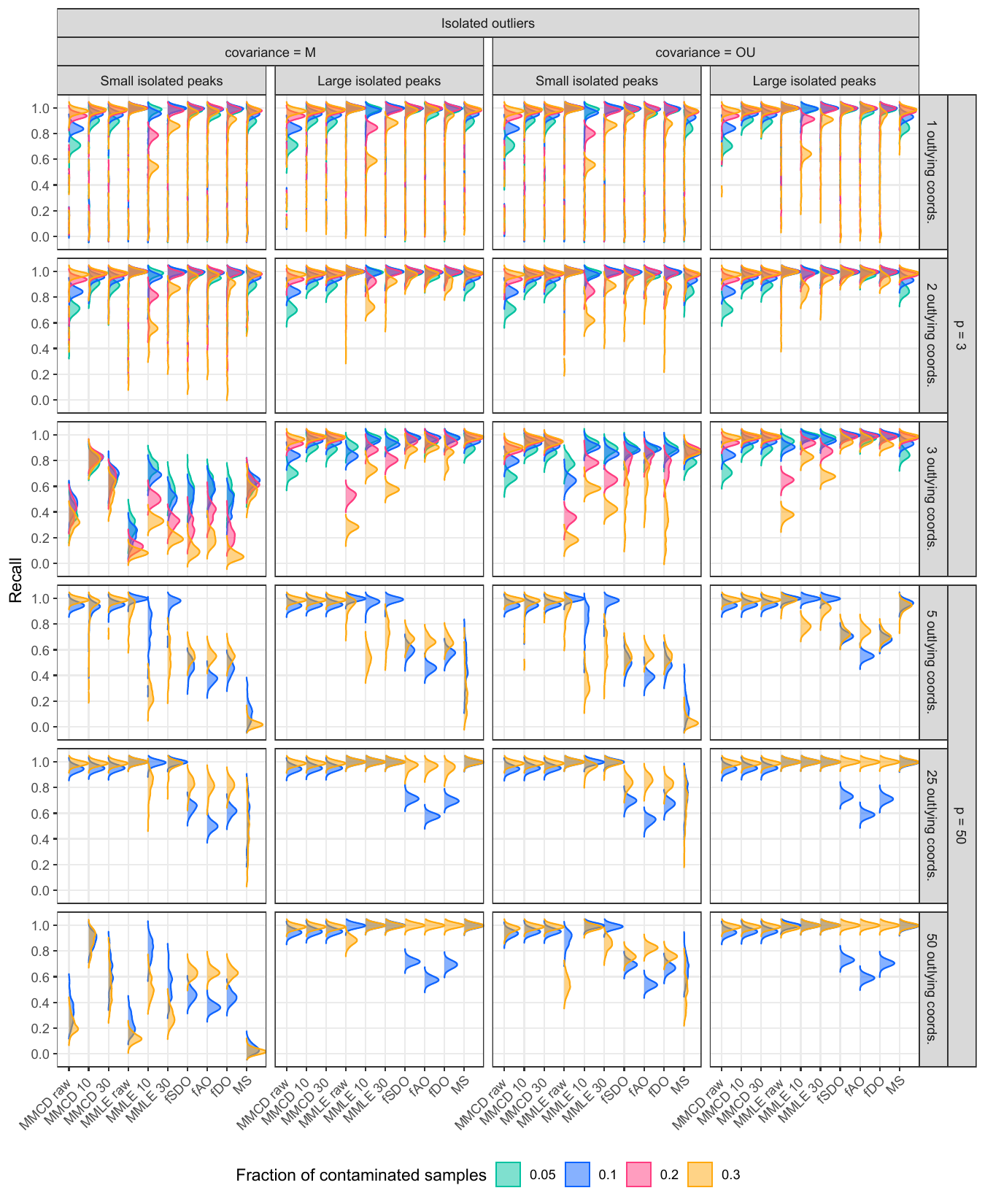}
    \caption{{Density plots of recall for isolated outliers in the Gaussian setting with $n = 1000$ across $p = 3,50$, Ornstein-Uhlenbeck $\kappa_{\text{OU}}$ and Matérn-type $\kappa_{\text{Matérn}}$ covariance structures, small and large outlier magnitudes, coordinate contamination levels $\lfloor{\varepsilon_{\mathrm{cord}} \cdot p\rfloor}$ ($\varepsilon_{\mathrm{cord}} = 0.1, 0.5, 1$), and outlier proportions ($\varepsilon = 0.05, 0.1, 0.2, 0.3$).}}
    \label{fig:recall_iso}
\end{figure}

\begin{figure}[p]
    \centering
    \includegraphics[width=1\linewidth]{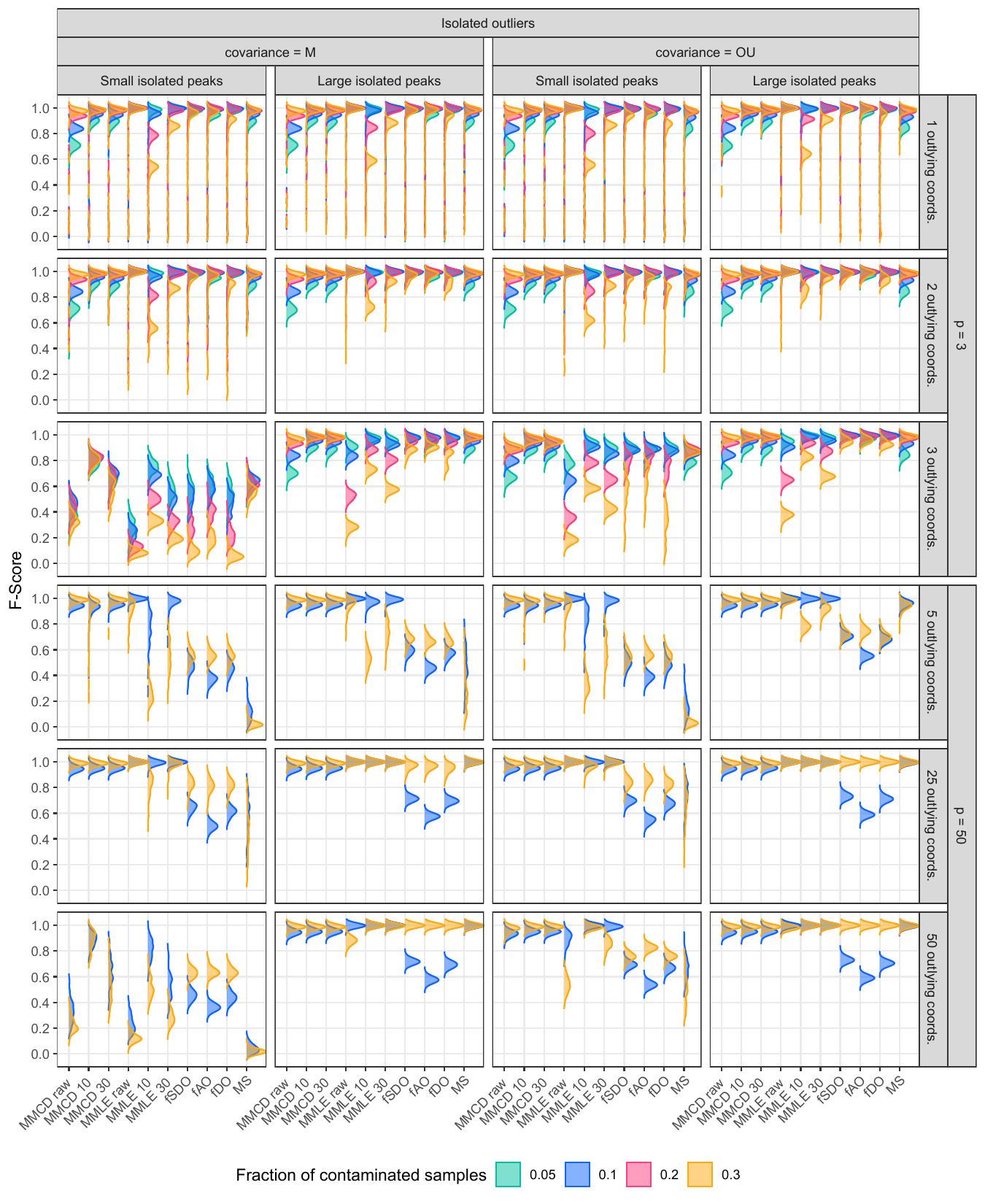}
    \caption{{Density plots of F-score for isolated outliers in the Gaussian setting with $n = 1000$ across $p = 3,50$, Ornstein-Uhlenbeck $\kappa_{\text{OU}}$ and Matérn-type $\kappa_{\text{Matérn}}$ covariance structures, small and large outlier magnitudes, coordinate contamination levels $\lfloor{\varepsilon_{\mathrm{cord}} \cdot p\rfloor}$ ($\varepsilon_{\mathrm{cord}} = 0.1, 0.5, 1$), and outlier proportions ($\varepsilon = 0.05, 0.1, 0.2, 0.3$).}}
    \label{fig:fscore_iso}
\end{figure}

\begin{figure}[p]
    \centering
    \includegraphics[width=1\linewidth]{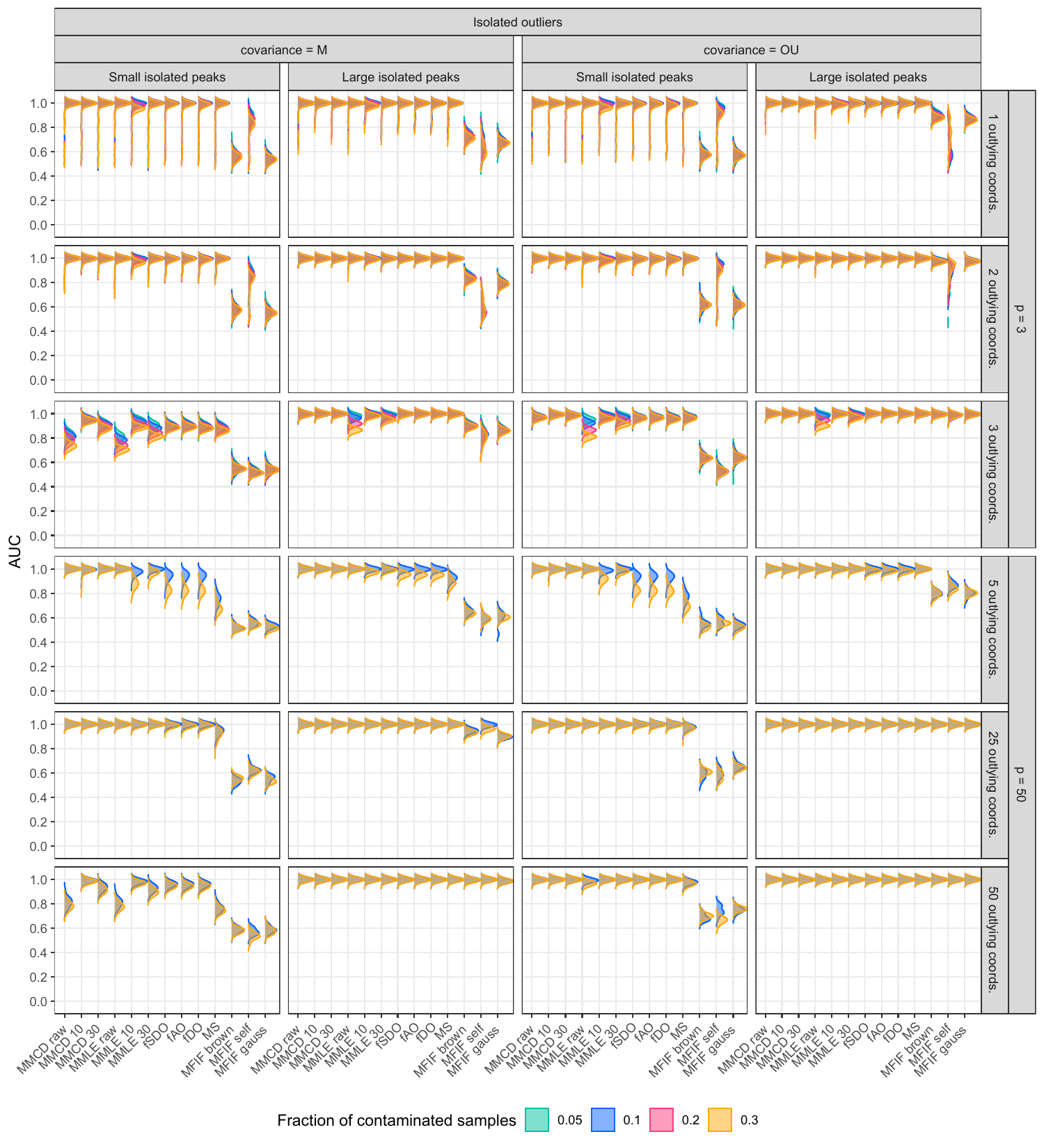}
    \caption{{Density plots of AUC values for isolated outliers in the Gaussian setting with $n = 1000$ across $p = 3,50$, Ornstein-Uhlenbeck $\kappa_{\text{OU}}$ and Matérn-type $\kappa_{\text{Matérn}}$ covariance structures, small and large outlier magnitudes, coordinate contamination levels $\lfloor{\varepsilon_{\mathrm{cord}} \cdot p\rfloor}$ ($\varepsilon_{\mathrm{cord}} = 0.1, 0.5, 1$), and outlier proportions ($\varepsilon = 0.05, 0.1, 0.2, 0.3$).} For $p = 50$, MFIF densities are based on 10 replications instead of 100.}
    \label{fig:AUC_iso}
\end{figure}

\begin{figure}[p]
    \centering
    \includegraphics[width=1\linewidth]{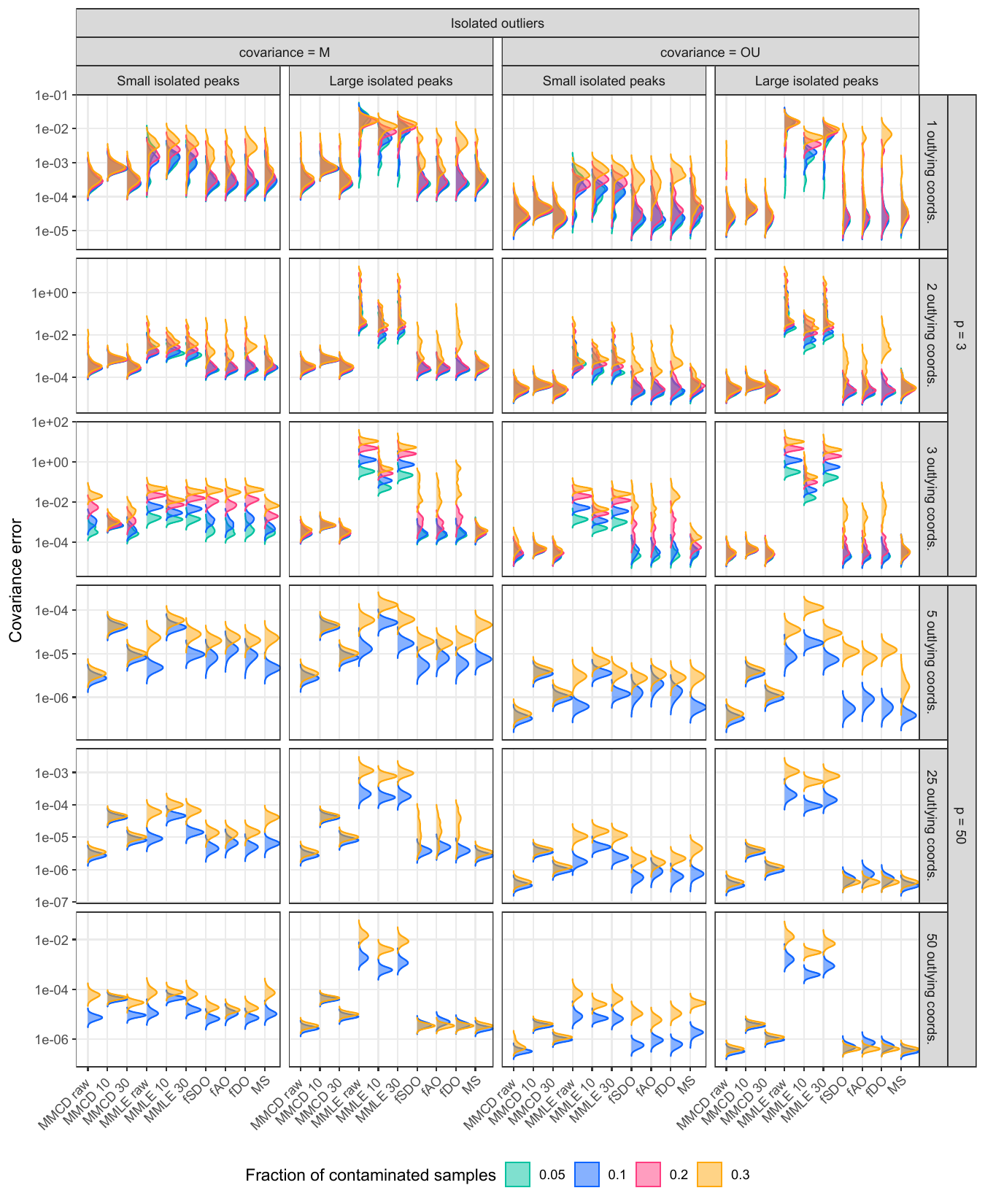}
    \caption{{Density plots of log covariance estimation error for isolated outliers in the Gaussian setting with $n = 1000$ across $p = 3,50$, Ornstein-Uhlenbeck $\kappa_{\text{OU}}$ and Matérn-type $\kappa_{\text{Matérn}}$ covariance structures, small and large outlier magnitudes, coordinate contamination levels $\lfloor{\varepsilon_{\mathrm{cord}} \cdot p\rfloor}$ ($\varepsilon_{\mathrm{cord}} = 0.1, 0.5, 1$), and outlier proportions ($\varepsilon = 0.05, 0.1, 0.2, 0.3$).}}
    \label{fig:cov_iso}
\end{figure}

\begin{figure}[p]
    \centering
    \includegraphics[width=1\linewidth]{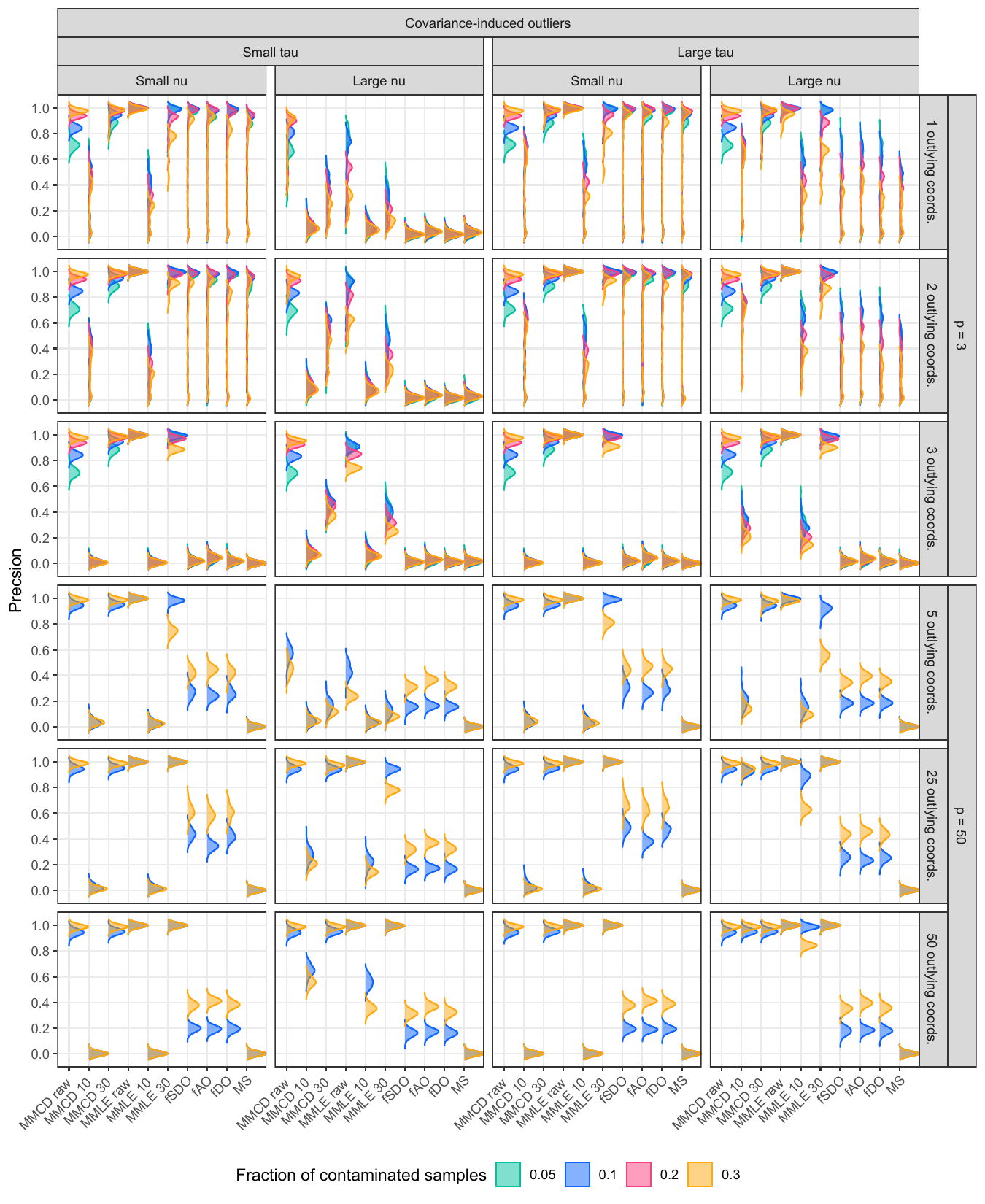}
    \caption{{Density plots of precision for shift outliers in the Gaussian setting with $n = 1000$ across $p = 3,50$, Ornstein-Uhlenbeck $\kappa_{\text{OU}}$ and Matérn-type $\kappa_{\text{Matérn}}$ covariance structures, small and large outlier magnitudes, coordinate contamination levels $\lfloor{\varepsilon_{\mathrm{cord}} \cdot p\rfloor}$ ($\varepsilon_{\mathrm{cord}} = 0.1, 0.5, 1$), and outlier proportions ($\varepsilon = 0.05, 0.1, 0.2, 0.3$).}}
    \label{fig:precision_cov}
\end{figure}

\begin{figure}[p]
    \centering
    \includegraphics[width=1\linewidth]{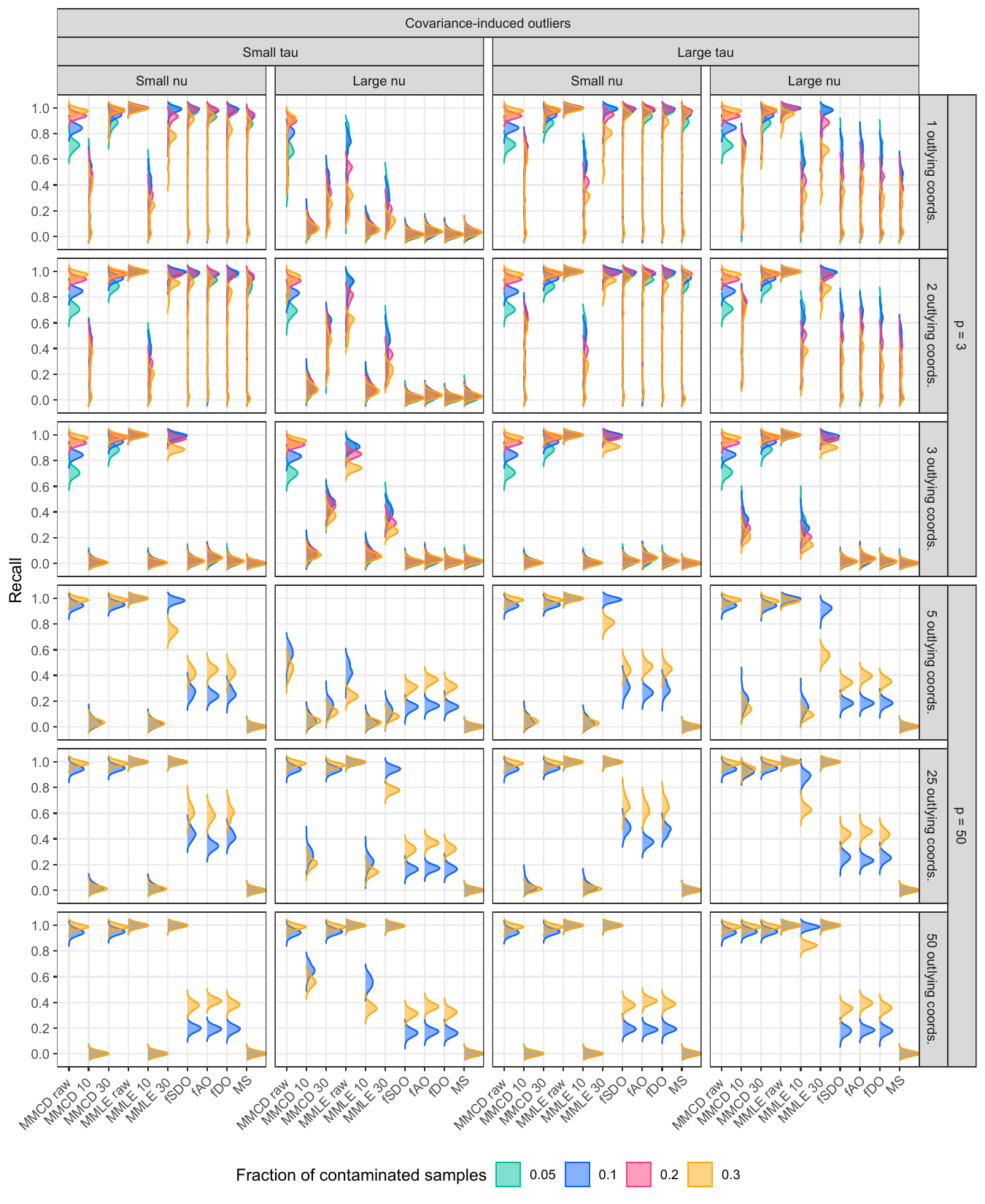}
    \caption{{Density plots of recall for shift outliers in the Gaussian setting with $n = 1000$ across $p = 3,50$, Ornstein-Uhlenbeck $\kappa_{\text{OU}}$ and Matérn-type $\kappa_{\text{Matérn}}$ covariance structures, small and large outlier magnitudes, coordinate contamination levels $\lfloor{\varepsilon_{\mathrm{cord}} \cdot p\rfloor}$ ($\varepsilon_{\mathrm{cord}} = 0.1, 0.5, 1$), and outlier proportions ($\varepsilon = 0.05, 0.1, 0.2, 0.3$).}}
    \label{fig:recall_cov}
\end{figure}

\begin{figure}[p]
    \centering
    \includegraphics[width=1\linewidth]{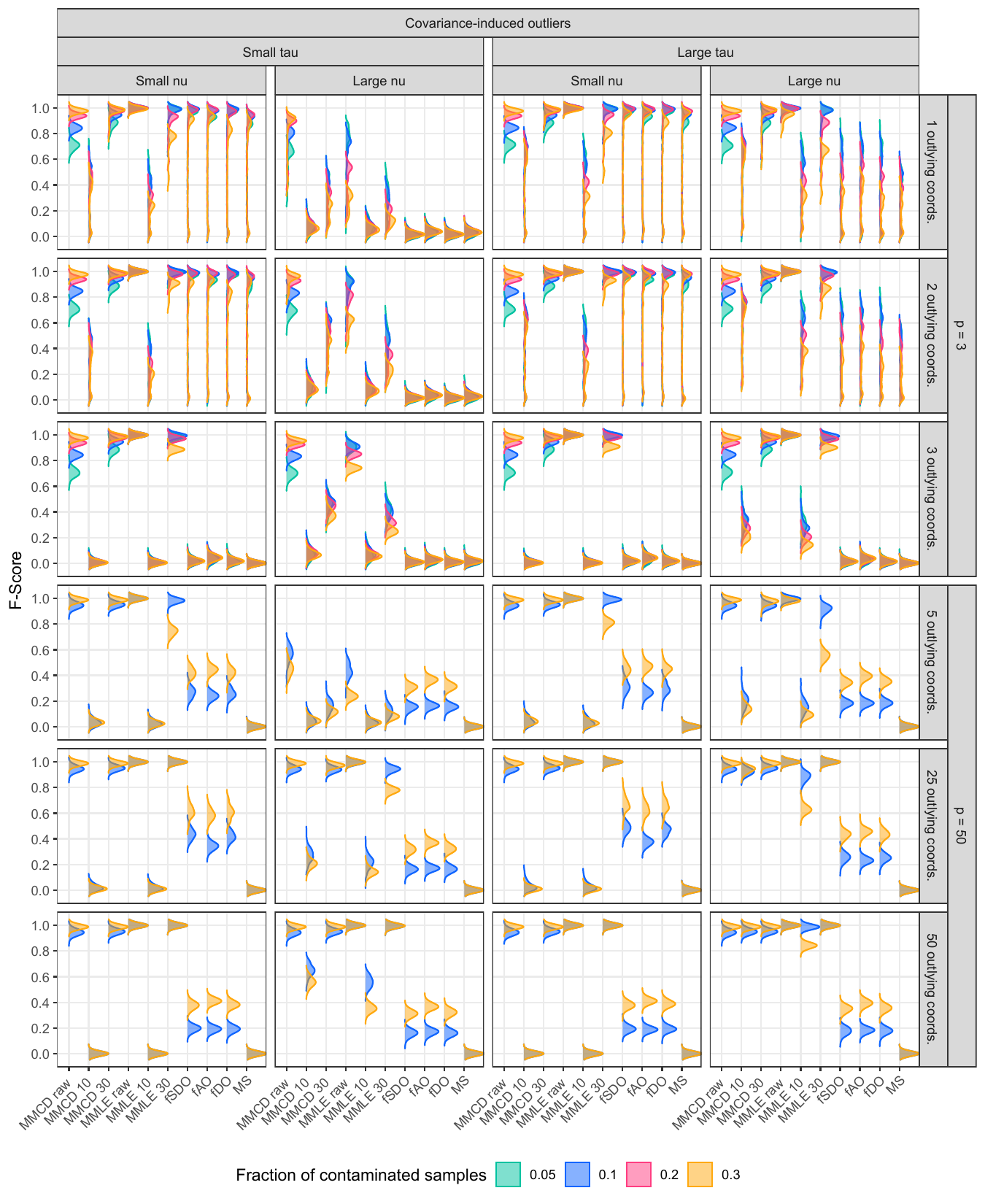}
    \caption{{Density plots of F-score for shift outliers in the Gaussian setting with $n = 1000$ across $p = 3,50$, Ornstein-Uhlenbeck $\kappa_{\text{OU}}$ and Matérn-type $\kappa_{\text{Matérn}}$ covariance structures, small and large outlier magnitudes, coordinate contamination levels $\lfloor{\varepsilon_{\mathrm{cord}} \cdot p\rfloor}$ ($\varepsilon_{\mathrm{cord}} = 0.1, 0.5, 1$), and outlier proportions ($\varepsilon = 0.05, 0.1, 0.2, 0.3$).}}
    \label{fig:fscore_cov}
\end{figure}

\begin{figure}[p]
    \centering
    \includegraphics[width=1\linewidth]{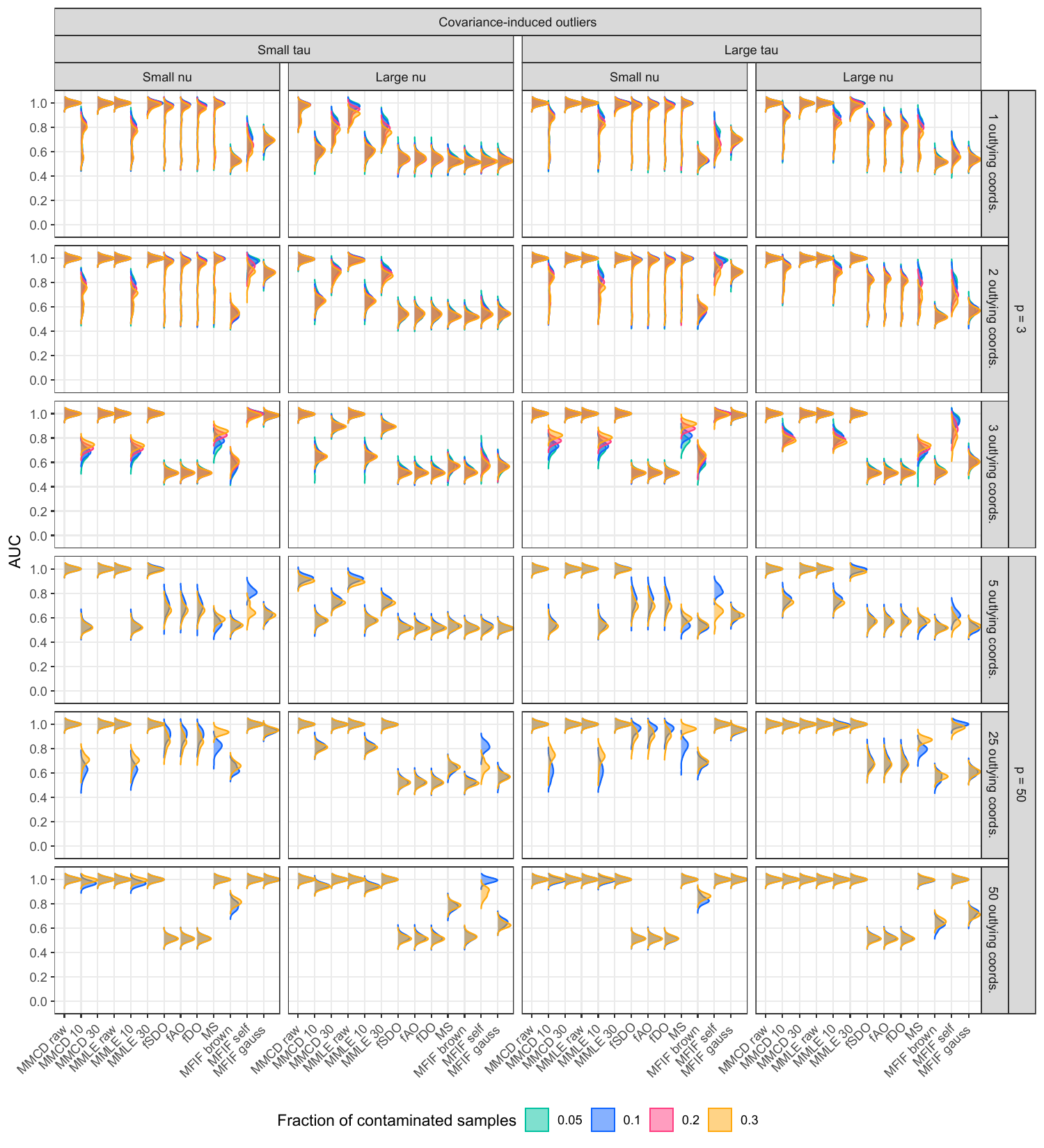}
    \caption{{Density plots of AUC values for shift outliers in the Gaussian setting with $n = 1000$ across $p = 3,50$, Ornstein-Uhlenbeck $\kappa_{\text{OU}}$ and Matérn-type $\kappa_{\text{Matérn}}$ covariance structures, small and large outlier magnitudes, coordinate contamination levels $\lfloor{\varepsilon_{\mathrm{cord}} \cdot p\rfloor}$ ($\varepsilon_{\mathrm{cord}} = 0.1, 0.5, 1$), and outlier proportions ($\varepsilon = 0.05, 0.1, 0.2, 0.3$).} For $p = 50$, MFIF densities are based on 10 replications instead of 100.}
    \label{fig:AUC_cov}
\end{figure}

\begin{figure}[p]
    \centering
    \includegraphics[width=1\linewidth]{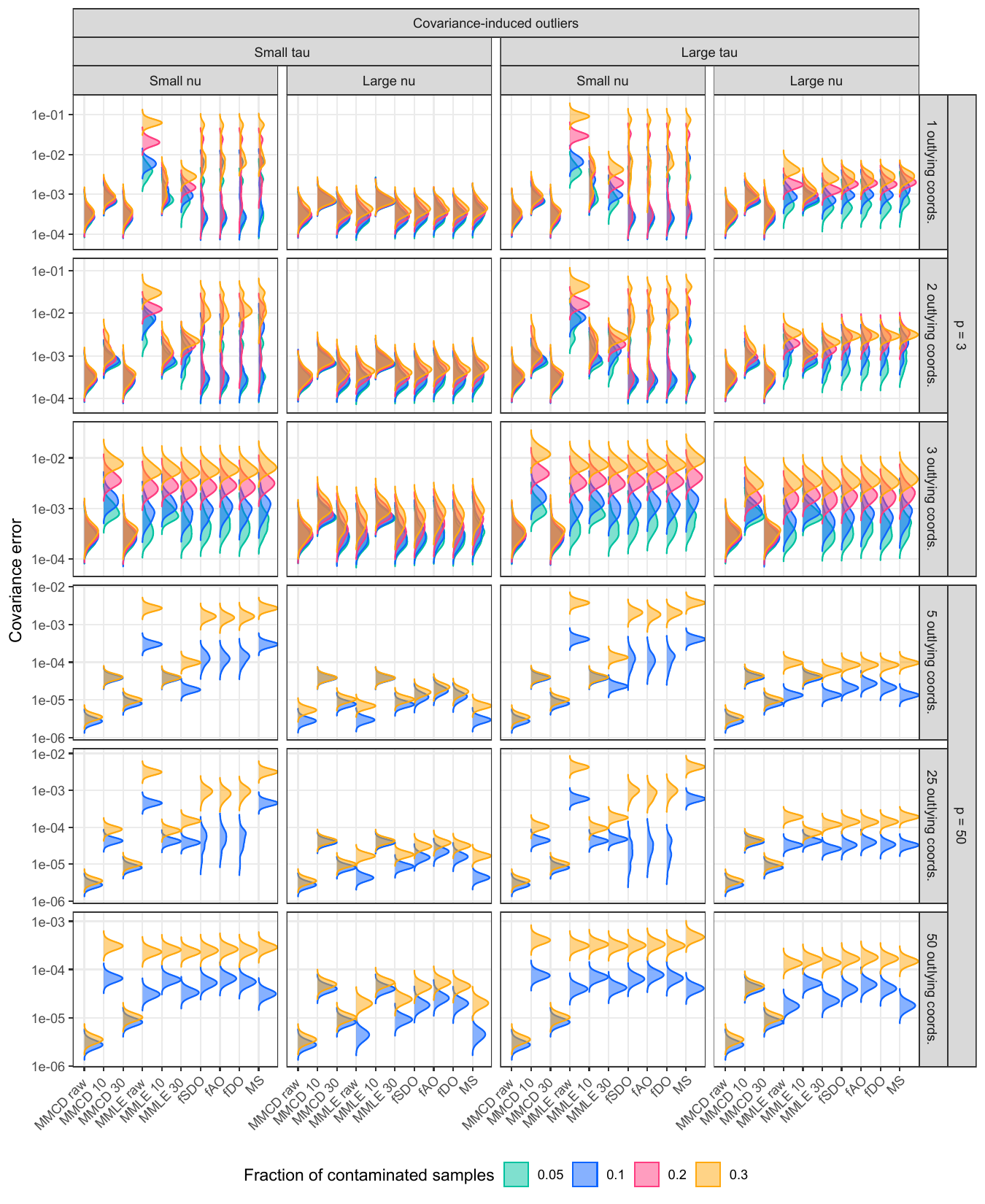}
    \caption{{Density plots of log covariance estimation error for shift outliers in the Gaussian setting with $n = 1000$ across $p = 3,50$, Matérn-type $\kappa_{\text{Matérn}}$ covariance structures with outliers introduced varying the smoothness parameter $\nu$ and range parameter $\tau$, coordinate contamination levels $\lfloor{\varepsilon_{\mathrm{cord}} \cdot p\rfloor}$ ($\varepsilon_{\mathrm{cord}} = 0.1, 0.5, 1$), and outlier proportions ($\varepsilon = 0.05, 0.1, 0.2, 0.3$).}}
    \label{fig:cov_cov}
\end{figure}

\cleardoublepage

\end{document}